\documentclass[twoside,leqno]{article}

\usepackage[letterpaper]{geometry}

\usepackage{ltexpprt}
\usepackage{hyperref}

\usepackage{libertine}
\usepackage{array}
\usepackage{datetime}
\usepackage{url}
\usepackage{paralist}
\usepackage{multicol}

\usepackage{array}
    \newcolumntype{P}[1]{>{\centering\arraybackslash}p{#1}}
    \newcolumntype{M}[1]{>{\centering\arraybackslash}m{#1}}

\usepackage[font=small,labelfont=bf,justification=centering]{caption}

\usepackage{amsmath}

\newtheorem{definition}{Definition}

\usepackage{cleveref}
\crefname{section}{\S}{\S\S}
\Crefname{section}{\S}{\S\S}
\crefformat{section}{\S#2#1#3}
\Crefname{line}{Line}{line}
\crefname{line}{Line}{line}
\Crefname{assumption}{Assumption}{assumption}
\Crefname{@theorem}{Theorem}{theorem}
\crefname{@theorem}{Theorem}{theorem}

\crefname{table}{Table}{Tables}

\usepackage{graphicx}
\usepackage{subcaption}
\usepackage{arydshln}

\usepackage{multirow}

\usepackage{threeparttable}

\usepackage{algorithmicx}
\usepackage{algorithm}
\usepackage{algpseudocode}
\usepackage{algcompatible}
\newfloat{module}{htbp}{loa}
\floatname{module}{Module}

\usepackage[normalem]{ulem}

\usepackage{xspace}

\newcommand{\name}{\textsc{Oper}\xspace}

\newcommand{\block}{\textsc{Crux}\xspace}
\newcommand{\longreb}{\textsc{LongReb3}\xspace}
\newcommand{\longgct}{\textsc{LongGC3}\xspace}
\newcommand{\longgcf}{\textsc{LongGC5}\xspace}
\newcommand{\shortvbt}{\textsc{ShortVB3}\xspace}
\newcommand{\longvbt}{\textsc{LongVB3}\xspace}
\newcommand{\longvbf}{\textsc{LongVB5}\xspace}
\newcommand{\shortfin}{\textsc{ShortFin}\xspace}
\newcommand{\longfin}{\textsc{LongFin}\xspace}
\newcommand{\reduceacool}{\textsc{RedACOOL}\xspace}

\newcommand{\vmax}{V_{\mathit{max}}\xspace}
\newcommand{\vfinal}{V_{\mathit{final}}\xspace}

\newcommand{\MRoot}[1]{\mathsf{MR}(#1)}

\newcommand{\rebnotbottom}[1]{$#1 \in \mathsf{Value}\,(#1\ne\bot_{\mathit{reb}})$}

\usepackage{xcolor}
\usepackage{framed}
\definecolor{lightgray}{gray}{0.90}
\renewenvironment{leftbar}[1][\hsize]
{%
\MakeFramed{\hsize#1\advance\hsize-\width\FrameRestore}%
}
{\endMakeFramed}

\algsetblockdefx[UponSend]{UponSend}{EndUponSend}{}{}[1]{\textbf{upon} \textsc{send}($#1$)}{\phantom}
\algtext*{EndUponSend}

\algsetblockdefx[UponDeliver]{UponDeliver}{EndUponDeliver}{}{}[1]{\textbf{upon} \textsc{SecureDeliver}($#1$)}{\phantom}
\algtext*{EndUponDeliver}

\algsetblockdefx[UponInit]{UponInit}{EndUponInit}{}{}[0]{\textbf{upon} \textsc{init}}{\phantom}
\algtext*{EndUponInit}

\algsetblockdefx[UponExists]{UponExists}{EndUponExists}{}{}[2]{\textbf{upon exists} $#1$ such that $#2$}{\phantom}
\algtext*{EndUponExists}

\algsetblockdefx[Procedure]{Procedure}{EndProcedure}{}{}[2]{\textbf{procedure} {#1}($#2$)}{\phantom}
\algtext*{EndProcedure}
\algsetblockdefx[Function]{Function}{EndFunction}{}{}[2]{\textbf{function} {#1}($#2$)}{\phantom}
\algtext*{EndFunction}

\algnewcommand{\BlueComment}[1]{\textcolor{blue}{\hfill\(\triangleright\) #1}}
\algnewcommand{\LineComment}[1]{\State \(\triangleright\) #1}

\usepackage{listings}
\crefname{lstlisting}{listing}{listings}
\Crefname{lstlisting}{Listing}{Listings}

\crefname{code}{line}{lines}
\Crefname{code}{Line}{Lines}

\usepackage{xcolor}

\definecolor{mygreen}{rgb}{0.254,0.572,0.294}
\definecolor{mygray}{rgb}{0.5,0.5,0.5}
\definecolor{myorange}{rgb}{1,0.35,0}
\definecolor{mymauve}{rgb}{0.58,0,0.82}
\definecolor{myblue}{rgb}{0.2,0.4,0.6}

\definecolor{rakos4orange}{RGB}{255,165,0}
\definecolor{rakos4blue}{RGB}{14,48,173}
\definecolor{rakos4lblue}{RGB}{92,172,238}
\definecolor{rakos4dgray}{RGB}{77,77,77}
\definecolor{plainred}{RGB}{211,63,63}
\definecolor{plainorange}{RGB}{221,105,41}

\usepackage{bold-extra} 

\lstdefinelanguage{Golang}%
  {morekeywords=[1]{package,import,struct,defer,panic,%
     recover,select,var,const,iota,},%
   morekeywords=[2]{string,uint,uint8,uint16,uint32,uint64,int,int8,int16,%
     int32,int64,bool,float32,float64,complex64,complex128,byte,rune,uintptr,%
     error,interface,message,node},%
   morekeywords=[3]{map,slice,make,new,nil,len,cap,copy,close,true,false,%
     delete,append,real,imag,complex,chan,},%
   morekeywords=[4]{break,continue,goto,switch,case,fallthrough,%
    default,},%
   morekeywords=[5]{Println,Printf,Error,Send},%
   sensitive=true,%
   morecomment=[l]{//},%
   morecomment=[s]{/*}{*/},%
   morestring=[b]",%
   morestring=[s]{`}{`},%
   }

\usepackage{thmtools,thm-restate}

\usepackage{algorithmicx}
\usepackage{algorithm} 
\usepackage{algpseudocode}
\usepackage{xifthen}
\usepackage{graphicx}
\usepackage{wrapfig}
\usepackage{dblfloatfix}

\usepackage{stmaryrd}

\newcommand{\ayaz}[1]{{\color{cyan}{[Ayaz: #1}]}}

\newcommand{\remove}[1]{}
\newcommand{\changed}[1]{#1}

\usepackage{cancel}
\usepackage[normalem]{ulem}
\usepackage{soul}
\newif\ifcomments
\commentstrue

\newcommand{\ms}[1]{%
	    \relax\ifmmode
	        \mathord{\mathcode`\-="702D\it #1\mathcode`\-="2200}%
	    \else
	        {\it #1}%
	    \fi
}
\newcommand{\tup}[1]{%
	    \relax\ifmmode
	      \langle #1 \rangle%
	    \else
	        $\langle$ #1 $\rangle$%
	    \fi
}

\date{}

\begin{document}

\algsetblockdefx[LocalState]{LocalState}{EndLocalState}{}{}{\textbf{Local state $\alpha_p$:}}{}
\algsetblockdefx[Round]{Round}{EndRound}{}{}[1]{\textbf{Round:} $#1$}{}
\algsetblockdefx[SendStep]{SendStep}{EndSendStep}{}{}{$S_p^r$:}{}
\algsetblockdefx[TransitionStep]{TransitionStep}{EndTransitionStep}{}{}{$T_p^r$:}{}
\algsetblockdefx[Upon]{Upon}{EndUpon}{}{}[2]{\textbf{upon} $\mathit{deliver}(#1)$ from $#2$ \textbf{do}}{}
\algsetblockdefx[UponEvent]{UponEvent}{EndUponEvent}{}{}[1]{\textbf{upon event} $#1$ \textbf{do}}{}
\algsetblockdefx[UponReceipt]{UponReceipt}{EndUponReceipt}{}{}[2]{\textbf{upon receipt of} $#1$ from $#2$ \textbf{do}}{}









\title{Partial Synchrony for Free: New Upper Bounds for Byzantine Agreement}

\author{
Pierre Civit\footnote{École Polytechnique Fédérale de Lausanne (EPFL), Lausanne, Switzerland, \texttt{pierre.civit@epfl.ch}}
\and Muhammad Ayaz Dzulfikar\footnote{NUS Singapore, Singapore, Singapore, \texttt{ayaz.dzulfikar@u.nus.edu}}
\and Seth Gilbert\footnote{NUS Singapore, Singapore, Singapore, \texttt{gilbert@comp.nus.edu.sg}}
\and Rachid Guerraoui\footnote{École Polytechnique Fédérale de Lausanne (EPFL), Lausanne, Switzerland, \texttt{rachid.guerraoui@epfl.ch}}
\and Jovan Komatovic\footnote{École Polytechnique Fédérale de Lausanne (EPFL), Lausanne, Switzerland, \texttt{jovan.komatovic@epfl.ch}}
\and Manuel Vidigueira\footnote{École Polytechnique Fédérale de Lausanne (EPFL), Lausanne, Switzerland, \texttt{manuel.ribeirovidigueira@epfl.ch}}
\and Igor Zablotchi\footnote{Mysten Labs, Zurich, Switzerland, \texttt{igor@mystenlabs.com}}
}

\maketitle







Byzantine agreement allows $n$ processes to decide on a common value, in spite of arbitrary failures.
The seminal Dolev-Reischuk bound states that any deterministic solution to Byzantine agreement exchanges $\Omega(n^2)$ bits.
In \emph{synchronous} networks, with a known upper bound on message delays, solutions with optimal $O(n^2)$ bit complexity, optimal fault tolerance, and no cryptography have been established for over three decades.
However, these solutions lack robustness under adverse network conditions.
Therefore, research has increasingly focused on Byzantine agreement for \emph{partially synchronous} networks, which behave synchronously only eventually and are thus more reflective of real-world conditions.
Numerous solutions have been proposed for the partially synchronous setting.
However, these solutions are notoriously hard to prove correct, and the most efficient cryptography-free algorithms still require $O(n^3)$ exchanged bits in the worst case.
Even with cryptography, the state-of-the-art remains a $\kappa$-bit factor away from the $\Omega(n^2)$ lower bound (where $\kappa$ is the security parameter).
This discrepancy between synchronous and partially synchronous solutions has remained unresolved for decades.

In this paper, we tackle the discrepancy above by introducing \name, the first generic transformation of deterministic Byzantine agreement algorithms from synchrony to partial synchrony.
\name requires no cryptography, is optimally resilient ($n \geq 3t+1$, where $t$ is the maximum number of failures), and preserves the worst-case per-process bit complexity of the transformed synchronous algorithm.
Leveraging \name, we present the first partially synchronous Byzantine agreement algorithm that (1) achieves optimal $O(n^2)$ bit complexity, (2) requires no cryptography, and (3) is optimally resilient ($n \geq 3t+1$), thus showing that the Dolev-Reischuk bound is tight even in partial synchrony.
Moreover, we adapt \name for long values and obtain several new partially synchronous algorithms with improved complexity and weaker (or completely absent) cryptographic assumptions.
Finally, we demonstrate the broad applicability of the \name transformation by showcasing its use for randomized synchronous agreement algorithms.
Indirectly, \name contradicts the folklore belief that there is a fundamental gap between synchronous and partially synchronous agreement protocols. In a way, we show that there is no inherent trade-off between the robustness of partially synchronous algorithms on the one hand, and the simplicity/efficiency of synchronous ones on the other hand. 


\newpage
\setcounter{tocdepth}{1}
\tableofcontents

\clearpage\newpage\section{Introduction} \label{section:introduction}

Byzantine agreement~\cite{lamport2019byzantine} is a fundamental problem in distributed computing.
The emergence of blockchain systems~\cite{abraham2016solida,buchman2016tendermint,correia2019byzantine,CGL18,chen2016algorand,luu2015scp} and the widespread use of State Machine Replication (SMR)~\cite{abd2005fault,adya2002farsite,amir2006scaling,CL02,kotla2007zyzzyva,kotla2004high,malkhi2019flexible,momose2021multi,Momose2021,veronese2011efficient}, in which Byzantine agreement plays a vital role, has vastly increased the demand for efficient and robust solutions.
Byzantine agreement operates among $n$ processes: each process proposes its value, and all processes eventually agree on a common valid decision.
A process is either \emph{correct} or \emph{faulty}: correct processes follow the prescribed protocol, whereas faulty processes are controlled by the adversary, and can behave arbitrarily.
Byzantine agreement satisfies the following properties:
\begin{compactitem}     
    \item \emph{Agreement:} No two correct processes decide different values.

    \item \emph{Termination:} All correct processes eventually decide.

    \item \emph{Strong validity:} If all correct processes propose the same value $v$, then no correct process decides a value $v' \neq v$.

    \item \emph{External validity:} If a correct process decides a value $v$, then $\mathsf{valid}(v) = \mathit{true}$.
\end{compactitem}
Here, $\mathsf{valid}(\cdot)$ is any predefined logical predicate that indicates whether or not a value is valid.
The Byzantine agreement problem can be characterized by different types of validity properties~\cite{aboutvalidity,civit2023validity}.
In this work, for the sake of generality, we take into account (the conjunction of) two of the most commonly used validity properties: \emph{strong validity}~\cite{abraham2017brief,civit2022byzantine,CGL18,KMM03} and \emph{external validity}~\cite{BKM19,Cachin2001,yin2019hotstuff}.\footnote{\changed{It may seem as though external validity on its own makes the problem trivial: processes could agree on some predetermined valid value. However, the problem can be made non-trivial even with external validity alone, e.g., by assuming that the domain of valid values is not known to honest parties and cannot be computed efficiently from the $\mathsf{valid}(\cdot)$ predicate~\cite{spiegelman21search}.}}


\smallskip
\noindent \textbf{Synchronous Byzantine agreement.}
Byzantine agreement has been extensively studied in the synchronous network model~\cite{abraham2017brief,berman1992bit,CoanW92,Delporte-Gallet22,dolev1985bounds,dolev1983authenticated,kowalski2013synchronous,madsen2020subject,raynal2002consensus,schmid2004synchronous}.
According to this model, algorithms are provided with a strong, ``round-based'' notion of time: all processes start simultaneously, send messages at the beginning of a round, and receive all messages sent to them by the end of the round.
In essence, all processes are perfectly aligned and share the same global clock.
This model has several key advantages.
First, it is fairly easy to reason about synchronous algorithms as their executions are defined around well-delineated rounds.
Second, the synchronous environment provides a strong guarantee that each correct process receives all messages sent by other correct processes within the same round.
Lastly, crashes can be detected perfectly with synchrony~\cite{C96weakest,CT96,lynch1996distributed}.
For example, if process $A$ expects a message from process $B$ in a certain round and does not receive it, $A$ can safely deduce that $B$ is faulty. 
Synchronous Byzantine agreement can be solved with $n \geq 2t+1$ if we assume a cryptographic setup, but requires $n \geq 3t+1$ in the setup-free case we primarily consider in this paper~\cite{FLM85, pfitzmann1996information}.

A significant body of work has been produced on the cost of solving Byzantine agreement in synchrony.
The seminal Dolev-Reischuk bound~\cite{dolev1985bounds} proves that any deterministic synchronous Byzantine agreement solution exchanges $\Omega(n^2)$ bits in the worst case.
It has also been shown that any synchronous solution incurs $\Omega(n)$ worst-case latency~\cite{dolev1983authenticated}.
Notably, these two lower bounds have been proven tight over three decades ago: both~\cite{berman1992bit} and~\cite{CoanW92} have presented Byzantine agreement algorithms for constant-sized values with $O(n^2)$ exchanged bits and $O(n)$ latency.
For long $L$-bit ($L \in \Omega(\log n)$) 
values, where the Dolev-Reischuk bound translates to $\Omega(nL+n^2)$, the work of~\cite{chen2021optimal} and~\cite{errorFreeValidated} introduces optimal and near-optimal solutions for strong validity and external validity, respectively. 
All the aforementioned algorithms are \emph{error-free} in the sense that they are (1) secure against a \emph{computationally unbounded} adversary (no cryptography is employed), and (2) correct in \emph{all} executions (no incorrect execution exists even with a negligible probability).
In summary, synchronous Byzantine agreement algorithms offer two primary benefits:
\begin{compactenum}
    \item They are \emph{conceptually simple}.  The ``round-based'' nature of the synchronous model has yielded algorithms that are easy to understand and prove correct.

    \item They are \emph{efficient}. Many powerful solutions have been discovered, culminating in deterministic error-free algorithms with \emph{optimal} $O(n^2)$ exchanged bits and $O(n)$ latency in the worst case.
\end{compactenum}

\smallskip
\noindent \textbf{Partially synchronous Byzantine agreement.}
The main drawback of synchronous algorithms is their fragility. They are not robust to adverse network conditions and thus have limited applicability in practice.
Many real-world applications are built over the Internet (or some other unreliable network), and inevitably suffer from ``periods of asynchrony'', during which correct processes are disconnected.\footnote{It can be tempting to implement synchronous rounds 
by using big timeouts, but this induces slow reactions to crashes.}
On the other hand, while (fully) asynchronous (randomized) Byzantine agreement algorithms could present a robust alternative, they struggle to achieve the same performance, especially without significant cryptography.
(In the asynchronous, full information model against an adaptive adversary, the best error-free optimally resilient Byzantine agreement algorithm has $\tilde{O}(n^{12})$ expected latency~\cite{huang2023byzantine}.)
To cope with sporadic periods of asynchrony, the \emph{partially synchronous} network model was introduced~\cite{DLS88}.
According to this model, the network behaves asynchronously (i.e., with no bound on message delays) up until an unknown point in time GST (Global Stabilization Time) after which it behaves synchronously. \changed{Partially synchronous Byzantine agreement requires $n \geq 3t+1$, even with a cryptographic setup~\cite{DLS88}.}

Partially synchronous Byzantine agreement algorithms have been the subject of intense research ~\cite{bravo2022making,CL02,civit2022byzantine,civit2023every,gueta2019sbft,lewis2022quadratic,naor2019cogsworth,yin2019hotstuff} and a go-to choice in practice.
Notably, these algorithms are much more \emph{error-prone} and \emph{difficult to design} than their synchronous counterparts, owing to their network model; in partial synchrony, there are no clear rounds and perfect failure detection is impossible.
Moreover, optimal partially synchronous solutions are still unknown.
Recently, two near-optimal solutions~\cite{civit2022byzantine,lewis2022quadratic} were presented, achieving $O(n^2\kappa)$ bit complexity (where $\kappa$ is a security parameter).\footnote{In practice, $\kappa \approx 256$ (the size of a hash).}
These algorithms are, however, not error-free as they rely on cryptography (such as threshold signatures~\cite{shoup2000practical}).
The most efficient known error-free solutions~\cite{CGL18,stern2020information,yu2024tetrabft} achieve $O(n^3)$ bit complexity, which presents a linear factor gap to both the lower bound~\cite{dolev1985bounds} and the complexity attainable in synchrony~\cite{berman1992bit,CoanW92}.
Historically, the fundamental differences between the synchronous and the partially synchronous network models have cultivated the belief that, not only is a complexity gap inevitable, but most synchronous Byzantine agreement algorithms have little to no use in partial synchrony. 
For example, the beautiful recursive approach of~\cite{berman1992bit, CoanW92} appears to be unusable without synchrony.
Clearly, synchronous agreement algorithms are unreliable in partial synchrony.
However, does this mean they are useless? Is the synchrony/partial synchrony gap fundamental?

\smallskip
\noindent \textbf{Contributions.}
In this paper, contrary to popular belief, we show that \emph{any} synchronous Byzantine agreement algorithm can be translated to partial synchrony, by introducing a novel transformation we call \name.\footnote{\name stands for ``Optimistic PERseverance'', which we believe is an adequate short description of our transformation.} 
Not only that, but by applying our transformation to efficient synchronous algorithms, we obtain efficient partially synchronous algorithms. 
Concretely, we prove the following theorem.

\begin{theorem} [Main]
Given any $t$-resilient $(t < n / 3)$ deterministic synchronous Byzantine agreement algorithm $\mathcal{A}^S$ with worst-case per-process bit complexity $\mathcal{B}$ and worst-case latency $\mathcal{L}$, $\name(\mathcal{A}^S)$ is a $t$-resilient deterministic partially synchronous Byzantine agreement algorithm with $O(\mathcal{B})$ worst-case per-process bit complexity and $O(\mathcal{L})$ worst-case latency.
\end{theorem}

By taking the seminal work of~\cite{berman1992bit, CoanW92}, achieving optimal $O(n)$ worst-case per-process bit complexity and optimal $O(n)$ latency in synchrony, \name constructs the first partially synchronous worst-case bit-optimal Byzantine agreement algorithm, which is additionally (1) worst-case latency-optimal, (2) error-free, and (3) optimally resilient ($t < n/3$). 
The emergence of this algorithm closes a long-standing open question on the tightness of the Dolev-Reischuk~\cite{dolev1985bounds} bound on the bit complexity of Byzantine agreement in partial synchrony.
We underline that this algorithm's quadratic complexity is optimal in \emph{every} scenario (not only in the worst case): any signature-free algorithm exchanges $\Omega(n^2)$ messages even in failure-free executions~\cite{DBLP:journals/mst/HadzilacosH93}.
A summary of the state-of-the-art partially synchronous algorithms is given in \Cref{fig:complexities_summary_short}.

\begin{table}[th]
\centering
\footnotesize
\begin{tabular}{ |P{2.7cm}|P{2.1cm}|P{1.5cm}|P{2.1cm}|  }
 \hline
 Protocol & Bit complexity & Resilience & Cryptography \\

 \hline
 \hline
Binary \textsc{DBFT}~\cite{CGL18} & $O(n^3)$ & $t < n/3$ & None \\
IT-HotStuff~\cite{stern2020information} & $O(n^3)$ & $t < n/3$ & None \\
 SQuad~\cite{civit2022byzantine,lewis2022quadratic}                 & $O(n^2\kappa)$ & $t < n/3$ & T. Sig \\
 \textbf{This paper}                 & $O(n^2)$ & $t < n/3$ & None \\
 \hline
 \hline
 Lower bound~\cite{dolev1985bounds} & $\Omega(n^2)$ & $t \in \Omega(n)$ & Any \\
 \hline
\end{tabular}
    \caption{Performance of various Byzantine agreement algorithms with constant-sized inputs and $\kappa$-bit security parameter.
    We consider the binary version of \textsc{DBFT}~\cite{CGL18} for fairness since the multi-valued version, which would be $O(n^4)$, solves a stronger problem (i.e., vector consensus \cite{civit2023validity}). 
    All the algorithms have $O(n)$ worst-case latency.}
\label{fig:complexities_summary_short}
\end{table}

Furthermore, we show that our generic \name transformation can be adapted to accommodate for long values (values with $L \in \Omega(\log n)$ bits), thus yielding several new algorithms.
We summarize these new solutions, as well as the state-of-the-art, in \Cref{fig:complexities_summary_long}.
Algorithm L0 proves that \name can produce cryptography-free Byzantine agreement algorithms with optimal $t < n / 3 $ resilience, albeit with a sub-optimal bit complexity of $O(n^2 L)$.
Algorithms L1 and L3 are obtained by applying the \name transformation to the algorithm of~\cite{errorFreeValidated} (with both strong and external validity).
Algorithms L2 and L4 are obtained from transforming the algorithm of~\cite{chen2021optimal} (with only strong validity).
The biggest improvements over state-of-the-art lie in (1) the complexity, namely the removal of the $\textsf{poly}(k)$ factor of \textsc{DARE-Stark}~\cite{civit2023every} and the $n^{0.5}L$ factor of \textsc{DARE}~\cite{civit2023every}, and (2) the reduced (or absent) cryptographic requirements.

\begin{table}[ht]
\centering
\footnotesize
\begin{tabular}{ |P{2.8cm}|P{1.2cm}|P{3.4cm}|P{1.6cm}|P{2.4cm}|  }
 \hline
 Protocol & Validity & Bit complexity & Resilience & Cryptography \\
\hline
\hline
 IT-HotStuff~\cite{stern2020information} & E & $O(n^3L)$ & $t < n/3$ & None \\
 SQuad~\cite{civit2022byzantine,lewis2022quadratic} & S+E        & $O(n^2L + n^2\kappa)$ & $t < n/3$ & T. Sig \\
 \textsc{DARE}~\cite{civit2023every} & S+E & $O(n^{1.5}L + n^{2.5}\kappa)$ & $t < n/3$ & T. Sig \\
 \textsc{DARE-Stark}~\cite{civit2023every} & S+E & $O(nL + n^2 \textsf{poly}(\kappa))$ & $t < n/3$ & T. Sig + STARK \\
\hline
\hline
\textbf{This paper - L0} & S+E & $O(n^2L)$ & $t < n / 3$ & None \\
\textbf{This paper - L1} & S+E & $O(n\log(n)L + n^2\log(n)\kappa)$ & $t < n/3$ & Hash \\
\textbf{This paper - L2} & S & $O(nL + n^2\log(n)\kappa)$ & $t < n/3$ & Hash\\
 \hline
 \hline
\textbf{This paper - L3} & S+E & $O(n\log(n)L + n^2\log n)$ & $t < n/5$ & None \\
\textbf{This paper - L4} & S & $O(nL + n^2\log n)$ & $t < n/5$ & None \\
\hline\hline
Lower bound~\cite{civit2023validity} & Any & $\Omega(nL + n^2)$ & $t \in \Omega(n)$ & Any \\
 \hline
\end{tabular}
    \caption{Performance of partially synchronous Byzantine agreement algorithms with long ($L$-bit) values and $\kappa$-bit security parameter.
    (S stands for ``strong validity'', and E stands for ``external validity''.) We underline that IT-HotStuff~\cite{stern2020information} is not optimized for long inputs.
    All the algorithms have $O(n)$ worst-case latency.}
\label{fig:complexities_summary_long}
\end{table}

Finally, we emphasize that \name can be applied even to randomized synchronous protocols.
In the randomized world, the power of \name is most evident when applied to synchronous protocols tailored for relaxed models, where, for instance, the adversary is static (rather than adaptive, as in the deterministic world) or unable to observe the messages exchanged between correct processes.
(Without any relaxations, the best known partially synchronous protocols are the deterministic ones outlined above.)
Our new results are summarized in Table~\ref{table:state_of_the_art_protocols} of \Cref{sec:consequences}. 
An overview of these results can be found in Table~\ref{table:state_of_the_art_protocols_intro}. 

\begin{table}[h!]
    \footnotesize
    \centering
    \begin{tabular}{|c|c|c|c|c|c|c|c|}
        \hline
        \textbf{Protocol} & \textbf{Relaxation} & \textbf{Security} & \textbf{Resiliency} & \textbf{Latency} & \textbf{Communication} & \textbf{Problem} & \textbf{Synchrony} \\
        \hline 
        KKKSS \cite{KKKSS08,KKKSS10} & Static adv. & $1/\mathit{log}^c(n)$ & $n/(3 + \epsilon)$ & $O(\log^{27}(n))$ & $\tilde{O}(n^2)$ & LE & A \\
        \hline
        $\boldsymbol{\textbf{\name}}($GPV \cite{GPV06}$)$ & Static adv. & Perfect & $n/(3 + \epsilon)$ & $O(\log(n)/\epsilon^2)$ & $O(Ln^2+n^3)$ & S+E & PS \\
        \hline
        $\boldsymbol{\textbf{\name}}($GPV \cite{GPV06}$)$ & Static adv. & Perfect & $n/\log^{1.58}(n)$ & $O(1)$ & $O(Ln+n^3)$ & S+E & PS \\
        \hline
        $\boldsymbol{\textbf{\name}}($KSSV \cite{KSSV06}$)$ & Static adv. & Perfect & $n/(3 + \epsilon)$ & $\mathit{polylog}(n)$ & $O(Ln^2)$ & S+E & PS \\
        \hline
        \hline
        AAKS \cite{AAKS14,AAKS18} & Crash & Perfect & $n/2$ & $\tilde{O}(n)$ & $\tilde{O}(n^2)$ & binary & A \\
        \hline
        $\boldsymbol{\textbf{\name}}($HKO \cite{hajiaghayi2024nearly}$)$ & Omission & Perfect & $n/30$ & $O(\sqrt{n}\log^2(n))$ & $O(n^2\log^3(n))$ & binary & PS \\
        \hline \hline
        BCP \cite{BCP18} & Private chan. & Perfect & $n/3$ & $O(n)$ & $O(n^6)$ & binary & A \\
        \hline
        $\boldsymbol{\textbf{\name}}($AC \cite{AsharovChandramouli24}$)$ & Private chan. & Perfect & $n/3$ & $O(1)$ & $O(Ln^2+n^3 \log^2(n))$ & S+E & PS \\
\hline
    \end{tabular}
    \caption{
        Overview of state-of-the-art fully asynchronous and partially synchronous protocols in relaxed models with computationally unbounded adversary. 
        (LE stands for ``leader election'', S stands for ``strong validity'', and E stands for ``external validity'', binary refers to the Byzantine agreement problem where only 0 and 1 can be proposed/decided.)
        See Table~\ref{table:state_of_the_art_protocols} in \Cref{sec:consequences} for more details. 
    }
\label{table:state_of_the_art_protocols_intro}
\end{table}

\medskip
\noindent \textbf{Roadmap.}
We discuss related work in \Cref{section:related_work}.
In \Cref{section:technical_overview}, we detail our key idea and provide an intuitive overview of \name.
We define the formal system model and preliminaries in \Cref{section:preliminaries}.
We formally present \name, and its main component \block, in \Cref{section:block}.
In \Cref{sec:consequences}, we discuss how \name can be applied to randomized synchronous agreement algorithms.
Finally, we conclude the paper in \Cref{section:conclusion}.
The optional appendix includes all omitted algorithms and proofs.
\section{Related Work}\label{section:related_work}

This section discusses existing results on Byzantine agreement,
including previous attempts at translating synchronous algorithms to weaker network models and common techniques used to achieve agreement.

\smallskip
\noindent \textbf{Byzantine agreement.}
Byzantine agreement~\cite{LSP82} is the problem of agreeing on a common proposal in a distributed system of $n$ processes despite the presence of $t < n$ arbitrary failures.
Byzantine agreement has many variants~\cite{AbrahamAD04,abraham2017brief,Cachin2001,civit2022byzantine,CGG21,CGL18,fitzi2003efficient,KMM03,lamport1983weak,melnyk2018byzantine,PeaseSL80,siu1998reaching,stolz2016byzantine,yin2019hotstuff} depending on its validity property~\cite{aboutvalidity,civit2023validity}.
In this paper, we focus on (arguably) the two most widely employed validity properties, namely \emph{strong validity} ~\cite{abraham2017brief,civit2022byzantine,CGL18,KMM03} and \emph{external validity}~\cite{BKM19,Cachin2001,yin2019hotstuff}.
Byzantine agreement protocols are primarily concerned with two metrics: \emph{latency} and \emph{communication}.
Latency captures the required number of rounds (or message delays) before all correct processes decide.
Communication concerns the information sent by correct processes and can be measured in multiple ways, such as the total number of sent messages, bits, or words.\footnote{Word complexity is a simplification of bit complexity as it deems all values and cryptographic objects to be of constant bit-size.}
In the worst case, deterministic Byzantine agreement is impossible to solve with fewer than $\Omega(t^2)$ messages ~\cite{civit2023all,civit2023validity,dolev1985bounds}, which also applies to words and bits. 
For $L$-bit proposals and $t \in \Omega(n)$, the (best) bit complexity lower bound is $\Omega(nL + n^2)$~\cite{civit2023all,civit2023validity,dolev1985bounds}.
In partial synchrony (and asynchrony), it has been shown~\cite{rambaud2023adaptively} that no unauthenticated (and thus information-theoretic secure) protocol (even randomized) achieves sub-quadratic expected message complexity.
This holds even with secure channels, common random strings (CRS), and non-interactive zero-knowledge (NIZK).

\smallskip
\noindent \textbf{Byzantine agreement in synchrony.} 
Considering only strong validity, there exist word-optimal~\cite{berman1992bit, CoanW92} ($O(n^2)$) and near bit-optimal~\cite{chen2021optimal} ($O(nL + n^2\log n)$) deterministic error-free solutions.
Recently, 
a deterministic error-free solution that achieves near-optimal $O(nL\log n + n^2\log n)$ bit complexity was presented for external validity~\cite{errorFreeValidated}.

\smallskip
\noindent \textbf{Byzantine agreement in partial synchrony.}
In the authenticated setting (with employed cryptography) deterministic non-error-free word-optimal solutions were proposed~\cite{civit2022byzantine,lewis2022quadratic}.
In terms of bit complexity, a deterministic solution with $O(nL + n^2\mathsf{poly}(\kappa))$ bits was recently achieved~\cite{civit2023every}, albeit employing both threshold signatures~\cite{shoup2000practical} and STARK proofs~\cite{ben2018scalable}, which are computationally heavy and induce the $\mathsf{poly}(\kappa)$ factor.
The best deterministic error-free solution has $O(n^3)$ bit-complexity even for the binary case~\cite{CGL18,stern2020information,yu2024tetrabft}.

\smallskip
\noindent \textbf{Randomized Byzantine agreement in asynchrony.}
In the full information model (without private channels) with an adaptive adversary, fully asynchronous Byzantine agreement presents an immense challenge \cite{huang2022byzantine,huang2023byzantine,kimmett2020improvement,king2016byzantine,king2018correction,melnyk2020byzantine}.
A breakthrough came in 2018 with the introduction of the first polynomial algorithm with linear resilience \cite{king2018correction},  correcting an earlier claim~\cite{king2016byzantine} with a technical flaw \cite{melnyk2020byzantine}. 
Yet, the solution proposed in \cite{king2018correction} achieved a resilience no better than $1.14\cdot 10^{-9} \cdot n$. 
Very recently, the first polynomial algorithm achieving optimal resilience for this model was presented~\cite{huang2023byzantine}, building upon a near-optimally resilient result published by the same authors~\cite{huang2022byzantine}. 
However, this \emph{tour de force} comes at the cost of a discouraging expected latency complexity of $\tilde{O}(n^{12})$.

\smallskip
\noindent \textbf{View synchronization}.
In network models where synchrony is only sporadic, such as partial synchrony~\cite{DLS88}, many algorithms rely on a ``view-based'' paradigm.
Essentially, processes communicate and attempt to enter a ``view'' roughly simultaneously (within some fixed time of each other).
Once in a view, processes act as if in a synchronous environment and try to safely achieve progress, typically by electing a leader who drives it.
If the processes suspect that progress is blocked, e.g., due to faulty behavior or asynchrony, they may try to re-synchronize and enter a different view (with a potentially different leader).
The view synchronization problem is closely related to the concept of \emph{leader election}~\cite{C96weakest,CT96}.
View synchronization has been employed extensively in agreement protocols, both for crash~\cite{lamport2001paxos,oki1988viewstamped,ongaro2015raft} and Byzantine faults~\cite{bravo2022making,CL02,civit2022byzantine,civit2023every,gueta2019sbft,lewis2022quadratic,lewis2023fever,lewis2023lumiere,naor2019cogsworth,yin2019hotstuff}.

\smallskip
\noindent \textbf{Synchronizers}.
Synchronizers~\cite{awerbuch1985complexity,devarajan1993correctness,fekete1987modular,lynch1996distributed,raynal1988networks,schneider1982synchronization} are a technique used to simulate a synchronous network in an asynchronous environment.
The main goal is to design efficient distributed algorithms in asynchronous networks by employing their synchronous counterparts.
Examples of successful applications include breadth-first search, maximum flow, and cluster decompositions~\cite{awerbuch1985complexity,awerbuch1985reducing,awerbuch1993near,awerbuch1990sparse,awerbuch1992routing}.
The main limitation of synchronizers is that they work only in the absence of failures~\cite{lynch1996distributed}, or by enriching the model with strong notions of failure detection~\cite{C96weakest,CT96,lynch1996distributed}, such as a \emph{perfect failure detector}, as done in \cite{schneider1982synchronization} for processes that can crash and subsequently recover.
Unfortunately, perfect failure detectors cannot be implemented in asynchronous or partially synchronous networks even for crash faults without further assumptions~\cite{C96weakest,fischer1985impossibility}.
Thus, no general transformation (i.e., for any problem) from synchrony into partial synchrony exists in the presence of failures.
In ~\cite{alistarh2011generating}, the authors introduce an \emph{asynchrony detector} that works on some classes of distributed tasks with crash failures, including agreement, and can be used to transform synchronous algorithms into partially synchronous ones that perform better in optimistic network conditions.
The proposed technique however does not provide any improvement in less-than-ideal network conditions (some asynchrony, or in the worst case) and does not extend to Byzantine failures.

\smallskip
\noindent \textbf{Network agnostic Byzantine agreement and MPC}.
 Network agnostic Byzantine agreement and Multi-party Computation (MPC) have been addressed in various works \cite{ABKL22,ACC22, ACC23,AC23,BCLL23,BKL19,BKL20SMR, BZL20,DE24,LM18,spiegelman21search}. These protocols can tolerate $t_a$ Byzantine failures in an asynchronous network and $t_s$ Byzantine failures in a synchronous network, provided $t_a + 2t_s < n$~\cite{BKL19}. 
 To achieve optimal $n/3$ resiliency in asynchronous environments, $t_s < n/3$ is necessary. Conversely, if the goal is to have $t_s>n/3$, it becomes inevitable to limit the adversary's simulation capabilities due to the FLM impossibility result \cite{FLM85}. Consequently, a majority of these constructions rely on a public key infrastructure (PKI) and a bounded adversary \cite{ABKL22,BCLL23,BKL19,BKL20SMR,BZL20,DE24,LM18,spiegelman21search}.
The work of \cite{abraham2008almost,BCP20} focuses on network agnostic information-theoretic MPC. 
Given that they tackle a fully asynchronous network and an unbounded adversary, they must contend with the complexity of an information-theoretic (shunning) common coin, 
assuming secret channels, to circumvent the FLP impossibility result \cite{fischer1985impossibility}. 
These challenges hinder them from exploring techniques that could potentially be beneficial in partially synchronous settings.

\section{\name: Overview}\label{section:technical_overview}

This section provides the overview of the \name transformation.
First, we introduce the key idea behind \name (\Cref{subsection:name_key_idea}).
Then, we intuitively explain how we realize this idea (\Cref{subsection:key_idea_in_action}).

\subsection{Key Idea} \label{subsection:name_key_idea}

The key idea underlying our \name transformation is to sequentially repeat a synchronous Byzantine agreement algorithm $\mathcal{A}^{S}$ in a series of \emph{views}, one instance of $\mathcal{A}^{S}$ per view, until one succeeds. 
Clearly, when the synchronous algorithm $\mathcal{A}^{S}$ is run in partial synchrony, \emph{a priori} nothing is guaranteed due to the asynchronous period before GST (the network stabilization point).
Before GST, the output of $\mathcal{A}^{S}$ might be unreliable (no agreement nor validity) if there is an output at all (no termination).
However, if $\mathcal{A}^{S}$ is started after GST by all correct processes \emph{nearly simultaneously} (with only a constant delay between processes), the conditions become sufficiently similar to synchrony that $\mathcal{A}^{S}$ can be simulated, thus allowing processes to decide.
In essence, to efficiently translate $\mathcal{A}^S$ from synchrony to partial synchrony, our \name transformation needs to tackle the following challenges:
\begin{compactitem}
    \item \emph{Challenge 1:} Ensuring agreement among correct processes within and across views, i.e., guaranteeing that correct processes do not decide different values despite the unreliability of $\mathcal{A}^{S}$ before GST.

    \item \emph{Challenge 2:} Ensuring a successful simulation of $\mathcal{A}^S$ after GST, thus enabling processes to decide.
    
    \item \emph{Challenge 3:} Preserving the per-process bit complexity and latency of $\mathcal{A}^S$.
\end{compactitem}


\smallskip
\noindent\textbf{Challenge 1: ensuring agreement within and across views.}
Running and deciding from $\mathcal{A}^S$ directly would be risky as $\mathcal{A}^S$ provides no security guarantees if run before GST.
Instead, 
we run $\mathcal{A}^S$ sequentially in between two protocols that act as ``safety guards''.
The job of the first safety guard is to effectively ``disable'' $\mathcal{A}^S$ when appropriate, forcing processes to ignore its (potentially harming) output.
For example, if all correct processes start a view (even before GST) already in agreement, that view's $\mathcal{A}^S$ instance will be disabled by the first safety guard.
The job of the second safety guard is to trigger a decision if it detects agreement after running $\mathcal{A}^S$.
For instance, if all processes obtain the same value after running (or ignoring) $\mathcal{A}^S$, then running the second safety guard will allow correct processes to decide.
Crucially, the two safety guards work in \emph{tandem}: if the second safety guard triggers a decision in some view $V$ then, in all future views $V' > V$, the first safety guard disables $\mathcal{A}^S$, thus preventing any potential disagreement caused by an unreliable output of $\mathcal{A}^S$.
This collaboration between the safety guards is essential for ensuring agreement in our \name transformation.

\smallskip
\noindent\textbf{Challenge 2: ensuring a successful simulation of $\mathcal{A}^S$ after GST.}
To successfully simulate $\mathcal{A}^S$ after GST, we guarantee conditions that are analogous to synchrony.
To this end, we employ a \emph{view synchronization} mechanism~\cite{civit2022byzantine,lewis2022quadratic,lewis2023fever,lewis2023lumiere} to ensure that all correct processes start a view (and its $\mathcal{A}^S$ instance) nearly simultaneously, i.e., within a constant delay $\Delta_{\mathit{shift}}$ of each other.

However, this is not enough to guarantee successful simulation of $\mathcal{A}^S$ as, in synchrony, correct processes start executing $\mathcal{A}^S$ at \emph{exactly} the same time (i.e., without any misalignment).
To tackle the initial distortion, we expand the duration of each round of $\mathcal{A}^S$: specifically, each correct process $p_i$ executes a round of $\mathcal{A}^S$ for exactly $\Delta_{\mathit{shift}} + \delta$ time, where $\delta$ denotes the upper bound on message delays after GST (in partial synchrony).
The $\Delta_{\mathit{shift}} + \delta$ round duration ensures that, after GST, $p_i$ receives each round's messages from \emph{all} correct processes as (1) all correct processes start executing $\mathcal{A}^S$ at most $\Delta_{\mathit{shift}}$ time after $p_i$ (ensured by the view synchronization mechanism), and (2) the message delays are bounded by $\delta$.
(A similar simulation technique is proposed in~\cite{LS22} for error-free synchronous algorithms; to accommodate for cryptography-based algorithms, we provide a general simulation technique, similar to that of~\cite{DBLP:conf/podc/LindellLR02}, in \Cref{section:crux_correctness_complexity_formal}.)



\smallskip
\noindent\textbf{Challenge 3: preserving the complexity of $\mathcal{A}^S$.}
To minimize communication and preserve the per-process bit complexity and latency of $\mathcal{A}^S$ after GST, we bound both (1)  the complexity of each view, and (2) the number of views executed after GST.
To accomplish the first task, we limit the number of bits sent within any view, before or after GST, as follows.
(1) We set the maximum number of bits any process can send when simulating $\mathcal{A}^S$ to the worst-case per-process maximum in synchrony $\mathcal{B}$.
This prevents any correct process from inadvertently exploding its complexity beyond $\mathcal{B}$, even in the presence of asynchrony.
(2) We implement all our view schemes efficiently (i.e., with $O(\mathcal{B})$ per-process bit cost and constant latency), including the safety guards and the view synchronization protocol.


For the second task, we rely on our view synchronization mechanism, which (indirectly) guarantees that by the end of the first view started after GST, all processes will have decided (as that view's $\mathcal{A}^S$ instance will be correctly simulated).
Nonetheless, during periods of asynchrony (i.e., before GST), slow correct processes can fall arbitrarily many views behind fast correct processes. 
After GST, slow processes can catch up by advancing through \emph{all} stale views at a communication cost proportional to the number of stale views, which is costly.
We thus introduce a mechanism to allow processes to catch up by \emph{skipping} any number of stale views while preserving agreement.
At the end of each view (before moving on to the next view), processes run a ``safe skip'' protocol that provides all correct processes (even late ones!) with a safe value to adopt.
This value is guaranteed to be in line with any previous decision, i.e., it preserves agreement.
By adopting this value, late processes can immediately skip all stale views and synchronize with fast processes (by view synchronization), while preserving the safety of \name.
In explicit terms, if we denote the number of views traversed until GST by $V_{\text{GST}}$, the ``safe skip'' protocol ensures all correct processes terminate in $O(1)$ views after GST (by safely skipping views) instead of $O(V_{\text{GST}})$ views (by traversing each stale view).

\smallskip
\noindent \textbf{A note on view numbers.}
We remark that, like in all ``view-based'' work in partial synchrony (that we are aware of), our protocols include the view number in message headers to distinguish messages from different views.
This means that there is an implicit additive term of $\log(V_{\text{GST}})$ bits per message.
Previous work typically considers ``word complexity'' (instead of bit complexity) which assumes that all counters are bounded by, e.g., 64 bits, and implies that $\log(V_{\text{GST}}) \in O(1)$.
For comparative clarity (and simplicity), we also omit the $\log(V_{\text{GST}})$ term in our results throughout this paper.
However, we do not use word complexity, since that would hide much larger (and thus significant) factors such as $L$ and $\kappa$.
Nevertheless, even if we were to consider word complexity, our work still presents the first error-free protocol with $O(n^2)$ total word complexity in the worst case (and not just per view, as in~\cite{stern2020information}).



\subsection{Implementation of the Key Idea} \label{subsection:key_idea_in_action}

To show how \name implements its key idea, we start by introducing a spider-graph-based interpretation of each of \name's views (\Cref{subsubsection:spider_graph}).
Then, we present the structure of each view in \name and use the aforementioned interpretation to show how \name satisfies agreement and termination (\Cref{sec:phases}).

\subsubsection{Interpretation of \name's Views via Spider Graphs} \label{subsubsection:spider_graph}
Let us represent the internal states of correct processes throughout a view of \name using \emph{spider graphs}~\cite{koebe1992new}.
(This interpretation is inspired by~\cite{AttiyaWelch23}.)
A spider graph is a graph with a central clique with $|\mathsf{Value}|$ branches, where each branch is associated with a particular value $v \in \mathsf{Value}$.
(Let $\mathsf{Value}$ denote the set of all values.)
See \Cref{fig:spider-graph} for an example.
Within a branch, the distance from the clique ($c = 0,1,2$) indicates the level of ``confidence'' in the corresponding value.
Roughly, for a given value $v$, $c=0$ implies that at least one process holds $v$, $c=1$ implies that all (correct) processes hold $v$ (the system is ``convergent'' on $v$), and $c=2$ implies that the system is $v$\emph{-valent}~\cite{fischer1985impossibility} (the only possible decidable value will forever be $v$).
Consequently, each process starts any view of \name with confidence $c = 0$ in its proposal, and, if it reaches some position $(v, c = 2)$ within the view, it decides $v$.





\begin{figure}[h]
\centering
\begin{subfigure}{0.28\textwidth}
    \centering
    \includegraphics[width=0.9\textwidth]{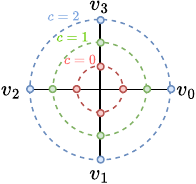}
    \caption{Positions in a spider graph.}
    \label{fig:spider-graph}
\end{subfigure}
\hfill
\begin{subfigure}{0.28\textwidth}
    \centering
    \includegraphics[width=0.9\textwidth]{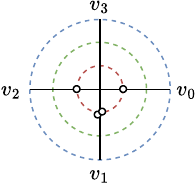}
    \caption{A divergent configuration.}
    \label{fig:divergent_config}
\end{subfigure}
\hfill
\begin{subfigure}{0.28\textwidth}
    \centering
    \includegraphics[width=0.9\textwidth]{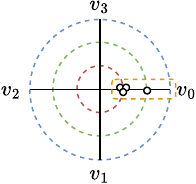}
    \caption{A convergent configuration.}
    \label{fig:convergent_config}
\end{subfigure}

\caption{Interpration of \name's views using spider graphs for $|\mathsf{Value}| = 4$. 
    In \Cref{fig:spider-graph}, all possible positions are represented with circles. 
    \Cref{fig:divergent_config,fig:convergent_config} illustrate examples of divergent and convergent configurations, respectively; each process is represented with a circle.
    }
\label{fig:graph_and_configurations}
\end{figure}

A set of positions of correct processes is called a \emph{configuration}.
We distinguish two types of configurations:
\begin{compactitem}
    \item \emph{Convergent:} all processes are on the same branch. 
    \item \emph{Divergent:} not all processes are on the same branch, and the maximum confidence is $c=0$.
\end{compactitem}
Note that the separation above is not exhaustive.
Concretely, it excludes {inconsistent} configurations where some process has high confidence in a value ($v, c \geq 1$) and some other process disagrees ($v' \neq v$).
Importantly, our \name transformation avoids such inconsistent configurations: if there exists a correct process with high confidence $c \geq 1$ in some value $v$, then it is guaranteed that \emph{all} correct processes hold $v$.
In summary, in a convergent configuration, all processes hold the same value, knowingly ($c > 0$) or unknowingly ($c = 0$), and in a divergent configuration, at least two processes disagree, but only with low confidence ($c = 0$).
\Cref{fig:divergent_config,fig:convergent_config} illustrate divergent and convergent configurations, respectively.







\subsubsection{Structure of \name's Views.}
\label{sec:phases}

The structure of each \name's view is illustrated in \Cref{fig:view_example}.
A view has four main components:
\begin{compactenum}
    \item \emph{First safety guard:}
    this component ensures that if the system is convergent at the start of the view, all correct processes ignore the (unreliable-before-GST) $\mathcal{A}^S$ instance.
    
    
    \item \emph{$\mathcal{A}^S$ simulation:}
    as the simulation can successfully be performed after GST, this component ensures that the system converges after GST.
    
    
    \item \emph{Second safety guard:}
    the second safety guard ensures that if the system is convergent before this step, all correct processes decide (and remain convergent).\footnote{Strictly speaking, the second safety guard not only ensures safety but also provides termination (liveness) in the sense that processes decide at this step. However, in the bigger picture, its role in \emph{achieving} liveness is minor: the ``heavy-lifting'' of converging is done by the preceding $\mathcal{A}^S$ simulation.}
    
    
    \item \emph{Safe skip:}
    this component enables lagging processes to skip views and immediately ``jump'' ahead to the next view.
    Importantly, if the system is convergent, the convergence is preserved.\footnote{Note that without this component, we could technically rely on just the first and second safety guards to guarantee safety transitively across views. 
    However, in that case (late) processes would have to traverse every (stale) view sequentially after GST, which is undesirable as previously explained.}
    
\end{compactenum}

\begin{figure}[ht]
\centering
\includegraphics[width=0.8\textwidth]{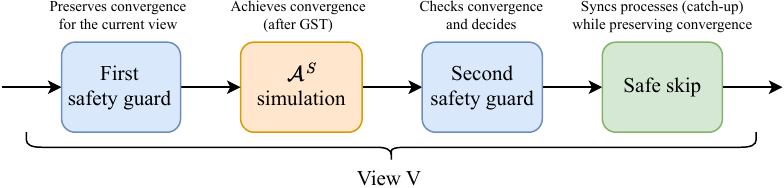}
\caption{Structure of a view in \name.}
\label{fig:view_example}
\end{figure}

Finally, let us informally show why \name satisfies the agreement and termination properties.

\begin{figure}[h]
\centering
\includegraphics[width=0.75\textwidth]{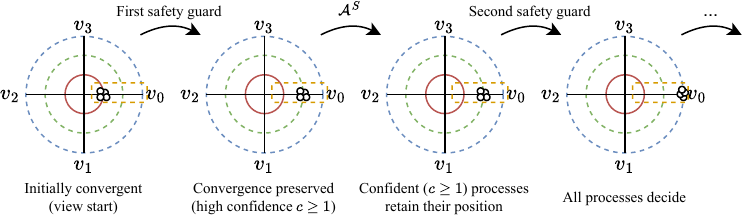}
\caption{Illustration of the preservation of (initial) convergence in a view.}
\label{fig:convergence_preservation_example}
\end{figure}

\smallskip
\noindent \textbf{Satisfying agreement.}
Suppose a correct process $p_i$ decides some value $v_0$ in some view $V$.
Hence, the second safety guard of view $V$ ensures that the system is convergent on $v_0$ after $p_i$'s decision.
Moreover, the safe skip component of view $V$ guarantees that all correct processes start view $V + 1$ with $v_0$ (i.e., convergence is preserved).
Therefore, the first safety guard of view $V + 1$ ensures that (1) the convergence is preserved, and (2) the output of the $\mathcal{A}^S$ simulation is ignored (this is crucial, as, before GST, the simulation can break the established convergence).
Therefore, the system remains convergent after the simulation step, which implies that all correct processes decide $v_0$ after executing the second safety guard.
We illustrate this agreement-preserving mechanism in \Cref{fig:convergence_preservation_example}.

\smallskip
\noindent \textbf{Satisfying termination.}
Let $\vfinal$ be the first view started after GST.
Suppose the system is divergent at the start of $\vfinal$.
Hence, the first safety guard might preserve the existing divergence (since the first safety guard only preserves already-existing convergence).
However, as a successful simulation of $\mathcal{A}^S$ can be performed after GST, this simulation step \emph{achieves} convergence.
Therefore, the system is convergent before starting the second safety guard, which then ensures that all correct processes decide.
We illustrate this convergence-achieving concept in \Cref{fig:achieving_convergence_example}.

\begin{figure}[ht]
\centering
\includegraphics[width=0.75\textwidth]{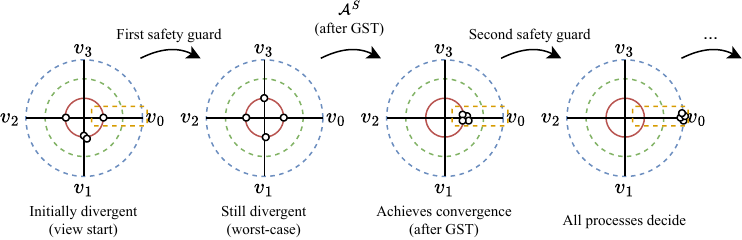}
\caption{Illustration of a view achieving convergence after GST.
}
\label{fig:achieving_convergence_example}
\end{figure}

\section{Preliminaries} \label{section:preliminaries}

\noindent \textbf{Processes.}
We consider a static system $\Pi = \{p_1, ..., p_n\}$ of $n$ processes that communicate by sending messages; each process acts as a deterministic state machine.
Each process has its local clock.
At most $0 < t < n / 3$ processes are Byzantine and controlled by the adaptive adversary capable of corrupting processes throughout the entire execution (and not only at its beginning).
(If $t \geq n / 3$, Byzantine agreement cannot be solved in partial synchrony~\cite{DLS88}.)
A Byzantine process behaves arbitrarily, whereas a non-Byzantine process behaves according to its state machine.
Byzantine processes are said to be \emph{faulty}; non-faulty processes are said to be \emph{correct}. 
The adversary is aware of the internal states of all processes (the full information model) and, unless stated otherwise, is computationally unbounded.
Lastly, we assume that local steps of processes take zero time, as the time needed for local computation is negligible compared to message delays.

\smallskip
\noindent \textbf{Values.} We denote the set of values by $\mathsf{Value}$.
Recall that $\mathsf{valid}: \mathsf{Value} \to \{\mathit{true}, \mathit{false}\}$ indicates whether or not a value is valid.
For the sake of simplicity, unless otherwise stated, we consider only constant-sized values ($L \in O(1)$ bits) throughout the rest of the main body of the paper.
(All results, including for long values, can be found in the appendix.)

\smallskip
\noindent \textbf{Communication network.} We assume a point-to-point communication network. Furthermore, we assume that the communication network is \emph{reliable}: if a correct process sends a message to a correct process, the message is eventually received.
Finally, we assume authenticated channels: the receiver of a message is aware of the sender's identity.

\smallskip
\noindent \textbf{Partial synchrony.}
We consider the standard partially synchronous environment~\cite{DLS88}.
Specifically, there exists an unknown Global Stabilization Time (GST) and a positive duration $\delta$ such that message delays are bounded by $\delta$ after GST: a message sent at time $\tau$ is received by time $\max(\tau, \text{GST}) + \delta$.
We assume that $\delta$ is known.
Moreover, we assume that all correct processes start executing their local algorithm before GST.
Finally, the local clocks of processes may drift arbitrarily before GST, but do not drift thereafter.

\smallskip
\noindent \textbf{Complexity of synchronous Byzantine agreement.}
Let $\mathcal{A}^S$ be any synchronous Byzantine agreement algorithm, and let $\mathit{execs}(\mathcal{A}^S)$ be the set of executions of $\mathcal{A}^S$.
The bit complexity of any correct process $p_i$ in any execution $\mathcal{E} \in \mathit{execs}(\mathcal{A}^S)$ is the number of bits sent by $p_i$ in $\mathcal{E}$.
The per-process bit complexity $\mathit{pbit}(\mathcal{A}^S)$ of $\mathcal{A}^S$ is then defined as
\begin{equation*}
    \mathit{pbit}(\mathcal{A}^S) = \max_{\mathcal{E} \in \mathit{execs}(\mathcal{A}^S), p_i \in \Pi} \bigg\{\text{the bit complexity of } p_i \text{ in } \mathcal{E} \bigg\}.
\end{equation*}
The latency of any execution $\mathcal{E} \in \mathit{execs}(\mathcal{A}^S)$ is the number of synchronous rounds before all correct processes decide in $\mathcal{E}$.
The latency $\mathit{latency}(\mathcal{A}^S)$ of $\mathcal{A}^S$ is then defined as
\begin{equation*}
    \mathit{latency}(\mathcal{A}^S) = \max_{\mathcal{E} \in \mathit{execs}(\mathcal{A}^S)} \bigg\{\text{the latency of } \mathcal{E} \bigg\}.
\end{equation*}

\smallskip
\noindent \textbf{Complexity of partially synchronous Byzantine agreement.}
Let $\mathcal{A}^{PS}$ be any partially synchronous Byzantine agreement algorithm, and let $\mathit{execs}(\mathcal{A}^{PS})$ be the set of executions of $\mathcal{A}^{PS}$.
The bit complexity of any correct process $p_i$ in any execution $\mathcal{E} \in \mathit{execs}(\mathcal{A}^{PS})$ is the number of bits sent by $p_i$ during the time period $[\text{GST}, \infty)$.\footnote{The number of bits any correct process sends before GST is unbounded in the worst case~\cite{spiegelman21search}.}
The per-process bit complexity $\mathit{pbit}(\mathcal{A}^{PS})$ of $\mathcal{A}^{PS}$ is then defined as
\begin{equation*}
    \mathit{pbit}(\mathcal{A}^{PS}) = \max_{\mathcal{E} \in \mathit{execs}(\mathcal{A}^{PS}), p_i \in \Pi} \bigg\{\text{the bit complexity of } p_i \text{ in } \mathcal{E} \bigg\}.
\end{equation*}
The latency of any execution $\mathcal{E} \in \mathit{execs}(\mathcal{A}^{PS})$ is equal to \changed{$\max(\frac{1}{\delta}(\tau^* - \text{GST}), 0)$}, where $\tau^*$ is the first time by which all correct processes decide in $\mathcal{E}$.
The latency $\mathit{latency}(\mathcal{A}^{PS})$ of $\mathcal{A}^{PS}$ is then defined as
\begin{equation*}
    \mathit{latency}(\mathcal{A}^{PS}) = \max_{\mathcal{E} \in \mathit{execs}(\mathcal{A}^{PS})} \bigg\{\text{the latency of } \mathcal{E} \bigg\}.
\end{equation*}

\smallskip
\noindent \textbf{Complexity of asynchronous algorithms.}
In our \name transformation, we utilize two asynchronous algorithms (see~\Cref{section:block}).
Therefore, we define the complexity of asynchronous algorithms as well.
Let $\mathcal{A}^{A}$ be any asynchronous algorithm, and let $\mathit{execs}(\mathcal{A}^{A})$ be the set of executions of $\mathcal{A}^{A}$.
The bit complexity of any correct process $p_i$ in any execution $\mathcal{E} \in \mathit{execs}(\mathcal{A}^{A})$ is the number of bits sent by $p_i$ in $\mathcal{E}$.
The per-process bit complexity $\mathit{pbit}(\mathcal{A}^{A})$ of $\mathcal{A}^{A}$ is then defined as
\begin{equation*}
    \mathit{pbit}(\mathcal{A}^{A}) = \max_{\mathcal{E} \in \mathit{execs}(\mathcal{A}^{A}), p_i \in \Pi} \bigg\{\text{the bit complexity of } p_i \text{ in } \mathcal{E} \bigg\}.
\end{equation*}

For latency, we adopt the standard definition of~\cite{AW04,AttiyaWelch23}.
Formally, a timed execution is an execution in which non-decreasing non-negative integers (``times'') are assigned to the events, with no two events by the same process having the same time.
The first event of each correct process is assigned 0.
For each timed execution, we consider the prefix ending when the last correct process terminates, and then scale the times so that the maximum time that elapses between the sending and receipt of any message between correct processes is $1$.
Therefore, we define the latency $\mathit{latency}(\mathcal{A}^A)$ of $\mathcal{A}^A$ as the maximum time, over all such scaled timed execution prefixes, assigned to the last event.
The latency of an asynchronous algorithm is also known as the number of asynchronous rounds that the algorithm requires~\cite{DBLP:journals/jacm/MostefaouiMR15}.
We use these two terms interchangeably.

\section{\block: The View Logic of \name} \label{section:block}

This section formally introduces \block, a distributed protocol run by processes in every view of the \name transformation.
First, we present \block's formal specification (\Cref{subsection:block_problem_definition}).
Second, we introduce the building blocks of \block: graded consensus and validation broadcast (\Cref{section:building_blocks}).
Third, we present \block's pseudocode and a proof sketch (\Cref{subsection:block_pseudocode}).
Finally, we explain how to obtain \name from \block (\Cref{section:from_block_to_name}).

\subsection{\block's Specification} \label{subsection:block_problem_definition}

Module~\ref{mod:crux} captures  \block's specification.
An instance of \block is parameterized with two time durations: $\Delta_{\mathit{shift}}$ and $\Delta_{\mathit{total}}$.
Moreover, each correct process $p_i$ is associated with its default value $\mathsf{def}(p_i)$. 
In brief, \block guarantees the safety of \name always (even if \block is run before GST), and it ensures the liveness of \name (by guaranteeing synchronicity) after GST, provided that all correct processes enter \block within $\Delta_{\mathit{shift}}$ of each other and do not abandon \block within $\Delta_{\mathit{total}}$ of entering.


\begin{module} [h]
\caption{\block}
\label{mod:crux}
\footnotesize
\begin{algorithmic}[1]

\Statex \textbf{Parameters:}
\Statex \hskip2em $\mathsf{Time\_Duration}$ $\Delta_{\mathit{shift}}$ \BlueComment{common for all processes}
\Statex \hskip2em $\mathsf{Time\_Duration}$ $\Delta_{\mathit{total}}$ \BlueComment{common for all processes}
\Statex \hskip2em $\mathsf{Value}$ $\mathsf{def}(p_i)$ such that $\mathsf{valid}\big( \mathsf{def}(p_i) \big) = \mathit{true}$, for every correct process $p_i$ \BlueComment{each process $p_i$ has its default value}

\medskip
\Statex \textbf{Events:}
\Statex \hskip2em \emph{request} $\mathsf{propose}(v \in \mathsf{Value})$: a process proposes value $v$.
\Statex \hskip2em \emph{request} $\mathsf{abandon}$: a process abandons (i.e., stops participating in) \block.
\Statex \hskip2em \emph{indication} $\mathsf{validate}(v' \in \mathsf{Value})$: a process validates value $v'$.
\Statex \hskip2em \emph{indication} $\mathsf{decide}(v' \in \mathsf{Value})$: a process decides value $v'$.
\Statex \hskip2em \emph{indication} $\mathsf{completed}$: a process is notified that \block has completed.

\medskip 
\Statex \textbf{Notes:} 
\Statex \hskip2em We assume that every correct process proposes at most once and it does so with a valid value. 
We do not assume that all correct processes \hphantom{aallll}propose.
Note that a correct process can validate a value from \block even if (1) it has not previously proposed, or (2) it has previously \hphantom{aallll}abandoned \block, or (3) it has previously received a $\mathsf{completed}$ indication.
Moreover, a correct process can receive both a $\mathsf{validate}(\cdot)$ and \hphantom{aallll}a $\mathsf{decide}(\cdot)$ indication from \block.
Finally, observe that two correct processes can validate (but not decide!) different values.

\medskip 
\Statex \textbf{Properties:}

\Statex \hskip2em  \emph{Strong validity:} If all correct processes that propose do so with the same value $v$, then no correct process decides or validates any value $v' \neq v$.

\smallskip
\Statex \hskip2em \emph{External validity:} If any correct process decides or validates any value $v$, then $\mathsf{valid}(v) = \mathit{true}$.
    


\smallskip 
\Statex \hskip2em \emph{Agreement:} If any correct process decides a value $v$, then no correct process validates or decides any value $v' \neq v$.
    
\smallskip 
\Statex \hskip2em \emph{Integrity:} No correct process decides or receives a $\mathsf{completed}$ indication unless it has previously proposed.

\smallskip 
\Statex \hskip2em \emph{Termination:} If all correct processes propose and no correct process abandons \block, then every correct process eventually receives a $\mathsf{completed}$ \hphantom{allllll}indication.

\smallskip 
\Statex \hskip2em \emph{Totality:} If any correct process receives a $\mathsf{completed}$ indication at some time $\tau$, then every correct process validates a value by time \hphantom{allllll}$\max(\tau, \text{GST}) + 2\delta$.
    
\smallskip 
\Statex \hskip2em \emph{Synchronicity:} Let $\tau$ denote the first time a correct process proposes to \block.
If (1) $\tau \geq \text{GST}$, (2) all correct processes propose by time \hphantom{allllll}$\tau + \Delta_{\mathit{shift}}$, and (3) no correct process abandons \block by time $\tau + \Delta_{\mathit{total}}$, then every correct process decides by time $\tau + \Delta_{\mathit{total}}$.

\smallskip 
\Statex \hskip2em \emph{Completion time:} If a correct process $p_i$ proposes to \block at some time $\tau \geq \text{GST}$, then $p_i$ does not receive a $\mathsf{completed}$ indication by time \hphantom{allllll}$\tau + \Delta_{\mathit{total}}$.

\end{algorithmic}
\end{module}


\subsection{\block's Building Blocks} \label{section:building_blocks} 

In this subsection, we formally present two building blocks that \block utilizes in a ``closed-box'' manner.
Namely, we introduce graded consensus (\Cref{subsection:graded_consensus_building_blocks}) and validation broadcast (\Cref{subsection:validation_broadcast_building_blocks}).
Roughly, graded consensus is used to implement the \emph{safety guards}, while validation broadcast fulfills the role of the ``safe skip'' mechanism (see \Cref{section:technical_overview}).

\subsubsection{Graded Consensus (Module~\ref{mod:graded-consensus})} \label{subsection:graded_consensus_building_blocks}

Graded consensus~\cite{abraham2022gradecast,AttiyaWelch23,DBLP:journals/siamcomp/FeldmanM97} (also known as Adopt-Commit~\cite{DBLP:journals/jpdc/Delporte-Gallet21,DBLP:journals/siamcomp/MostefaouiRRT08}) is a problem in which processes propose their input value and decide on some value with some binary grade.
Graded consensus and similar primitives~\cite{AttiyaWelch23} are often employed in consensus protocols~\cite{Abraham22crusader,gradecastphaseking}.
In brief, the graded consensus primitive ensures agreement among the correct processes only if some correct process has decided a value with (higher) grade $1$.
If no such correct process exists, graded consensus does not guarantee agreement.
(Thus, graded consensus is a weaker primitive than Byzantine agreement.)
In the context of \Cref{sec:phases}, graded consensus is the core primitive of the first and second safety guards.
Importantly, to utilize a graded consensus algorithm in \block, the algorithm must have a known worst-case latency. This requirement is satisfied by all graded consensus implementations that we consider in this paper, as well as all implementations that we are aware of in general.

\begin{module} [h]
\caption{Graded consensus}
\label{mod:graded-consensus}
\footnotesize
\begin{algorithmic}[1]

\Statex \textbf{Events:}

\Statex \hskip2em \emph{request} $\mathsf{propose}(v \in \mathsf{Value})$: a process proposes value $v$.

\Statex \hskip2em \emph{request} $\mathsf{abandon}$: a process abandons (i.e., stops participating in) graded consensus.

\Statex \hskip2em \emph{indication} $\mathsf{decide}(v' \in \mathsf{Value}, g' \in \{0, 1\})$: a process decides value $v'$ with grade $g'$.

\medskip 
\Statex \textbf{Notes:} 
\Statex \hskip2em We assume that every correct process proposes at most once and it does so with a valid value.
We do not assume that all correct processes \hphantom{allllll}propose.

\medskip 
\Statex \textbf{Properties:}
\Statex \hskip2em \emph{Strong validity:} If all correct processes that propose do so with the same value $v$ and a correct process decides a pair $(v', g')$, then $v' = v$ \hphantom{allllll}and $g' = 1$.

\smallskip 
\Statex \hskip2em \emph{External validity:} If any correct process decides a pair $(v', \cdot)$, then $\mathsf{valid}(v') = \mathit{true}$.

\smallskip 
\Statex \hskip2em \emph{Consistency:} If any correct process decides a pair $(v, 1)$, then no correct process decides any pair $(v' \neq v, \cdot)$.

\smallskip
\Statex \hskip2em \emph{Integrity:} No correct process decides more than once.

\smallskip 
\Statex \hskip2em \emph{Termination:} If all correct processes propose and no correct process abandons graded consensus, then every correct process eventually decides.
\end{algorithmic}
\end{module}

\subsubsection{Validation Broadcast (Module~\ref{mod:validation-broadcast})} \label{subsection:validation_broadcast_building_blocks}
Validation broadcast is a novel primitive that we introduce to allow processes to skip views in \name while preserving its safety.
In the context of \Cref{sec:phases}, validation broadcast plays the role of the ``safe skip'' component.
Intuitively, processes broadcast their input value and eventually validate some value.
In a nutshell, validation broadcast ensures that, if all correct processes broadcast the same value, no correct process validates another value.
(This preserves convergence among views.)
Furthermore, if any correct process completes the validation broadcast, all correct processes (even those that have not broadcast) will validate some value shortly after (in two message delays).
(This enables catch-up of processes arbitrarily far behind.)

\begin{module} [h]
\caption{Validation broadcast}
\label{mod:validation-broadcast}
\footnotesize
\begin{algorithmic}[1]

\Statex \textbf{Parameters:}
\Statex \hskip2em $\mathsf{Value}$ $\mathsf{def}(p_i)$ \BlueComment{each process $p_i$ has its default value}

\medskip
\Statex \textbf{Events:}
\Statex \hskip2em \emph{request} $\mathsf{broadcast}(v \in \mathsf{Value})$: a process broadcasts value $v$.

\Statex \hskip2em \emph{request} $\mathsf{abandon}$: a process abandons (i.e., stops participating in) validation broadcast.

\Statex \hskip2em \emph{indication} $\mathsf{validate}(v' \in \mathsf{Value})$: a process validates value $v'$.

\Statex \hskip2em \emph{indication} $\mathsf{completed}$: a process is notified that validation broadcast has completed.

\medskip 
\Statex \textbf{Notes:} 
\Statex \hskip2em We assume that every correct process broadcasts at most once and it does so with a valid value.
We do not assume that all correct processes \hphantom{allllll}broadcast.
Note that a correct process might validate a value even if (1) it has not previously broadcast, or (2) it has previously abandoned the \hphantom{allllll}primitive, or (3) it has previously received a $\mathsf{completed}$ indication. 
Moreover, a correct process may validate multiple values, and two correct \hphantom{allllll}processes may validate different values.

\medskip \Statex \textbf{Properties:}
\Statex \hskip2em \emph{Strong validity:} If all correct processes that broadcast do so with the same value $v$, then no correct process validates any value $v' \neq v$.

\smallskip 
\Statex \hskip2em \emph{Safety:} If a correct process $p_i$ validates a value $v'$, then a correct process has previously broadcast $v'$ or $v' = \mathsf{def}(p_i)$.

\smallskip 
\Statex \hskip2em \emph{Integrity:} No correct process receives a $\mathsf{completed}$ indication unless it has previously broadcast a value.

\smallskip 
\Statex \hskip2em \emph{Termination:} If all correct processes broadcast and no correct process abandons validation broadcast, then every correct process eventually \hphantom{allllll}receives a $\mathsf{completed}$ indication.

\smallskip 
\Statex \hskip2em \emph{Totality:} If any correct process receives a $\mathsf{completed}$ indication at some time $\tau$, then every correct process validates a value by time \hphantom{allllll}$\max(\tau, \text{GST}) + 2\delta$.
\end{algorithmic}
\end{module}

\subsection{\block's Pseudocode} \label{subsection:block_pseudocode}

\block's pseudocode is presented in \Cref{algorithm:block}, and it consists of three independent tasks.
Moreover, a flowchart of \block is depicted in \Cref{fig:block_overview}.
\block internally utilizes the following three primitives: (1) asynchronous graded consensus with two instances $\mathcal{GC}_1$ and $\mathcal{GC}_2$ (line~\ref{line:asynchronous_graded_consensus}), (2) synchronous Byzantine agreement with one instance $\mathcal{A}^S$ (line~\ref{line:synchronous_byzantine_agreement}), and (3) validation broadcast with one instance $\mathcal{VB}$ (line~\ref{line:asynchronous_validation_broadcast}).



\begin{algorithm}
\caption{\block: Pseudocode (for process $p_i$)}
\footnotesize
\label{algorithm:block}
\begin{algorithmic} [1]

\State \textbf{Uses:}
\State \hskip2em Asynchronous graded consensus, \textbf{instances} $\mathcal{GC}_1$, $\mathcal{GC}_2$ \BlueComment{see \Cref{subsection:graded_consensus_building_blocks}} \label{line:asynchronous_graded_consensus}
\State \hskip2em Synchronous Byzantine agreement, \textbf{instance} $\mathcal{A}^S$ \BlueComment{the synchronous agreement algorithm used as a closed-box} \label{line:synchronous_byzantine_agreement}
\State \hskip2em Asynchronous validation broadcast, \textbf{instance} $\mathcal{VB}$ \BlueComment{$\mathcal{VB}$ is initialized with $\mathsf{def}(p_i)$; see \Cref{subsection:validation_broadcast_building_blocks}} \label{line:asynchronous_validation_broadcast}

\medskip
\State \textbf{Comment:}
\State \hskip2em Whenever $p_i$ measures time, it does so locally. Recall that, as $p_i$'s local clock drifts arbitrarily before GST (see \Cref{section:preliminaries}), and $p_i$ \hphantom{aaaal}accurately measures time after GST.

\medskip
\State \textbf{Constants:}
\State \hskip2em $\Delta_1 = \mathit{latency}(\mathcal{GC}_1) \cdot \delta$ \BlueComment{$\mathit{latency}(\mathcal{GC}_1)$ denotes the number of asynchronous rounds of $\mathcal{GC}_1$ (see \Cref{section:preliminaries})}
\State \hskip2em $\Delta_2 = \mathit{latency}(\mathcal{GC}_2) \cdot \delta$ \BlueComment{$\mathit{latency}(\mathcal{GC}_2)$ denotes the number of asynchronous rounds of $\mathcal{GC}_2$ (see \Cref{section:preliminaries})}
\State \hskip2em $\Delta_{\mathit{sync}} = \Delta_{\mathit{shift}} + \delta$ 
\State \hskip2em $\mathcal{B} = \mathit{pbit}(\mathcal{A}^S)$ \BlueComment{$\mathcal{B}$ denotes the per-process bit complexity of $\mathcal{A}^S$ (see \Cref{section:preliminaries})}
\State \hskip2em $\mathcal{R} = \mathit{latency}(\mathcal{A}^S)$ \BlueComment{$\mathcal{R}$ denotes the number of synchronous rounds of $\mathcal{A}^S$ (see \Cref{section:preliminaries})}


\medskip
\State \textbf{Parameters:}
\State \hskip2em $\Delta_{\mathit{shift}}$ = any value (configurable) \label{line:delta_shift}

\State \hskip2em $\Delta_{\mathit{total}} = ( \Delta_{\mathit{shift}} + \Delta_1 ) +  (\mathcal{R} \cdot \Delta_{\mathit{sync}}) +  (\Delta_{\mathit{shift}} + \Delta_2 )$ \label{line:delta_total}


\medskip
\State \textbf{Task 1:}
\State \hskip2em \textbf{When to start:} upon an invocation of a $\mathsf{propose}(v \in \mathsf{Value})$ request \label{line:start_task_1}

\smallskip
\State \hskip2em \textbf{Steps:}
\State \hskip4em 1) Process $p_i$ proposes $v$ to $\mathcal{GC}_1$. Process $p_i$ runs $\mathcal{GC}_1$ until (1) $\Delta_{\mathit{shift}} + \Delta_1$ time has elapsed since $p_i$ proposed, and (2) $p_i$ decides from \hphantom{aaaaaaaaaaa}$\mathcal{GC}_1$.
Let $(v_1, g_1)$ be $p_i$'s decision from $\mathcal{GC}_1$. \label{line:step_gc_1}
 
\smallskip
\State \hskip4em 2) Process $p_i$ proposes $v_1$ to $\mathcal{A}^S$. Process $p_i$ runs (i.e., simulates) $\mathcal{A}^S$ in the following way: (1) $p_i$ executes $\mathcal{A}^S$ for exactly $\mathcal{R}$ rounds, (2) \hphantom{aaaaaaaaaaa}each round lasts for exactly $\Delta_{\mathit{sync}}$ time, and (3) $p_i$ does not send more than $\mathcal{B}$ bits.
Let $v_{A}$ be $p_i$'s decision from $\mathcal{A}^S$. 
If $p_i$ did not \hphantom{aaaaaaaaaaa}decide in time (i.e., there is no decision after running $\mathcal{A}^S$ for $\mathcal{R}$ rounds), then $v_{A} \gets \bot$. \label{line:step_ba}

\smallskip
\State \hskip4em 3) Process $p_i$ initializes a local variable $\mathit{est}_i$.
If $g_1 = 1$, then $\mathit{est}_i \gets v_1$.
Else if $v_{A} \neq \bot$ and $\mathsf{valid}(v_{A}) = \mathit{true}$, then $\mathit{est}_i \gets v_{A}$.
\hphantom{aaaaaaaaaaa}Else, when neither of the previous two cases applies, then $\mathit{est}_i \gets v$. \label{line:step_est}

\smallskip
\State \hskip4em 4) Process $p_i$ proposes $\mathit{est}_i$ to $\mathcal{GC}_2$. Process $p_i$ runs $\mathcal{GC}_2$ until (1) $\Delta_{\mathit{shift}} + \Delta_2$ time has elapsed since $p_i$ proposed, and (2) $p_i$ decides \hphantom{aaaaaaaaaaa}from $\mathcal{GC}_2$.
Let $(v_2, g_2)$ be $p_i$'s decision from $\mathcal{GC}_2$. \label{line:step_gc_2}

\smallskip
\State \hskip4em 5) If $g_2 = 1$, then process $p_i$ triggers $\mathsf{decide}(v_2)$. \BlueComment{process $p_i$ decides from \block} \label{line:step_decide}

\smallskip
\State \hskip4em 6) Process $p_i$ broadcasts $v_2$ via $\mathcal{VB}$, and it runs $\mathcal{VB}$ until it receives a $\mathsf{completed}$ indication from $\mathcal{VB}$. \label{line:step_validation}

\smallskip
\State \hskip4em 7) Process $p_i$ triggers $\mathsf{completed}$. \BlueComment{process $p_i$ completes \block} \label{line:step_complete}

\medskip
\State \textbf{Task 2:}
\State \hskip2em \textbf{When to start:} upon an invocation of an $\mathsf{abandon}$ request \label{line:start_task_2}
\State \hskip2em \textbf{Steps:}
\State \hskip4em 1) Process $p_i$ stops executing Task 1, i.e., process $p_i$ invokes an $\mathsf{abandon}$ request to $\mathcal{GC}_1$, $\mathcal{GC}_2$ and $\mathcal{VB}$ and stops running $\mathcal{A}^S$ (if it is \hphantom{aaaaaaaaaaa}currently doing so). \label{line:step_stop_task_1}

\medskip
\State \textbf{Task 3:}
\State \hskip2em \textbf{When to start:} upon a reception of a $\mathcal{VB}.\mathsf{validate}(v' \in \mathsf{Value})$ indication \label{line:start_task_3}
\State \hskip2em \textbf{Steps:}
\State \hskip4em 1) Process $p_i$ triggers $\mathsf{validate}(v')$. \BlueComment{process $p_i$ validates from \block} \label{line:step_validate_from_block}

\end{algorithmic}
\end{algorithm}

\smallskip
\noindent \textbf{Values of the $\Delta_{\mathit{shift}}$ and $\Delta_{\mathit{total}}$ parameters.} 
In \Cref{algorithm:block}, $\Delta_{\mathit{shift}}$ is a configurable parameter that can take any value (line~\ref{line:delta_shift}).
(Specifically, when employed in \name, the $\Delta_{\mathit{shift}}$ parameter is set to $2\delta$.)
The $\Delta_{\mathit{total}}$ parameter takes an exact value (i.e., it is not configurable) that depends on (1) $\Delta_{\mathit{shift}}$, (2) $\mathcal{GC}_1$, (3) $\mathcal{A}^S$, and (4) $\mathcal{GC}_2$ (line~\ref{line:delta_total}).

\smallskip
\noindent \textbf{Description of Task 1.}
Process $p_i$ starts executing Task 1 upon receiving a $\mathsf{propose}(v \in \mathsf{Value})$ request (line~\ref{line:start_task_1}).
As many of the design choices for Task 1 are driven by the synchronicity property of \block, let us denote the precondition of the property by $\mathcal{S}$.
Concretely, we say that ``$\mathcal{S}$ holds'' if and only if (1) the first correct process that proposes to \block does so at some time $\tau \geq \text{GST}$, (2) all correct processes propose by time $\tau + \Delta_{\mathit{shift}}$, and (3) no correct process abandons \block by time $\tau + \Delta_{\mathit{total}}$.
We now explain each of the seven steps of \block's Task 1:


\smallskip
\noindent$\triangleright$ Step 1 (line~\ref{line:step_gc_1}):
This step corresponds to the execution of the first graded consensus ($\mathcal{GC}_1$). 
As will become clear in Step 3, $\mathcal{GC}_1$ essentially acts as the \emph{first safety guard} (see~\Cref{section:technical_overview}).
Process $p_i$ inputs its proposal $v$ and outputs $(v_1, g_1)$.
Importantly, $p_i$ only moves on to the next step (Step 2) when enough time has elapsed in Step 1 ($\Delta_{\mathit{shift}} + \Delta_1$, where $\Delta_1$ is the maximum time it takes for $\mathcal{GC}_1$ to terminate after GST).
This way, when $\mathcal{S}$ holds, all processes initiate Step 2 nearly simultaneously (within at most $\Delta_{\mathit{shift}}$ time of each other).

\smallskip
\noindent$\triangleright$ Step 2 (line~\ref{line:step_ba}):
This step corresponds to the simulation of the synchronous Byzantine agreement algorithm ($\mathcal{A}^S$).
We are foremost concerned with correctly simulating $\mathcal{A}^S$ when $\mathcal{S}$ holds.
In this scenario, due to Step 1, all processes start Step 2 at most $\Delta_{\mathit{shift}}$ apart.
Therefore, instead of (normally) running each ``synchronous'' round of $\mathcal{A}^S$ for its regular duration ($\delta$), process $p_i$ runs each round for an increased duration that accounts for this shift ($\delta + \Delta_{\mathit{shift}}$).
Hence, process $p_i$ will receive all messages sent for the round, even by ``$\Delta_{\mathit{shift}}$-late'' processes, before moving on to the next round.
After processes execute exactly $\mathcal{R}$ rounds of $\mathcal{A}^S$ in this way, each correct process is guaranteed to decide a valid value from $\mathcal{A}^S$.
Importantly, for $p_i$ (and every other process), we can deduce the maximum number $\mathcal{B}$ of bits sent during a correct synchronous execution of $\mathcal{A}^S$.
Hence, we limit $p_i$ to sending no more than $\mathcal{B}$ bits when simulating $\mathcal{A}^S$.
This prevents $\mathcal{A}^S$ from overshooting its (per-process) budget when faced with asynchronous behavior before GST.\footnote{It may appear as though $\mathcal{A}^S$ cannot exceed its budget due to asynchrony, since we exclude all pre-GST communication from the budget. But since nothing is guaranteed about $\mathcal{A}^S$ in asynchrony, if GST occurs while $\mathcal{A}^S$ is executing, it is possible for the post-GST portion of $\mathcal{A}^S$'s execution to perform more communication than any fully-synchronous execution of $\mathcal{A}^S$.}

\smallskip
\noindent$\triangleright$ Step 3 (line~\ref{line:step_est}):
This step relates to the ``convergence-preservation'' aspect of the \emph{first safety guard} (see~\Cref{section:technical_overview}). 
If $p_i$ decides with grade $1$ from $\mathcal{GC}_1$ (i.e., $g_1 = 1$), then $\mathit{est}_i$ takes the value decided from $\mathcal{GC}_1$ (i.e., $\mathit{est}_1 = v_1$), essentially ignoring the output of $\mathcal{A}^S$.
Due to the strong validity property of $\mathcal{GC}_1$, this will always be the case if processes are already convergent (propose the same value) before executing \block (and $\mathcal{GC}_1$).
Otherwise, if $g_1 = 0$, $p_i$ will adopt the value decided from $\mathcal{A}^S$ instead ($v_{A}$).
If there was no value from $\mathcal{A}^S$ at all, or $v_{A}$ is invalid, it must mean that \block was started before GST ($\mathcal{S}$ does not hold).
In that case, $p_i$ simply adopts its original proposal, which is at least valid.
Notice how, when $\mathcal{S}$ holds, every process adopts the same value by the end of Step 3 (i.e., processes converge).
If all processes have $g_1 = 0$, all processes adopt $v_{A}$ which is the same for all processes ($\mathcal{A}^S$ ensures agreement when $\mathcal{S}$ holds).
Else, if some process has $g_1 = 1$, all processes proposed $v_1$ in Step 2 (as $\mathcal{GC}_1$ ensures consistency), thus $v_1 = v_{A}$ (as $\mathcal{A}^S$ ensures strong validity when $\mathcal{S}$ holds). 
Thus, all processes adopt the same value at the end of Step 3 when $\mathcal{S}$ holds.

\begin{figure} [h]
   \centering
   \includegraphics[width=\textwidth]{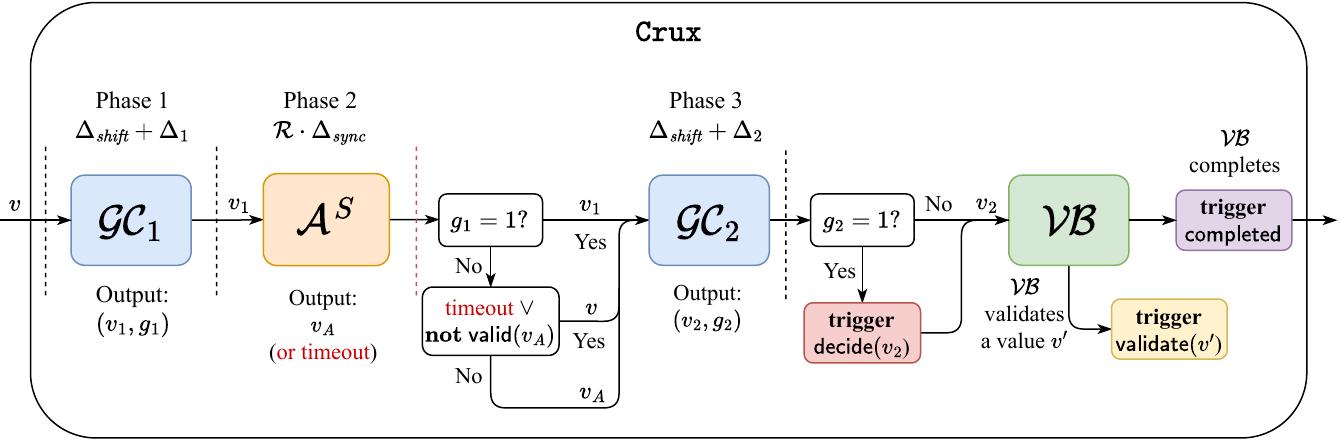}
   \caption{Overview of \block.}
   \label{fig:block_overview}
\end{figure}

\smallskip
\noindent$\triangleright$ Step 4 (line~\ref{line:step_gc_2}):
This step corresponds to the execution of the second graded consensus ($\mathcal{GC}_2$), which acts as the \emph{second safety guard} (as will be seen in Step 5).
Process $p_i$ inputs its estimate $\mathit{est}_i$ (obtained from Step 3) to $\mathcal{GC}_2$, and outputs $(v_2, g_2)$.
As in Step 1, $p_i$ waits enough time ($\Delta_{\mathit{shift}} + \Delta_2$) before moving on to the next step.
This waiting step serves only to ensure the completion time property of \block.

\smallskip
\noindent$\triangleright$ Step 5 (line~\ref{line:step_decide}): 
In this step, we see $\mathcal{GC}_2$'s role as the second safety guard play out: if $p_i$ decided with grade $1$ from $\mathcal{GC}_2$ (Step 4), then $p_i$ decides from \block the value decided from $\mathcal{GC}_2$ (i.e., $p_i$ decides $v_2$), and all other processes either decide or (at least) adopt $v_2$.
Importantly, when $\mathcal{S}$ holds, all correct processes input the same value to $\mathcal{GC}_2$ (as we detailed in Step 3).
This means all correct processes will obtain $g_2 = 1$ in Step 4 and decide $v_2$ in Step 5, due to the strong validity property of $\mathcal{GC}_2$.
This ensures the synchronicity property of \block.

\smallskip
\noindent$\triangleright$ Step 6 (line~\ref{line:step_validation}): 
This step corresponds to the execution of the validation broadcast ($\mathcal{VB}$), which performs the role of the ``\emph{safe skip}'' mechanism (see~\Cref{section:technical_overview}).
Before $p_i$ completes the \block instance, it helps all correct processes obtain a valid value, even if they have not participated in the \block instance (e.g., slow processes before GST).
Moreover, any value obtained here is \emph{safe}: if some correct process decides $v_2$ in Step 5, all processes are guaranteed to broadcast the same $v_2$ value in Step 6, and the only value that can be obtained is precisely $v_2$, due to the strong validity property of validation broadcast.

\smallskip
\noindent$\triangleright$ Step 7 (line~\ref{line:step_complete}): Finally, process $p_i$ completes \block and triggers a $\mathsf{completed}$ indication. 

\smallskip
\noindent \textbf{Description of Task 2.}
A correct process $p_i$ starts executing Task 2 upon receiving an $\mathsf{abandon}$ request (line~\ref{line:start_task_2}).
Task 2 instructs process $p_i$ to stop executing Task 1: process $p_i$ invokes $\mathsf{abandon}$ requests to $\mathcal{GC}_1$, $\mathcal{GC}_2$ and $\mathcal{VB}$ and it stops running $\mathcal{A}^S$ (line~\ref{line:step_stop_task_1}).

\smallskip
\noindent \textbf{Description of Task 3.}
A correct process $p_i$ starts executing Task 3 upon receiving a $\mathsf{validate}(v' \in \mathsf{Value})$ indication from $\mathcal{VB}$ (line~\ref{line:start_task_3}).
When that happens, process $p_i$ validates $v'$ from \block, i.e., $p_i$ triggers a $\mathsf{validate}(v')$ indication (line~\ref{line:step_validate_from_block}).

\smallskip
\noindent \textbf{\changed{Proof sketch.}}
We relegate a formal proof of \block's correctness and complexity to \Cref{section:crux_correctness_complexity_formal}.
Here, we give a proof sketch.

\begin{itemize}
\setlength{\itemsep}{5pt}

\item \emph{Strong validity} (see \Cref{theorem:block_strong_validity}) is derived directly from the strong validity of \block's submodules, namely $\mathcal{GC}_1$, $\mathcal{GC}_2$, and $\mathcal{VB}$. 
We recall that the updating rule of the estimation variable ($est_i$, Step 3), combined with the strong validity property of $\mathcal{GC}_1$, ensures that the output of the synchronous algorithm is ignored if a value is unanimously proposed.

\item \emph{External validity} (see \Cref{theorem:block_external_validity}) is ensured by the external validity property of $\mathcal{GC}_2$ and the safety property of $\mathcal{VB}$. 
Any value decided by a correct process is valid due to the external validity of $\mathcal{GC}_2$. 
If a correct process validates a value, it has either been previously validated from $\mathcal{VB}$, or it is the process's proposal ($v$), which is (assumed) valid.

\item \emph{Agreement} (see \Cref{theorem:block_agreement}) is guaranteed by the consistency property of $\mathcal{GC}_2$ and the strong validity property of $\mathcal{VB}$. 
No two correct processes decide different values from \block due to the consistency property of $\mathcal{GC}_2$. 
Moreover, the strong validity property of $\mathcal{VB}$ ensures that no correct process validates any value different from a potential decided value, which would have been unanimously broadcast through $\mathcal{VB}$.

\item \emph{Integrity} (see \Cref{theorem:block_integrity}) is satisfied as any correct process that decides or completes \block does so while executing Task 1, which it starts only after proposing to \block.

\item \emph{Termination} (see \Cref{theorem:block_termination}) is ensured by the simulation of $\mathcal{A}^S$ within bounded time (timeout) and the termination properties of \block's remaining submodules, $\mathcal{GC}_1$, $\mathcal{GC}_2$, and $\mathcal{VB}$. 

\item \emph{Totality} (see \Cref{theorem:block_totality}) comes as a direct consequence of the totality property of $\mathcal{VB}$. 

\item \emph{Completion time} (see \Cref{theorem:block_completion_time}). 
The earliest time at which a correct process can broadcast via $\mathcal{VB}$ is lower-bounded by the sum ($\Delta_{\mathit{total}}$) of the time it takes to complete each previous step of the algorithm (particularly, Steps 1, 2, and 4). 
Therefore, the integrity property of $\mathcal{VB}$ ensures that no process receives a completed indication from \block before this time.

\item \emph{Synchronicity} (see \Cref{theorem:block_synchronicity}). 
As detailed in the description of Task 1, when $\mathcal{S}$ holds (i.e., the precondition of the synchronicity property), all correct processes (1) execute a correct simulation of $\mathcal{A}^S$ (Step 2), (2) obtain the same estimate (Step 3), (3) propose and obtain the same value with grade 1 from $\mathcal{GC}_2$ (Step 4), and (4) decide (Step 5).
For the complete proof that the simulation of $\mathcal{A}^S$ (Step 2) is correct, we refer the reader to \Cref{section:crux_correctness_complexity_formal}.


\item \emph{Per-process bit complexity} (see \Cref{theorem:block_communication}). 
Let us consider the worst-case bit complexity \emph{per-process}.
\block only sends messages through its $\mathcal{GC}_1$, $\mathcal{GC}_2$, and $\mathcal{VB}$ instances, and when simulating $\mathcal{A^S}$.
Thus, \block's worst-case per-process bit complexity is equal to the sum of its parts.
For $\mathcal{GC}_1$, $\mathcal{GC}_2$, and $\mathcal{VB}$, which are originally asynchronous algorithms, the worst-case complexity is identical.
For the simulation of $\mathcal{A^S}$, crucially, communication is bounded in Step 3: the per-process communication complexity when simulating $\mathcal{A^S}$ in asynchrony is bounded by the worst-case per-process complexity of executing $\mathcal{A^S}$ directly in synchrony.
Therefore, before and after GST, the per-process bit complexity of \block is bounded by
\begin{equation*}
    \mathit{pbit}(\mathcal{GC}_1) + \mathit{pbit}(\mathcal{GC}_2) + \mathit{pbit}(\mathcal{VB}) + \mathit{pbit}(\mathcal{A}^S) \text{ bits, where}
\end{equation*}
$\mathit{pbit}(\mathcal{X})$ denotes the maximum number of bits any correct process sends in $\mathcal{X} \in \{\mathcal{GC}_1, \mathcal{GC}_2, \mathcal{VB}, \mathcal{A}^S\}$.
When $\mathit{pbit}(\mathcal{GC}_1) = \mathit{pbit}(\mathcal{GC}_2) = \mathit{pbit}(\mathcal{VB}) \in O(n)$, \block preserves the per-process bit complexity of $\mathcal{A}^S$, given that $\mathit{pbit}(\mathcal{A}^S) \in \Omega(n)$, due to the Dolev-Reischuk lower bound~\cite{dolev1985bounds}.

\item \emph{Latency after GST}.
The latency of \block $(\mathit{latency}(\block)$) is defined such that, if all correct processes start executing \block by some time $\tau$, then all correct processes complete by time $\textsf{max}(\tau, GST) + \mathit{latency}(\block)$.
In particular $\mathit{latency}(\block) = (\mathit{latency}(\mathcal{GC}_1) \cdot \delta) + (\mathcal{R} \cdot \Delta_{\mathit{sync}}) + (\mathit{latency}(\mathcal{GC}_2) \cdot \delta) + (\mathit{latency}(\mathcal{VB}) \cdot \delta)$.
This corresponds directly to the sum of the maximum latency of each step after GST.
(Notice that, given that the latency concerns time spent executing steps after GST, starting or executing steps before GST can only decrease latency.)

\end{itemize}



\subsection{From \block to \name}\label{section:from_block_to_name}


This subsection briefly presents \name, our generic transformation that maps \emph{any} synchronous Byzantine agreement algorithm into a partially synchronous one.
\name consists of a sequential composition of an arbitrarily 
 long series of \block instances (see \Cref{section:block}) through views.
Whenever a process starts a \block instance within a view, it simply proposes a value validated by the \block instance in the preceding view (or its initial proposal, if it is the first view).
\name's safety is directly ensured by this sequential composition, since each \block instance guarantees agreement and strong validity.
Therefore, \name has two main objectives: (1) ensuring \emph{liveness} by providing the necessary precondition for \emph{synchronicity} (\block's property, see \Cref{subsection:block_problem_definition}), and (2) halting processes without jeopardizing liveness or safety.
These objectives are achieved efficiently via two components, namely a \emph{view synchronizer} (for liveness) and a \emph{finisher} (for halting), which we briefly describe next.
The full pseudocode and proof of \name (and all its components) can be found in \Cref{section:name_correctness_complexity_formal}. 



\smallskip
\noindent \textbf{View synchronization.}
The view synchronizer is run once before each view and ensures that, after GST, correct processes (1) enter a common view nearly simultaneously (within $\Delta_{\mathit{shift}} = 2\delta$ time of each other), and (2) remain there for a sufficiently long duration ($\Delta_{\mathit{total}}$).
This matches exactly the precondition for synchronicity (\Cref{subsection:block_problem_definition}), resulting in all processes deciding by the end of that common view (liveness).
It is known that view synchronization can be implemented efficiently  (using $O(n^2)$ bits) following Bracha's double-echo approach \cite{bravo2022making,naor2019cogsworth}.
Briefly, (1) when a process wishes to advance to the next view $V$, it broadcasts $\langle \textsc{start-view}, V \rangle$, (2) when at least $t+1$ $\langle \textsc{start-view}, V \rangle$ messages are received by a correct process, it echoes (broadcasts) that $\langle \textsc{start-view}, V \rangle$ message, and (3) a process finally advances to view $V$ upon receiving $2t+1$
$\langle \textsc{start-view}, V \rangle$. 
The amplification mechanism (2) ensures that, if a process enters a view at time $\tau \geq \mathsf{GST}$, it will be followed by the remaining correct processes by time $\tau + 2\delta$.
The $t+1$ threshold (2) prevents Byzantine processes from arbitrarily pushing correct processes to skip views.




\smallskip
\noindent \textbf{Finisher.}
After a correct process has decided, it cannot arbitrarily halt, as other correct processes might depend upon it to terminate.
To halt, a correct process $p_i$ first executes a \emph{finisher}, which guarantees upon completion that all correct processes will obtain the correct decided value (leaving $p_i$ free to halt).
For short inputs ($L\in O(1)$), an efficient finisher can be implemented straightforwardly by following the double-echo approach of Bracha's reliable broadcast \cite{B87}. 
We note that implementing an efficient finisher for long values becomes non-trivial, since broadcasting the decided value is prohibitive in terms of communication.
We relegate an efficient implementation of a finisher for long values to \Cref{subsection:finisher}.

\smallskip
\noindent \textbf{Proof sketch.}
The formal proof of \name's correctness and complexity is relegated to \Cref{section:name_correctness_complexity_formal}.
Here, we provide a proof sketch of \name's correctness:
\begin{itemize}
\setlength{\itemsep}{5pt}

\item \emph{Agreement} (see \Cref{theorem:name_agreement}).
Suppose a correct process $p_i$ decides a value $v$ in some view $V$ and another correct process $p_j$ decides a value $v'$ in some view $V'$.
Without loss of generality, let $V \leq V'$.
If $V = V'$, the agreement property of the \block instance of view $V = V'$ ensures that $v = v'$.
If $V < V'$, it is ensured that no correct process validates any non-$v$ value from the \block instance of view $V$ (due to the agreement property of \block).
Therefore, no correct process proposes a non-$v$ value to the \block instance of view $V + 1$, which implies that no correct process can decide or validate a non-$v$ value in view $V + 1$.
By inductively applying the same argument, we reach a conclusion that process $p_j$ cannot decide any non-$v$ value in view $V' > V$, thus showing the agreement property of \name.

\item \emph{External validity} (see \Cref{theorem:name_external_validity}) follows directly from the external validity property of \block.

\item \emph{Strong validity} (see \Cref{theorem:name_strong_validity}) follows from the strong validity property of \block.
Suppose all correct processes propose the same value $v$ to \name.
Therefore, no correct process proposes any value different from $v$ to the \block instance of view $1$, which implies that (1) no correct process decides or validates any value different from $v$ in view $1$ (due to the strong validity property of \block), and (2) no correct process proposes any value different from $v$ to the \block instance of view $2$.
By applying the same argument to view $2$ (and all subsequent views), we show that the strong validity of \name is preserved.

\item \emph{Termination} (see \Cref{theorem:name_termination}).
Let $\vfinal$ denote the smallest view entered by any correct process after GST; let $\tau_{\vfinal}$ denote the time at which the first correct process enters view $\vfinal$.
Our view synchronization mechanism ensures that all correct processes enter view $\vfinal$ within $\Delta_{\mathit{shift}} = 2\delta$ time of each other.
Moreover, the completion time property of the \block instance of view $\vfinal$ ensures that no correct process leaves view $\vfinal$ before time $\tau_{\vfinal} + \Delta_{\mathit{total}}$.
Therefore, the preconditions of the synchronicity property of the \block instance associated with view $\vfinal$ are satisfied, which implies that all correct processes decide in view $\vfinal$ (unless they have already done so).
\end{itemize}

\smallskip
\noindent \textbf{New algorithms obtained by \name.}
We conclude this section by presenting a few efficient signature-free partially synchronous Byzantine agreement algorithms that \name yields (see \Cref{table:concrete_algorithms}).
As we formally prove in \Cref{section:name_correctness_complexity_formal}, each correct process executes only $O(1)$ views after GST.
Therefore, $\mathit{pbit}(\name) = O\big( n + \mathit{pbit}(\block) \big)$ for constant-sized inputs.
Since $\mathit{pbit}(\block) = \mathit{pbit}(\mathcal{GC}_1) + \mathit{pbit}(\mathcal{GC}_2) + \mathit{pbit}(\mathcal{VB}) + \mathit{pbit}(\mathcal{A}^S)$, the per-process bit complexity $\mathit{pbit}(\name)$ of \name for constant-sized values can be defined as
\begin{equation*}
    \mathit{pbit}(\name) = O\big( n + \mathit{pbit}(\mathcal{GC}_1) + \mathit{pbit}(\mathcal{GC}_2) + \mathit{pbit}(\mathcal{VB}) + \mathit{pbit}(\mathcal{A}^S) \big).
\end{equation*}
Similarly, \Cref{section:name_correctness_complexity_formal} proves that the per-process bit complexity $\mathit{pbit}(\name)$ of \name for long $L$-bit values can be defined as
\begin{equation*}
    \mathit{pbit}(\name) = O\big( L + n \log n + \mathit{pbit}(\mathcal{GC}_1) + \mathit{pbit}(\mathcal{GC}_2) + \mathit{pbit}(\mathcal{VB}) + \mathit{pbit}(\mathcal{A}^S) \big).
\end{equation*}
In \Cref{table:concrete_algorithms}, we specify, for each \name-obtained Byzantine agreement algorithm, the concrete implementations of (1) asynchronous graded consensus ($\mathcal{GC}_1$ and $\mathcal{GC}_2$), (2) synchronous Byzantine agreement ($\mathcal{A}^S$), and (3) asynchronous validation broadcast ($\mathcal{VB}$) required to construct the algorithm.

\begin{table}[ht]
\centering
\footnotesize
\tabcolsep=0.08cm
\begin{tabular}{|p{3.75cm}|p{1.3cm}|p{2.6cm}|p{3cm}|p{2.6cm}|p{1.8cm}|  }
 \hline
 \centering \textbf{Total bit complexity of \\the final algorithm} &
\centering Resilience & \centering $\mathcal{GC}_1 = \mathcal{GC}_2$ \\ total bits & \centering $\mathcal{A}^S$ \\ $n$ $\cdot$ (bits per process)  & \centering $\mathcal{VB}$ \\ total bits & \centering Cryptography \tabularnewline
 \hline
 \hline

\centering $O(n^2)$ \\(with $L \in O(1)$)  & \centering $t < n/3$ & \centering \cite{AttiyaWelch23} \\ $O(n^2)$ & \centering \cite{berman1992bit,CoanW92} \\ $O(n^2)$ & \centering \Cref{subsection:validation_broadcast_basic_implementation} \\ $O(n^2)$ & \centering None \tabularnewline

\hline

\centering $O(n^2L)$  & \centering $t < n/3$ & \centering \cite{AttiyaWelch23} \\ $O(n^2L)$ & \centering \cite{berman1992bit,CoanW92} \\ $O(n^2L)$ & \centering \Cref{subsection:validation_broadcast_basic_implementation} \\ $O(n^2L)$ & \centering None \tabularnewline



 
 \hline
  \centering $O(nL + n^2\log(n)\kappa)$ \\ (only strong validity) & \centering $t < n/3$ & \centering \Cref{subsection:graded_consensus_bounded} \\ $O(nL + n^2\log(n)\kappa)$  & \centering \cite{chen2021optimal} \\ $O(nL + n^2\log n)$ & \centering \Cref{subsection:validation_bcast_long_3_4} \\ $O(nL + n^2\log(n)\kappa)$ & \centering Hash \tabularnewline
 \hline

 \centering $O(n\log(n)L + n^2\log(n)\kappa)$ & \centering $t < n/3$ & \centering \Cref{subsection:graded_consensus_bounded} \\ $O(nL + n^2\log(n)\kappa)$  & \centering \cite{errorFreeValidated}
 \\ $O(
  n\log(n)L + n^2\log n)$ & \centering \Cref{subsection:validation_bcast_long_3_4} \\ $O(nL + n^2\log(n)\kappa)$ & \centering Hash \tabularnewline
 \hline

   \centering $O(nL + n^2\log n)$ \\ (only strong validity) & \centering $t < n/5$ & \centering \Cref{subsection:its_graded_consensus} \\ $O(nL + n^2\log n)$  & \centering \cite{chen2021optimal} \\ $O(nL + n^2\log n)$ & \centering \Cref{subsection:its_validation_broadcast} \\ $O(nL + n^2\log n)$ & \centering None \tabularnewline
 \hline

 \centering $O(n\log(n)L + n^2\log n)$ & \centering $t < n/5$ & \centering \Cref{subsection:its_graded_consensus} \\ $O(nL + n^2\log n)$  & \centering \cite{errorFreeValidated}
 \\ $O(
  n\log(n)L + n^2\log n)$ & \centering \Cref{subsection:its_validation_broadcast} \\ $O(nL + n^2\log n)$ & \centering None \tabularnewline
 \hline
\end{tabular}
    \caption{Concrete partially synchronous Byzantine agreement algorithms obtained by \name.
    We emphasize that rows 3 and 5 satisfy only strong validity (i.e., they do not satisfy external validity).
    All mentioned algorithms are balanced in terms of total bit complexity.
    \\($L$ denotes the bit-size of a value, whereas $\kappa$ denotes the bit-size of a hash value. We consider $\kappa \in \Omega(\log n)$.)}
\label{table:concrete_algorithms}
\end{table}

\smallskip
\noindent \changed{\textbf{Discussion.} It is well-known that synchronous Byzantine Agreement can be solved with $n \geq 2t+1$ under the assumption of a cryptographic setup, so the interested reader may ask themselves the following question: If we assume cryptography and use a $2t+1$ protocol for $\mathcal{A}^S$, would we obtain a $2t+1$ version of \block? The answer is no: the other components of \block, specifically graded consensus, require $n \geq 3t+1$ in partial synchrony \cite{AW04}, even with cryptography. In fact, this is not specific to \block: no transformation could achieve $2t+1$ Byzantine Agreement in partial synchrony, since this has been shown to be impossible~\cite{DLS88}}.

\changed{\section{\name in the Randomized World}}
\label{sec:consequences}

While the presentation of \name has been made for an underlying deterministic synchronous protocol, the transformation can be easily adapted to an underlying \emph{randomized} synchronous protocol by using a more general simulation technique~\cite{DBLP:conf/podc/LindellLR02} which is geared towards cryptography-based algorithms (detailed in Algorithm \ref{algorithm:simulation}, Lemmas \ref{lemma:simulation_correct} and \ref{lemma:cryptography_based_simulation_correct}). 
In the most challenging model (i.e., when faced with a full-information adaptive adversary), the best known (potentially randomized) solution is actually a deterministic one~\cite{errorFreeValidated}. 
However, with some relaxations (e.g., against a static adversary or assuming private channels), we can see how \name can be useful for randomized protocols, where significant gaps exist between synchronous and
asynchronous variants.

\subsection{Adaptation of the \name Transformation}

The $\mathsf{CryptoSim}$ simulation (Algorithm \ref{algorithm:simulation}) enables us to simulate any adversarial probability space over executions of any synchronous randomized protocol $\mathcal{A}^S$ in a partially synchronous environment, as long as the precondition $\mathcal{S}$ (see \Cref{subsection:block_pseudocode}) of view synchronization holds (i.e., correct processes start and execute the simulation within $\Delta_{\mathit{shift}} = 2\delta$ time of each other and for sufficiently long). Moreover, if $\mathcal{A}^S$ has a per-process bit-complexity of $\mathcal{B}$, the simulation will preserve the same (fixed) round complexity as the original protocol and exhibit a per-process bit-complexity of $\max(O(n), \mathcal{B})$, even if the Global Stabilization Time (GST) occurs during the execution.

\smallskip
\noindent \textbf{Randomized simulation robustness.}
The robustness of the randomized simulation arises because the adversary cannot adapt its behavior in a (simulated) round $r$ based on information from round $r+1$. 
This ensures that the simulated protocol retains its probabilistic guarantees in the partially synchronous model.
The proof amounts to a reformulation of \Cref{lemma:cryptography_based_simulation_correct}, which is written for deterministic \emph{cryptography-based} solutions, but done instead in a probabilistic fashion.
Formally, let $\mathcal{A}^S$ be a (randomized) synchronous algorithm with a fixed number of rounds $\mathcal{R}$ and fixed per-process bit-complexity $\mathcal{B}$.\footnote{We explain why this assumption is not a problem in the paragraph ``From synchronous Las Vegas to synchronous Monte Carlo''.} Let $\mathcal{A}^{*} = \mathsf{CryptoSim}(\mathcal{A}^S)$. For every adversary $\mathsf{Adv}^{*}$ for $\mathcal{A}^{*}$ that guarantees that the condition $\mathcal{S}$ holds and that $\mathcal{R}$ rounds elapse, there exists an adversary $\mathsf{Adv}^{S}$ for $\mathcal{A}^S$ and an associated bijective homomorphism

$$\mu : (\mathit{execs}(\mathcal{A}^{*}), \mathcal{F}_{\mathit{execs}(\mathcal{A}^{*})}, \mathsf{Pr}_{\mathsf{Adv}^{*}}) \rightarrow (\mathit{execs}(\mathcal{A}^S), \mathcal{F}_{\mathit{execs}(\mathcal{A}^S)}, \mathsf{Pr}_{\mathsf{Adv}^S})\footnote{Following a standard definition of probability space over executions engendered by the adversary \cite{Mitra07,Segala95}.}$$
such that for every pair of executions $(\mathcal{E}^{*}, \mathcal{E}^S) \in \mathit{execs}(\mathcal{A}^{*}) \times \mathit{execs}(\mathcal{A}^S)$ with $\mathcal{E}^S = \mu(\mathcal{E}^{*})$, we have $\mathsf{Pr}_{\mathsf{Adv}^S}(\mathcal{E}^S) = \mathsf{Pr}_{\mathsf{Adv}^{*}}(\mathcal{E}^{*})$ and if $\mathsf{Pr}_{\mathsf{Adv}^{*}}(\mathcal{E}^{*}) > 0$:

\begin{compactitem}
    \item $\mathcal{E}^S$ and $\mathcal{E}^{*}$ count the same number of rounds $\mathcal{R}$.
    \item The sets of correct processes in $\mathcal{E}^S$ and $\mathcal{E}^{*}$ are identical.
    \item The proposals of correct processes in $\mathcal{E}^S$ and $\mathcal{E}^{*}$ are identical.
    \item The sets of messages sent by correct processes in $\mathcal{E}^S$ and $\mathcal{E}^{*}$ are identical.
    \item For each correct process $p_i$ and every $k \in [1, \mathcal{R} + 1]$, $s_i^k(\mathcal{E}^{*}) = s_i^k(\mathcal{E}^S)$, where $s_i^k(\mathcal{E}^{S})$ (resp., $s_i^k(\mathcal{E}^{*})$) is the state of $p_i$ at the beginning of the $k$-th (i.e., at the end of the $(k - 1)$-st) round (resp., simulated round) in $\mathcal{E}^{S}$ (resp., $\mathcal{E}^{*}$). 
\end{compactitem}

\smallskip
\noindent \textbf{From synchronous Monte Carlo to partially synchronous Las Vegas.}
In the context of randomized algorithms, Las Vegas algorithms always produce a correct result (with probability 1), but their runtime may vary, while Monte Carlo algorithms have a fixed runtime but may produce incorrect results with a non-zero probability.
Ensuring correctness in a synchronized view (where processes enter the view within $\Delta_{\mathit{shift}} = 2\delta$ from each other) with a constant probability $\rho$ is acceptable, as safety is guaranteed deterministically, and Las Vegas termination will be achieved with probability $1$ after $O(1)$ expected views after GST. Indeed, a reduction to a geometric distribution yields an expected latency of $O(\mathcal{R}/\rho) = O(\mathcal{R})$ (after GST), where $\mathcal{R}$ is the known worst-case round complexity of the underlying synchronous algorithm $\mathcal{A}^S$. The per-process communication complexity (after GST)  remains \(\max(\mathcal{B}, O(n))\).
This is particularly relevant for protocols solving leader or committee election against a static adversary with a constant probability of success (e.g., versions of \cite{GPV06,KSSV06}).

We can note that there is no obvious complexity-preserving transformation to convert an asynchronous Monte Carlo protocol with a non-zero probability of termination violation (e.g., \cite{KKKSS08,KKKSS10}) into a Las Vegas protocol with probability 1 of liveness and safety. This limitation emphasizes the value of our approach in partially synchronous settings.

\smallskip
\noindent \textbf{From synchronous Las Vegas to synchronous Monte Carlo.}
A synchronous Las Vegas algorithm with constant probability $\rho$ of success and known expected latency  $\mathcal{R}^e$ but no bound on worst-case latency can be converted into a Monte Carlo algorithm with worst-case latency $\mathcal{R}_x = x \mathcal{R}^e$ and success probability $\rho (1 - \frac{1}{x})$, via a straightforward Markov inequality. Let us note that a similar observation can be made for the per-process bit-complexity, and its utilization in $\mathsf{CryptoSim}$.
Therefore, in the following discussion, we will not distinguish between Las Vegas and Monte Carlo synchronous protocols, nor will we specify whether the mentioned latencies and communication complexities are expected or worst-case latencies.

\smallskip
\noindent \textbf{Reductions between Byzantine agreement problems in the synchronous setting.}
Interactive consistency (IC) is an agreement variant where correct processes agree on the proposals of all processes.
Any solvable Byzantine agreement problem, including multi-valued validated Byzantine agreement (MVBA, guaranteeing external validity) and strong MVBA (SMVBA, guaranteeing both strong and external validity, as in \Cref{section:introduction}), can be reduced to interactive consistency at no cost \cite{civit2023all}. 
Interactive consistency in turn can be reduced to $n$ (parallel) instances of Byzantine broadcast (BB), hence its alternative name, ``parallel broadcast'' \cite{AsharovChandramouli24}. 
In our setting $(n>3t)$, Byzantine broadcast (BB) and Byzantine agreement with strong validity (BA) are computationally equivalent.\footnote{Our strong validity property is sometimes called \emph{weak validity} \cite{abraham2019asymptotically}, while \emph{strong validity} is used when the decision must be proposed by a correct process.} Moreover, since the breakthrough of \cite{chen2021optimal}, $L$-bit BB (resp., BA) can be reduced to binary BB (resp., BA) at the cost of $O(1)$ additional rounds and $O((L+n)\log(n))$ per-process bit-complexity. The parallel composition of $n$ instances of black-box Monte Carlo protocols that succeed with probability $1 - 1/n^c$ for some constant $c \geq 1$ (e.g., \cite{hajiaghayi2024nearly,KingSaia10,King2011}), succeeds with at least constant probability. 
Therefore, we can derive (S)MVBA protocols from binary BA/BB protocols, with the corresponding overhead taken into account.

Other protocols facing a static adversary can elect a good leader or a good committee with some probability $\rho$ (as done in $\cite{GPV06}$). From such a protocol, we can derive a version solving MVBA by a reduction
to the deterministic committee broadcast of \cite{errorFreeValidated}, at the cost of $O(1)$ additional rounds and $O((L+n)\log(n))$ additional bits per-process. 
Accordingly, we can construct MVBA protocols from leader or committee election protocols, factoring in the associated overhead.

Some (subquadratic) protocols, like \cite{KSSV06}, achieve \emph{almost-everywhere} agreement, where the agreement is reached only among a set of $1-O(\log^{-1}(n))$ processes. Such a protocol can be directly transformed into ``full agreement'' (i.e., going from almost-everywhere to everywhere) by a single quadratic all-to-all round of communication. 
Other work has aimed to solve subquadratic almost-everywhere-to-everywhere \cite{braud2013fast,Gelles23}, but this is not directly useful in our case as \name's transformation involves a quadratic graded consensus implementation.

\smallskip
\noindent \textbf{Reductions between Byzantine agreement problems in the asynchronous setting.}
Agreement on a Core Set (ACS) (a.k.a. vector consensus or asynchronous common subset) is an agreement variant where correct processes agree on the proposals of $n-t$ processes. Any asynchronously solvable Byzantine agreement problem, including SMVBA, can be reduced to ACS at no cost \cite{civit2023validity}.
This problem can further be reduced \cite{ben1994asynchronous} to $n$ (parallel) instances of Byzantine Reliable Broadcast \cite{Alhaddad2022} and $n$ (parallel) instances of binary asynchronous BA (ABA) with strong validity (e.g., \cite{BCP18}).

\subsection{Implications for protocol performance}
\label{sec:state_of_the_art}

In this subsection, we present a non-exhaustive comparison of state-of-the-art protocols designed for different adversarial models with a computationally unbounded adversary (see Table~\ref{table:state_of_the_art_protocols}). These protocols are evaluated based on their resiliency, latency, communication complexity, and any specific relaxations of the adversary model. 
For setup-free synchronous protocols, \name immediately yields partially-synchronous counterparts in the same model with (1) the same resiliency, (2) the same asymptotic complexity (counted post-GST), (3) unconditional safety, and (4) Las Vegas liveness ensured with probability 1 and preserved expected latency (modulo a constant overhead) after GST. 
Essentially, these protocols inherit the performance characteristics of their synchronous sources while avoiding safety issues in non-synchronous environments. The performance metrics (evaluated post-GST), surpass those of state-of-the-art asynchronous counterparts (operating under too challenging conditions). We acknowledge that pre-GST cost is not considered in this comparison, but we still believe that the yielded results offer an interesting trade-off in practice.
Table~\ref{table:state_of_the_art_protocols} provides a comprehensive overview of these protocols.

\begin{table}[h!]
    \footnotesize
    \centering
    \begin{tabular}{|c|c|c|c|c|c|c|c|}
        \hline
        \textbf{Protocol} & \textbf{Relaxation} & \textbf{Security} & \textbf{Resiliency} & \textbf{Latency} & \textbf{Communication} & \textbf{Problem} & \textbf{Sync} \\
        \hline
        KS \cite{king2016byzantine,king2018correction} & None & Perfect & $n/(0.87\cdot 10^{9}) $ & $O(n^3)$ & $O(n^5)$ & binary & A \\
        \hline
        HPZ \cite{huang2023byzantine} & None & Perfect & $n/(3 + \epsilon)$ & $\tilde{O}(n^4/\epsilon^{8})$ & $\tilde{O}(n^6/\epsilon^{8})$ & binary & A \\
        \hline
        HPZ \cite{huang2023byzantine} & None & Perfect & $n/3$ & $\tilde{O}(n^{12})$ & $\tilde{O}(n^{14})$ & binary & A \\
        \hline
        $\boldsymbol{\textbf{\name}}($CDGGKVZ \cite{errorFreeValidated}$)$ & None & Perfect & $n/3$ & $O(n)$ & $O(Ln^2)$ & S+E & PS \\
        \hline \hline
        KKKSS \cite{KKKSS08,KKKSS10} & Static adv. & $1/\mathit{log}^c(n)$ & $n/(3 + \epsilon)$ & $O(\log^{27}(n))$ & $\tilde{O}(n^2)$ & LE & A \\
        \hline
        KKKSS \cite{KKKSS08,KKKSS10} & Static adv. & $1/n^c$ & $n/(3 + \epsilon)$ & $2^{\Theta(\log^8(n))}$ & $\tilde{O}(n^2)$ & LE & A \\
        \hline
        $\boldsymbol{\textbf{\name}}($GPV \cite{GPV06}$)$ & Static adv. & Perfect & $n/(3 + \epsilon)$ & $O(\log(n)/\epsilon^2)$ & $O(Ln^2+n^3)$ & S+E & PS \\
        \hline
        $\boldsymbol{\textbf{\name}}($GPV \cite{GPV06}$)$ & Static adv. & Perfect & $n/\log^{1.58}(n)$ & $O(1)$ & $O(Ln+n^3)$ & S+E & PS \\
        \hline
        $\boldsymbol{\textbf{\name}}($KSSV \cite{KSSV06}$)$ & Static adv. & Perfect & $n/(3 + \epsilon)$ & $\mathit{polylog}(n)$ & $O(Ln^2)$ & S+E & PS \\
        \hline \hline
        BB \cite{Bar-JosephB98} & Crash & - & - & $\Omega(\sqrt{n/\log(n)})$ & - & LB & S \\
        \hline
        ACDNP0S \cite{AbrahamCDNP0S19} & Omission & - & - & - & $\Omega(n^2)$ & LB & S \\
        \hline
        AAKS \cite{AAKS14,AAKS18} & Crash & Perfect & $n/2$ & $\tilde{\Omega}(n)$ & $\tilde{O}(n^2)$ & binary & A \\
        \hline
        $\boldsymbol{\textbf{\name}}($HKO \cite{hajiaghayi2024nearly}$)$ & Omission & Perfect & $n/30$ & $O(\sqrt{n}\log^2(n))$ & $O(n^2\log^3(n))$ & binary & PS \\

\hline
        AAKS \cite{AAKS14,AAKS18} & Crash only & Perfect & $n/2$ & $\tilde{\Omega}(n)$ & $\tilde{O}(Ln^2+n^3)$ & ACS & A\\
        \hline
        $\boldsymbol{\textbf{\name}}($HKO \cite{hajiaghayi2024nearly}$)$ & Omission & Perfect & $n/30$ & $O(\sqrt{n}\log^2(n))$ & $O(Ln^2+n^3\log^3(n))$ & S+E & PS \\
        
        \hline \hline

Rambaud \cite{rambaud2023adaptively} & Private chan. & - & $t$ & - & $\Omega(nt)$ & LB & PS \\
        \hline
        BCP \cite{BCP18} & Private chan. & Perfect & $n/3$ & $O(n)$ & $O(n^6)$ & binary & A \\
        \hline
        BCP \cite{BCP18} & Private chan. & Perfect & $n/(3+\epsilon)$ & $O(1/\epsilon)$ & $O(n^6)$ & binary & A \\
        \hline
        BCP \cite{BCP18} & Private chan. & Perfect & $n/3$ & $O(n \log(n))$ & $O(Ln^2+n^7)$ & ACS & A \\
        \hline
        BCP \cite{BCP18} & Private chan. & Perfect & $n/(3+\epsilon)$ & $O(\log(n)/\epsilon)$ & $O(Ln^2+n^7)$ & ACS & A \\
        \hline
        AAPS \cite{AAPS23} & Private chan. & Perfect & $n/4$ & $O(1)$ & $\tilde{O}(Ln^2+n^4)$ & ACS & A \\
        \hline
        AAPS \cite{AAPS23} + \cite{ChoudhuryPatra23}'s AISS & Private chan. & $\mathit{neg}(\lambda)$ & $n/3$ & $O(1)$ & $O(Ln^2+n^6 + \lambda^2 n^5)$ & ACS & A \\
        \hline
        $\boldsymbol{\textbf{\name}}($AC \cite{AsharovChandramouli24}$)$ & Private chan. & Perfect & $n/3$ & $O(1)$ & $O(Ln^2+n^3 \log^2(n))$ & S+E & PS \\
\hline
        $\boldsymbol{\textbf{\name}}($KS \cite{KingSaia10,King2011}$)$  & Private chan. & Perfect & $n/(3+\epsilon)$ & $\mathit{polylog(n)}$ & $\tilde{O}(Ln^{5/2})$ & S+E & PS \\    
        \hline
    \end{tabular}
    \caption{Overview of state-of-the-art fully asynchronous and partially synchronous protocols in different models with computationally unbounded adversary. 
        In the Problem column: ACS stands for ``Agreement on a Core Set'', LE stands for ``Leader/Committee Election'', LB stands for ``lower bound'', S stands for ``strong validity'', E stands for ``external validity'', and binary refers to the Byzantine agreement problem (with strong validity) where only 0 and 1 can be proposed/decided.
        A value different from 'Perfect' in the Security column corresponds to the probability of violating termination. 
        A (Asynchrony), S (Synchrony), and PS (Partial Synchrony) in the Sync column indicate the synchrony assumption. 
        $\epsilon$ is always a non-zero constant. 
        For protocols yielded by $\name$, we present the total communication, while it is balanced (the per-process counterpart can be divided by $n$). 
        The communication complexity of GPV \cite{GPV06} is indeed cubic, while it is sometimes presented to be $\tilde{O}(n^2)$ (see \cite[Lemma 10]{Vaikuntanathan09}), since the first step of the underlying $\Pi_{\mathsf{Elect}\text{-}\mathsf{GPV}}$ involves $n$ parallel instances of (balanced) quadratic gradecast, which is cubic \cite{AsharovChandramouli24,BG12,ZLC23}. 
    }
    \label{table:state_of_the_art_protocols}
\end{table}

For instance, in the non-relaxed model, no fully asynchronous consensus protocol with good resiliency ($>n/10^{8}$) achieves better than $O(n^4)$ expected latency, while $\name(\mathsf{ErrorFreeExt} \cite{errorFreeValidated})$ achieves linear latency. In the static full information adversary model, $\name(\mathsf{GPV} \cite{GPV06})$ can achieve $\log(n)$ latency (and even $O(1)$ latency with a minor logarithmic resiliency concession), outperforming its asynchronous counterpart \cite{KKKSS08,KKKSS10} by an exponent of $27$. Furthermore, in the adaptive model with private channels, $\name(\mathsf{AC} \cite{AC23})$ achieves optimal $n/3$ resiliency, $O(1)$ latency, and cubic communication complexity, outperforming any fully asynchronous solution in the same model with comparable (near)-optimal resiliency by a cubic multiplicative factor.

\subsection{Extension to Agreement on a Core Set}
\label{sec:acs_extension}

In \Cref{subsec:from_ole_to_acs}, we also go over how the (randomized) \name transformation can be adapted to solve Agreement on a Core Set (ACS)~\cite{AAPS23,ben1993asynchronous,ben1994asynchronous,cohen2023concurrent,das2024asynchronous,duan2023practical,shoup2024theoretical} with quadratic per-process bit-complexity. This approach is particularly relevant when applied to \cite{AsharovChandramouli24} and the near-optimally resilient version of \cite{GPV06}, as it enables achieving ACS instead of SMVBA without incurring any additional cost. Results are summarized in \Cref{table:state_of_the_art_protocols_acs} and allow a fairer comparison.

\section{Conclusion}
\label{section:conclusion}


This paper introduces \name, the first generic transformation of deterministic Byzantine agreement algorithms from synchrony to partial synchrony.
\name requires no cryptography, is optimally resilient ($n \geq 3t+1$, where $t$ is the maximum number of failures), and preserves the worst-case per-process bit complexity of the transformed synchronous algorithm.
Leveraging \name, we present the first partially synchronous Byzantine agreement algorithm that (1) achieves optimal $O(n^2)$ bit complexity, (2) requires no cryptography, and (3) is optimally resilient ($n \geq 3t+1$), thus showing that the Dolev-Reischuk bound is tight even in partial synchrony.
By adapting \name for long values, we obtain several new partially synchronous algorithms with improved complexity and weaker (or completely absent) cryptographic assumptions.
\changed{We also show how to adapt the transformation to the randomized setting, immediately yielding new partially synchronous results (presented in Table~\ref{table:state_of_the_art_protocols} in \Cref{sec:consequences}) that can offer a good alternative to their asynchronous counterparts that evolve in a too challenging model, without sacrificing correctness due to periods of asynchrony.}
Indirectly, \name contradicts the folklore belief that there is a fundamental gap between synchronous and partially synchronous agreement protocols.
We show that there is no inherent trade-off between the robustness of partially synchronous agreement protocols on the one hand, and the simplicity and efficiency of synchronous ones on the other hand. 
Concretely, we prove that partially synchronous algorithms can be automatically derived from synchronous ones, combining thereby simplicity, efficiency, and robustness. 

Interesting future research directions include:
\begin{compactitem}
    \item Achieving adaptive latency (e.g., $O(f)$, where $f$ is the \emph{actual} number of failures in an execution) while preserving the worst-case bit complexity.

    \item Modifying our \name transformation to attain optimistic responsiveness that allows decisions to be made at the actual speed of the network rather than at the pessimistic bound $\delta$.

    \item Improving the results for long values, e.g., by finding a worst-case bit-optimal (for long values) error-free Byzantine agremeent algorithm for $n \geq 3t+1$ (which would dominate all other solutions for long values).
    We emphasize that, given our \name transformation, achieving optimal bit complexity in Byzantine agreement requires focusing solely on solving the graded consensus and validation broadcast primitives with optimal bit complexity in the error-free setting with $n \geq 3t + 1$.
\end{compactitem}



\newpage
\appendix
\section*{APPENDIX}
The appendix can be separated into three parts.
In the first part, we formally prove the correctness and complexity of \block and \name.
Concretely, we provide a formal proof of correctness and complexity for \block in \Cref{section:crux_correctness_complexity_formal}.
Then, we give the pseudocode of \name and prove its correctness and complexity in \Cref{section:name_correctness_complexity_formal}.

The second part of the appendix focuses on the concrete implementations of graded consensus and validation broadcast that we employ in our \name transformation (concretely, in \block).
We review the existing primitives we utilize in \Cref{section:existing_primitives}.
Then, we introduce \reduceacool (\Cref{section:acool_reduction}), a primitive inspired by the A-COOL Byzantine algorithm protocol~\cite{li2021communication}; \reduceacool plays an important role in our graded consensus and validation broadcast implementations.
Next, we define and implement the rebuilding broadcast primitive (\Cref{section:rebuilding_broadccast}), another primitive that allows us to efficiently implement graded consensus and validation broadcast.
Finally, we give our implementations of graded consensus (\Cref{section:graded_consensus_concrete_implementations}) and validation broadcast (\Cref{section:validation_broadcast_concrete_implementations}).

The last part of the appendix shows how \block and \name can be used to solve other variants of Byzantine agreement. Specifically, in \Cref{subsec:from_ole_to_acs}, we show how agreement on a core set (ACS), arguably the strongest variant of Byzantine agreement, can be achieved efficiently with \block and \name.

\section{\block's Correctness \& Complexity: Formal Proof}\label{section:crux_correctness_complexity_formal}

In this section, we provide formal proof of \block's correctness and complexity.

\subsection{Review of the Specification of \block}

For the reader’s convenience, we first review the specification of \block.
Two durations parameterize \block's specification: (1) $\Delta_{\mathit{shift}}$, and (2) $\Delta_{\mathit{total}} > \Delta_{\mathit{shift}}$.
Moreover, each process $p_i$ is associated with its default value $\mathsf{def}(p_i)$ such that $\mathsf{valid}(\mathsf{def}(p_i)) = \mathit{true}$.
\block exposes the following interface:
\begin{compactitem}
    \item \textbf{request} $\mathsf{propose}(v \in \mathsf{Value})$: a process proposes value $v$.

    \item \textbf{request} $\mathsf{abandon}$: a process abandons (i.e., stops participating in) \block.

    \item \textbf{indication} $\mathsf{validate}(v' \in \mathsf{Value})$: a process validates value $v'$.

    \item \textbf{indication} $\mathsf{decide}(v' \in \mathsf{Value})$: a process decides value $v'$.

    \item \textbf{indication} $\mathsf{completed}$: a process is notified that \block has completed.
\end{compactitem}
Every correct process proposes to \block at most once and it does so with a valid value.
Observe that it is not guaranteed that all correct processes propose to \block.

The following properties are satisfied by \block:
\begin{compactitem}
    \item \emph{Strong validity:} If all correct processes that propose do so with the same value $v$, then no correct process validates or decides any value $v' \neq v$.
    
    \item \emph{External validity:} If any correct process decides or validates any value $v$, then $\mathsf{valid}(v) = \mathit{true}$.

    


    \item \emph{Agreement:} If any correct process decides a value $v$, then no correct process decides or validates any value $v' \neq v$.
    
    \item \emph{Integrity:} No correct process decides or receives a $\mathsf{completed}$ indication unless it has proposed.

    \item \emph{Termination:} If all correct processes propose and no correct process abandons \block, then every correct process eventually receives a $\mathsf{completed}$ indication.

    \item \emph{Totality:} If any correct process receives a $\mathsf{completed}$ indication at some time $\tau$, then every correct process validates a value by time $\max(\tau, \text{GST}) + 2\delta$.
    
    \item \emph{Synchronicity:} Let $\tau$ denote the first time a correct process proposes to \block.
    If (1) $\tau \geq \text{GST}$, (2) all correct processes propose by time $\tau + \Delta_{\mathit{shift}}$, and (3) no correct process abandons \block by time $\tau + \Delta_{\mathit{total}}$, then every correct process decides by time $\tau + \Delta_{\mathit{total}}$.

    \item \emph{Completion time:} If a correct process $p_i$ proposes at some time $\tau \geq \text{GST}$, then $p_i$ does not receive a $\mathsf{completed}$ indication by time $\tau + \Delta_{\mathit{total}}$.
\end{compactitem}

\subsection{Proof of Correctness \& Complexity}
We now prove \block's correctness and complexity.

\smallskip
\noindent \textbf{Proof of correctness.}
First, we prove that correct processes propose only valid values to $\mathcal{GC}_1$ and $\mathcal{GC}_2$.
Recall that $\mathcal{GC}_1$ and $\mathcal{GC}_2$ are two instances of the graded consensus primitive (see \Cref{subsection:graded_consensus_building_blocks}) utilized by \block (see \Cref{algorithm:block}).

\begin{lemma} \label{lemma:crux_valid_only_1}
Let $p_i$ be any correct process that proposes a value $v$ to $\mathcal{GC}_1$ or $\mathcal{GC}_2$ in \block (\Cref{algorithm:block}).
Then, $\mathsf{valid}(v) = \mathit{true}$.
\end{lemma}
\begin{proof}
To prove the lemma, we consider all possible cases:
\begin{compactitem}
    \item Let $p_i$ propose $v$ to $\mathcal{GC}_1$ (Step 1 of Task 1).
    Therefore, $p_i$ has previously proposed $v$ to \block.
    Due to the assumption that correct processes only propose valid values to \block, $v$ is valid.


    \item Let $p_i$ propose $v$ to $\mathcal{GC}_2$ (Step 4 of Task 1).
    Hence, $v$ is the value of $p_i$'s local variable $\mathit{est}_i$ updated in Step 3 of Task 1.
    Let us investigate all possible scenarios for $v$ according to Step 3 of Task 1:
    \begin{compactitem}
        \item Let $v$ be the value decided from $\mathcal{GC}_1$.
        In this case, the external validity property of $\mathcal{GC}_1$ guarantees that $v$ is valid.

        \item Let $v$ be the value decided from the $\mathcal{A}^S$ instance of synchronous Byzantine agreement.
        In this case, $p_i$ explicitly checks that $v$ is valid before assigning $v$ to $\mathit{est}_i$.

        \item Let $v$ be $p_i$'s proposal to \block.
        Here, $v$ is valid due to the assumption that no correct process proposes an invalid value to \block.
    \end{compactitem}
\end{compactitem}
The lemma holds as its statement is true for all possible cases.
\end{proof}

\Cref{lemma:crux_valid_only_1} proves that $\mathcal{GC}_1$ and $\mathcal{GC}_2$ behave according to their specification as correct processes indeed propose only valid values to them.
Next, we prove a direct consequence of \Cref{lemma:crux_valid_only_1}: any correct process broadcasts only valid values via the $\mathcal{VB}$ instance of the validation broadcast primitive (see \Cref{subsection:validation_broadcast_building_blocks}).

\begin{lemma} \label{lemma:crux_valid_only_2}
Let $p_i$ be any correct process that broadcasts a value $v$ via $\mathcal{VB}$ in \block (\Cref{algorithm:block}).
Then, $\mathsf{valid}(v) = \mathit{true}$.
\end{lemma}
\begin{proof}
As $p_i$ broadcasts $v$ via $\mathcal{VB}$ (Step 6 of Task 1), $p_i$ has previously decided $v$ from $\mathcal{GC}_2$ (Step 4 of Task 1).
As $\mathcal{GC}_2$ satisfies external validity (due to \Cref{lemma:crux_valid_only_1}), $v$ is valid.
\end{proof}

Note that \Cref{lemma:crux_valid_only_2} shows that $\mathcal{VB}$ behaves according to its specification.
The following theorem proves that \block satisfies strong validity.

\begin{theorem} [Strong validity] \label{theorem:block_strong_validity}
\block (\Cref{algorithm:block}) satisfies strong validity.
\end{theorem}
\begin{proof}
Suppose all correct processes that propose to \block do so with the same value denoted by $v$.
This implies that all correct processes that propose to $\mathcal{GC}_1$ do propose value $v$ (Step 1 of Task 1).
Hence, due to the strong validity property of $\mathcal{GC}_1$, every correct process that decides from $\mathcal{GC}_1$ decides $(v, 1)$.
Therefore, every correct process $p_i$ that reaches Step 3 of Task 1 sets its $\mathit{est}_i$ local variable to $v$, and proposes $v$ to $\mathcal{GC}_2$ (Step 4 of Task 1).
The strong validity property of $\mathcal{GC}_2$ further ensures that every correct process that decides from $\mathcal{GC}_2$ does decide $(v, 1)$, which implies that no correct process decides any value $v' \neq v$ from \block (Step 5 of Task 1).
Furthermore, every correct process that broadcasts using $\mathcal{VB}$ does broadcast value $v$ (Step 6 of Task 1).
Due to the strong validity property of $\mathcal{VB}$, no correct process validates any value $v' \neq v$ (Step 1 of Task 3), thus ensuring strong validity.
\end{proof}

Next, we prove \block's external validity.

\begin{theorem} [External validity] \label{theorem:block_external_validity}
\block (\Cref{algorithm:block}) satisfies external validity.
\end{theorem}
\begin{proof}
Let a correct process decide a value $v$ from \block (Step 5 of Task 1).
Hence, that process has previously decided $v$ from $\mathcal{GC}_2$.
Due to the external validity property of $\mathcal{GC}_2$ (ensured by \Cref{lemma:crux_valid_only_1}), $v$ is a valid value.

If a correct process $p_i$ validates a value $v'$ (Step 1 of Task 3), the process has previously validated $v'$ from $\mathcal{VB}$.
There are two possibilities to analyze according to $\mathcal{VB}$'s safety property:
\begin{compactitem}
    \item Let $v' = \mathsf{def}(p_i)$. 
    In this case, $v' = \mathsf{def}(p_i)$ is valid due to the assumption that $\mathsf{valid}(\mathsf{def}(p_i)) = \mathit{true}$.

    \item Otherwise, $v'$ has been broadcast via $\mathcal{VB}$ by a correct process.
    In this case, $\mathsf{valid}(v') = \mathit{true}$ by \Cref{lemma:crux_valid_only_2}.
\end{compactitem}
The theorem holds.
\end{proof}

The following theorem shows \block's agreement.

\begin{theorem} [Agreement] \label{theorem:block_agreement}
\block (\Cref{algorithm:block}) satisfies agreement.
\end{theorem}
\begin{proof}
No two correct processes decide different values from \block (Step 5 of Task 1) due to the consistency property of $\mathcal{GC}_2$.
Moreover, if a correct process decides some value $v$ from \block (Step 5 of Task 1), every correct process that decides from $\mathcal{GC}_2$ (Step 4 of Task 1) does so with value $v$ (due to the consistency property of $\mathcal{GC}_2$).
Thus, every correct process that broadcasts via $\mathcal{VB}$ does so with value $v$ (Step 6 of Task 1).
Due to the strong validity property of $\mathcal{VB}$, no correct process validates any value $v' \neq v$ from $\mathcal{VB}$.
Hence, no correct process validates any non-$v$ value from \block (Step 1 of Task 3).
\end{proof}

Next, we prove that \block satisfies integrity.

\begin{theorem} [Integrity] \label{theorem:block_integrity}
\block (\Cref{algorithm:block}) satisfies integrity.
\end{theorem}
\begin{proof}
Any correct process $p_i$ that decides or completes \block does so while executing Task 1. 
As $p_i$ starts executing Task 1 only after it has proposed to \block, the integrity property is satisfied.
\end{proof}

The following theorem proves \block's termination.

\begin{theorem} [Termination] \label{theorem:block_termination}
\block (\Cref{algorithm:block}) satisfies termination.
\end{theorem}
\begin{proof}
Let all correct processes propose to \block and let no correct process ever abandon \block.
Hence, every correct process proposes to $\mathcal{GC}_1$ (Step 1 of Task 1), and no correct process ever abandons it.
This implies that every correct process eventually decides from $\mathcal{GC}_1$ (due to the termination property of $\mathcal{GC}_1$), and proposes to $\mathcal{A}^S$ (Step 2 of Task 1).
As every correct process executes $\mathcal{A}^S$ for a limited time only (i.e., for exactly $\mathcal{R}$ rounds of finite time), every correct process eventually concludes Step 2 of Task 1.
Therefore, every correct process proposes to $\mathcal{GC}_2$ (Step 4 of Task 1), and no correct process ever abandons it.
Hence, the termination property of $\mathcal{GC}_2$ ensures that every correct process eventually decides from $\mathcal{GC}_2$, which implies that every correct process broadcasts its decision via $\mathcal{VB}$ (Step 6 of Task 1).
Lastly, as no correct process ever abandons $\mathcal{VB}$, every correct process eventually receives a $\mathsf{completed}$ indication from $\mathcal{VB}$ (Step 6 of Task 1) and completes \block (Step 7 of Task 1).
\end{proof}

Next, we prove \block's totality.

\begin{theorem} [Totality] \label{theorem:block_totality}
\block (\Cref{algorithm:block}) satisfies totality.
\end{theorem}
\begin{proof}
Suppose a correct process receives a $\mathsf{completed}$ indication from \block at some time $\tau$ (Step 7 of Task 1).
Hence, that correct process has previously received a $\mathsf{completed}$ indication from $\mathcal{VB}$ at time $\tau$ (Step 6 of Task 1).
Therefore, the totality property of $\mathcal{VB}$ ensures that every correct process validates a value from $\mathcal{VB}$ by time $\max(\tau, \text{GST}) + 2\delta$.
Therefore, every correct process validates a value from \block by time $\max(\tau, \text{GST}) + 2\delta$ (Step 1 of Task 3).
\end{proof}

The theorem below proves the completion time property of \block.

\begin{theorem} [Completion time] \label{theorem:block_completion_time}
\block (\Cref{algorithm:block}) satisfies completion time.
\end{theorem}
\begin{proof}
Let $p_i$ be any correct process that proposes to \block at some time $\tau \geq \text{GST}$.
As $\tau \geq \text{GST}$, $p_i$'s local clock does not drift (see \Cref{section:preliminaries}).
Process $p_i$ does not complete Step 1 of Task 1 by time $\tau + ( \Delta_{\mathit{shift}} + \Delta_1 )$.
Similarly, $p_i$ does not complete Step 2 of Task 1 by time $\tau + ( \Delta_{\mathit{shift}} + \Delta_1 ) + ( \mathcal{R} \cdot \Delta_{\mathit{sync}})$.
Lastly, $p_i$ does not complete Step 4 of Task 1 by time $\tau + ( \Delta_{\mathit{shift}} + \Delta_1 ) + ( \mathcal{R} \cdot \Delta_{\mathit{sync}} ) + ( \Delta_{\mathit{shift}} + \Delta_2 ) = \tau + \Delta_{\mathit{total}}$.
Hence, the earliest time at which $p_i$ broadcasts via $\mathcal{VB}$ (Step 6 of Task 1) is $\tau' > \tau + \Delta_{\mathit{total}}$.
Thus, due to the integrity property of $\mathcal{VB}$, $p_i$ cannot receive a $\mathsf{completed}$ indication from $\mathcal{VB}$ (and, thus, from \block at Step 1 of Task 3) before time $\tau' > \tau + \Delta_{\mathit{total}}$, which proves \block's completion time property.
\end{proof}

Lastly, we need to prove the synchronicity property of \block.
First, we explicitly state how processes simulate a synchronous agreement algorithm $\mathcal{A}^S$ in \block.
Concretely, we propose two simulation approaches: (1) $\mathsf{CryptoFreeSim}$, when $\mathcal{A}^S$ is cryptography-free, which is conceptually simpler, and (2) $\mathsf{CryptoSim}$, when $\mathcal{A}^S$ is cryptography-based, which is more general.

\smallskip
\noindent \textbf{Simulating cryptography-free $\mathcal{A}^S$.}
We explicitly define our simulation $\mathsf{CryptoFreeSim}$ in \Cref{algorithm:crypto_free_sim}.
As mentioned in \block's pseudocode (\Cref{algorithm:block}), $\mathsf{CryptoFreeSim}$ roughly works as follows.
(1) A correct process $p_i$ runs each simulated round for exactly $\Delta_{\mathit{sync}} = \Delta_{\mathit{shift}} + \delta$ time.
(2) If process $p_i$ sends a message $m$ in an even (resp., odd) round $r$ of $\mathcal{A}^S$, then $p_i$ appends parity bit $0$ (resp., $1$) to $m$ in the simulated round $r$.
(3) Process $p_i$ executes exactly $\mathcal{R}$ simulated rounds; recall that $\mathcal{R}$ denotes the number of rounds $\mathcal{A}^S$ takes to terminate when run in synchrony.
(4) Process $p_i$ does not send more than $2\mathcal{B}$ bits; recall that $\mathcal{B}$ is the maximum number of bits any correct process sends when $\mathcal{A}^S$ is run in synchrony.

\begin{algorithm}
\caption{$\mathsf{CryptoFreeSim}$: Pseudocode (for process $p_i$)}
\label{algorithm:crypto_free_sim}
\footnotesize
\begin{algorithmic} [1] 
\State \textbf{Local variables:}
\State \hskip2em $\mathsf{Local\_State}$ $s_i \gets$ the initial state corresponding to $p_i$'s proposal to $\mathcal{A}^S$
\State \hskip2em $\mathsf{Integer}$ $\mathit{round}_i \gets 1$
\State \hskip2em $\mathsf{Integer}$ $\mathit{sent\_bits}_i \gets 0$
\State \hskip2em $\mathsf{Set}(\mathsf{Message})$ $\mathit{received}_i \gets \emptyset$ \BlueComment{received messages are stored here}

\medskip
\State \textbf{while} $\mathit{round}_i \leq \mathcal{R}$:
\State \hskip2em \textbf{for each} $\mathsf{Process}$ $p_j$:
\State \hskip4em let $M_j \gets$ the messages $\mathcal{A}^S$ instructs $p_i$ to send to $p_j$ when $p_i$'s local state is $s_i$
\State \hskip4em let $B_j \gets$ the number of bits in $M_j$
\State \hskip4em \textbf{if} $M_j \neq \bot$: \BlueComment{there exists a message to be sent to $p_j$}
\State \hskip6em \textbf{if} $\mathit{sent\_bits}_i + B_j \leq 2\mathcal{B}$:   \label{line:check_sent_messages_crypto_free} \BlueComment{$p_i$ can still send messages}
\State \hskip8em send $\langle \mathit{round}_i \text{ mod } 2, M_j \rangle$ to $p_j$ \BlueComment{$p_i$ sends $M_j$ to process $p_j$ with the parity bit $\mathit{round}_i \text{ mod } 2$}
\State \hskip8em $\mathit{sent\_bits}_i \gets \mathit{sent\_bits}_i + B_j$

\medskip
\State \hskip2em \textbf{wait for} $\Delta_{\mathit{sync}} = \Delta_{\mathit{shift}} + \delta$ time

\medskip
\State \hskip2em let $\mathit{received\_current\_round} \gets $ every message $m$ that belongs to $\mathit{received}_i$ with the parity bit $\mathit{round}_i \text{ mod } 2$
\State \hskip2em $s_i \gets$ the state $\mathcal{A}^S$ instructs $p_i$ to transit to based on (1) $p_i$'s previous state $s_i$, and (2) $\mathit{received\_current\_round}$ 
\State \hskip2em $\mathit{round}_i \gets \mathit{round}_i + 1$
\end{algorithmic} 
\end{algorithm}

In the rest of the proof, we say that ``$\mathcal{S}^*$ holds'' if and only if (1) the first correct process that starts $\mathsf{CryptoFreeSim}$ does so at some time $\tau^* \geq \text{GST}$, (2) all correct processes start $\mathsf{CryptoFreeSim}$ by time $\tau^* + \Delta_{\mathit{shift}}$, and (3) no correct process stops $\mathsf{CryptoFreeSim}$ by time $\tau^* + \mathcal{R}(\Delta_{\mathit{shift}} + \delta) = \tau^* + \mathcal{R} \cdot \Delta_{\mathit{sync}}$.
The following lemma proves that $\mathsf{CryptoFreeSim}$ indeed simulates $\mathcal{A}^S$ when $\mathcal{S}^*$ holds.

\begin{lemma} [$\mathsf{CryptoFreeSim}$ simulates $\mathcal{A}^S$] \label{lemma:cryptography_free_simulation_correct}
Let $\mathcal{S}^*$ hold.
For each execution $\mathcal{E}$ of $\mathsf{CryptoFreeSim}$, there exists an $\mathcal{R}$-rounds-long synchronous execution $\mathcal{E}'$ of $\mathcal{A}^S$ such that:
\begin{compactitem}
    \item the sets of correct processes in $\mathcal{E}$ and $\mathcal{E}'$ are identical, and

    \item the proposals of correct processes in $\mathcal{E}$ and $\mathcal{E}'$ are identical, and

    \item the sets of messages sent by correct processes in $\mathcal{E}$ and $\mathcal{E}'$ are identical (modulo the parity bits), and

    \item for each correct process $p_i$ and every $k \in [1, \mathcal{R} + 1]$, $s_i^k(\mathcal{E}) = s_i^k(\mathcal{E}')$, where (1) $s_i^k(\mathcal{E})$ is the state of $p_i$ at the beginning of the $k$-th (i.e., at the end of the $(k - 1)$-st) simulated round in $\mathcal{E}$, and (2) $s_i^k(\mathcal{E}')$ is the state of $p_i$ at the beginning of the $k$-th (i.e., at the end of the $(k - 1)$-st) round in $\mathcal{E}'$.
\end{compactitem}    
\end{lemma}
\begin{proof}
To prove the lemma, we go through a sequence of intermediate results.

\medskip
\noindent \emph{Intermediate result 1: Let $\mathsf{CryptoFreeSim}^-$ be identical to $\mathsf{CryptoFreeSim}$ except that correct processes are allowed to send any number of bits (i.e., the check at line~\ref{line:check_sent_messages_crypto_free} is removed).
Moreover, let the condition $\mathcal{S}^*$ be adapted to $\mathsf{CryptoFreeSim}^-$.
Then, the lemma holds for $\mathsf{CryptoFreeSim}^-$.}
\\ The result is proven in~\cite[Theorem 4.1]{LS22}.

\medskip
\noindent \emph{Intermediate result 2: Let $\mathsf{CryptoFreeSim}^-$ be identical to $\mathsf{CryptoFreeSim}$ except that correct processes are allowed to send any number of bits (i.e., the check at line~\ref{line:check_sent_messages_crypto_free} is removed).
Moreover, let the condition $\mathcal{S}^*$ be adapted to $\mathsf{CryptoFreeSim}^-$.
Then, no correct process sends more than $2\mathcal{B}$ bits in any execution $\mathcal{E}^-$ of $\mathsf{CryptoFreeSim}^-$ when $\mathcal{S}^*$ holds.}
\\ By contradiction, suppose there exists an execution $\mathcal{E}^-$ of $\mathsf{CryptoFreeSim}^-$ in which some correct process $p_i$ sends more than $2\mathcal{B}$ bits.
The first intermediate result proves that $\mathcal{E}^-$ simulates an execution $\mathit{sim}(\mathcal{E}^-)$ of $\mathcal{A}^S$.
Hence, a message $m$ is sent by $p_i$ in $\mathcal{E}^-$ if and only if a message $m'$ is sent by $p_i$ in $\mathit{sim}(\mathcal{E}^-)$ such that $|m| = |m'| + 1$, where $|m|$ (resp., $|m'|$) denotes the bit-size of message $m$ (resp., $m'$).
As each sent message contains at least a single bit, $p_i$ sends at most $\mathcal{M} \leq \mathcal{B}$ messages in $\mathit{sim}(\mathcal{E}^-)$.
Therefore, process $p_i$ can send at most $\mathcal{M} \leq \mathcal{B}$ parity bits in $\mathcal{E}^-$ (not sent in $\mathit{sim}(\mathcal{E}^-)$).
Thus, it is impossible for $p_i$ to send more than $2\mathcal{B}$ bits in $\mathcal{E}^-$.

\medskip
\noindent \emph{Proof of \Cref{lemma:cryptography_free_simulation_correct}.}
To prove that $\mathsf{CryptoFreeSim}$ correctly simulates $\mathcal{A}^S$ when $\mathcal{S}^*$ holds, it suffices to show that $\mathsf{CryptoFreeSim} \equiv \mathsf{CryptoFreeSim}^-$ as the lemma would follow from the first intermediate result, where $\mathsf{CryptoFreeSim}^-$ is defined above.
By contradiction, suppose $\mathsf{CryptoFreeSim} \not\equiv \mathsf{CryptoFreeSim}^-$ when $\mathcal{S}^*$ holds.
This is only possible if there exists an execution $\mathcal{E}$ of $\mathsf{CryptoFreeSim}$ in which a correct process does not send some message $m$ it was supposed to send according to $\mathcal{A}^S$ because the sending would exceed the $2\mathcal{B}$ bits limit.
However, this implies that there exists an execution of $\mathsf{CryptoFreeSim}^-$ in which this correct process does send more than $2\mathcal{B}$ bits, which represents a contradiction with the second intermediate result.
Therefore, $\mathsf{CryptoFreeSim} \equiv \mathsf{CryptoFreeSim}^-$ when $\mathcal{S}^*$ holds.
\end{proof}

\smallskip
\noindent \textbf{Simulating cryptography-based $\mathcal{A}^S$.}
$\mathsf{CryptoSim}$ (\Cref{algorithm:simulation}) represents our simulation of a cryptography-based synchronous algorithm $\mathcal{A}^S$ (Step 2 of Task 1).
Importantly, when $\mathsf{CryptoSim}$ is utilized in \block, $\Delta_{\mathit{sync}} = 2\Delta_{\mathit{shift}} + \delta$.

\begin{algorithm}
\caption{$\mathsf{CryptoSim}$: Pseudocode (for process $p_i$)}
\label{algorithm:simulation}
\footnotesize
\begin{algorithmic} [1] 
\State \textbf{Local variables:}
\State \hskip2em $\mathsf{Local\_State}$ $s_i \gets$ the initial state corresponding to $p_i$'s proposal to $\mathcal{A}^S$
\State \hskip2em $\mathsf{Integer}$ $\mathit{round}_i \gets 1$
\State \hskip2em $\mathsf{Integer}$ $\mathit{sent\_bits}_i \gets 0$
\State \hskip2em $\mathsf{Set}(\mathsf{Message})$ $\mathit{received}_i \gets \emptyset$

\medskip
\State \textbf{while} $\mathit{round}_i \leq \mathcal{R}$:
\State \hskip2em $\mathit{received}_i \gets \emptyset$
\State \hskip2em \textbf{measure} $\Delta_{\mathit{sync}} = 2\Delta_{\mathit{shift}} + \delta$ time 
\State \hskip2em for every message $m$ received in the following $\Delta_{\mathit{sync}} = 2\Delta_{\mathit{shift}} + \delta$ time period, add $m$ to $\mathit{received}_i$ \label{line:receive_simulation}

\medskip
\State \hskip2em \textbf{wait for} $\Delta_{\mathit{shift}}$ time
\State \hskip2em \textbf{for each} $\mathsf{Process}$ $p_j$:
\State \hskip4em let $M_j \gets$ the messages $\mathcal{A}^S$ instructs $p_i$ to send to $p_j$ when $p_i$'s local state is $s_i$
\State \hskip4em let $B_j \gets$ the number of bits in $M_j$
\State \hskip4em \textbf{if} $M_j \neq \bot$: \BlueComment{there exists a message to be sent to $p_j$}
\State \hskip6em \textbf{if} $\mathit{sent\_bits}_i + B_j \leq \mathcal{B}$:  \BlueComment{$p_i$ can still send messages} \label{line:check_sent_bits}
\State \hskip8em send $\langle M_j \rangle$ to $p_j$ \label{line:send_simulation} \BlueComment{$p_i$ sends $M_j$ to process $p_j$}
\State \hskip8em $\mathit{sent\_bits}_i \gets \mathit{sent\_bits}_i + B_j$

\medskip
\State \hskip2em \textbf{upon} the measured $\Delta_{\mathit{sync}} = 2\Delta_{\mathit{shift}} + \delta$ time elapses:
\State \hskip4em $s_i \gets$ the state $\mathcal{A}^S$ instructs $p_i$ to transit to based on (1) $p_i$'s previous state $s_i$, and (2) $\mathit{received}_i$ \label{line:update_state}
\State \hskip4em $\mathit{round}_i \gets \mathit{round}_i + 1$
\end{algorithmic} 
\end{algorithm}

As we did for $\mathsf{CryptoFreeSim}$, we say that ``$\mathcal{S}^*$ holds'' if and only if (1) the first correct process that starts $\mathsf{CryptoSim}$ does so at some time $\tau^* \geq \text{GST}$, (2) all correct processes start $\mathsf{CryptoSim}$ by time $\tau^* + \Delta_{\mathit{shift}}$, and (3) no correct process stops $\mathsf{CryptoSim}$ by time $\tau^* + \mathcal{R} \cdot \Delta_{\mathit{sync}}$ (recall that $\Delta_{\mathit{sync}} = 2\Delta_{\mathit{shift}} + \delta$).
The following lemma is crucial in proving that $\mathsf{CryptoSim}$ successfully simulates a synchronous algorithm $\mathcal{A}^S$.

\begin{lemma} \label{lemma:simulation_correct}
Let $\mathsf{CryptoSim}^-$ be identical to $\mathsf{CryptoSim}$ except that correct processes are allowed to send any number of bits (i.e., the check at line~\ref{line:check_sent_bits} is removed).
Moreover, let $\mathcal{S}^*$ hold for $\mathsf{CryptoSim}^-$.
For each execution $\mathcal{E}$ of $\mathsf{CryptoSim}^-$, there exists an $\mathcal{R}$-rounds-long synchronous execution $\mathcal{E}'$ of $\mathcal{A}^S$ such that:
\begin{compactitem}
    \item the sets of correct processes in $\mathcal{E}$ and $\mathcal{E}'$ are identical, and

    \item the proposals of correct processes in $\mathcal{E}$ and $\mathcal{E}'$ are identical, and

    \item the sets of messages sent by correct processes in $\mathcal{E}$ and $\mathcal{E}'$ are identical, and

    \item for each correct process $p_i$ and every $k \in [1, \mathcal{R} + 1]$, $s_i^k(\mathcal{E}) = s_i^k(\mathcal{E}')$, where (1) $s_i^k(\mathcal{E})$ is the state of $p_i$ at the beginning of the $k$-th (i.e., at the end of the $(k - 1)$-st) simulated round in $\mathcal{E}$, and (2) $s_i^k(\mathcal{E}')$ is the state of $p_i$ at the beginning of the $k$-th (i.e., at the end of the $(k - 1)$-st) round in $\mathcal{E}'$.
\end{compactitem}
\end{lemma}
\begin{proof}
Let $\mathcal{C}$ denote the set of correct processes in $\mathcal{E}$.
Recall that $\tau^*$ denotes the time the first correct process starts $\mathsf{CryptoSim}^-$.
For each process $p_i \in \mathcal{C}$, we introduce the following notation:
\begin{compactitem}
    \item Let $\tau_i$ denote the time at which $p_i$ starts executing \Cref{algorithm:simulation}; as $\mathcal{S}^*$ holds, $\tau_i \in [\tau^*, \tau^* + \Delta_{\mathit{shift}}]$.
    Moreover, $\tau_i \in [\tau_j - \Delta_{\mathit{shift}}, \tau_j + \Delta_{\mathit{shift}}]$, for any process $p_j \in \mathcal{C}$.


    \item For every $k \in [1, \mathcal{R}]$, let $s_i^k$ denote the value of $p_i$'s local variable $s_i$ at the beginning of $k$-th iteration $k$ of the while loop in $\mathcal{E}$.
    Moreover, let $s_i^{\mathcal{R} + 1}$ denote the value of $p_i$'s local variable $s_i$ at the end of $\mathcal{R}$-th iteration of the while loop (after executing line~\ref{line:update_state}).

    \item For every $k \in [1, \mathcal{R}]$, let $\mathit{sent}_i(k)$ denote the set of messages $m$ such that $k$ is the value of the $\mathit{round}_i$ variable when $p_i$ sends $m$ in $\mathcal{E}$ (line~\ref{line:send_simulation}).
    Importantly, $m \in \mathit{sent}_i(k)$ if and only if $m$ is sent at time $\tau_i + (k - 1)\Delta_{\mathit{sync}} + \Delta_{\mathit{shift}}$ in $\mathcal{E}$; recall that, as $\tau_i \geq \tau^* \geq \text{GST}$ (since $\mathcal{S}^*$ holds), the local clocks of processes do not drift.

    \item For every $k \in [1, \mathcal{R}]$, let $\mathit{received}_i(k)$ denote the set of messages $m$ such that $k$ is the value of the $\mathit{round}_i$ variable when $p_i$ receives $m$ in $\mathcal{E}$ (line~\ref{line:receive_simulation}).
    Observe that $m \in \mathit{received}_i(k)$ if and only if $m$ is received by $p_i$ during the time period $\mathcal{T}_i^k = [\tau_i + (k - 1)\Delta_{\mathit{sync}}, \tau_i + k \cdot \Delta_{\mathit{sync}}]$ in $\mathcal{E}$.\footnote{
    For the sake of simplicity, and without loss of generality, we assume that when $\mathcal{S}^{*}$ holds, no correct process $p_i$ can receive a message sent by another correct process exactly at time $\tau_i + k' \cdot \Delta_{\mathit{sync}}$, where $k'$ is an integer. 
    To satisfy this assumption, we can define $\Delta_{\mathit{sync}} = 2\Delta_{\mathit{shift}} + \delta + \epsilon$, for any arbitrarily small constant $\epsilon>0$.
    (We avoid doing so for the simplicity of presentation.)} 
    
\end{compactitem}

We construct $\mathcal{E}'$ in the following way:
\begin{compactenum}
    \item For every process $p_i \in \mathcal{C}$, we perform the following steps:
    \begin{compactenum}
        \item For every round $k \in [1, \mathcal{R}]$, $p_i$ starts round $k$ in state $s_i^k$.

        \item Process $p_i$ concludes round $\mathcal{R}$ in state $s_i^{\mathcal{R} + 1}$.
    
        \item For every round $k \in [1, \mathcal{R}]$, $p_i$ sends $\mathit{sent}_i(k)$ in round $k$.

        \item For every round $k \in [1, \mathcal{R}]$, $p_i$ receives $\mathit{received}_i(k)$ in round $k$. 
    \end{compactenum}

    \item For every process $p_j \notin C$, we perform the following steps:
    \begin{compactenum}
        \item For every message $m \in \mathit{sent}_i(k)$ with the receiver being $p_j$, for some process $p_i \in \mathcal{C}$ and some round $k \in [1, \mathcal{R}]$, $p_j$ receives $m$ in round $k$.

        \item For every message $m \in \mathit{received}_i(k)$ with the sender being $p_j$, for some process $p_i \in \mathcal{C}$ and some round $k \in [1, \mathcal{R}]$, $p_j$ sends $m$ in round $k$.
    \end{compactenum}
\end{compactenum}
Due to the construction of $\mathcal{E}'$, the statement of the lemma is indeed satisfied.
It is only left to prove that $\mathcal{E}'$ is a valid synchronous execution of $\mathcal{A}^S$.
To this end, we show that $\mathcal{E}'$ satisfies the properties of a valid synchronous execution:
\begin{compactitem}
    \item \emph{If a message $m$ is sent by a process in round $k$, then the message is received in round $k$.}
    
    Consider any message $m$ sent in some round $k$ of $\mathcal{E}'$.
    Let the sender of $m$ be denoted by $p_s$ and let the receiver of $m$ be denoted by $p_r$.
    We consider four possibilities:
    \begin{compactitem}
        \item Let $p_s \in \mathcal{C}$ and $p_r \in \mathcal{C}$.
        In this case, $p_s$ sends $m$ at time $\tau^S(m) = \tau_s + (k - 1)\Delta_{\mathit{sync}} + \Delta_{\mathit{shift}}$ in $\mathcal{E}$.
        Importantly, message $m$ reaches process $p_r$ by time $\tau_s + (k - 1)\Delta_{\mathit{sync}} + \Delta_{\mathit{shift}} + \delta$.
        As $\tau_s \leq \tau_r + \Delta_{\mathit{shift}}$, $\tau_s + (k - 1)\Delta_{\mathit{sync}} + \Delta_{\mathit{shift}} + \delta \leq \tau_r + k \cdot \Delta_{\mathit{sync}}$.
        Finally, as $k \leq \mathcal{R}$, $p_r$ indeed receives $m$ in $\mathcal{E}$ as $p_r$ does not stop executing $\mathsf{CryptoSim}^-$ by time $\tau_r + \mathcal{R} \cdot \Delta_{\mathit{sync}}$.
        
        Let $\tau^R(m) \in [\tau^S(m), \tau^S(m) + \delta]$ denote the time at which $p_r$ receives $m$ in $\mathcal{E}$.
        As $\tau_s \geq \tau_r - \Delta_{\mathit{shift}}$, $\tau^R(m) \geq \tau_r + (k - 1)\Delta_{\mathit{sync}}$.
        Similarly, as $\tau_s \leq \tau_r + \Delta_{\mathit{shift}}$, $\tau^R(m) \leq \tau_r + k \cdot \Delta_{\mathit{sync}}$.
        Therefore, $p_r$ receives $m$ in $\mathcal{E}$ during the time period $[\tau_r + (k - 1)\Delta_{\mathit{sync}}, \tau_r + k \cdot \Delta_{\mathit{sync}}]$, which proves that $m \in \mathit{received}_r(k)$.
        Thus, $m$ is indeed received in round $k$ of $\mathcal{E'}$ due to step 1d of the construction.
        
        \item Let $p_s \in \mathcal{C}$ and $p_r \notin \mathcal{C}$.
        Here, message $m$ is indeed received by $p_r$ in round $k$ of $\mathcal{E}'$ due to step 2a of the construction.

        \item Let $p_s \notin \mathcal{C}$ and $p_r \in \mathcal{C}$.
        As $m$ is sent in round $k$ of $\mathcal{E}'$, this is done due to step 2b of the construction.
        Hence, $m \in \mathit{received}_r(k)$.
        Therefore, step 1d ensures $m$'s reception in round $k$ of $\mathcal{E}'$.
        
        \item Let $p_s \notin \mathcal{C}$ and $p_r \notin \mathcal{C}$.
        This case is impossible as our construction (step 2b) dictates $p_s$ to send $m$ only if $p_r \in \mathcal{C}$.

    \end{compactitem}
    In any possible scenario, the property is satisfied.
    \smallskip
    
    \item \emph{If a message $m$ is received by a process in round $k$, then the message is sent in round $k$.}

    Consider any message $m$ sent in some round $k$ of $\mathcal{E}'$.
    Let the sender of $m$ be denoted by $p_s$ and let the receiver of $m$ be denoted by $p_r$.
    Let us distinguish four scenarios:
    \begin{compactitem}
        \item Let $p_s \in \mathcal{C}$ and $p_r \in \mathcal{C}$.
        As $m$ is received in round $k$ of $\mathcal{E}'$, $m \in \mathit{received}_r(k)$.
        This implies that $m$ is received at some time $\tau^R(m) \in [\tau_r + (k - 1)\Delta_{\mathit{sync}}, \tau_r + k \cdot \Delta_{\mathit{sync}}]$.
        Moreover, $m \in \mathit{sent}_s(k')$, for some $k' \in [1, \mathcal{R}]$.
        
        If $k' = k$, step 1c of our construction ensures that $m$ is indeed sent in round $k$ of $\mathcal{E}'$.
        By contradiction, let $k' \neq k$.
        Recall that process $p_s$ sends $m$ at time $\tau^S(m) = \tau_s + (k' - 1)\Delta_{\mathit{sync}} + \Delta_{\mathit{shift}}$.
        We separate two cases:
        \begin{compactitem}
            \item Let $k' < k$.
            First, note that $\tau^R(m) \leq \tau^S(m) + \delta$.
            Hence, $\tau^R(m) \leq \tau_s + (k' - 1)\Delta_{\mathit{sync}} + \Delta_{\mathit{shift}} + \delta$.
            As $\tau_s \leq \tau_r + \Delta_{\mathit{shift}}$, $\tau^R(m) \leq \tau_r + \Delta_{\mathit{shift}} + (k' - 1)\Delta_{\mathit{sync}} + \Delta_{\mathit{shift}} + \delta \leq \tau_r + k' \cdot \Delta_{\mathit{sync}}$.
            As $k' < k$, $p_r$ receives $m$ before entering the $k'$-th iteration of the while loop in $\mathcal{E}$, therefore proving that $m \notin \mathit{received}_r(k)$, which is a contradiction.

            \item Let $k' > k$.
            Observe that $\tau^R(m) \geq \tau^S(m)$.
            Thus, $\tau^R(m) \geq \tau_s + (k' - 1)\Delta_{\mathit{sync}} + \Delta_{\mathit{shift}}$.
            Since $\tau_s \geq \tau_r - \Delta_{\mathit{shift}}$, we have that $\tau^R(m) \geq \tau_r - \Delta_{\mathit{shift}} + (k' - 1)\Delta_{\mathit{sync}} + \Delta_{\mathit{shift}} = \tau_r + (k' - 1)\Delta_{\mathit{sync}} \geq \tau_r + k \cdot \Delta_{\mathit{sync}}$.
            Thus, we reach a contradiction that $m \in \mathit{received}_r(k)$.
        \end{compactitem}

        \item Let $p_s \in \mathcal{C}$ and $p_r \notin \mathcal{C}$.
        In this case, the property holds due to step 2a of the construction.

        \item Let $p_s \notin \mathcal{C}$ and $p_r \in \mathcal{C}$.
        In this case, the property holds due to step 2b of the construction.
        
        \item Let $p_s \notin \mathcal{C}$ and $p_r \notin \mathcal{C}$.
        This case cannot occur as process $p_s \notin \mathcal{C}$ only sends messages to processes in $\mathcal{C}$ (see step 2b of the construction).
    \end{compactitem}

    \medskip
    \item \emph{The local behavior of every process $p_i \in \mathcal{C}$ is correct according to $\mathcal{A}^S$.}

    This property holds as $p_i$ transfers its states and sent and received messages from $\mathcal{E}$ to $\mathcal{E}'$.

    \medskip
    \item \emph{The execution is computationally feasible.} We aim to prove that for every message $m$ sent by a (Byzantine) process, the computation of $m$ does not need more computational assumptions than those required for executions in $\mathcal{A}^{S}$. 

    To prove this property, we focus on a specific process $p_j \notin \mathcal{C}$.
    Let $m$ be any message sent by $p_j$ in some round $k$ of $\mathcal{E}$ (in the case $p_j \in \mathcal{C}$, $p_j$ would exhibit a correct behavior that is, by definition, computationally feasible).
    As $m$ is sent in round $k$, $m \in \mathit{received}_r(k)$, for some process $p_r \in \mathcal{C}$ (see step 2b of the construction).
    Since $m \in \mathit{received}_r(k)$, $p_r$ receives $m$ in $\mathcal{E}$ at some time $\tau^R(m) \in [\tau_r + (k - 1)\Delta_{\mathit{sync}}, \tau_r + k \cdot \Delta_{\mathit{sync}}]$.

    Let $\mathcal{M}(m)$ denote the set of messages $m'$ such that (1) any process $p_z \notin \mathcal{C}$ has received $m'$ before $p_j$ sends $m$ in $\mathcal{E}$, and (2) the sender $p_s$ of $m'$ belongs to $\mathcal{C}$.
    Consider any message $m' \in \mathcal{M}(m)$.
    As $m'$ is sent by $p_s$ in $\mathcal{E}$, $m' \in \mathit{sent}_s(k')$, for some $k' \in [1, \mathcal{R}]$.
    Let $\tau^S(m')$ denote the time process $p_s$ sends $m'$ in $\mathcal{E}$.
    Importantly, $\tau^S(m') < \tau^R(m)$.
    First, we show that $k' \leq k$.
    By contradiction, let $k' > k$.
    We know that $\tau^S(m') = \tau_s + (k' - 1)\Delta_{\mathit{sync}} + \Delta_{\mathit{shift}}$.
    As $\tau_s \geq \tau_r - \Delta_{\mathit{shift}}$, $\tau^S(m') \geq \tau_r - \Delta_{\mathit{shift}} + (k' - 1)\Delta_{\mathit{sync}} + \Delta_{\mathit{shift}} = \tau_r + (k' - 1)\Delta_{\mathit{sync}} \geq \tau_r + k \cdot \Delta_{\mathit{sync}}$.
    Thus, we reach a contradiction with the fact that $\tau^S(m) < \tau^R(m)$, thus proving that $k' \leq k$.

    Finally, for every message $m' \in \mathcal{M}(m)$, $m'$ is received by process $p_z \notin \mathcal{C}$ in round no greater than $k$ in $\mathcal{E}'$ (due to step 2a of the construction).
    As $p_j$ is capable of sending $m$ once processes that do not belong to $\mathcal{C}$ have received messages from the $\mathcal{M}(m)$ set (it does so in $\mathcal{E}$), $\mathcal{E}'$ is indeed computationally feasible.
\end{compactitem}
As $\mathcal{E}'$ satisfies all aforementioned properties, it is indeed a valid synchronous execution $\mathcal{A}^S$, thus concluding the proof of the lemma.
\end{proof}

The following lemma proves that $\mathsf{CryptoSim}$ indeed simulates a cryptography-based (and, thus, even a cryptography-free) synchronous algorithm $\mathcal{A}^S$ when $\mathcal{S}^*$ holds.

\begin{lemma} [$\mathsf{CryptoSim}$ simulates $\mathcal{A}^S$] \label{lemma:cryptography_based_simulation_correct}
Let $\mathcal{S}^*$ hold.
For each execution $\mathcal{E}$ of $\mathsf{CryptoSim}$, there exists an $\mathcal{R}$-rounds-long synchronous execution $\mathcal{E}'$ of $\mathcal{A}^S$ such that:
\begin{compactitem}
    \item the sets of correct processes in $\mathcal{E}$ and $\mathcal{E}'$ are identical, and

    \item the proposals of correct processes in $\mathcal{E}$ and $\mathcal{E}'$ are identical, and

    \item the sets of messages sent by correct processes in $\mathcal{E}$ and $\mathcal{E}'$ are identical, and

    \item for each correct process $p_i$ and every $k \in [1, \mathcal{R} + 1]$, $s_i^k(\mathcal{E}) = s_i^k(\mathcal{E}')$, where (1) $s_i^k(\mathcal{E})$ is the state of $p_i$ at the beginning of the $k$-th (i.e., at the end of the $(k - 1)$-st) simulated round in $\mathcal{E}$, and (2) $s_i^k(\mathcal{E}')$ is the state of $p_i$ at the beginning of the $k$-th (i.e., at the end of the $(k - 1)$-st) round in $\mathcal{E}'$.
\end{compactitem}    
\end{lemma}
\begin{proof}
To prove the lemma, we go through a sequence of intermediate results.

\medskip
\noindent \emph{Intermediate result 1: Let $\mathsf{CryptoSim}^-$ be identical to $\mathsf{CryptoSim}$ except that correct processes are allowed to send any number of bits (i.e., the check at line~\ref{line:check_sent_bits} is removed).
Moreover, let the condition $\mathcal{S}^*$ be adapted to $\mathsf{CryptoSim}^-$.
Then, the lemma holds for $\mathsf{CryptoSim}^-$.}
\\ Follows from \Cref{lemma:simulation_correct}.

\medskip
\noindent \emph{Intermediate result 2: Let $\mathsf{CryptoSim}^-$ be identical to $\mathsf{CryptoSim}$ except that correct processes are allowed to send any number of bits (i.e., the check at line~\ref{line:check_sent_bits} is removed).
Moreover, let the condition $\mathcal{S}^*$ be adapted to $\mathsf{CryptoSim}^-$.
Then, no correct process sends more than $\mathcal{B}$ bits in any execution $\mathcal{E}^-$ of $\mathsf{CryptoSim}^-$ when $\mathcal{S}^*$ holds.}
\\ By contradiction, suppose there exists an execution $\mathcal{E}^-$ of $\mathsf{CryptoSim}^-$ in which some correct process $p_i$ sends more than $\mathcal{B}$ bits.
The first intermediate result proves that $\mathcal{E}^-$ simulates an execution $\mathit{sim}(\mathcal{E}^-)$ of $\mathcal{A}^S$.
Hence, a message $m$ is sent by $p_i$ in $\mathcal{E}^-$ if and only if a message $m$ is sent by $p_i$ in $\mathit{sim}(\mathcal{E}^-)$.
Thus, $p_i$ sends more than $\mathcal{B}$ bits in $\mathit{sim}(\mathcal{E}^-)$, which is impossible as $p_i$ sends at most $\mathcal{B}$ bits in any execution of $\mathcal{A}^S$.

\medskip
\noindent \emph{Proof of \Cref{lemma:cryptography_based_simulation_correct}.}
To prove that $\mathsf{CryptoSim}$ correctly simulates $\mathcal{A}^S$ when $\mathcal{S}^*$ holds, it suffices to show that $\mathsf{CryptoSim} \equiv \mathsf{CryptoSim}^-$ as the lemma would follow from the first intermediate result, where $\mathsf{CryptoSim}^-$ is defined above.
By contradiction, suppose $\mathsf{CryptoSim} \not\equiv \mathsf{CryptoSim}^-$ when $\mathcal{S}^*$ holds.
This is only possible if there exists an execution $\mathcal{E}$ of $\mathsf{CryptoSim}$ in which a correct process does not send some message $m$ it was supposed to send according to $\mathcal{A}^S$ because the sending would exceed the $\mathcal{B}$ bits limit.
However, this implies that there exists an execution of $\mathsf{CryptoSim}^-$ in which this correct process does send more than $\mathcal{B}$ bits, which represents a contradiction with the second intermediate result.
Therefore, $\mathsf{CryptoSim} \equiv \mathsf{CryptoSim}^-$ when $\mathcal{S}^*$ holds.
\end{proof}

Now that we have explicitly introduced our simulation techniques for cryptography-free ($\mathsf{CryptoFreeSim}$) and cryptography-based ($\mathsf{CryptoSim}$) synchronous algorithms, we are ready to prove that \block satisfies the synchronicity property.

\begin{theorem} [Synchronicity] \label{theorem:block_synchronicity}
\block (\Cref{algorithm:block}) satisfies synchronicity.
\end{theorem}
\begin{proof}
Suppose $\tau$ denotes the first time a correct process proposes to \block.
Let the following hold: (1) $\tau \geq \text{GST}$, (2) all correct processes propose to \block by time $\tau + \Delta_{\mathit{shift}}$, and (3) no correct process abandons \block by time $\tau + \Delta_{\mathit{total}}$.
(Hence, let the precondition of the synchronicity property be satisfied.)

As $\mathcal{GC}_1$ terminates in $\mathit{latency}(\mathcal{GC}_1)$ asynchronous rounds, every correct process decides from $\mathcal{GC}_1$ by time $\tau + \Delta_{\mathit{shift}} + \Delta_1$ (as all correct processes overlap for $\Delta_1 = \mathit{latency}(\mathcal{GC}_1) \cdot \delta$ time in $\mathcal{GC}_1$).
Moreover, all correct processes start executing $\mathcal{A}^S$ within $\Delta_{\mathit{shift}}$ time of each other (as they execute $\mathcal{GC}_1$ for at least $\Delta_{\mathit{shift}} + \Delta_1$ time even if they decide from $\mathcal{GC}_1$ before).
Due to \cref{lemma:cryptography_free_simulation_correct,lemma:cryptography_based_simulation_correct}, $\mathcal{A}^S$ exhibits a valid synchronous execution.
Hence, all correct processes decide the same valid (non-$\bot$) value from $\mathcal{A}^S$ by time $\tau + (\Delta_{\mathit{shift}} + \Delta_1) + (\mathcal{R} \cdot \Delta_{\mathit{sync}})$.
To prove \block's synchronicity property, we show that, at the end of Task 1's Step 3, the local variables $\mathit{est}_i$ and $\mathit{est}_j$, for any two correct processes $p_i$ and $p_j$, are identical.
\begin{compactitem}
    \item Assume a correct process $p_i$ decides $\mathit{est}_i$ with grade $1$ from $\mathcal{GC}_1$ (Step 1 of Task 1). 
    Hence, by the consistency property of $\mathcal{GC}_1$, all correct processes decide $(\mathit{est}_i, \cdot)$ from $\mathcal{GC}_1$, and then propose $\mathit{est}_i$ to $\mathcal{A}^S$.
    As $\mathcal{A}^S$ satisfies strong validity, every correct process $p_j$ decides $\mathit{est}_i$ from $\mathcal{A}^S$.
    As stated above, $\mathit{est}_i \neq \bot$ and $\mathsf{valid}(\mathit{est}_i) = \mathit{true}$. 
    Let $p_j$ be any correct process.
    \begin{compactitem}
    \item  If process $p_j$ has decided $(\mathit{est}_j, 1)$ from $\mathcal{GC}_1$, then $\mathit{est}_j = \mathit{est}_i$ due to the consistency property of $\mathcal{GC}_1$.
    
    \item  If process $p_j$ has decided $(\cdot, 0)$ from $\mathcal{GC}_1$, then $\mathit{est}_j = \mathit{est}_i$ due to the fact that $\mathit{est}_i$ is decided by $p_j$ from $\mathcal{A}^S$.
    \end{compactitem}
    In both cases, $\mathit{est}_i = \mathit{est}_j$ at the end of Task 1's Step 3.

    \item Let both $p_i$ and $p_j$ decide with grade $0$ from $\mathcal{GC}_1$ (Step 1 of Task 1).
    In this case, $\mathit{est}_i = \mathit{est}_j$ due to the agreement property of $\mathcal{A}^S$.
\end{compactitem} 
Thus, all correct processes propose to $\mathcal{GC}_2$ the same valid value $v$, and they do so within $\Delta_{\mathit{shift}}$ time of each other.
Every correct process decides $(v, 1)$ by time $\tau + (\Delta_{\mathit{shift}} + \Delta_1) + (\mathcal{R} \cdot \Delta_{\mathit{sync}}) + (\Delta_{\mathit{shift}} + \Delta_2)$ as $\Delta_2 = \mathit{latency}(\mathcal{GC}_2) \cdot \delta$; $v$ is decided with grade $1$ due to the strong validity property of $\mathcal{GC}_2$.
Therefore, every correct process decides (Step 5 of Task 1) by time $\tau + ( \Delta_{\mathit{shift}} + \Delta_1 ) + ( \mathcal{R} \cdot \Delta_{\mathit{sync}} ) + ( \Delta_{\mathit{shift}} + \Delta_2 ) = \tau + \Delta_{\mathit{total}}$, thus ensuring synchronicity.
\end{proof}

\smallskip
\noindent \textbf{Proof of complexity.}
To conclude the section, we prove \block's per-process complexity. Recall that $\mathit{pbit}(\mathcal{X})$ is the maximum number of bits sent by a correct process in $\mathcal{X} \in \{\mathcal{GC}_1, \mathcal{GC}_2, \mathcal{VB}\}$, whereas $\mathcal{B} = \mathit{pbit}(\mathcal{A}^S)$ is the maximum number of bits sent by a correct process in $\mathcal{A}^S$ (see \cref{section:preliminaries,section:block}).

\begin{theorem} [Exchanged bits]\label{theorem:block_communication}
Any correct process sends
\begin{equation*}
    \mathit{pbit}(\block) = \mathit{pbit}(\mathcal{GC}_1) + \mathit{pbit}(\mathcal{GC}_2) + \mathit{pbit}(\mathcal{VB}) + 2 \mathcal{B} \text{ bits in \block.}
\end{equation*}
\end{theorem}
\begin{proof}
Any correct process sends (1) $\mathit{pbit}(\mathcal{GC}_1)$ bits in $\mathcal{GC}_1$, (2) $\mathit{pbit}(\mathcal{GC}_2)$ bits in $\mathcal{GC}_2$, (3) $\mathit{pbit}(\mathcal{VB})$ in $\mathcal{VB}$, and (4) at most $2\mathcal{B}$ bits in the simulation of $\mathcal{A}^S$.
\end{proof}

Lastly, we define $\mathit{latency}(\block)$ in the following way:
\begin{equation*}
    \mathit{latency}(\block) = (\mathit{latency}(\mathcal{GC}_1) \cdot \delta) + (\mathcal{R} \cdot \Delta_{\mathit{sync}}) + (\mathit{latency}(\mathcal{GC}_2) \cdot \delta) + (\mathit{latency}(\mathcal{VB}) \cdot \delta).
\end{equation*}
We rely on $\mathit{latency}(\block)$ in \Cref{section:name_correctness_complexity_formal}, where we show that \name does not add more than a constant factor overhead in latency over \block, i.e., that \name has latency in $O(\mathit{latency}(\block))$.


\section{\name: Pseudocode \& Proof of Correctness and Complexity} \label{section:name_correctness_complexity_formal}

In this section, we give the pseudocode of \name.
Moreover, we prove \name's correctness and complexity.

\subsection{Finisher} \label{subsection:finisher}

First, we formally define the finisher primitive that \name utilizes to allow correct processes to decide and halt (i.e., stop sending and receiving messages).
The finisher primitive exposes the following interface:
\begin{compactitem}
    \item \textbf{request} $\mathsf{to\_finish}(v \in \mathsf{Value})$: a process aims to finish with value $v$.

    \item \textbf{indication} $\mathsf{finish}(v' \in \mathsf{Value})$: a process finishes with value $v'$.
\end{compactitem}
Every correct process invokes $\mathsf{to\_finish}(\cdot)$ at most once.
Moreover, if any correct process invokes $\mathsf{to\_finish}(v_1)$ and any other correct process invokes $\mathsf{to\_finish}(v_2)$, then $v_1 = v_2$.
We do not assume that all correct processes invoke $\mathsf{to\_finish}(\cdot)$.

The following properties are satisfied by the finisher primitive:
\begin{compactitem}
    \item \emph{Integrity:} If a correct process receives a $\mathsf{finish}(v')$ indication, then a correct process has previously invoked a $\mathsf{to\_finish}(v')$ request.

    \item \emph{Termination:} Let $\tau$ be the first time such that all correct processes have invoked a $\mathsf{to\_finish}(\cdot)$ request by time $\tau$.
    Then, every correct process receives a $\mathsf{finish}(\cdot)$ indication by time $\max(\tau, \text{GST}) + 2\delta$.

    \item \emph{Totality:} If any correct process receives a $\mathsf{finish}(\cdot)$ indication at some time $\tau$, then every correct process receives a $\mathsf{finish}(\cdot)$ indication by time $\max(\tau, \text{GST}) + 2\delta$.
\end{compactitem}

\subsubsection{\shortfin: implementation for constant-sized values}\label{subsubsection:shortfin}
\Cref{algorithm:finisher_small_inputs} is the pseudocode of \shortfin, our implementation of the finisher primitive for constant-sized values (i.e., the size of each value $v \in \mathsf{Value}$ is $O(1)$ bits).
\shortfin tolerates up to $t < n / 3$ Byzantine processes and exchanges $O(n^2)$ bits.

\shortfin operates as follows.
Once a correct process $p_i$ invokes a $\mathsf{to\_finish}(v)$ request (line~\ref{line:to_finish_small}), $p_i$ disseminates its value $v$ to all processes (line~\ref{line:broadcast_finish_1_small}).
Moreover, process $p_i$ disseminates some value (line~\ref{line:broadcast_finish_2_small}) once it receives that value from at least $t + 1$ processes (line~\ref{line:pluratility_finish_small}).
Finally, once $p_i$ receives some value $v'$ from $2t + 1$ processes (line~\ref{line:quorum_finish_small}), $p_i$ triggers $\mathsf{finish}(v')$ (line~\ref{line:finish_small}).

\begin{algorithm} [h]
\caption{\shortfin: Pseudocode (for process $p_i$)}
\label{algorithm:finisher_small_inputs}
\footnotesize
\begin{algorithmic} [1] 
\State \textbf{Local variables:}
\State \hskip2em $\mathsf{Boolean}$ $\mathit{started}_i \gets \mathit{false}$

\medskip
\State \textbf{upon} $\mathsf{to\_finish}(v \in \mathsf{Value})$: \label{line:to_finish_small}
\State \hskip2em $\mathit{started}_i \gets \mathit{true}$
\State \hskip2em \textbf{broadcast} $\langle \textsc{finish}, v \rangle$ \label{line:broadcast_finish_1_small}

\medskip
\State \textbf{upon} $\langle \textsc{finish}, v' \rangle$ is received from $t + 1$ processes, for some $v' \in \mathsf{Value}$, and $\mathit{started}_i = \mathit{false}$: \label{line:pluratility_finish_small}
\State \hskip2em $\mathit{started}_i \gets \mathit{true}$
\State \hskip2em \textbf{broadcast} $\langle \textsc{finish}, v' \rangle$ \label{line:broadcast_finish_2_small}

\medskip
\State \textbf{upon} $\langle \textsc{finish}, v' \rangle$ is received from $2t + 1$ processes, for some $v' \in \mathsf{Value}$: \label{line:quorum_finish_small}
\State \hskip2em \textbf{trigger} $\mathsf{finish}(v')$ \label{line:finish_small}

\end{algorithmic} 
\end{algorithm}

\smallskip
\noindent \textbf{Proof of correctness \& complexity.}
Let $v^{\star}$ denote the value such that if any correct process invokes $\mathsf{to\_finish}(v)$, then $v = v^{\star}$.
We start by proving that the first correct process that broadcasts a \textsc{finish} message does so for value $v^{\star}$.

\begin{lemma} \label{lemma:finisher_small_first_correct_message}
The first correct process that broadcasts a \textsc{finish} message does so for value $v^{\star}$.
\end{lemma}
\begin{proof}
Let $p_i$ be that correct process.
Process $p_i$ cannot broadcast the message at line~\ref{line:broadcast_finish_2_small} as that would contradict the fact that $p_i$ is the first correct process to broadcast a \textsc{finish} message.
Hence, process $p_i$ broadcasts its $\langle \textsc{finish}, v \rangle$ message at line~\ref{line:broadcast_finish_1_small}, which implies that $p_i$ has previously invoked a $\mathsf{to\_finish}(v)$ request (line~\ref{line:to_finish_small}).
Therefore, $v = v^{\star}$ due to the assumption that no correct process invokes a $\mathsf{to\_finish}(\cdot)$ request with a value different from $v^{\star}$.
\end{proof}

Next, we prove that no correct process broadcasts a \textsc{finish} message for a non-$v^{\star}$ value.

\begin{lemma} \label{lemma:finisher_small_all_correct_message}
If a correct process broadcasts a $\langle \textsc{finish}, v \rangle$ message, then $v = v^{\star}$.
\end{lemma}
\begin{proof}
We prove the lemma by induction.

\medskip
\noindent \emph{Base step: We prove that if $p_i$ is the first correct process to broadcast a \textsc{finish} message, then $v = v^{\star}$.}
\\ The base step follows directly from \Cref{lemma:finisher_small_first_correct_message}.

\medskip
\noindent \emph{Inductive step: The first $j$ correct processes to broadcast a \textsc{finish} message do so for value $v^{\star}$, for some $j \geq 1$.
We prove that the $(j + 1)$-st correct process to broadcast a \textsc{finish} message does so for value $v^{\star}.$}
\\ Let $p_{j + 1}$ be the $(j + 1)$-st correct process to broadcast a \textsc{finish} message, and let that message be for value $v$.
We distinguish two possibilities:
\begin{compactitem}
    \item Let $p_{j + 1}$ broadcast the \textsc{finish} message at line~\ref{line:broadcast_finish_1_small}.
    In this case, $v = v^{\star}$ as no correct process invokes a $\mathsf{to\_finish}(\cdot)$ request with a value different from $v^{\star}$.

    \item Let $p_{j + 1}$ broadcast the \textsc{finish} message at line~\ref{line:broadcast_finish_2_small}.
    Hence, $p_{j + 1}$ has previously received a \textsc{finish} message for $v$ from a correct process (due to the rule at line~\ref{line:pluratility_finish_small}).
    Therefore, $v = v^{\star}$.

\end{compactitem}
As $v = v^{\star}$ in all possible cases, the inductive step is concluded.
\end{proof}

We are now ready to prove that \shortfin satisfies the integrity property.

\begin{theorem} [Integrity]
\shortfin (\Cref{algorithm:finisher_small_inputs}) satisfies integrity.
\end{theorem}
\begin{proof}
Let $p_i$ be any correct process that receives a $\mathsf{finish}(v')$ indication, for some value $v'$ (line~\ref{line:finish_small}).
Hence, $p_i$ has previously received a $\langle \textsc{finish}, v' \rangle$ message from $2t + 1$ processes (line~\ref{line:quorum_finish_small}).
Thus, $p_i$ has received a \textsc{finish} message for value $v'$ from a correct process.
Given that no correct process sends a \textsc{finish} message for a non-$v^{\star}$ value (by \Cref{lemma:finisher_small_all_correct_message}), $v' = v^{\star}$.
Moreover, as the first correct process that broadcasts a \textsc{finish} value does so at line~\ref{line:broadcast_finish_1_small}, some correct process has invoked a $\mathsf{to\_finish}(v^{\star})$ request prior to $p_i$ receiving the aforementioned $\mathsf{finish}(\cdot)$ indication.
\end{proof}

Next, we prove the termination property.

\begin{theorem} [Termination]
\shortfin (\Cref{algorithm:finisher_small_inputs}) satisfies termination.
\end{theorem}
\begin{proof}
Recall that $\tau$ is the first time such that all correct processes have invoked a $\mathsf{to\_finish}(\cdot)$ request by time $\tau$.
Hence, as there are at least $n - t \geq 2t + 1$ correct processes, every correct process $p_i$ receives a \textsc{finish} message from $2t + 1$ processes by time $\max(\tau, \text{GST}) + \delta$.
As all these messages are for the same value (namely, $v^{\star}$) due to \Cref{lemma:finisher_small_all_correct_message}, process $p_i$ does receive a $\mathsf{finish}(\cdot)$ indication by time $\max(\tau, \text{GST}) + \delta$ (line~\ref{line:finish_small}), which concludes the proof.
\end{proof}

The following theorem proves the totality property.

\begin{theorem} [Totality]
\shortfin (\Cref{algorithm:finisher_small_inputs}) satisfies totality.
\end{theorem}
\begin{proof}
Let $p_i$ be any correct process that receives a $\mathsf{finish}(\cdot)$ indication at some time $\tau$ (line~\ref{line:finish_small}).
Therefore, $p_i$ has received a $\langle \textsc{finish}, v^{\star} \rangle$ message from $2t + 1$ processes by time $\tau$ (line~\ref{line:quorum_finish_small}); recall that the integrity property is satisfied by \shortfin.
Hence, by time $\max(\tau, \text{GST}) + \delta$, every correct process receives $t + 1$ $\langle \textsc{finish}, v^{\star} \rangle$ messages.

Consider any correct process $p_j$.
As mentioned above, $p_j$ receives a $\langle \textsc{finish}, v^{\star} \rangle$ message from $t + 1$ processes by time $\max(\tau, \text{GST}) + \delta$.
Hence, the rule at line~\ref{line:pluratility_finish_small} activates at $p_j$ by time $\max(\tau, \text{GST}) + \delta$ (otherwise, $p_j$ has already broadcast a \textsc{finish} message at line~\ref{line:broadcast_finish_1_small}).
Hence, $p_j$ indeed broadcast a \textsc{finish} message for $v^{\star}$ by time $\max(\tau, \text{GST}) + \delta$.

Finally, as every correct process broadcasts a $\langle \textsc{finish}, v^{\star} \rangle$ message by time $\max(\tau, \text{GST}) + \delta$, the rule at line~\ref{line:quorum_finish_small} activates at every correct process by time $\max(\tau, \text{GST}) + 2\delta$.
Thus, the totality property is satisfied.
\end{proof}

Finally, we prove that any correct process sends $O(n)$ bits in \shortfin.

\begin{theorem} [Exchanged bits]
Any correct process sends $O(n)$ bits in \shortfin.
\end{theorem}
\begin{proof}
Each correct process broadcasts only $O(1)$ \textsc{finish} messages, each of constant size; recall that values are constant-sized.
Hence, each correct process sends $O(n)$ bits. 
\end{proof}

\subsubsection{\longfin: implementation for long values.} \label{subsubsection:finisher_long_inputs_implementation}
\Cref{algorithm:finisher_long_inputs} is the pseudocode of \longfin, our implementation of the finisher primitive for values of size $L \notin O(1)$ bits.
\longfin tolerates up to $t < n / 3$ Byzantine processes and it exchanges $O\big( nL + n^2 \log(n) \big)$ bits.

The crucial element of \longfin is asynchronous data dissemination (ADD)~\cite{das2021asynchronous}, an asynchronous information-theoretic secure primitive tolerating $t < n / 3$ Byzantine failures.
ADD ensures the following: 
Let $M$ be a data blob of size $L$ that is the input of at least $t + 1$ correct processes.
The remaining correct processes do not input any value.
It is guaranteed that all correct processes eventually output (only) $M$.
In terms of complexity, the ADD protocol incurs $2$ asynchronous rounds and $O\big( L + n \log(n) \big)$ per-process bit complexity.

We describe \longfin from the perspective of a correct process $p_i$.
Once $p_i$ invokes a $\mathsf{to\_finish}(v)$ request (line~\ref{line:to_finish_long}), $p_i$ inputs its value to ADD (line~\ref{line:start_add_finisher_long}) and notifies all processes about this (line~\ref{line:broadcast_started_add_finisher_long}).
When process $p_i$ learns that $2t + 1$ processes have started ADD (line~\ref{line:started_add_quorum_finisher_long}), process $p_i$ knows that at least $t + 1$ correct processes have started ADD with a non-$\bot$ value (recall that this represents a precondition of the ADD primitive).
Hence, $p_i$ informs all other processes that at least $t + 1$ correct processes have started ADD with a non-$\bot$ value via a $\langle \text{``plurality started ADD''} \rangle$ message (line~\ref{line:broadcast_plurality_finisher_long}).
If $p_i$ receives a $\langle \text{``plurality started ADD''} \rangle$ message from $t + 1$ processes and it has not previously disseminated this message (line~\ref{line:plurality_plurality_2_finisher_long}), $p_i$ does so (line~\ref{line:broadcast_plurality_2_finisher_long}).
Finally, once $p_i$ outputs a value $v'$ from ADD and receives a $\langle \text{``plurality started ADD''} \rangle$ message from $2t + 1$ processes (line~\ref{line:rule_to_finish_finisher_long}), $p_i$ 
triggers $\mathsf{finish}(v')$ (line~\ref{line:finish_finisher_long}).

\begin{algorithm}
\caption{\longfin: Pseudocode (for process $p_i$)}
\label{algorithm:finisher_long_inputs}
\footnotesize
\begin{algorithmic} [1] 
\State \textbf{Local variables:}
\State \hskip2em $\mathsf{Boolean}$ $\mathit{echoed}_i \gets \mathit{false}$

\medskip
\State \textbf{upon} $\mathsf{to\_finish}(v \in \mathsf{Value})$: \label{line:to_finish_long}
\State \hskip2em input $v$ to ADD \label{line:start_add_finisher_long}
\State \hskip2em \textbf{broadcast} $\langle \text{``started ADD''} \rangle$ \label{line:broadcast_started_add_finisher_long}

\medskip
\State \textbf{upon} $\langle \text{``started ADD''} \rangle$ is received from $2t + 1$ processes and $\mathit{echoed}_i = \mathit{false}$:\label{line:started_add_quorum_finisher_long}
\State \hskip2em $\mathit{echoed}_i \gets \mathit{true}$
\State \hskip2em \textbf{broadcast} $\langle \text{``plurality started ADD''} \rangle$ \label{line:broadcast_plurality_finisher_long}


\medskip
\State \textbf{upon} $\langle \text{``plurality started ADD''} \rangle$ is received from $t + 1$ processes and $\mathit{echoed}_i = \mathit{false}$:\label{line:plurality_plurality_2_finisher_long}
\State \hskip2em $\mathit{echoed}_i \gets \mathit{true}$
\State \hskip2em \textbf{broadcast} $\langle \text{``plurality started ADD''} \rangle$ \label{line:broadcast_plurality_2_finisher_long}

\medskip
\State \textbf{upon} $\mathsf{Value}$ $v'$ is output from ADD and $\langle \text{``plurality started ADD''} \rangle$ is received from $2t + 1$ processes: \label{line:rule_to_finish_finisher_long}
\State \hskip2em \textbf{trigger} $\mathsf{finish}(v')$ \label{line:finish_finisher_long}

\end{algorithmic} 
\end{algorithm}

\smallskip
\noindent \textbf{Proof of correctness \& complexity.}
Let us denote by $v^{\star}$ the common value of all correct processes that invoke a $\mathsf{to\_finish}(\cdot)$ request.
We start by proving that if a correct process broadcasts a $\langle \text{``plurality started ADD''} \rangle$ message, then at least $t + 1$ correct processes have previously started ADD with $v^{\star}$.

\begin{lemma} \label{lemma:enough_to_ADD}
If a correct process broadcasts a $\langle \text{``plurality started ADD''} \rangle$ message, then at least $t + 1$ correct processes have previously input $v^{\star}$ to ADD.
\end{lemma}
\begin{proof}
The first correct process to broadcast a $\langle \text{``plurality started ADD''} \rangle$ message does so at line~\ref{line:broadcast_plurality_finisher_long}.
Let us denote this process by $p_i$.
Hence, before broadcasting the aforementioned message, $p_i$ has received a $\langle \text{``started ADD''} \rangle$ messages from $2t + 1$ processes (line~\ref{line:started_add_quorum_finisher_long}).
Therefore, at least $t + 1$ correct processes have sent a $\langle \text{``started ADD''} \rangle$ message.
Finally, as any correct process $p_j$ sends a $\langle \text{``started ADD''} \rangle$ message (line~\ref{line:broadcast_started_add_finisher_long}) only after inputting a value $v^{\star} \neq \bot$ to ADD (line~\ref{line:start_add_finisher_long}), the statement of the lemma holds.
\end{proof}

Next, we show that no correct process inputs to ADD a non-$v^{\star}$ value.

\begin{lemma} \label{lemma:no_other_to_ADD}
If a correct process inputs a value $v$ to ADD, then $v = v^{\star}$.
\end{lemma}
\begin{proof}
The lemma follows from the fact that no correct process invokes a $\mathsf{to\_finish}(\cdot)$ request with a non-$v^{\star}$ value (line~\ref{line:to_finish_long}).
\end{proof}

We now prove the integrity property of \longfin.

\begin{theorem} [Integrity]
\longfin (\Cref{algorithm:finisher_long_inputs}) satisfies integrity.
\end{theorem}
\begin{proof}
Let $p_i$ be any correct process that receives a $\mathsf{finish}(v')$ indication, for some value $v'$ (line~\ref{line:finish_finisher_long}).
Hence, $p_i$ has previously output $v'$ from ADD and received a $\langle \text{``plurality started ADD''} \rangle$ message from $2t + 1$ processes (line~\ref{line:rule_to_finish_finisher_long}).
As (1) at least $t + 1$ correct processes have previously input $v^{\star}$ to ADD (by \Cref{lemma:enough_to_ADD}), and (2) no correct process inputs any other value to ADD (by \Cref{lemma:no_other_to_ADD}), the precondition of ADD is satisfied.
Therefore, ADD ensures that $v^{\star} = v'$.
Finally, as the first correct process to input a value to ADD does so at line~\ref{line:start_add_finisher_long}, a correct process has indeed invoked a $\mathsf{to\_finish}(v^{\star})$ request (line~\ref{line:to_finish_long}) prior to $p_i$ receiving the aforementioned $\mathsf{finish}(v' = v^{\star})$ indication.
\end{proof}

Next, we prove the termination property.

\begin{theorem} [Termination]
\longfin (\Cref{algorithm:finisher_long_inputs}) satisfies termination.
\end{theorem}
\begin{proof}
Recall that $\tau$ is the first time such that all correct processes have invoked a $\mathsf{to\_finish}(\cdot)$ request by time $\tau$.
Hence, by time $\tau$ at least $n - t \geq t + 1$ correct processes input $v^{\star}$ to ADD (line~\ref{line:start_add_finisher_long}) and send a $\langle \text{``started ADD''} \rangle$ message (line~\ref{line:broadcast_started_add_finisher_long}).
As no correct process inputs any non-$v^{\star}$ value to ADD (by \Cref{lemma:no_other_to_ADD}), the precondition of ADD is satisfied.
Therefore, by time $\max(\tau, \text{GST}) + 2\delta$, every correct process outputs a value from ADD (since ADD incurs two asynchronous rounds).
Moreover, by time $\max(\tau, \text{GST}) + \delta$, every correct process sends a $\langle \text{``plurality started ADD''} \rangle$ message (line~\ref{line:broadcast_plurality_finisher_long} or line~\ref{line:broadcast_plurality_2_finisher_long}).
Thus, every correct process receives a $\langle \text{``plurality started ADD''} \rangle$ from $n - t \geq 2t + 1$ processes by time $\max(\tau, \text{GST}) + 2\delta$.
This implies that the rule at line~\ref{line:rule_to_finish_finisher_long} activates at each correct process by time $\max(\tau, \text{GST}) + 2\delta$, thus concluding the proof.
\end{proof}

The following theorem proves the totality property.

\begin{theorem} [Totality]
\longfin (\Cref{algorithm:finisher_long_inputs}) satisfies totality.
\end{theorem}
\begin{proof}
Let $p_i$ be a correct process that receives a $\mathsf{finish}(\cdot)$ indication at some time $\tau$; as guaranteed by the integrity property, the indication is for $v^{\star}$.
Hence, by time $\tau$, process $p_i$ has output $v^{\star}$ from ADD and received a $\langle \text{``plurality started ADD''} \rangle$ from $2t + 1$ processes (due to the rule at line~\ref{line:rule_to_finish_finisher_long}).
Let us focus on any correct process $p_j$.

Due to \cref{lemma:enough_to_ADD,lemma:no_other_to_ADD} and the fact that ADD incurs two asynchronous rounds, process $p_j$ outputs $v^{\star}$ from ADD by time $\max(\tau, \text{GST}) + 2\delta$.
Moreover, every correct process broadcasts a $\langle \text{``plurality started ADD''} \rangle$ message by time $\max(\tau, \text{GST}) + \delta$ (as $p_i$ has received such messages from at least $t + 1$ correct processes by time $\tau$), which implies that $p_j$ receives $n - t \geq 2t + 1$ such messages by time $\max(\tau, \text{GST}) + 2\delta$.
Therefore, the rule at line~\ref{line:rule_to_finish_finisher_long} activates at $p_j$ by time $\max(\tau, \text{GST}) + 2\delta$, which concludes the proof.
\end{proof}

Lastly, we prove the number of bits correct processes send in \longfin.


\begin{theorem} [Exchanged bits]
Any correct process sends $O\big(L + n \log(n) \big)$ bits in \longfin.
\end{theorem}
\begin{proof}
Each correct process $p_i$ sends $O(n)$ bits via $\langle \text{``started ADD''} \rangle$ and $\langle \text{``plurality started ADD''} \rangle$ messages.
Moreover, the ADD primitive incurs $O\big( L + n \log(n) \big)$ bits per-process.
\end{proof}

\subsection{Pseudocode}

The pseudocode of \name is given in \Cref{algorithm:name_pseudocode_2}.
\name's executions unfold in views; $\mathsf{View} = \{1, 2, ...\}$ denotes the set of views.
Moreover, each view is associated with its instance of \block (see \Cref{section:block}); the instance of \block associated with view $V \in \mathsf{View}$ is denoted by $\mathcal{CX}(V)$ (line~\ref{line:block_instance}).
Each instance of \block is parametrized with $\Delta_{\mathit{shift}} = 2\delta$.
To guarantee liveness, \name ensures that all correct processes are brought to the same instance of \block for sufficiently long after GST, thus allowing \block to decide (due to its synchronicity property).
The safety of \name is ensured by the careful utilization of the \block instances.
We proceed to describe \name's pseudocode (\Cref{algorithm:name_pseudocode_2}) from the perspective of a correct process $p_i$.

\begin{algorithm} [h]
\caption{\name: Pseudocode (for process $p_i$)}
\label{algorithm:name_pseudocode_2}
\footnotesize
\begin{algorithmic} [1] 
\State \textbf{Uses:}
\State \hskip2em \block with parameter $\Delta_{\mathit{shift}} = 2\delta$, \textbf{instances} $\mathcal{CX}(V)$, for every $V \in \mathsf{View}$ \label{line:block_instance} \BlueComment{see \Cref{section:block}}
\State \hskip2em Finisher, \textbf{instance} $\mathcal{F}$


\medskip
\State \textbf{Local variables:}
\State \hskip2em $\mathsf{Map}(\mathsf{View} \to \mathsf{Boolean})$ $\mathit{helped}_i \gets \{\mathit{false}, \mathit{false}, ..., \mathit{false}\}$
\State \hskip2em $\mathsf{View}$ $\mathit{view}_i \gets 1$ \label{line:init_view}

\medskip
\State \textbf{upon} $\mathsf{propose}(v \in \mathsf{Value})$: \label{line:name_propose} \BlueComment{start participating in \name}
\State \hskip2em initialize $\mathcal{CX}(V)$ with $\mathsf{def}(p_i) = v$, for every view $V$
\State \hskip2em \textbf{invoke} $\mathcal{CX}(1).\mathsf{propose}(v)$ \BlueComment{start \block associated with view $1$ (i.e., enter view $1$)} \label{line:enter_1}

\medskip
\State \textbf{upon} $\mathcal{CX}(\mathit{view}_i).\mathsf{completed}$: \label{line:view_timer_expire} \BlueComment{current \block instance (i.e., current view) has completed}
\State \hskip2em \textbf{broadcast} $\langle \textsc{start-view}, \mathit{view}_i + 1 \rangle$ \label{line:send_complete_1} \BlueComment{start transiting to the next view}

\medskip
\State \textbf{upon} exists $\mathsf{View}$ $V$ such that $\langle \textsc{start-view}, V \rangle$ is received from $t + 1$ processes and $\mathit{helped}_i[V] = \mathit{false}$: \label{line:receive_plurality}
\State \hskip2em $\mathit{helped}_i[V] \gets \mathit{true}$
\State \hskip2em \textbf{broadcast} $\langle \textsc{start-view}, V \rangle$ \label{line:send_complete_2}



\medskip
\State \textbf{upon} exists $\mathsf{View}$ $V > \mathit{view}_i$ such that $\langle \textsc{start-view}, V \rangle$ is received from $2t + 1$ processes: \label{line:receive_quorum}
\label{line:send_complete_3_prime}
\State \hskip2em \textbf{wait for} $\mathcal{CX}(V - 1).\mathsf{validate}(v \in \mathsf{Value})$ \label{line:wait_to_validate} \BlueComment{wait for a value to propose to the new \block instance}
\State \hskip2em \textbf{invoke} $\mathcal{CX}(\mathit{view}_i).\mathsf{abandon}$ \BlueComment{stop participating in the current \block instance} \label{line:leave_v}
\State \hskip2em \textbf{invoke} $\mathcal{CX}(V).\mathsf{propose}(v)$ \BlueComment{start the new \block instance (i.e., enter new view)}  \label{line:enter_v}
\State \hskip2em $\mathit{view}_i \gets V$ \label{line:update_view_i} \BlueComment{update the current view}


\medskip
\State \textbf{upon} $\mathcal{CX}(\mathit{view}_i).\mathsf{decide}(v' \in \mathsf{Value})$: \label{line:decide_from_core} \BlueComment{decided from \block}
\State \hskip2em \textbf{invoke} $\mathcal{F}.\mathsf{to\_finish}(v')$ \label{line:to_finish_name}


\medskip 
\State \textbf{upon} $\mathcal{F}.\mathsf{finish}(v')$: \label{line:receive_finish_name}
\State \hskip2em \textbf{trigger} $\mathsf{decide}(v')$ \label{line:name_decide} \BlueComment{decide from \name}
\State \hskip2em \textbf{invoke} $\mathcal{CX}(\mathit{view}_i).\mathsf{abandon}$ \BlueComment{stop participating in the current \block instance}
\State \hskip2em \textbf{halt} \BlueComment{stop sending any messages and reacting to any received messages}  \label{line:name_stop} 


\end{algorithmic} 
\end{algorithm}

\smallskip
\noindent \textbf{Pseudocode description.}
We say that process $p_i$ \emph{enters} view $V$ once $p_i$ invokes a $\mathcal{CX}(V).\mathsf{propose}(\cdot)$ request (line~\ref{line:enter_1} or line~\ref{line:enter_v}).
Moreover, a process $p_i$ \emph{completes} view $V$ once $p_i$ receives a $\mathsf{completed}$ indication from $\mathcal{CX}(V)$ (line~\ref{line:view_timer_expire}).
Process $p_i$ keeps track of its \emph{current view} using the $\mathit{view}_i$ variable: $\mathit{view}_i$ is the last view entered by $p_i$. 
When process $p_i$ proposes to \name (line~\ref{line:name_propose}), $p_i$ forwards the proposal to $\mathcal{CX}(1)$ (line~\ref{line:enter_1}), i.e., $p_i$ enters view $1$.
Once process $p_i$ completes its current view (line~\ref{line:view_timer_expire}), $p_i$ starts transiting to the next view: process $p_i$ sends a \textsc{start-view} message for the next view (line~\ref{line:send_complete_1}), illustrating its will to enter the next view.
When $p_i$ receives $t + 1$ \textsc{start-view} messages for the same view (line~\ref{line:receive_plurality}), $p_i$ ``helps'' a transition to that view by broadcasting its own \textsc{start-view} message (line~\ref{line:send_complete_2}).
Finally, when $p_i$ receives $2t + 1$ \textsc{start-view} messages for any view $V$ greater than its current view (line~\ref{line:receive_quorum}), $p_i$ performs the following steps: (1) $p_i$ waits until it validates any value $v$ from $\mathcal{CX}(V - 1)$ (line~\ref{line:wait_to_validate}), (2) $p_i$ abandons its current (stale) view (line~\ref{line:leave_v}), (3) $p_i$ enters view $V$ with value $v$ (line~\ref{line:enter_v}), and (4) $p_i$ updates its current view to $V$ (line~\ref{line:update_view_i}).

Once process $p_i$ decides some value $v'$ from a \block instance associated with its current view (line~\ref{line:decide_from_core}), $p_i$ inputs $v'$ to the finisher primitive (line~\ref{line:to_finish_name}).
Lastly, when $p_i$ receives a $\mathsf{finish}(v^*)$ indication from the finisher primitive (line~\ref{line:receive_finish_name}), $p_i$ decides $v^*$ from \name (line~\ref{line:name_decide}) and halts (line~\ref{line:name_stop}).


\subsection{Proof of Correctness \& Complexity} We now prove the correctness and complexity of \name.




\smallskip
\noindent \textbf{Proof of correctness.} First, we show that if a correct process decides a value $v$ from $\mathcal{CX}(V)$, for any view $V$, then all correct processes that propose to $\mathcal{CX}(V')$ do propose value $v$, for any view $V' > V$.

\begin{lemma} \label{lemma:decide_from_block_propose_the_same}
Let a correct process decide a value $v$ from $\mathcal{CX}(V)$, where $V$ is any view.
If a correct process proposes a value $v'$ to $\mathcal{CX}(V')$, for any view $V' > V$, then $v' = v$. 
\end{lemma}
\begin{proof}
We prove the lemma by induction.

\medskip
\noindent \emph{Base step: We prove that if a correct process proposes $v'$ to $\mathcal{CX}(V + 1)$, then $v' = v$.}
\\ Let $p_i$ be any correct process that proposes $v'$ to $\mathcal{CX}(V + 1)$ (line~\ref{line:enter_v}).
Hence, $p_i$ has previously validated $v'$ from $\mathcal{CX}(V)$ (line~\ref{line:wait_to_validate}).
As a correct process decides $v$ from $\mathcal{CX}(V)$, the agreement property of $\mathcal{CX}(V)$ ensures that $v' = v$.

\medskip
\noindent \emph{Inductive step: If a correct process proposes $v'$ to $\mathcal{CX}(V')$, for some $V' > V$, then $v' = v$. We prove that if a correct process proposes $v''$ to $\mathcal{CX}(V' + 1)$, then $v'' = v$.}
\\ Let $p_i$ be any correct process that proposes $v''$ to $\mathcal{CX}(V' + 1)$ (line~\ref{line:enter_v}).
Hence, $p_i$ has previously validated $v''$ from $\mathcal{CX}(V')$ (line~\ref{line:wait_to_validate}).
Due to the inductive hypothesis, all correct processes that propose to $\mathcal{CX}(V')$ do so with value $v$.
Therefore, the strong validity property of $\mathcal{CX}(V')$ ensures that $v'' = v$.
\end{proof}

The following lemma proves that no two correct processes decide different values from (potentially different) instances of \block.

\begin{lemma} \label{lemma:decide_from_block_same}
Let a correct process $p_i$ decide a value $v_i$ from $\mathcal{CX}(V_i)$, where $V_i$ is any view.
Moreover, let another correct process $p_j$ decide a value $v_j$ from $\mathcal{CX}(V_j)$, where $V_j$ is any view.
Then, $v_i = v_j$.
\end{lemma}
\begin{proof}
If $V_i = V_j$, the lemma holds due to the agreement property of $\mathcal{CX}(V_i = V_j)$.
Suppose $V_i \neq V_j$; without loss of generality, let $V_i < V_j$.
Due to \Cref{lemma:decide_from_block_propose_the_same}, all correct processes that propose to $\mathcal{CX}(V_j)$ do so with value $v_i$.
Therefore, due to the strong validity property of $\mathcal{CX}(V_j)$, $v_j = v_i$. 
\end{proof}

Next, we prove that there exists a common value $v^{\star}$ such that if a correct process invokes a $\mathcal{F}.\mathsf{to\_finish}(v)$ request, then $v = v^{\star}$.

\begin{lemma} \label{lemma:same_to_finish}
Let a correct process $p_i$ invoke a $\mathcal{F}.\mathsf{to\_finish}(v_i)$ request.
Moreover, let another correct process $p_j$ invoke a $\mathcal{F}.\mathsf{to\_finish}(v_j)$ request.
Then, $v_i = v_j$.
\end{lemma}
\begin{proof}
As $p_i$ invokes a $\mathcal{F}.\mathsf{to\_finish}(v_i)$ request (line~\ref{line:to_finish_name}), $p_i$ has previously decided $v_i$ from $\mathcal{CX}(V_i)$, for some view $V_i$.
Similarly, $p_j$ has decided $v_j$ from $\mathcal{CX}(V_j)$, for some view $V_j$.
Therefore, $v_i = v_j$ due to \Cref{lemma:decide_from_block_same}.
\end{proof}

We are finally ready to prove that \name satisfies agreement.

\begin{theorem} [Agreement] \label{theorem:name_agreement}
\name (\Cref{algorithm:name_pseudocode_2}) satisfies agreement.
\end{theorem}
\begin{proof}
Suppose a correct process $p_i$ decides a value $v_i \in \mathsf{Value}$ (line~\ref{line:name_decide}).
Moreover, suppose another correct process $p_j$ decides a value $v_j \in \mathsf{Value}$ (line~\ref{line:name_decide}).
As \Cref{lemma:same_to_finish} guarantees that $\mathcal{F}$ works according to its specification, $v_i = v_j$.
\end{proof}

Next, we prove that \name satisfies external validity.

\begin{theorem} [External validity] \label{theorem:name_external_validity}
\name (\Cref{algorithm:name_pseudocode_2}) satisfies external validity.
\end{theorem}
\begin{proof}
Suppose a correct process decides a value $v$ (line~\ref{line:name_decide}).
Hence, that correct process has previously received a $\mathsf{finish}(v)$ indication from $\mathcal{F}$ (line~\ref{line:receive_finish_name}).
Moreover, as \Cref{lemma:same_to_finish} guarantees that $\mathcal{F}$ works according to its specification, a correct process has invoked a $\mathcal{F}.\mathsf{to\_finish}(v)$ request (due to the integrity property of $\mathcal{F}$) upon deciding $v$ from $\mathcal{CX}(V)$ (line~\ref{line:decide_from_core}), for some view $V$.
Therefore, due to the external validity property of $\mathcal{CX}(V)$, $v$ is valid.
\end{proof}



The following theorem proves the strong validity property of \name.

\begin{theorem} [Strong validity] \label{theorem:name_strong_validity}
\name (\Cref{algorithm:name_pseudocode_2}) satisfies strong validity.
\end{theorem}
\begin{proof}
Suppose all correct processes propose the same value $v$ to \name.
Moreover, let a correct process $p_i$ decide some value $v'$ (line~\ref{line:name_decide}).
Hence, process $p_i$ has received a $\mathcal{F}.\mathsf{finish}(v')$ indication (line~\ref{line:receive_finish_name}).
Due to the integrity property of $\mathcal{F}$, a correct process had invoked a $\mathcal{F}.\mathsf{to\_finish}(v')$ upon deciding $v'$ from $\mathcal{CX}(V')$ (line~\ref{line:decide_from_core}), for some view $V'$.
To conclude the proof, we show by induction that all correct processes must have proposed $v$ to $\mathcal{CX}(V')$.

\medskip
\noindent \emph{Base step: We prove that if a correct process proposes $v^*$ to $\mathcal{CX}(1)$, then $v^* = v$.}
\\ The statement holds as all correct processes propose $v$ to \name.

\medskip
\noindent \emph{Inductive step: If a correct process proposes $v''$ to $\mathcal{CX}(V'')$, for some $V'' \geq 1$, then $v'' = v$. We prove that if a correct process proposes $v^*$ to $\mathcal{CX}(V'' + 1)$, then $v^* = v$.}
\\ Let $p_i$ be any correct process that proposes $v^*$ to $\mathcal{CX}(V'' + 1)$ (line~\ref{line:enter_v}).
Therefore, $p_i$ has previously validated $v^*$ from $\mathcal{CX}(V'')$ (line~\ref{line:wait_to_validate}).
Due to the inductive hypothesis, all correct processes that propose to $\mathcal{CX}(V'')$ do  so with value $v$.
Therefore, the strong validity property of $\mathcal{CX}(V'')$ ensures that $v^* = v$.

\medskip
As shown above, all correct processes propose $v$ to $\mathcal{CX}(V')$.
Therefore, $v' = v$ due to the strong validity property of $\mathcal{CX}(V')$.
\end{proof}

To prove the termination property of \name, we start by showing that if a correct process decides, all correct processes eventually decide.

\begin{lemma} \label{lemma:one_decide_all_decide}
If any correct process decides at some time $\tau$, then all correct processes decide by time $\max(\tau, \text{GST}) + 2\delta$.
\end{lemma}
\begin{proof}
The lemma follows directly from the totality property of $\mathcal{F}$.
\end{proof}

The following lemma proves that, for any view $V$, the first $\langle \textsc{start-view}, V \rangle$ message broadcast by a correct process is broadcast at line~\ref{line:send_complete_1}.

\begin{lemma} \label{lemma:first_start_view_message}
For any view $V$, the first $\langle \textsc{start-view}, V \rangle$ message broadcast by a correct process is broadcast at line~\ref{line:send_complete_1}.
\end{lemma}
\begin{proof}
By contradiction, suppose the first \textsc{start-view} message for view $V$ broadcast by a correct process is broadcast at line~\ref{line:send_complete_2}; let $p_i$ be the sender of the message.
Prior to sending the message, $p_i$ has received a \textsc{start-view} message for $V$ from a correct process (due to the rule at line~\ref{line:receive_plurality}).
Therefore, we reach a contradiction.
\end{proof}

Next, we prove that if a correct process enters a view $V > 1$, view $V - 1$ was previously entered and completed by a correct process.
Recall that a correct process enters (resp., completes) some view $V^*$ if and only if that process invokes a $\mathcal{CX}(V^*).\mathsf{propose}(\cdot)$ request (resp., receives a $\mathcal{CX}(V^*).\mathsf{completed}$ indication).

\begin{lemma} \label{lemma:enter_next_previous_enter_complete}
If any correct process enters any view $V > 1$, then a correct process has previously entered and completed view $V - 1$.
\end{lemma}
\begin{proof}
Let a correct process $p_i$ enter view $V > 1$ (line~\ref{line:enter_v}).
Hence, $p_i$ has previously received a \textsc{start-view} message for view $V$ from a correct process (due to the rule at line~\ref{line:receive_quorum}).
As the first correct process to broadcast such a message does so at line~\ref{line:send_complete_1} (by \Cref{lemma:first_start_view_message}), that process has previously completed view $V - 1$ (line~\ref{line:view_timer_expire}).
Moreover, due to the integrity property of $\mathcal{CX}(V - 1)$, that correct process had entered view $V - 1$ prior to $p_i$ entering view $V$.
\end{proof}

The following lemma proves that if no correct process ever decides from \name, every view is eventually entered by a correct process.

\begin{lemma} \label{lemma:no_decision_every_view}
If no correct process ever decides, then every view is eventually entered by a correct process.
\end{lemma}
\begin{proof}
By contradiction, suppose this is not the case.
Let $V + 1$ be the smallest view that is not entered by any correct process.
As each correct process initially enters view $1$ (line~\ref{line:enter_1}), $V + 1 \geq 2$.
Moreover, by \Cref{lemma:enter_next_previous_enter_complete}, no correct process enters any view greater than $V + 1$.
Lastly, as no correct process enters any view greater than $V$, the $\mathit{view}_i$ variable cannot take any value greater than $V$ at any correct process $p_i$.
We prove the lemma through a sequence of intermediate results.

\medskip
\noindent \emph{Step 1.} \emph{If $V > 1$, then every correct process $p_i$ eventually broadcasts a $\langle \textsc{start-view}, V \rangle$ message.}
\\ Let $p_j$ be any correct process that enters view $V > 1$; such a process exists as $V$ is entered by a correct process.
Prior to entering view $V$ (line~\ref{line:enter_v}), $p_j$ has received $2t + 1$ $\langle \textsc{start-view}, V \rangle$ messages (due to the rule at line~\ref{line:receive_quorum}), out of which (at least) $t + 1$ are sent by correct processes.
Therefore, every correct process eventually receives the aforementioned $t + 1$ \textsc{start-view} messages (line~\ref{line:receive_plurality}), and broadcasts a $\langle \textsc{start-view}, V \rangle$ message at line~\ref{line:send_complete_2} (if it has not previously done so).

\medskip
\noindent \emph{Step 2.} \emph{Every correct process $p_i$ eventually enters view $V$.}
\\ If $V = 1$, the statement of the lemma holds as every correct process enters view $1$ (line~\ref{line:enter_1}) immediately upon starting. 

Hence, let $V > 1$.
By the statement of the first step, every correct process eventually broadcasts a $\langle \textsc{start-view}, V \rangle$ message.
Therefore, every correct process $p_i$ eventually receives $2t + 1$ $\langle \textsc{start-view}, V \rangle$ messages.
When this happens, there are two possibilities:
\begin{compactitem}
    \item Let $\mathit{view}_i < V$: In this case, the rule at line~\ref{line:receive_quorum} activates.
    Moreover, as view $V - 1$ has been completed by a correct process (by \Cref{lemma:enter_next_previous_enter_complete}), the totality property of $\mathcal{CX}(V - 1)$ ensures that $p_i$ eventually validates a value from $\mathcal{CX}(V - 1)$ (line~\ref{line:wait_to_validate}).
    Therefore, $p_i$ indeed enters $V$ in this case (line~\ref{line:enter_v}).

    \item Let $\mathit{view}_i = V$: In this case, $p_i$ has already entered view $V$.
\end{compactitem}

\medskip
\noindent \emph{Epilogue.} Due to the statement of the second step, every correct process eventually enters view $V$.
Moreover, no correct process ever abandons view $V$ (i.e., invokes $\mathcal{CX}(V).\mathsf{abandon}$ at line~\ref{line:leave_v}) as no correct process ever enters a view greater than $V$ (or halts).
The termination property of $\mathcal{CX}(V)$ ensures that every correct process eventually completes view $V$ (line~\ref{line:view_timer_expire}), and broadcasts a $\langle \textsc{start-view}, V + 1 \rangle$ message (line~\ref{line:send_complete_1}).
Therefore, every correct process eventually receives $n - t \geq 2t + 1$ $\langle \textsc{start-view}, V + 1 \rangle$ messages.
When that happens, (1) the rule at line~\ref{line:receive_quorum} activates at every correct process $p_i$ as $\mathit{view}_i < V + 1$, (2) $p_i$ eventually validates a value from $\mathcal{CX}(V)$ (line~\ref{line:wait_to_validate}) due to the totality property of $\mathcal{CX}(V)$ (recall that view $V$ is completed by a correct process), and (3) $p_i$ enters view $V + 1$ (line~\ref{line:enter_v}).
This represents a contradiction with the fact that view $V + 1$ is never entered by any correct process, which concludes the proof of the lemma.
\end{proof}


We now define the set of views that are entered by a correct process.

\begin{definition} [Entered views]
Let $\mathcal{V} = \{V \in \mathsf{View} \,|\, V \text{ is entered by a correct process}\}$.
\end{definition}

Moreover, we define the first time any correct process enters any view $V \in \mathcal{V}$.

\begin{definition} [First-entering time] \label{definition:first_time}
For any view $V \in \mathcal{V}$, $\tau_V$ denotes the time at which the first correct process enters $V$.
\end{definition}


Finally, we define the smallest view that is entered by every correct process at or after GST.

\begin{definition} [View $\vfinal$] \label{definition:vfinal}
We denote by $\vfinal$ the smallest view that belongs to $\mathcal{V}$ for which $\tau_{\vfinal} \geq \text{GST}$.
If such a view does not exist, then $\vfinal = \bot$.
\end{definition}

Recall that $\Delta_\mathit{total}$ is the minimum time a process has to spend on \block before completing it. Next, observe that, as all correct processes start executing \name before GST, $\vfinal > 1$.
The following lemma proves that no correct process enters any view greater than $\vfinal$ by time $\tau_{\vfinal} + \Delta_{\mathit{total}}$ (if $\vfinal \neq \bot$).

\begin{lemma} \label{lemma:time_greater_than_final}
Let $\vfinal \neq \bot$.
For any view $V \in \mathcal{V}$ such that  $V > \vfinal$, $\tau_V > \tau_{\vfinal} + \Delta_{\mathit{total}} > \tau_{\vfinal} + 2\delta$.
\end{lemma}
\begin{proof}
For view $\vfinal + 1$ to be entered by a correct process, there must exist a correct process that has previously completed view $\vfinal$ (by \Cref{lemma:enter_next_previous_enter_complete}).
As $\tau_{\vfinal} \geq \text{GST}$, the completion time property of $\mathcal{CX}(\vfinal)$ ensures that no correct process completes view $\vfinal$ by time $\tau_{\vfinal} + \Delta_{\mathit{total}}$.
Therefore, $\tau_{\vfinal + 1} > \tau_{\vfinal} + \Delta_{\mathit{total}}$.
Moreover, due to \Cref{lemma:enter_next_previous_enter_complete}, $\tau_V > \tau_{\vfinal} + \Delta_{\mathit{total}}$, for any view $V > \vfinal + 1$. Finally, let us note that $\Delta_{\mathit{total}} > \Delta_{\mathit{shift}} = 2\delta$.
\end{proof}

Assuming that no correct process decides by time $\tau_{V_{\mathit{final}}} + \Delta_{\mathit{total}}$ and $\vfinal \neq \bot$, every correct process decides from $\mathcal{CX}(\vfinal)$ by time $\tau_{V_{\mathit{final}}} + \Delta_{\mathit{total}}$.

\begin{lemma} \label{lemma:enter_vfinal_2_delta}
Let $\vfinal \neq \bot$ and let no correct process decide by time $\tau_{\vfinal} + \Delta_{\mathit{total}}$.
Then, every correct process decides the same value from $\mathcal{CX}(\vfinal)$ by time $\tau_{\vfinal} + \Delta_{\mathit{total}}$.
\end{lemma}
\begin{proof}
We prove the lemma through a sequence of intermediate steps.

\medskip
\noindent \emph{Step 1.} \emph{Every correct process enters view $\vfinal$ by time $\tau_{\vfinal} + 2\delta$.}
\\ Recall that $\vfinal > 1$.
Let $p_i$ be the correct process that enters view $\vfinal$ (line~\ref{line:enter_v}) at time $\tau_{\vfinal} \geq \text{GST}$.
Therefore, $p_i$ has received $2t + 1$ $\langle \textsc{start-view}, \vfinal \rangle$ messages (due to the rule at line~\ref{line:receive_quorum}) by time $\tau_{\vfinal}$.
Among the aforementioned $2t + 1$ \textsc{start-view} messages, at least $t + 1$ are broadcast by correct processes.
Note that \Cref{lemma:enter_next_previous_enter_complete} shows that some correct process $p_l$ has completed view $\vfinal - 1$ by time $\tau_{\vfinal}$.

Now consider any correct process $p_j$.
We prove that $p_j$ broadcasts a \textsc{start-view} message for view $\vfinal$ by time $\tau_{\vfinal} + \delta$.
Indeed, by time $\tau_{\vfinal} + \delta$, $p_j$ receives $t + 1$ $\langle \textsc{start-view}, \vfinal \rangle$ messages (line~\ref{line:receive_plurality}), and broadcasts a $\langle \textsc{start-view}, \vfinal \rangle$ message (line~\ref{line:send_complete_2}) assuming that it has not already done so.

As we have proven, all correct processes broadcast a \textsc{start-view} message for view $\vfinal$ by time $\tau_{\vfinal} + \delta$.
Therefore, every correct process $p_k$ receives $2t + 1$ $\langle \textsc{start-view}, \vfinal \rangle$ messages by time $\tau_{\vfinal} + 2\delta$.
Importantly, when this happens, the rule at line~\ref{line:receive_quorum} activates at process $p_k$ (unless $p_k$ has already entered view $\vfinal$) as the value of the $\mathit{view}_k$ variable cannot be greater than $\vfinal$ due to \Cref{lemma:time_greater_than_final} and the fact that $\Delta_{\mathit{total}} > 2\delta$.
Moreover, due to the totality property of $\mathcal{CX}(\vfinal - 1)$, $p_k$ validates a value from $\mathcal{CX}(\vfinal - 1)$ by time $\tau_{\vfinal} + 2\delta$ (line~\ref{line:wait_to_validate}); recall that some correct process $p_l$ has completed view $\vfinal - 1$ by time $\tau_{\vfinal}$.
Therefore, $p_k$ indeed enters view $\vfinal$ by time $\tau_{\vfinal} + 2\delta$ (line~\ref{line:enter_v}).

\medskip
\noindent \emph{Step 2.} \emph{No correct process abandons view $\vfinal$ by time $\tau_{\vfinal} + \Delta_{\mathit{total}}$.}
\\ As no correct process decides by time $\tau_{\vfinal} + \Delta_{\mathit{total}}$, no correct process halts by time $\tau_{\vfinal} + \Delta_{\mathit{total}}$.
Moreover, no correct process enters any view greater than $\vfinal$ by time $\tau_{\vfinal} + \Delta_{\mathit{total}}$ (due to \Cref{lemma:time_greater_than_final}).
Therefore, the statement holds.

\medskip
\noindent \emph{Epilogue.} Due to the aforementioned two intermediate steps, the precondition of the synchronicity property of $\mathcal{CX}(\vfinal)$ is fulfilled.
Therefore, the synchronicity and agreement properties of $\mathcal{CX}(\vfinal)$ directly imply the lemma.
\end{proof}

We are finally ready to prove the termination property of \name.

\begin{theorem} [Termination] \label{theorem:name_termination}
\name (\Cref{algorithm:name_pseudocode_2}) satisfies termination.
Concretely, if $\vfinal \neq \bot$, every correct process decides by time $\tau_{\vfinal} + \Delta_{\mathit{total}} + 2\delta$.
\end{theorem}
\begin{proof}
If $\vfinal = \bot$, then at least one correct process decides.
(Indeed, if no correct process decides, then \Cref{lemma:no_decision_every_view} proves that $\vfinal \neq \bot$.)
Hence, termination is ensured by \Cref{lemma:one_decide_all_decide}.

Let us now consider the case in which $\vfinal \neq \bot$.
We study two scenarios:
\begin{compactitem}
    \item Let a correct process decide by time $\tau_{\vfinal} + \Delta_{\mathit{total}}$.
    In this case, the theorem holds due to \Cref{lemma:one_decide_all_decide}.

    \item Otherwise, all correct processes decide the same value from $\mathcal{CX}(\vfinal)$ by time $\tau_{\vfinal} + \Delta_{\mathit{total}}$ (by \Cref{lemma:enter_vfinal_2_delta}) and invoke a $\mathcal{F}.\mathsf{to\_finish}(\cdot)$ request (line~\ref{line:to_finish_name}).
    Therefore, the theorem holds due to the termination property of $\mathcal{F}$.
\end{compactitem}
Hence, the termination property is ensured even if $\vfinal \neq \bot$.
\end{proof}

\smallskip
\noindent \textbf{Proof of complexity.}
First, we define the greatest view entered by a correct process before GST.

\begin{definition} [View $\vmax$] \label{definition:vmax}
We denote by $\vmax$ the greatest view that belongs to $\mathcal{V}$ for which $\tau_{\vmax} < \text{GST}$.
\end{definition}

Observe that $\vmax$ is well-defined due to the assumption that all correct processes start executing \name before GST.
Importantly, if $\vfinal \neq \bot$ (see \Cref{definition:vfinal}), then $\vfinal = \vmax + 1$ (by \Cref{lemma:enter_next_previous_enter_complete}).
The following lemma shows that if a correct process broadcasts a \textsc{start-view} message for a view $V$, then $V \in \mathcal{V}$ or $V - 1 \in \mathcal{V}$.

\begin{lemma} \label{lemma:v_in_v}
If a correct process broadcasts a \textsc{start-view} message for view $V$, then $V \in \mathcal{V}$ or $V - 1 \in \mathcal{V}$.
\end{lemma}
\begin{proof}
If $|\mathcal{V}| = \infty$, the lemma trivially holds.
Hence, let $|\mathcal{V}| \neq \infty$; let $V^*$ denote the greatest view that belongs to $\mathcal{V}$.
\Cref{lemma:enter_next_previous_enter_complete} guarantees that $V' \in \mathcal{V}$, for every view $V' < V^*$.
By contradiction, suppose there exists a correct process that broadcasts a \textsc{start-view} message for a view $V$ such that $V > V^* + 1$.
Let $p_i$ be the first correct process to broadcast a $\langle \textsc{start-view}, V > V^* + 1 \rangle$ message.
By \Cref{lemma:first_start_view_message}, $p_i$ has previously completed view $V - 1 \geq V^* + 1$.
Due to the integrity property of $\mathcal{CX}(V - 1)$, $p_i$ has entered view $V - 1 \geq V^* + 1$.
Therefore, $V^* + 1 \in \mathcal{V}$, which contradicts the fact that $V^*$ is the greatest view that belongs to $\mathcal{V}$.
\end{proof}

Next, we prove that any correct process broadcasts at most two \textsc{start-view} messages for any view $V$.

\begin{lemma} \label{lemma:at_most_two_start_view}
Any correct process broadcasts at most two \textsc{start-view} messages for any view $V$.
\end{lemma}
\begin{proof}
Let $p_i$ be any correct process.
Process $p_i$ sends at most one $\langle \textsc{start-view}, V \rangle$ message at line~\ref{line:send_complete_1} as $p_i$ enters monotonically increasing views (i.e., it is impossible for $p_i$ to complete view $V$ more than once).
Moreover, process $p_i$ sends at most one $\langle \textsc{start-view}, V \rangle$ message at line~\ref{line:send_complete_2} due to the $\mathit{helped}_i[]$ variable, which concludes the proof.
\end{proof}

We next prove that $\vmax \in O(1)$ (i.e., it does not depend on $n$).

\begin{lemma} \label{lemma:vmax_constant}
$\vmax \in O(1)$.
\end{lemma}
\begin{proof}
The lemma holds as $\vmax$ does not depend on $n$; $\vmax$ depends on $\text{GST}$, the message delays before GST and the clock drift.
\end{proof}

The following lemma proves that if $\vfinal \neq \bot$, then $\vfinal \in O(1)$.

\begin{lemma} \label{lemma:vfinal_constant}
If $\vfinal \neq \bot$, then $\vfinal \in O(1)$.
\end{lemma}
\begin{proof}
Recall that if $\vfinal \neq \bot$, $\vfinal = \vmax + 1$.
As $\vmax \in O(1)$ (by \Cref{lemma:vmax_constant}), $\vfinal \in O(1)$.
\end{proof}

Next, we prove that if $\vfinal = \bot$, then $\vmax$ is the greatest view that belongs to $\mathcal{V}$.

\begin{lemma} \label{lemma:vfinal_bot}
If $\vfinal = \bot$, then $\vmax$ is the greatest view that belongs to $\mathcal{V}$.
\end{lemma}
\begin{proof}
By contradiction, suppose there exists a view $V^* \in \mathcal{V}$ such that $V^* > \vmax$.
We distinguish two possibilities regarding $\tau_{V^*}$:
\begin{compactitem}
    \item Let $\tau_{V^*} < \text{GST}$: This case is impossible as $\vmax$ is the greatest view that belongs to $\mathcal{V}$ entered by a correct process before GST (see \Cref{definition:vmax}).

    \item Let $\tau_{V^*} \geq \text{GST}$: This case is impossible as $\vfinal = \bot$ (see \Cref{definition:vfinal}).
\end{compactitem}
Therefore, the lemma holds.
\end{proof}

The following lemma gives the earliest entering time for each view greater than $\vfinal$ (assuming that $\vfinal \neq \bot$).

\begin{lemma} \label{lemma:starting_time}
If $\vfinal \neq \bot$, then $\tau_{V} > \tau_{V - 1} + \Delta_{\mathit{total}}$, for every view $V \in \mathcal{V}$ such that $V > \vfinal$.
\end{lemma}
\begin{proof}
The proof is similar to that of \Cref{lemma:time_greater_than_final}.
For view $V > \vfinal$ to be entered by a correct process, there must exist a correct process that has previously completed view $V - 1 \geq \vfinal$ (by \Cref{lemma:enter_next_previous_enter_complete}).
As $\tau_{V - 1} \geq \text{GST}$ (due to \Cref{lemma:enter_next_previous_enter_complete} and the fact that $\tau_{\vfinal} \geq \text{GST}$), the completion time property of $\mathcal{CX}(V - 1)$ ensures that no correct process completes view $V - 1$ by time $\tau_{V - 1} + \Delta_{\mathit{total}}$.
Therefore, $\tau_{V} > \tau_{V - 1} + \Delta_{\mathit{total}}$.
\end{proof}

Next, we give an upper bound on the greatest view entered by a correct process assuming that $\vfinal \neq \bot$.

\begin{lemma} \label{lemma:vfinal+2}
Let $\vfinal \neq \bot$, and let $V^*$ be the greatest view that belongs to $\mathcal{V}$.
Then, $V^* < \vfinal + 2$.
\end{lemma}
\begin{proof}
By \Cref{theorem:name_termination}, all correct processes decide (and halt) by time $\tau_{\vfinal} + \Delta_{\mathit{total}} + 2\delta$.
Moreover, $\tau_{\vfinal + 1} > \tau_{\vfinal} + \Delta_{\mathit{total}}$ (by \Cref{lemma:starting_time}).
Furthermore, \Cref{lemma:starting_time} shows that $\tau_{\vfinal + 2} > \tau_{\vfinal + 1} + \Delta_{\mathit{total}} > \tau_{\vfinal} + 2\Delta_{\mathit{total}}$.
As $\Delta_{\mathit{total}} > 2\delta$, we have that $\tau_{\vfinal} + \Delta_{\mathit{total}} + 2\delta < \tau_{\vfinal} + 2\Delta_{\mathit{total}}$, which concludes the proof.
\end{proof}

The last intermediate result shows that the greatest view entered by a correct process does not depend on $n$ (i.e., it is a constant).

\begin{lemma} \label{lemma:greatest_constant}
Let $V^*$ be the greatest view that belongs to $\mathcal{V}$.
Then, $V^* \in O(1)$.
\end{lemma}
\begin{proof}
If $\vfinal = \bot$, then $V^* = \vmax$ (by \Cref{lemma:vfinal_bot}).
Therefore, \Cref{lemma:vmax_constant} concludes the proof.
Otherwise, $V^* < \vfinal + 2$ (by \Cref{lemma:vfinal+2}).
In this case, the lemma holds due to \Cref{lemma:vfinal_constant} in this case.
\end{proof}

We are finally ready to prove the bit complexity of \name.
Recall that $\mathit{bit}(\block)$ denotes the number of bits correct processes collectively send in \block.
Moreover, we denote by $\mathit{bit}(\mathcal{F})$ the number of bits correct processes collectively send in $\mathcal{F}$.

\begin{theorem} [Per-process bit complexity] \label{theorem:name_bit_complexity}
\name achieves $O\big( n + \mathit{pbit}(\block) + \mathit{pbit}(\mathcal{F}) \big)$ per-process bit complexity.
\end{theorem}
\begin{proof}
Every correct process broadcasts at most two \textsc{start-view} messages for any view (by \Cref{lemma:at_most_two_start_view}).
Moreover, \Cref{lemma:v_in_v} proves that, if a correct process sends a \textsc{start-view} message for a value $V$, then (1) $V \in \mathcal{V}$, or (2) $V - 1 \in \mathcal{V}$.
As the greatest view $V^*$ of $\mathcal{V}$ is a constant (due to \Cref{lemma:greatest_constant}), every correct process sends $O(1) \cdot 2 \cdot n = O(n)$ bits via \textsc{start-view} messages.
Moreover, there are $O(1)$ executed instances of \block (due to \Cref{lemma:greatest_constant}).
Finally, every correct process sends $\mathit{pbit}(\mathcal{F})$ bits in $\mathcal{F}$.
Therefore, the per-process bit complexity of \name is $O(n) + O(1) \cdot \mathit{pbit}(\block) + \mathit{pbit}(\mathcal{F}) = O\big( n + \mathit{pbit}(\block) + \mathit{pbit}(\mathcal{F}) \big)$.
\end{proof}

To prove the latency of \name, we first prove a specific property of $\mathcal{CX}(\vmax)$.

\begin{lemma} \label{lemma:time_complete_vmax}
Let (1) all correct processes enter view $\vmax$ by some time $\tau$, and (2) no correct process abandon view $\vmax$ before time $\tau' = \max(\tau, \text{GST}) + \mathit{latency}(\block)$.
Then, all correct processes complete view $\vmax$ by time $\tau'$.
\end{lemma}
\begin{proof}
By time $\max(\tau, \text{GST}) + (\mathit{latency}(\mathcal{GC}_1) \cdot \delta)$, all correct processes decide from $\mathcal{GC}_1$ (i.e., conclude Step 1 of Task 1) as they all overlap while executing $\mathcal{GC}_1$ for at least $\mathit{latency}(\mathcal{GC}_1) \cdot \delta$ time.
Similarly, all correct processes stop executing $\mathcal{A}^S$ (conclude Step 2 of Task 1) by time $\max(\tau, \text{GST}) + (\mathit{latency}(\mathcal{GC}_1) \cdot \delta) + (\mathcal{R} \cdot \Delta_{\mathit{sync}})$.
Furthermore, all correct processes decide from $\mathcal{GC}_2$ by time $\max(\tau, \text{GST}) + (\mathit{latency}(\mathcal{GC}_1) \cdot \delta) + (\mathcal{R} \cdot \Delta_{\mathit{sync}}) + (\mathit{latency}(\mathcal{GC}_2) \cdot \delta)$.
Lastly, all correct processes receive a $\mathsf{completed}$ indication from $\mathcal{VB}$ (and, thus, complete $\vmax$) by time $\max(\tau, \text{GST}) + (\mathit{latency}(\mathcal{GC}_1) \cdot \delta) + (\mathcal{R} \cdot \Delta_{\mathit{sync}}) + (\mathit{latency}(\mathcal{GC}_2) \cdot \delta) + (\mathit{latency}(\mathcal{VB}) \cdot \delta) = \max(\tau, \text{GST}) + \mathit{latency}(\block) = \tau'$.
\end{proof}

Next, we prove that $\tau_{\vfinal} - \text{GST} \leq 2\delta + \mathit{latency}(\block) + 2\delta$ (assuming $\vfinal \neq \bot$).

\begin{lemma} \label{lemma:tau_vfinal_value}
Let $\vfinal \neq \bot$.
Then, $\tau_{\vfinal} - \text{GST} \leq 2\delta + \mathit{latency}(\block) + 2\delta$.
\end{lemma}
\begin{proof}
By contradiction, suppose $\tau_{\vfinal} > \text{GST} + 2\delta + \mathit{latency}(\block) + 2\delta$.
Hence, no correct process enters any view greater than $\vmax$ by time $\text{GST} + 2\delta + \mathit{latency}(\block) + 2\delta$ (by \Cref{lemma:enter_next_previous_enter_complete}).
First, we prove that all correct processes enter view $\vmax$ by time $\text{GST} + 2\delta$.

\medskip
\noindent \emph{Intermediate result: All correct processes enter view $\vmax$ by time $\text{GST} + 2\delta$.}
\\ If $\vmax = 1$, then every correct process enters view $\vmax$ (line~\ref{line:enter_1}) before GST, which proves the statement.

Let $\vmax > 1$.
Let $p_i$ be the correct process that enters view $\vmax$ (line~\ref{line:enter_v}) at time $\tau_{\vmax} < \text{GST}$.
Therefore, $p_i$ has received $2t + 1$ $\langle \textsc{start-view}, \vmax \rangle$ messages (due to the rule at line~\ref{line:receive_quorum}) by time $\tau_{\vmax}$.
Among the aforementioned $2t + 1$ \textsc{start-view} messages, at least $t + 1$ are broadcast by correct processes.
Note that \Cref{lemma:enter_next_previous_enter_complete} shows that some correct process $p_l$ has completed view $\vmax - 1$ by time $\tau_{\vmax}$.

Now consider any correct process $p_j$.
We prove that $p_j$ broadcasts a \textsc{start-view} message for view $\vmax$ by time $\text{GST} + \delta$.
Indeed, by time $\text{GST} + \delta$, $p_j$ receives $t + 1$ $\langle \textsc{start-view}, \vmax \rangle$ messages (line~\ref{line:receive_plurality}), and broadcasts a $\langle \textsc{start-view}, \vmax \rangle$ message (line~\ref{line:send_complete_2}) assuming that it has not already done so.

As we have proven, all correct processes broadcast a \textsc{start-view} message for view $\vmax$ by time $\text{GST} + \delta$.
Therefore, every correct process $p_k$ receives $2t + 1$ $\langle \textsc{start-view}, \vfinal \rangle$ messages by time $\text{GST} + 2\delta$.
Importantly, when this happens, the rule at line~\ref{line:receive_quorum} activates at process $p_k$ (unless $p_k$ has already entered view $\vmax$) as the value of the $\mathit{view}_k$ variable cannot be greater than $\vmax$ due to the fact that no correct process enters any view greater than $\vmax$ by time $\text{GST} + 2\delta + \mathit{latency}(\block) + 2\delta > \text{GST} + 2\delta$.
Moreover, due to the totality property of $\mathcal{CX}(\vmax - 1)$, $p_k$ validates a value from $\mathcal{CX}(\vmax - 1)$ by time $\text{GST} + 2\delta$ (line~\ref{line:wait_to_validate}); recall that some correct process $p_l$ has completed view $\vmax - 1$ by time $\text{GST}$.
Therefore, $p_k$ indeed enters view $\vmax$ by time $\text{GST} + 2\delta$ (line~\ref{line:enter_v}).

\medskip
\noindent \emph{Epilogue.}
Due to the intermediate result and \Cref{lemma:time_complete_vmax}, all correct processes complete view $\vmax$ by time $\text{GST} + 2\delta + \mathit{latency}(\block)$ (line~\ref{line:view_timer_expire}).
Therefore, every correct process broadcasts a \textsc{start-view} message for $\vmax + 1 = \vfinal$ by time $\text{GST} + 2\delta + \mathit{latency}(\block)$ (line~\ref{line:send_complete_1}), which implies that every correct process receives $n - t \geq 2t + 1$ \textsc{start-view} messages for view $\vfinal$ by time $\text{GST} + 2\delta + \mathit{latency}(\block) + \delta$ (line~\ref{line:receive_quorum}).
Moreover, as all correct processes complete view $\vmax$ by time $\text{GST} + 2\delta + \mathit{latency}(\block)$, all correct processes validate a value from $\mathcal{CX}(\vmax)$ by time $\text{GST} + 2\delta + \mathit{latency}(\block) + 2\delta$ (line~\ref{line:wait_to_validate}), which proves that $\tau_{\vfinal}  - \text{GST} \leq 2\delta + \mathit{latency}(\block) + 2\delta$.
\end{proof}

Finally, we are ready to prove \name's latency.

\begin{theorem} [Latency] \label{theorem:name_latency}
\name (\Cref{algorithm:name_pseudocode_2}) achieves $O\big( \mathit{latency}(\block) \big)$ latency. 
\end{theorem}
\begin{proof}
If $\vfinal = \bot$, \name's latency is $0$.
Hence, let $\vfinal \neq \bot$.
By \Cref{theorem:name_termination}, all correct processes decide by time $\tau_{\vfinal} + \Delta_{\mathit{total}} + 2\delta$.
Due to \Cref{lemma:tau_vfinal_value}, $\tau_{\vfinal} - \text{GST} \in O\big( \mathit{latency}(\block) \big)$.
Therefore, the latency of \name is $\tau_{\vfinal} + \Delta_{\mathit{total}} + 2\delta - \text{GST} \in O(\mathit{latency}(\block) + \Delta_{\mathit{total}}) = O(\mathit{latency}(\block))$.
\end{proof}

\smallskip
\noindent \textbf{On limiting the number of views for which \textsc{start-view} messages are sent.}
Recall that our implementation ensures that all correct processes enter monotonically increasing views, i.e., if a correct process enters a view $v'$ after it has previously entered a view $v$, then $v' > v$.
Moreover, recall that $|\mathcal{V}| \leq \vmax + 2$ (due to \cref{lemma:vfinal_bot,lemma:vfinal+2}).
Therefore, once a correct process enters view $\vmax$ (or any greater view), only $O(1)$ views are left for the process to go through before it terminates.
Let us denote by $\tau_{\geq \vmax}(p_i)$ the time at which process $p_i$ enters a view greater than or equal to $\vmax$, for every correct process $p_i$.
Importantly, our implementation of \name allows for any correct process $p_i$ to visit \emph{all} views smaller than $\vmax$ during the time period $\mathcal{T}_{\mathit{unstable}}(p_i) = [\text{GST}, \tau_{\geq \vmax}(p_i))$.
Therefore, every correct process $p_i$ can visit unboundedly many views (though independent of $n$) during the time period $\mathcal{T}_{\mathit{unstable}}(p_i)$.

Importantly, the problem above can easily be tackled.
Employing the ``waiting'' strategy proposed in~\cite{civit2022byzantine} suffices to guarantee that any correct process $p_i$ visits only $O(1)$ views during the time period $\mathcal{T}_{\mathit{unstable}}(p_i)$.
Let us briefly describe the aforementioned strategy.
When a correct process $p_i$ learns about a new view, $p_i$ does not immediately enter that view (as is the case in our current implementation).
Instead, process $p_i$ waits $\delta$ time; this $\delta$ time is used to learn about other (potentially more advanced) views.
Hence, at least $\delta$ time elapses (after GST) between any two entrances performed by process $p_i$.
As it is ensured that $p_i$ enters $\vmax$ (or a greater view) within $O(1) \cdot \delta$ time after GST (i.e., $\tau_{\geq \vmax}(p_i) - \text{GST} \leq O(1) \cdot \delta$), the proposed strategy ensures that $p_i$ visit only $O(1)$ views during the time period $\mathcal{T}_{\mathit{unstable}}(p_i)$.
We opted not to include this logic in our implementation of \name for the sake of simplicity and presentation.

\section{Existing Primitives} \label{section:existing_primitives}

This section outlines the fundamental building blocks utilized in our implementations of (1) rebuilding broadcast (\Cref{section:rebuilding_broadccast}), (2) graded consensus (\Cref{section:graded_consensus_concrete_implementations}), and (3) validation broadcast (\Cref{section:validation_broadcast_concrete_implementations}).

\smallskip
\noindent \textbf{Error-correcting codes.}
We use error-correcting codes.
Concretely, we use the standard Reed-Solomon (RS) codes~\cite{reed1960polynomial}.
We denote by $\mathsf{RSEnc}$ and $\mathsf{RSDec}$ the encoding and decoding algorithms.
Briefly, $\mathsf{RSEnc}(M, m, k)$ takes as input a message $M$ consisting of $k$ symbols, treats it as a polynomial of degree $k - 1$, and outputs $m$ evaluations of the corresponding polynomial. Moreover, each symbol consists of $O\big( \frac{|M|}{k} + \log(m) \big)$ bits.
On the other hand, $\mathsf{RSDec}(k, r, T)$ takes as input a set of symbols $T$ (some of which may be incorrect), and outputs a polynomial of degree $k - 1$ (i.e., $k$ symbols) by correcting up to $r$ errors (incorrect symbols) in $T$.
Importantly, $\mathsf{RSDec}$ can correct up to $r$ errors in $T$ and output the original message if $|T| \geq k + 2r$~\cite{macwilliams1977theory}. 
One concrete instantiation of RS codes is the Gao algorithm~\cite{gao2003new}.

\smallskip
\noindent \textbf{Collision-resistant hash function.}
We assume a cryptographic collision-resistant hash function $\mathsf{hash}(\cdot)$ that guarantees that a computationally bounded adversary cannot devise two inputs $i_1$ and $i_2$ such that $\mathsf{hash}(i_1) = \mathsf{hash}(i_2)$, except with a negligible probability.
Each hash value is of size $\kappa$ bits; we assume $\kappa > \log(n)$.\footnote{If $\kappa \leq \log(n)$, $t \in O(n)$ faulty processes would have computational power exponential in $\kappa$.}

\smallskip
\noindent \textbf{Cryptographic accumulators.}
We use standard cryptographic accumulators~\cite{abraham2023communication,nayak2020improved}.
A cryptographic accumulator scheme constructs an accumulation value for a set of values and produces a witness for each value in the set.
Given the accumulation value and a witness, any process can verify if a value is indeed in the set.
More formally, given a parameter $\kappa$ and a set $\mathcal{D}$ of $n$ values $d_1, ..., d_n$, an accumulator has the following components:
\begin{compactitem}
    \item $\mathsf{Gen}(1^\kappa, n)$: This algorithm takes a parameter $\kappa$ represented in the unary form $1^\kappa$ and an accumulation threshold $n$ (an upper bound on the number of values that can be accumulated securely); returns an accumulator key $a_k$.
    The accumulator key $a_k$ is public.

    \item $\mathsf{Eval}(a_k, \mathcal{D})$: This algorithm takes an accumulator key $a_k$ and a set $\mathcal{D}$ of values to be accumulated; returns an accumulation value $z$ for the value set $\mathcal{D}$.

    \item $\mathsf{CreateWit}(a_k, z, d_i, \mathcal{D})$: This algorithm takes an accumulator key $a_k$, an accumulation value $z$ for $\mathcal{D}$ and a value $d_i$; returns $\bot$ if $d_i \notin \mathcal{D}$, and a witness $\omega_i$ if $d_i \in \mathcal{D}$.

    \item $\mathsf{Verify}(a_k, z, \omega_i, d_i)$: This algorithm takes an accumulator key $a_k$, an accumulation value $z$ for $\mathcal{D}$ and a value $d_i$; returns $\mathit{true}$ if $\omega_i$ is the witness for $d_i \in \mathcal{D}$, and $\mathit{false}$ otherwise.
\end{compactitem}
Concretely, we use Merkle trees~\cite{merkle1987digital} as our cryptographic accumulators as they are purely hash-based.
Elements of $\mathcal{D}$ form the leaves of a Merkle tree, the accumulator key is a specific hash function, an accumulation value is the Merkle tree root, and a witness is a Merkle tree proof. 
Importantly, the size of an accumulation value is $O(\kappa)$ bits, and the size of a witness is $O(\log(n) \kappa)$ bits, where $\kappa$ denotes the size of a hash value.
Throughout the remainder of the paper, we refrain from explicitly mentioning the accumulator key $a_k$ as we assume that the associated hash function is fixed. 
Moreover, the accumulator scheme is assumed to be collision-free, i.e., for any accumulator key $ak \gets \mathsf{Gen}(1^{\kappa}, n)$, it is computationally impossible to establish $(\{d_1, ..., d_n\}, d', w')$ such that (1) $d' \notin \{d_1, ..., d_n\}$, (2) $z \leftarrow \mathsf{Eval}(ak, \{d_1, ..., d_n\})$, and (3) $\mathsf{Verify}(ak, z, w', d') = \mathit{true}$. In our case, this property is reduced to the collision resistance of the underlying hash function. 
It is important to mention that bilinear accumulators allow for witnesses of size $\kappa$ bits; however, they require a trusted powers-of-tau setup to establish $q$-SDH public parameters \cite{DXR22a}.

In the rest of the paper, $\MRoot{v} = \mathsf{Eval}\Big( \{\big( 1,P_v(1) \big), ..., \big( n,P_v(n) \big)\} \Big)$, where $\mathsf{Eval}$ is the accumulator evaluation function (see above) and $[P_v(1), ..., P_v(n)] = \mathsf{RSEnc}(v, n, t + 1)$ is the Reed-Solomon encoding of value $v$ (see the paragraph ``Error-correcting codes'' in this section).
We underline that this construction is standard in the literature (see, e.g., \cite{nayak2020improved}). 

\smallskip
\noindent \textbf{Attiya-Welch graded consensus.}
We utilize a graded consensus algorithm proposed by Attiya and Welch~\cite{AttiyaWelch23}; recall that the specification of the graded consensus problem is given in \Cref{subsection:graded_consensus_building_blocks}.
Specifically, the Attiya-Welch (AW, for short) graded consensus algorithm tolerates up to $t < n / 3$ Byzantine processes, incurs $O(n L)$ per-process bit complexity with $L$-sized values, and terminates in $9$ asynchronous rounds.
Crucially, the AW graded consensus algorithm, in addition to the properties specified in \Cref{subsection:graded_consensus_building_blocks}, satisfies the following property:
\begin{compactitem}
    \item \emph{Safety:} If a correct process decides a pair $(v', \cdot)$, then $v'$ has been proposed by a correct process.
\end{compactitem}

\smallskip
\noindent \textbf{Reducing broadcast.}
The reducing broadcast primitive is proposed in~\cite{DBLP:journals/acta/MostefaouiR17}. The corresponding implementation tolerates up to $t < n / 3$ Byzantine processes, incurs $O(n L)$ per-process bit complexity with $L$-sized values, and terminates in $2$ asynchronous rounds. 
The goal of the primitive is to reduce the number of different values held by correct processes to a constant.
The specification is associated with the default value $\bot_{rd} \notin \mathsf{Value}$.
Reducing broadcast exposes the following interface:
\begin{compactitem}
    \item \textbf{request} $\mathsf{broadcast}(v \in \mathsf{Value})$: a process broadcasts value $v$.

    \item \textbf{request} $\mathsf{abandon}$: a process abandons (i.e., stops participating in) reducing broadcast.

    \item \textbf{indication} $\mathsf{deliver}(v' \in \mathsf{Value} \cup \{\bot_{rd}\})$: a process delivers value $v'$ ($v'$ can be $\bot_{rd}$).
\end{compactitem}
Every correct process broadcasts at most once.

The following properties are ensured by reducing broadcast:
\begin{compactitem}
    \item \emph{Validity:} If all correct processes that broadcast do so with the same value, no correct process delivers $\bot_{rd}$.

    \item \emph{Safety:} If a correct process delivers a value $v' \in \mathsf{Value}$ ($v' \neq \bot_{rd}$), then a correct process has previously broadcast $v'$.

    \item \emph{Reduction:} The number of values (including $\bot_{rd}$) that are delivered by correct processes, in total, is $O(1)$.

    \item \emph{Termination:} If all correct processes broadcast and no correct process abandons reducing broadcast, then every correct process eventually delivers a value.
\end{compactitem}

\section{\reduceacool: A-COOL Reduction} \label{section:acool_reduction}

This section introduces a distributed algorithm named \reduceacool that is crucial in our cryptography-less implementations of graded consensus and validation broadcast optimized for long values.
\reduceacool tolerates up to $t < n / 5$ Byzantine processes, and it can be seen as a version of the A-COOL Byzantine agreement protocol introduced by Li and Chen~\cite{li2021communication} without its underlying reliance on a one-bit Byzantine agreement.
Moreover, the \reduceacool algorithm, inspired by the approach taken in~\cite{alhaddad2022balanced}, clarifies the handling of messages made implicit in the original A-COOL algorithm.
Concretely, a message from a process $p$ for phase $q > 1$ is not processed by a correct process unless a message from $p$ for phase $q - 1$ has previously been processed.\footnote{We underline that this detail is not explicitly mentioned in~\cite{li2021communication}.}
\Cref{tab:reduce_cool_summary} outlines the key features of \reduceacool.

\begin{table}[ht]
\centering
\footnotesize
\begin{tabular}{ |p{3.9cm}|p{2.3cm}|p{2.4cm}|p{2cm}|p{1.5cm}|p{1.8cm}|  }
 \hline
 \centering Algorithm & \centering Section & \centering Exchanged bits & \centering Async. rounds & \centering Resilience & 
\centering Cryptography \tabularnewline
 \hline
 \hline

\centering \textbf{ \reduceacool (\Cref{algorithm:a_cool_reduce}) } & \centering \Cref{subsection:reduceacool_pseudocode} & \centering $O\big( nL + n^2 \log(n) \big)$ & \centering $5$ & \centering $t < n / 5$ & \centering None \tabularnewline
\hline
\end{tabular}
    \caption{Relevant aspects of the \reduceacool algorithm proposed by Li and Chen~\cite{li2021communication}.
    \\($L$ denotes the bit-size of a value.)}
\label{tab:reduce_cool_summary}
\end{table}


\subsection{Specification of \reduceacool}

Each correct process can input its value to \reduceacool.
Moreover, each correct process can abandon (i.e., stop participating in) \reduceacool.
Lastly, each correct process can output a pair $(\mathit{success} \in \{0, 1\}, v \in \mathsf{Value})$ from \reduceacool.



The following properties are ensured by \reduceacool:
\begin{compactitem}
    \item \emph{Safety:} Let any correct process output a pair $(1, v)$ from \reduceacool.
    Then, a correct process has previously input $v$ to \reduceacool.

    \item \emph{Agreement:} Let any correct process output a pair $(1, v)$ from \reduceacool.
    If a correct process outputs a pair $(\cdot, v')$ from \reduceacool, then $v' = v$.

    \item \emph{Strong validity:} If all correct processes that input to \reduceacool do so with the same value $v$, then no correct process outputs a pair different from $(1, v)$.
    
    \item \emph{Termination:} If all correct processes input to \reduceacool and no correct process abandons \reduceacool, then every correct process eventually outputs from \reduceacool.
\end{compactitem}

\subsection{\reduceacool: Pseudocode} \label{subsection:reduceacool_pseudocode}

The pseudocode of \reduceacool is given in \Cref{algorithm:a_cool_reduce}.
Recall that \reduceacool (1) tolerates up to $t<n/5$
Byzantine processes, (2) is error-free (i.e., no execution, even with a negligible probability, violates the correctness of \reduceacool), and (3) exchanges $O\big( nL + n^2 \log (n))$ bits. 
\reduceacool internally utilizes Reed-Solomon codes (see \Cref{section:existing_primitives}).


\begin{algorithm}
\caption{\reduceacool: Pseudocode (for process $p_i$)}
\label{algorithm:a_cool_reduce}
\footnotesize
\begin{multicols}{2}
\begin{algorithmic} [1] 
\State \textbf{Constants:}
\State \hskip2em $\mathsf{Integer}$ $k = \lfloor \frac{t}{5} \rfloor + 1$

\medskip
\State \textbf{Local variables:}
\State \hskip2em $\mathsf{Value}$ $\omega^{(i)} \gets $ $p_i$'s input value
\State \hskip2em $\mathsf{Set}(\mathsf{Process})$ $\mathcal{S}^{1}_{0}, \mathcal{S}^{1}_{1}, \mathcal{S}^{2}_{0}, \mathcal{S}^{2}_{1}, \mathcal{S}^{3}_{0}, \mathcal{S}^{3}_{1}, \mathcal{S}^{4}_{0}, \mathcal{S}^{4}_{1}$ $\gets \emptyset$
\State \hskip2em $\mathsf{Array}(\mathsf{RS\_Symbol})$ $\mathit{sym} \gets [\bot, ..., \bot]$

\medskip
\State \textbf{Phase 1:}
\State let $[y^{(i)}_1, y^{(i)}_2, \ldots, y^{(i)}_n] \gets \textsc{RSEnc}(\omega^{(i)}, n, k)$
\State \textbf{send} $\langle\textsc{symbols}, (y^{(i)}_j, y^{(i)}_i)\rangle$ to every process $p_j$

\smallskip
\State \textbf{upon receiving} $\langle\textsc{symbols}, (y^{(j)}_i, y^{(j)}_j)\rangle$ from process $p_j$: 
\State \hskip2em \textbf{if} $(y^{(j)}_i, y^{(j)}_j) = (y^{(i)}_i, y^{(i)}_j)$:
\State \hskip4em  $\mathcal{S}^{1}_{1} \gets \mathcal{S}^{1}_{1} \cup \{ j \}$
\State \hskip2em \textbf{else:}
\State \hskip4em  $\mathcal{S}^{1}_{0} \gets \mathcal{S}^{1}_{0} \cup \{ j \}$

\smallskip
\State \textbf{upon} $|\mathcal{S}^{1}_{0} \cup \mathcal{S}^{1}_{1}| \geq 4t + 1$:

\State \hskip2em \textbf{if} $|\mathcal{S}^{1}_{1}| \geq 3t + 1$: 

\State \hskip4em let $s^{1}_{i} \gets 1$
\State \hskip4em \textbf{broadcast} $\langle P1, s^{1}_{i} \rangle$

\State \hskip2em \textbf{else:}

\State \hskip4em let $s^{1}_{i} \gets 0$, $\omega^{(i)} \gets \phi$
\State \hskip4em \textbf{broadcast} $\langle P1, s^{1}_{i} \rangle$

\smallskip
\State \textbf{upon receiving} $\langle P1, s^{1}_{j} \rangle$ from process $p_j$: 
\State \hskip2em \textbf{if} $s^{1}_{j} = 1$:
\State \hskip4em wait until $j \in \mathcal{S}^{1}_{0} \cup \mathcal{S}^{1}_{1}$
\State \hskip4em \textbf{if} $j \in \mathcal{S}^{1}_{1}$:
\State \hskip6em  $\mathcal{S}^{2}_{1} \gets \mathcal{S}^{2}_{1} \cup \{ j \}$
\State \hskip4em \textbf{else:}
\State \hskip6em  $\mathcal{S}^{2}_{0} \gets \mathcal{S}^{2}_{0} \cup \{ j \}$
\State \hskip2em \textbf{else:}
\State \hskip4em  $\mathcal{S}^{2}_{0} \gets \mathcal{S}^{2}_{0} \cup \{ j \}$

\medskip
\State \textbf{Phase 2:}
\State \textbf{if} $s^{1}_{i} = 1$:

\smallskip
\State \hskip2em \textbf{upon} $|\mathcal{S}^{2}_{0} \cup \mathcal{S}^{2}_{1}| \geq 4t + 1$:

\State \hskip4em \textbf{if} $|\mathcal{S}^{2}_{1}| \geq 3t + 1$: 

\State \hskip6em let $s^{2}_{i} \gets 1$
\State \hskip6em \textbf{broadcast} $\langle P2, s^{2}_{i} \rangle$

\State \hskip4em \textbf{else:}

\State \hskip6em let $s^{2}_{i} \gets 0$, $\omega^{(i)} \gets \phi$
\State \hskip6em \textbf{broadcast} $\langle P2, s^{2}_{i} \rangle$

\State \textbf{else:}
\State \hskip2em let $s^{2}_{i} \gets 0$, $\omega^{(i)} \gets \phi$
\State \hskip2em \textbf{broadcast} $\langle P2, s^{2}_{i} \rangle$

\smallskip
\State \textbf{upon receiving} $\langle P2, s^{2}_{j} \rangle$ from process $p_j$: 
\State \hskip2em \textbf{if} $s^{2}_{j} = 1$:
\State \hskip4em wait until $j \in \mathcal{S}^{2}_{0} \cup \mathcal{S}^{2}_{1}$
\State \hskip4em \textbf{if} $j \in \mathcal{S}^{2}_{1}$:
\State \hskip6em  $\mathcal{S}^{3}_{1} \gets \mathcal{S}^{3}_{1} \cup \{ j \}$
\State \hskip4em \textbf{else:}
\State \hskip6em $\mathcal{S}^{3}_{0} \gets \mathcal{S}^{3}_{0} \cup \{ j \}$
\State \hskip2em \textbf{else:}
\State \hskip4em  $\mathcal{S}^{3}_{0} \gets \mathcal{S}^{3}_{0} \cup \{ j \}$

\medskip
\State \textbf{Phase 3:}
\State \textbf{if} $s^{2}_{i} = 1$:

\smallskip
\State \hskip2em \textbf{upon} $|\mathcal{S}^{3}_{0} \cup \mathcal{S}^{3}_{1}| \geq 4t + 1$:

\State \hskip4em \textbf{if} $|\mathcal{S}^{3}_{1}| \geq 3t + 1$:

\State \hskip6em let $s^{3}_{i} \gets 1$
\State \hskip6em \textbf{send} $\langle P3, s^{3}_{i}, y^{(i)}_j \rangle$ to every process $p_j$ \label{line:reduce_acool_send_symbols_phase_3}

\State \hskip4em \textbf{else}:

\State \hskip6em let $s^{3}_{i} \gets 0$, $\omega^{(i)} \gets \phi$
\State \hskip6em \textbf{broadcast} $\langle P3, s^{3}_{i}, \bot \rangle$

\State \textbf{else:}
\State \hskip2em let $s^{3}_{i} \gets 0$, $\omega^{(i)} \gets \phi$
\State \hskip2em \textbf{broadcast} $\langle P3, s^{3}_{i}, \bot \rangle$

\smallskip
\State \textbf{upon receiving} $\langle P3, s^{3}_{j}, y^{(j)}_i \rangle$ from process $p_j$: 
\State \hskip2em \textbf{if} $s^{3}_{j} = 1$:
\State \hskip4em wait until $j \in \mathcal{S}^{3}_{0} \cup \mathcal{S}^{3}_{1}$
\State \hskip4em \textbf{if} $j \in \mathcal{S}^{3}_{1}$:
\State \hskip6em  $\mathcal{S}^{4}_{1} \gets \mathcal{S}^{4}_{1} \cup \{ j \}$
\State \hskip6em  $sym[j] \gets y^{(j)}_i$
\State \hskip4em \textbf{else:}
\State \hskip6em $\mathcal{S}^{4}_{0} \gets \mathcal{S}^{4}_{0} \cup \{ j \}$

\medskip
\State \textbf{Phase 4:}
\State \textbf{if} $s^{3}_{i} = 1$:

\smallskip
\State \hskip2em \textbf{upon} $|\mathcal{S}^{4}_{0} \cup \mathcal{S}^{4}_{1}| \geq 4t + 1$:

\State \hskip4em \textbf{if} $|\mathcal{S}^{4}_{1}| \geq 3t + 1$: \label{line:reduce_acool_receive_n-t_success_indicators}

\State \hskip6em let $s^{4}_{i} \gets 1$
\State \hskip6em \textbf{broadcast} $\langle P4, y^{(i)}_i \rangle$ \label{line:reduce_acool_share_symbol_validation_unchanged_1}

\State \hskip4em \textbf{else:}

\State \hskip6em let $s^{4}_{i} \gets 0$
\State \hskip6em \textbf{broadcast} $\langle P4, y^{(i)}_i \rangle$ \label{line:reduce_acool_share_symbol_validation_unchanged_0}

\State \textbf{else:}

\State \hskip2em \textbf{upon} $|\mathcal{S}^{4}_{0} \cup \mathcal{S}^{4}_{1}| \geq 4t + 1$: \label{line:reduce_acool_receive_succes_4_phase_3}

\State \hskip4em let $y_i^{(i)} \gets \mathsf{majority}(\mathit{sym})$ \label{line:reduce_acool_majority_validation}

\State \hskip4em  let $s^{4}_{i} \gets 0$
\State \hskip4em \textbf{broadcast} $\langle P4, y^{(i)}_i \rangle$ \label{line:reduce_acool_share_symbol_validation_majority}

\medskip
\State \textbf{Phase 5:}
\State \textbf{if} $s^{3}_{i} = 1$: \label{line:reduce_acool_final_check_success_indicator}
\State \hskip2em \textbf{return} $(s^{4}_{i}, \omega^{(i)})$ \label{line:reduce_acool_deliver_1}
\State \textbf{else:}
\State \hskip2em \textbf{upon} receiving $4t + 1$ RS symbols in $P4$ messages:
\State \hskip4em \textbf{return} $(0, \mathsf{RSDec}(k,t,\text{received symbols}))$ \label{line:reduce_acool_deliver_0} 

\end{algorithmic}
\end{multicols}
\end{algorithm}


\smallskip
\noindent \textbf{Pseudocode description.}
\reduceacool consists of five phases.
In brief, the crucial idea behind \reduceacool (borrowed from the A-COOL algorithm~\cite{li2021communication}) is to reduce the number of possible non-default (i.e., non-$\phi$) values to at most one.
Concretely, after finishing the third phase of \reduceacool, there exists at most one value $\omega \neq \phi$ such that every correct process that has a non-$\phi$ value as its estimation has value $\omega$.
Importantly, if such a value $\omega$ indeed exists (i.e., not all correct processes reach the fourth phase with value $\phi$), then at least $2t+1$ correct processes have value $\omega$.
This is essential as it ensures the successful reconstruction of $\omega$ at all correct processes that have $\phi$ as their estimated decision.
At each phase, processes update and exchange some \emph{success indicators} binary variables $\{s^\rho_i\}^{\rho \in [1:4]}_{i \in [1:n] }$. The event of a negative success indicator $(s^\rho_i = 0)$ means that the number of mismatched observations is high enough to imply that the initial message of processor $p_i$ doesn’t match the majority of other processors’ initial messages. On the opposite, the event of a final positive success indicator $(s^4_i = 1)$ means that the corresponding non-$\phi$ value is reconstructed by every correct process.
We refer the reader to~\cite{li2021communication} for the full details on how \reduceacool (using the logic of A-COOL) reduces the number of non-$\phi$ values to at most one before reaching the fourth phase.

\subsection{\reduceacool: Proof of Correctness \& Complexity}

We underline that \reduceacool's proof of correctness and complexity can be found in~\cite{li2021communication}.
For completeness, we summarize the proof in this subsection.

\smallskip
\noindent \textbf{Proof of correctness.}
First, we prove that \reduceacool satisfies termination.

\begin{theorem} [Termination]
\reduceacool (\Cref{algorithm:a_cool_reduce}) satisfies termination.
\end{theorem}
\begin{proof}
As there are at least $n-t=4t + 1$ correct processes and thresholds at each phase are set to at most $4t + 1$, no correct process gets stuck at any phase of \reduceacool. 
For the full proof, see~\cite[Lemma 3]{li2021communication}.
\end{proof}

Next, we prove the safety property.

\begin{theorem} [Safety]
\reduceacool (\Cref{algorithm:a_cool_reduce}) satisfies safety.
\end{theorem}
\begin{proof}
The following holds in \reduceacool, for each correct process $p_i$: (1) $\omega^{(i)} \in \{\phi, p_i\text{'s input value}\}$, and (2) $\omega^{(i)} = \phi$ implies $s^{3}_i = 0$.
The check at line~\ref{line:reduce_acool_final_check_success_indicator} ensures that a correct process $p_i$ returns $(1, \omega^{(i)})$ at line~\ref{line:reduce_acool_deliver_1} only if $s^{3}_i = 1$.
Hence, the two statements ensure that $p_i$ returns $(1, \omega^{(i)})$ only if $\omega^{(i)}$ is equal to $p_i$'s input value.
\end{proof}

Next, we restate the key lemma from the A-COOL algorithm~\cite{li2021communication}.

\begin{lemma}\label{lemma:non_duplicity_reduce_acool}
At the end of phase 3 of \reduceacool (\Cref{algorithm:a_cool_reduce}) with $n \geq 5t+1$, there exists at most one group of correct processes such that (1) the processes within this group have the same non-$\phi$ value, and (2) the correct processes outside
this group have $\phi$ as their value.
\end{lemma}
\begin{proof}
The lemma is proven in~\cite[Lemma 6]{li2021communication}.
\end{proof}

We are now ready to prove the agreement property.

\begin{theorem} [Agreement]
\reduceacool (\Cref{algorithm:a_cool_reduce}) satisfies agreement.
\end{theorem}
\begin{proof}
We follow the proof of \cite[Lemma 4]{li2021communication}.
Let $p_i$ be a correct process that outputs $(1, \omega_i)$ from \reduceacool.
Process $p_i$ has received $3t + 1$ positive success indicators 
(line~\ref{line:reduce_acool_receive_n-t_success_indicators}), which implies that a set $K$ of at least $3t + 1 - t =2t + 1$ correct processes $p_k$ have $s_k^{3} = 1$ and send a $\langle P3, 1, y_j^{(k)} \rangle$ message to each correct process $p_j$ (line~\ref{line:reduce_acool_send_symbols_phase_3}). 
Note that $p_i \in K$. 
By \Cref{lemma:non_duplicity_reduce_acool}, all correct processes that send a positive success indicator in the third phase hold the same value $\omega_i$ (that was proposed by them). 
Moreover, for each correct process $p_j$, the success indicators it receives at line~\ref{line:reduce_acool_receive_succes_4_phase_3} must contain at least $t+1$ positive success indicators as there are at least $t + 1$ correct processes that both $p_j$ and $p_i$ have heard from.
This implies that, for every correct process $p_l$ with $s_l^{3} = 0$, $p_l$ obtains a correctly-encoded RS symbol at line~\ref{line:reduce_acool_majority_validation} (as at most $t$ incorrect and at least $t + 1$ correct symbols are received by $p_l$).  
 Hence, every correct process that sends a symbol (lines~\ref{line:reduce_acool_share_symbol_validation_unchanged_1}, \ref{line:reduce_acool_share_symbol_validation_unchanged_0} or \ref{line:reduce_acool_share_symbol_validation_majority}) does send a correct symbol, which means that any correct process that outputs at line~\ref{line:reduce_acool_deliver_0} does output $\omega_i$. 
 Finally, by \Cref{lemma:non_duplicity_reduce_acool}, any correct process that outputs at line~\ref{line:reduce_acool_deliver_1} also output $\omega_i$.
\end{proof}

Finally, we prove the strong validity property.

\begin{theorem} [Strong validity]
\reduceacool (\Cref{algorithm:a_cool_reduce}) satisfies strong validity.
\end{theorem}
\begin{proof}
We follow the proof of \cite[Lemma 5]{li2021communication}. 
Recall that a message from a process $p$ for phase $q > 1$ is processed only after a message for phase $q-1$ from the same process has been processed.
When every correct process inputs the same value $v$, all correct processes set their success indicators to $1$ and maintain their value throughout the entire algorithm.
Therefore, for every correct process $p_i$, $\omega^{(i)} = v$.
Thus, every correct process outputs $(1, v)$ from \reduceacool (line~\ref{line:reduce_acool_deliver_1}). 
\end{proof}

\smallskip
\noindent \textbf{Proof of complexity}
We prove that any correct process sends $O\big(L + n\log(n) \big)$ bits in $\reduceacool$.

\begin{theorem} [Exchanged bits]
Any correct process sends $O\big( L + n\log(n) \big)$ bits in \reduceacool.
\end{theorem}
\begin{proof}
Each correct process sends $O\big( L + n\log(n) \big)$ bits via the $\textsc{symbols}$ messages and $P1$, $P2$, $P3$ and $P4$ messages.
\end{proof}

Lastly, we prove that \reduceacool requires $5$ asynchronous rounds until all correct processes output a pair from \reduceacool.

\begin{theorem}[Asynchronous rounds]
Assuming all correct processes input to $\reduceacool$ and no correct process abandons $\reduceacool$, $\reduceacool$ takes $5$ asynchronous rounds before all correct processes output.
\end{theorem}
\begin{proof}
\reduceacool incurs $1$ asynchronous round for each \textsc{symbols}, $P1$, $P2$, $P3$ and $P4$ message.
Therefore, \reduceacool incurs $5$ asynchronous rounds.
\end{proof}

\section{Rebuilding Broadcast} \label{section:rebuilding_broadccast}

In this section, we introduce rebuilding broadcast, a distributed primitive that plays a major role in our implementations of graded consensus (\Cref{section:graded_consensus_concrete_implementations}) and validation broadcast (\Cref{section:validation_broadcast_concrete_implementations}) optimized for long values.
Concretely, we present the following implementation of the rebuilding broadcast primitive.

\begin{table}[ht]
\centering
\footnotesize
\begin{tabular}{ |p{3.9cm}|p{2.3cm}|p{2.6cm}|p{2cm}|p{1.5cm}|p{1.8cm}|  }
 \hline
 \centering Algorithm & \centering Section & \centering Exchanged bits & \centering Async. rounds & \centering Resilience & 
\centering Cryptography \tabularnewline
 \hline
 \hline

\centering \textbf{ \longreb (\Cref{algorithm:longreb_optimal_resilience}) } & \centering \Cref{subsection:longreb_3} & \centering $O( nL + n^2 \log(n) \kappa )$ & \centering $2$ & \centering $t < n / 3$ & \centering Hash \tabularnewline
\hline
\end{tabular}
    \caption{Relevant aspects of a rebuilding broadcast algorithm we propose.
    \\($L$ denotes the bit-size of a value, whereas $\kappa$ denotes the bit-size of a hash value.)}
\label{tab:rebuilding-broadcast-summary}
\end{table}

We start by defining the problem of rebuilding broadcast (\Cref{subsection:rebuilding_broadcast_problem_definition}).
Then, we give \longreb's pseudocode (\Cref{subsection:longreb_3}).
Finally, we prove the correctness and complexity of \longreb (\Cref{subsection:longreb_proof}).

\subsection{Problem Definition} \label{subsection:rebuilding_broadcast_problem_definition}

The rebuilding broadcast primitive allows each process to broadcast its input value and eventually deliver and rebuild some values.
The specification of the problem is associated with the default value $\bot_{\mathit{reb}} \notin \mathsf{Value}$.
Rebuilding broadcast exposes the following interface:
\begin{compactitem}
    \item \textbf{request} $\mathsf{broadcast}(v \in \mathsf{Value})$: a process broadcasts value $v$.

    \item \textbf{request} $\mathsf{abandon}$: a process abandons (i.e., stops participating in) rebuilding broadcast.

    \item \textbf{indication} $\mathsf{deliver}(v' \in \mathsf{Value} \cup \{\bot_{\mathit{reb}}\})$: a process delivers value $v'$ ($v'$ can be $\bot_{\mathit{reb}}$).

    \item \textbf{indication} $\mathsf{rebuild}(v' \in \mathsf{Value})$: a process rebuilds value $v'$ ($v'$ cannot be $\bot_{\mathit{reb}}$).
\end{compactitem}
Any correct process broadcasts at most once.
We do not assume that all correct processes broadcast.

The rebuilding broadcast primitive requires the following properties to be satisfied:
\begin{compactitem}
    \item \emph{Strong validity:} If all correct processes that broadcast do so with the same value, then no correct process delivers $\bot_{\mathit{reb}}$.

    \item \emph{Safety:} If a correct process delivers a value \rebnotbottom{v'}, then a correct process has previously broadcast $v'$.

    \item \emph{Rebuilding validity:} If a correct process delivers a value \rebnotbottom{v'} at some time $\tau$, then every correct process rebuilds $v'$ by time $\max(\tau, \text{GST}) + \delta$.

    \item \emph{Integrity:} A correct process delivers at most once and only if it has previously broadcast.

    \item \emph{Termination:} If all correct processes broadcast and no correct process abandons rebuilding broadcast, then every correct process eventually delivers.
\end{compactitem}
Note that a correct process can rebuild a value even if (1) it has not previously broadcast, or (2) it has previously abandoned rebuilding broadcast, or (3) it has previously delivered a value (or $\bot_{\mathit{reb}}$).
Moreover, multiple values can be rebuilt by a correct process.

\subsection{\longreb: Pseudocode} \label{subsection:longreb_3}

In this subsection, we introduce \longreb (\Cref{algorithm:longreb_optimal_resilience}), our implementation of the rebuilding broadcast primitive.
\longreb (1) tolerates up to $t < n / 3$ Byzantine processes, (2) exchanges $O(nL + n^2 \log(n) \kappa)$ bits, and (3) delivers a value in 2 asynchronous rounds.
Internally, \longreb utilizes cryptographic accumulators (see \Cref{section:existing_primitives}).
We underline that \longreb is highly inspired by an implementation of the reducing broadcast primitive presented in~\cite{DBLP:journals/acta/MostefaouiR17}.

\smallskip
\noindent \textbf{Pseudocode description.}
Let $p_i$ be any correct process.
Process $p_i$ relies on the following local functions:
\begin{compactitem}
    \item $\mathsf{total}(\mathcal{H})$: returns the set of processes from which $p_i$ has received an \textsc{init} or an \textsc{echo} message with $\mathcal{H}$ as the Merkle root (line~\ref{line:function_total_hash}); these two types of messages are explained in the rest of the subsection.

    \item $\mathsf{init}(\mathcal{H})$: returns the set of processes from which $p_i$ has received an \textsc{init} message with $\mathcal{H}$ as the Merkle root (line~\ref{line:function_init_hash}).

    \item $\mathsf{total\_init}$: returns the union of $\mathsf{init}(\mathcal{H})$, for every Merkle root $\mathcal{H}$ (line~\ref{line:function_total_init}).

    \item $\mathsf{most\_frequent}$: returns the most frequent Merkle root according to the $\mathsf{init}(\cdot)$ function (line~\ref{line:function_most_frequent}).
\end{compactitem}
When $p_i$ broadcasts its value $v$ (line~\ref{line:rebuilding_broadcast}), $p_i$ (1) encodes $v$ into $n$ RS symbols $[m_1, m_2, ..., m_n]$ (line~\ref{line:reblong3-enc}), (2) computes the Merkle root of the aforementioned RS symbols (line~\ref{line:reblong3-eval}), and (3) sends each symbol $m_j$ to process $p_j$ via an \textsc{init} message (line~\ref{line:reblong3-init-send}); this message also contains the Merkle root and its witness (i.e., a Merkle proof).
Moreover, once $p_i$ receives an \textsc{init} message for a specific Merkle root and RS symbol from $t + 1$ processes (line~\ref{line:rebuilding_check_echo}), it disseminates the Merkle root and the RS symbol to all processes via an \textsc{echo} message (line~\ref{line:rebuilding_broadcast_echo}).

If $p_i$'s local function $\mathsf{total}(\cdot)$ returns a set of size $t + 1$ for some Merkle root different from the Merkle root $p_i$ has computed (line~\ref{line:different_hash_t+1}), $p_i$ knows that it is impossible that all correct processes have broadcast the same value.
Therefore, $p_i$ delivers the default value $\bot_{\mathit{reb}}$ in this case (line~\ref{line:decide_bot_1}).
Once $p_i$ receives $t + 1$ different RS symbols for the same Merkle root via \textsc{init} and \textsc{echo} messages (line~\ref{line:rule_rebuild_rebuilding}), $p_i$ rebuilds that value by decoding the received RS symbols (line~\ref{line:rebuild_rebuilding}).
Similarly, when $p_i$ receives $2t + 1$ different RS symbols for the same Merkle root (line~\ref{line:rule_deliver_rebuilding}), $p_i$ delivers that value (line~\ref{line:deliver_rebuilding}).
Finally, once $|\mathsf{total\_init}| - |\mathsf{init}(\mathsf{most\_frequent})| \geq t + 1$ (line~\ref{line:complex_check}), $p_i$ delivers the default value $\bot_{\mathit{reb}}$ (line~\ref{line:deliver_4_rebuilding}) as it is impossible that all correct processes have previously broadcast the same value.


\begin{algorithm}
\caption{\longreb: Pseudocode (for process $p_i$)}
\label{algorithm:longreb_optimal_resilience}
\footnotesize
\begin{algorithmic} [1] 
\State \textbf{Rules:} 
\State \hskip2em Any \textsc{init} or \textsc{echo} message with an invalid witness is ignored.
\State \hskip2em Only one \textsc{init} message is processed per process. \label{line:one_init_accept}

\medskip
\State \textbf{Local variables:}
\State \hskip2em $\mathsf{Merkle\_Root}$ $\mathcal{H}_i \gets \bot$
\State \hskip2em $\mathsf{Boolean}$ $\mathit{delivered}_i \gets \mathit{false}$
\State \hskip2em $\mathsf{Map}(\mathsf{Merkle\_Root} \to \mathsf{Boolean})$ $\mathit{rebuilt}_i \gets \{\mathit{false}, \mathit{false}, ..., \mathit{false} \}$

\medskip
\State \textbf{Local functions:}
\State \hskip2em $\mathsf{total}(\mathcal{H} \in \mathsf{Merkle\_Root}) \gets $ the set of processes from which $p_i$ has received $\langle \textsc{init}, \mathcal{H}, \cdot, \cdot \rangle$ or $\langle \textsc{echo}, \mathcal{H}, \cdot, \cdot \rangle$ messages \label{line:function_total_hash}



\State \hskip2em $\mathsf{init}(\mathcal{H} \in \mathsf{Merkle\_Root}) \gets$ the set of processes from which $p_i$ has received a $\langle \textsc{init}, \mathcal{H}, \cdot, \cdot \rangle$ message \label{line:function_init_hash}

\State \hskip2em $\mathsf{total\_init} \gets \bigcup_{\mathcal{H}} \mathsf{init}(\mathcal{H})$ \label{line:function_total_init}

\State \hskip2em $\mathsf{most\_frequent} \gets \mathcal{H}$ such that $|\mathsf{init}(\mathcal{H})| \geq |\mathsf{init}(\mathcal{H}')|$, for every $\mathcal{H}' \in \mathsf{Merkle\_Root}$ \label{line:function_most_frequent}

\medskip
\State \textbf{upon} $\mathsf{broadcast}(v \in \mathsf{Value})$: \label{line:rebuilding_broadcast}
\State \hskip2em let $[m_1, m_2, ..., m_n] \gets \mathsf{RSEnc}(v, n, t + 1)$\label{line:reblong3-enc}
\State \hskip2em let $\mathcal{H}_i \gets \mathsf{Eval}([(1,m_1), (2,m_2), ..., (n,m_n)])$\label{line:reblong3-eval}
\State \hskip2em \textbf{for each} $j \in [1, n]$:
\State \hskip4em let $\mathcal{P}_j \gets \mathsf{CreateWit}(\mathcal{H}_i, (j,m_j), [(1,m_1), (2,m_2), ..., (n,m_n)])$
\State \hskip4em \textbf{send} $\langle \textsc{init}, \mathcal{H}_i, m_j, \mathcal{P}_j \rangle$ to process $p_j$ \label{line:reblong3-init-send}

\medskip
\State \textbf{when} $\langle \textsc{init}, \mathcal{H}, m_i, \mathcal{P}_i \rangle$ or $\langle \textsc{echo}, \mathcal{H}, m_i, \mathcal{P}_i \rangle$ \textbf{is received:}
\State \hskip2em \textbf{if} (1) $\mathcal{H} \neq \mathcal{H}_i$, and (2) $\langle \textsc{init}, \mathcal{H}, m_i, \mathcal{P}_i \rangle$ is received from $t + 1$ processes, and (3) $\langle \textsc{echo}, \mathcal{H}, m_i, \mathcal{P}_i \rangle$ is not broadcast yet: \label{line:rebuilding_check_echo}
\State \hskip4em \textbf{broadcast} $\langle \textsc{echo}, \mathcal{H}, m_i, \mathcal{P}_i \rangle$ \label{line:rebuilding_broadcast_echo}

\medskip
\State \hskip2em \textbf{if} exists $\mathcal{H}' \neq \mathcal{H}_i$ such that $|\mathsf{total}(\mathcal{H}')| \geq t + 1$ and $\mathit{delivered}_i = \mathit{false}$: \label{line:different_hash_t+1}
\State \hskip4em $\mathit{delivered}_i \gets \mathit{true}$
\State \hskip4em \textbf{trigger} $\mathsf{deliver}(\bot_{\mathit{reb}})$ \label{line:decide_bot_1}

\smallskip
\State \hskip2em \textcolor{blue}{\(\triangleright\) if $\mathsf{broadcast}(\cdot)$ has not been invoked, $p_i$ only performs the following check}
\State \hskip2em \textbf{if} exists $\mathcal{H}'$ such that $|\mathsf{total}(\mathcal{H}')| \geq t + 1$ and $\mathit{rebuilt}_i[\mathcal{H}'] = \mathit{false}$: \label{line:rule_rebuild_rebuilding} 
\State \hskip4em $\mathit{rebuilt}_i[\mathcal{H}'] \gets \mathit{true}$
\State \hskip4em \textbf{trigger} $\mathsf{rebuild}\big( \mathsf{RSDec}(t + 1, 0, \text{any } t + 1 \text{ received RS symbols for } \mathcal{H}') \big)$ \label{line:rebuild_rebuilding}

\smallskip
\State \hskip2em \textbf{if} exists $\mathcal{H}'$ such that $|\mathsf{total}(\mathcal{H}')| \geq 2t + 1$ and $\mathit{delivered}_i = \mathit{false}$:  \label{line:rule_deliver_rebuilding} 
\State \hskip4em $\mathit{delivered}_i \gets \mathit{true}$
\State \hskip4em \textbf{trigger} $\mathsf{deliver}\big( \mathsf{RSDec}(t + 1, 0, \text{any } t + 1 \text{ received RS symbols for } \mathcal{H}') \big)$ \label{line:deliver_rebuilding}

\smallskip
\State \hskip2em \textbf{if} $|\mathsf{total\_init}| - |\mathsf{init}(\mathsf{most\_frequent})| \geq t + 1$ and $\mathit{delivered}_i = \mathit{false}$: \label{line:complex_check}
\State \hskip4em $\mathit{delivered}_i \gets \mathit{true}$
\State \hskip4em \textbf{trigger} $\mathsf{deliver}(\bot_{\mathit{reb}})$ \label{line:deliver_4_rebuilding}
\end{algorithmic} 
\end{algorithm}

\subsection{\longreb: Proof of Correctness \& Complexity} \label{subsection:longreb_proof}

We now prove \longreb's correctness and complexity.

\smallskip
\noindent \textbf{Proof of correctness.}
Recall that $\MRoot{v} = \mathsf{Eval}\big([(1,m_1), ..., (n,m_n)]\big)$, where $[m_1, ..., m_n] =  \mathsf{RSEnc}(v, n, t + 1)$ (see \Cref{section:existing_primitives}).
We start by showing that \longreb satisfies strong validity. 

\begin{theorem} [Strong validity] \label{theorem:long_reb_strong_validity_3}
\longreb (\Cref{algorithm:longreb_optimal_resilience}) satisfies strong validity.
\end{theorem}
\begin{proof}
Suppose all correct processes that propose to \longreb do so with the same value $v \in \mathsf{Value}$; let $\mathcal{H} = \mathsf{MR}(v)$.
Observe that no correct process sends any \textsc{init} or \textsc{echo} message for any Merkle root different from $\mathcal{H}$ due to the check at line~\ref{line:rebuilding_check_echo}.

Let us consider any correct process $p_i$.
Process $p_i$ does not deliver $\bot_{\mathit{reb}}$ at line~\ref{line:decide_bot_1} as the check at line~\ref{line:different_hash_t+1} never activates (given that no correct process sends any message for any Merkle root $\mathcal{H}' \neq \mathcal{H})$.
It is left to prove that the check at line~\ref{line:complex_check} never activates at process $p_i$.
By contradiction, suppose it does. 
Let $\omega$ be the most frequent Merkle root when the check at line~\ref{line:complex_check} activates; let $x = |\mathsf{init}(\omega)|$. 
Let $g$ (resp., $b$) be the set of correct (resp., Byzantine) processes from which $p_i$ has received an $\textsc{init}$ message. 
The activation of the check at line~\ref{line:complex_check} implies $|g| + |b| - x \geq t+1$. 
Hence, $|g| > x$ and $|\mathsf{init}(\mathcal{H})| > x$, which contradicts the fact that $\omega = \mathsf{most\_frequent}$.
\end{proof}

Next, we prove the safety property.

\begin{theorem} [Safety]
\longreb (\Cref{algorithm:longreb_optimal_resilience}) satisfies safety.
\end{theorem}
\begin{proof}
If a correct process delivers a value \rebnotbottom{v'} (line~\ref{line:deliver_rebuilding}), then the process has previously received an \textsc{init} or \textsc{echo} message for the Merkle root $\mathcal{H} = \mathsf{MR}(v')$ from a correct process (due to the check at line~\ref{line:rule_deliver_rebuilding} and the $\MRoot{\cdot}$'s collision resistance).
As any correct process sends an \textsc{init} or \textsc{echo} message for $\mathcal{H}$ only if a correct process has previously broadcast $v'$ (due to the rule at line~\ref{line:rebuilding_check_echo} and the $\MRoot{\cdot}$'s collision resistance), the safety property is guaranteed.
\end{proof}

The following theorem proves rebuilding validity.

\begin{theorem} [Rebuilding validity]
\longreb (\Cref{algorithm:longreb_optimal_resilience}) satisfies rebuilding validity.
\end{theorem}
\begin{proof}
Suppose any correct process $p_j$ delivers a value \rebnotbottom{v'} (line~\ref{line:deliver_rebuilding}) at time $\tau$.
Therefore, $p_j$ has previously received $2t + 1$ correctly-encoded (as they are accompanied by valid Merkle proofs) RS symbols (line~\ref{line:rule_deliver_rebuilding}) by time $\tau$, out of which $t + 1$ are broadcast by correct processes (via \textsc{init} or \textsc{echo} messages).
Hence, every correct process receives $t + 1$ correctly-encoded RS symbols (line~\ref{line:rule_rebuild_rebuilding}) by time $\max(\tau, \text{GST}) + \delta$ and rebuilds $v'$ (line~\ref{line:rebuild_rebuilding}), also by time $\max(\tau, \text{GST}) + \delta$.
\end{proof}

We continue our proof by showing that \longreb satisfies integrity.

\begin{theorem} [Integrity]
\longreb (\Cref{algorithm:longreb_optimal_resilience}) satisfies integrity.
\end{theorem}
\begin{proof}
Any correct process $p_i$ delivers at most once due to the $\mathit{delivered}_i$ variable.
Moreover, process $p_i$ delivers only if it has previously broadcast due to the fact that only the check at line~\ref{line:rule_rebuild_rebuilding} is performed by $p_i$ unless $p_i$ has previously broadcast.
\end{proof}

Lastly, we prove \longreb's termination.

\begin{theorem} [Termination] \label{theorem:termination_longreb_3t}
\longreb (\Cref{algorithm:longreb_optimal_resilience}) satisfies termination.
\end{theorem}
\begin{proof}
To prove termination, we consider two cases:
\begin{compactitem}
    \item Suppose at least $t + 1$ correct process broadcast the same value \rebnotbottom{v}.
    Hence, every correct process that did not broadcast $v$ broadcasts an \textsc{echo} message for $\mathcal{H} = \mathsf{MR}(v)$.
    As there are at least $n - t \geq 2t + 1$ correct processes, the rule at line~\ref{line:rule_deliver_rebuilding} eventually activates at every correct process $p_i$ and enables $p_i$ to deliver a value (line~\ref{line:deliver_rebuilding}).

    \item Suppose no value \rebnotbottom{v} exists such that $t + 1$ correct processes broadcast $v$.
    Consider any correct process $p_i$.
    Let us assume that $p_i$ never activates the rule at line~\ref{line:different_hash_t+1} nor the rule at line~\ref{line:rule_deliver_rebuilding}.
    We now prove that the rule at line~\ref{line:complex_check} eventually activates at $p_i$ in this case.

    As no correct process abandons \longreb, $p_i$ eventually receives \textsc{init} messages from all correct processes.
    When that happens, $|\mathsf{total\_init}| \geq 2t + 1 + f$, where $f \leq t$ denotes the number of faulty processes $p_i$ receives (and does not ignore) \textsc{init} messages from. 
    Note that, because of the $\MRoot{\cdot}$'s collision resistance, $|\mathsf{init}(\mathsf{most\_frequent})| \leq t + f$.
    Hence, $|\mathsf{total\_total}| - |\mathsf{init}(\mathsf{most\_frequent})| \geq 2t + 1 + f - t - f \geq t + 1$, which implies that the rule at line~\ref{line:complex_check} activates.
    Therefore, $p_i$ delivers $\bot_{\mathit{reb}}$ at line~\ref{line:deliver_4_rebuilding}.
\end{compactitem}
Since termination is ensured in both possible cases, the proof is concluded.
\end{proof}



\smallskip
\noindent \textbf{Proof of complexity.}
First, we prove that any correct process broadcasts $O(1)$ \textsc{echo} messages. 

\begin{lemma} \label{lemma:two_echo}
Any correct process broadcasts $O(1)$ different \textsc{echo} messages. 
\end{lemma}
\begin{proof}
Any correct process $p_i$ can receive $t + 1$ identical \textsc{init} messages from as many different processes at most $O(1)$ times since $n > 3t$. 
(Recall that $p_i$ only ``accepts'' one \textsc{init} message per process due to the rule at line~\ref{line:one_init_accept}.)
Therefore, the lemma holds.
\end{proof}

The following theorem proves that each correct process exchanges $O(L + n\log(n)\kappa)$ bits.

\begin{theorem} [Exchanged bits]
A correct process sends $O(L + n\log(n)\kappa)$ bits in \longreb. 
\end{theorem}
\begin{proof}
Each message sent by a correct process is of size $O(\kappa + \frac{L}{n} + \log(n) + \log(n)\kappa) = O(\frac{L}{n} + \log(n)\kappa)$ bits.
As each correct process sends at most $O(1)$ messages (one \textsc{init} and $O(1)$ \textsc{echo} messages as proven by \Cref{lemma:two_echo}) to each process, each correct process sends $O(1) \cdot n \cdot O(\frac{L}{n} + \log(n)\kappa) = O(L + n\log(n)\kappa)$ bits.
\end{proof}

Finally, the following theorem proves that \longreb takes $2$ asynchronous rounds before correct processes deliver a value.

\begin{theorem} [Asynchronous rounds] \label{theorem:async_rounds_rebuilding_3}
Assuming all correct processes broadcast via \longreb and no correct process abandons \longreb, \longreb takes $2$ asynchronous rounds before all correct processes deliver.
\end{theorem}
\begin{proof}
Similarly to the proof of the termination property (\Cref{theorem:termination_longreb_3t}), we analyze two scenarios:
\begin{compactitem}
    \item There exists a value \rebnotbottom{v} such that at least $t + 1$ correct processes broadcast $v$ via \longreb.
    At the end of the first asynchronous round, every correct process whose value is not $v$ broadcasts an \textsc{echo} message for $\mathsf{MR}(v)$ (line~\ref{line:rebuilding_broadcast_echo}).
    Therefore, at the end of the second asynchronous round, every correct process receives $n - t \geq 2t + 1$ messages for $\mathsf{MR}(v)$, activates the rule at line~\ref{line:rule_deliver_rebuilding}, and delivers at line~\ref{line:deliver_rebuilding}.

    \item There does not exist a value \rebnotbottom{v} such that at least $t + 1$ correct processes broadcast $v$ via \longreb.
    In this case, every correct process activates a rule at line~\ref{line:complex_check} and delivers at line~\ref{line:deliver_4_rebuilding} upon receiving \textsc{init} messages from all correct processes.
    Hence, \longreb takes $1$ asynchronous round in this scenario.
\end{compactitem}
The proof is concluded as \longreb takes $2$ asynchronous rounds before all correct processes deliver.
\end{proof} 

\section{Graded Consensus: Concrete Implementations} \label{section:graded_consensus_concrete_implementations}

This section provides concrete implementations of the graded consensus primitive that we employ in \name to yield Byzantine agreement algorithms with various bit complexity.
Concretely, \Cref{tab:graded-consensus-summary} outlines the characteristics of two graded consensus implementations we introduce.

\begin{table}[ht]
\centering
\footnotesize
\begin{tabular}{ |p{3.7cm}|p{2.3cm}|p{2.5cm}|p{2cm}|p{1.5cm}|p{1.8cm}|  }
 \hline
 \centering Algorithm & \centering Section & \centering Exchanged bits & \centering Async. rounds & \centering Resilience & 
\centering Cryptography \tabularnewline
 \hline
 \hline

\centering \textbf{ \longgct (\Cref{algorithm:graded_consensus_long_3_4}) } & \centering \Cref{subsection:graded_consensus_bounded}     & \centering $O(nL + n^2 \log(n) \kappa)$ & \centering 11 & \centering $t < n / 3$ & \centering Hash \tabularnewline

 \hline

 \centering \textbf{ \longgcf (\Cref{algorithm:its_graded_consensus}) }  & \centering \Cref{subsection:its_graded_consensus}                 & \centering $O\big( nL + n^2\log(n) \big)$ & \centering 14 & \centering $t < n / 5$ & \centering None  \tabularnewline
 \hline
\end{tabular}
    \caption{Relevant aspects of the two graded consensus algorithms we propose.
    \\($L$ denotes the bit-size of a value, whereas $\kappa$ denotes the bit-size of a hash value.)}
\label{tab:graded-consensus-summary}
\end{table}

\subsection{Review of the Specification of Graded Consensus}\label{app:graded-consensus-definition}
First, we recall the definition of graded consensus.
Graded consensus exposes the following interface:
\begin{compactitem}
    \item \textbf{request} $\mathsf{propose}(v \in \mathsf{Value})$: a process proposes value $v$.

    \item \textbf{request} $\mathsf{abandon}$: a process abandons (i.e., stops participating in) graded consensus.

    \item \textbf{indication} $\mathsf{decide}(v' \in \mathsf{Value}, g' \in \{0, 1\})$: a process decides value $v'$ with grade $g'$.
\end{compactitem}
Every correct process proposes at most once and no correct process proposes an invalid value.
Importantly, not all correct processes are guaranteed to propose to graded consensus.

The graded consensus primitive satisfies the following properties:
\begin{compactitem}
    \item \emph{Strong validity:} If all correct processes that propose do so with the same value $v$ and a correct process decides a pair $(v', g')$, then $v' = v$ and $g' = 1$.

    \item \emph{External validity:} If any correct process decides a pair $(v', \cdot)$, then $\mathsf{valid}(v') = \mathit{true}$.

    \item \emph{Consistency:} If any correct process decides a pair $(v, 1)$, then no correct process decides a pair $(v' \neq v, \cdot)$.

    \item \emph{Integrity:} No correct process decides more than once.

    \item \emph{Termination:} If all correct processes propose and no correct process abandons graded consensus, then every correct process eventually decides.
\end{compactitem}



\subsection{\longgct: Pseudocode \& Proof of Correctness and Complexity}\label{subsection:graded_consensus_bounded} 

In this subsection, we present \longgct (\Cref{algorithm:graded_consensus_long_3_4}), our hash-based implementation of graded consensus that exchanges $O(nL + n^2 \log (n) \kappa)$ bits.
Notably, \longgct is optimally resilient (tolerates $t < n / 3$ faulty processes).
\longgct internally relies on (1) a collision-resistant hash function $\mathsf{hash}(\cdot)$, (2) the \longreb rebuilding broadcast algorithm (see \Cref{section:rebuilding_broadccast}), and (3) AW graded consensus~\cite{AttiyaWelch23} (see \Cref{section:existing_primitives}).


\begin{algorithm}
\caption{\longgct: Pseudocode (for process $p_i$)}
\label{algorithm:graded_consensus_long_3_4}
\footnotesize
\begin{algorithmic} [1] 
\State \textbf{Uses:}
\State \hskip2em \longreb rebuilding broadcast, \textbf{instance} $\mathcal{RB}$ \BlueComment{see \Cref{section:rebuilding_broadccast}}
\State \hskip2em AW graded consensus~\cite{AttiyaWelch23}, \textbf{instance} $\mathcal{AW}$ \BlueComment{hash values and $\bot_{\mathit{reb}}$ can be proposed and decided; see \Cref{section:existing_primitives}}

\medskip
\State \textbf{Local variables:}
\State \hskip2em $\mathsf{Value}$ $\mathit{proposal}_i \gets \bot$

\medskip
\State \textbf{upon} $\mathsf{propose}(v \in \mathsf{Value})$: \label{line:graded_consensus_long_propose}
\State \hskip2em $\mathit{proposal}_i \gets v$
\State \hskip2em \textbf{invoke} $\mathcal{RB}.\mathsf{broadcast}(v)$ \label{line:graded_consensus_long_broadcast}

\medskip
\State \textbf{upon} $\mathcal{RB}.\mathsf{deliver}(v' \in \mathsf{Value} \cup \{\bot_{\mathit{reb}}\})$: \label{line:graded_consensus_long_deliver_rb}
\State \hskip2em \textbf{if} $v' \neq \bot_{\mathit{reb}}$:
\State \hskip4em \textbf{invoke} $\mathcal{AW}.\mathsf{propose}\big( \mathsf{hash}(v') \big)$ \label{line:graded_consensus_long_propose_aw} \BlueComment{propose the hash value of $v'$ if $v' \neq \bot_{\mathit{reb}}$}
\State \hskip2em \textbf{else:}
\State \hskip4em \textbf{invoke} $\mathcal{AW}.\mathsf{propose}(\bot_{\mathit{reb}})$ \label{line:graded_consensus_long_propose_bot_aw} \BlueComment{propose $\bot_{\mathit{reb}}$ if $v' = \bot_{\mathit{reb}}$}

\medskip
\State \textbf{upon} $\mathcal{AW}.\mathsf{decide}(\mathcal{H} \in \mathsf{Hash\_Value} \cup \{\bot_{\mathit{reb}}\}, g \in \{0, 1\})$: \label{line:graded_consensus_long_decide_aw}
\State \hskip2em \textbf{if} $\mathcal{H} = \bot_{\mathit{reb}}$: \label{line:graded_consensus_check_after_aw_decision}
\State \hskip4em \textbf{trigger} $\mathsf{decide}(\mathit{proposal}_i, 0)$ \label{line:graded_consensus_long_decide_1} \BlueComment{if $\mathcal{H} = \bot_{\mathit{reb}}$, decide $p_i$'s proposal with grade 0}
\State \hskip2em \textbf{else:}
\State \hskip4em \textbf{wait for} $\mathcal{RB}.\mathsf{rebuild}(v' \in \mathsf{Value})$ such that $\mathsf{hash}(v') = \mathcal{H}$ \label{line:graded_consensus_wait_for_rebuild} \BlueComment{some correct process delivered value $v'$ from $\mathcal{RB}$}
\State \hskip4em \textbf{trigger} $\mathsf{decide}(v', g)$ \label{line:graded_consensus_long_decide} \BlueComment{after rebuilding $v'$, decide $v'$ with the grade specified by $\mathcal{AW}$}
\end{algorithmic} 
\end{algorithm}

\smallskip
\noindent \textbf{Pseudocode description.}
We describe \longgct (\Cref{algorithm:graded_consensus_long_3_4}) from the perspective of a correct process $p_i$.
When $p_i$ proposes its value $v$ (line~\ref{line:graded_consensus_long_propose}), it broadcasts the value using the $\mathcal{RB}$ instance of the rebuilding broadcast (line~\ref{line:graded_consensus_long_broadcast}).
If $p_i$ delivers a value $v' \neq \bot_{\mathit{reb}}$ from $\mathcal{RB}$, it proposes $\mathsf{hash}(v')$ to the $\mathcal{AW}$ instance of the AW graded consensus algorithm (line~\ref{line:graded_consensus_long_propose_aw}).
If $p_i$ delivers $\bot_{\mathit{reb}}$ from $\mathcal{RB}$, it forwards $\bot_{\mathit{reb}}$ to $\mathcal{AW}$ (line~\ref{line:graded_consensus_long_propose_bot_aw}).
Eventually, $p_i$ decides a pair $(\mathcal{H}, g)$ from $\mathcal{AW}$ (line~\ref{line:graded_consensus_long_decide_aw}).
We distinguish two scenarios:
\begin{compactitem}
    \item If $\mathcal{H} = \bot_{\mathit{reb}}$, then $p_i$ decides its proposal with grade 0 (line~\ref{line:graded_consensus_long_decide_1}).

    \item Otherwise, $p_i$ waits to rebuild a value $v'$ such that $\mathsf{hash}(v') = \mathcal{H}$ (line~\ref{line:graded_consensus_wait_for_rebuild}).
    As $\mathcal{H} \neq \bot_{\mathit{reb}}$ is decided from $\mathcal{AW}$, the safety property of $\mathcal{AW}$ guarantees that $\mathcal{H}$ has previously been proposed to $\mathcal{AW}$ by a correct process (line~\ref{line:graded_consensus_long_propose_aw}).
    Therefore, $v'$ has been delivered from $\mathcal{RB}$ by a correct process (line~\ref{line:graded_consensus_long_deliver_rb}), which implies that $p_i$ eventually rebuilds $v'$ (due to the rebuilding validity property of $\mathcal{RB}$).
    After rebuilding $v'$, $p_i$ decides $v'$ with grade $g$ (line~\ref{line:graded_consensus_long_decide}).
\end{compactitem}

\smallskip
\noindent \textbf{Proof of correctness.}
We start by proving the strong validity property.

\begin{theorem} [Strong validity]
\longgct (\Cref{algorithm:graded_consensus_long_3_4}) satisfies strong validity.
\end{theorem}
\begin{proof}
Suppose all correct processes that propose to graded consensus do so with the same value $v$.
Hence, all correct processes that broadcast their proposal via $\mathcal{RB}$ do so with value $v$ (line~\ref{line:graded_consensus_long_broadcast}).
Therefore, every correct process that delivers a value from $\mathcal{RB}$ does deliver value $v \neq \bot_{\mathit{reb}}$ (due to the strong validity and safety properties of $\mathcal{RB}$), which further implies that all correct processes that propose to $\mathcal{AW}$ do so with hash value $\mathcal{H} = \mathsf{hash}(v)$ (line~\ref{line:graded_consensus_long_propose_aw}).
Due to the strong validity property of $\mathcal{AW}$, any correct process $p_i$ that decides from $\mathcal{AW}$ decides a pair $(\mathcal{H} \neq \bot_{\mathit{reb}}, 1)$ (line~\ref{line:graded_consensus_long_decide_aw}).
Hence, every correct process that decides from \longgct decides with grade 1 (line~\ref{line:graded_consensus_long_decide}).
Finally, as the $\mathsf{hash}(\cdot)$ function is collision-resistant, every correct process that decides from \longgct does decide value $v$.
\end{proof}

The following theorem proves the external validity property.

\begin{theorem} [External validity]
\longgct (\Cref{algorithm:graded_consensus_long_3_4}) satisfies external validity.
\end{theorem}
\begin{proof}
Suppose a correct process $p_i$ decides some value $v'$.
We consider two possibilities:
\begin{compactitem}
    \item Let $p_i$ decide $v'$ at line~\ref{line:graded_consensus_long_decide_1}.
    In this case, $v'$ is the proposal of $p_i$.
    As no correct process proposes an invalid value to \longgct, $v'$ is a valid value.

    \item Let $p_i$ decide $v'$ at line~\ref{line:graded_consensus_long_decide}.
    Hence, $p_i$ has previously decided $\mathcal{H} = \mathsf{hash}(v') \neq \bot_{\mathit{reb}}$ from $\mathcal{AW}$ (line~\ref{line:graded_consensus_long_decide_aw}).
    Due to the safety property of $\mathcal{AW}$, a correct process has previously proposed $\mathcal{H}$ to $\mathcal{AW}$ (line~\ref{line:graded_consensus_long_propose_aw}), which implies (due to the collision resistance of the $\mathsf{hash}(\cdot)$ function) that a correct process has delivered $v' \neq \bot_{\mathit{reb}}$ from $\mathcal{RB}$ (line~\ref{line:graded_consensus_long_deliver_rb}).
    The safety property of $\mathcal{RB}$ proves that $v'$ has been broadcast by a correct process (line~\ref{line:graded_consensus_long_broadcast}), which means that a correct process has proposed $v'$ to \longgct (line~\ref{line:graded_consensus_long_propose}).
    As no correct process proposes an invalid value to \longgct, $v'$ is a valid value.
\end{compactitem}
As $v'$ is a valid value in both possible scenarios, the proof is concluded.
\end{proof}

Next, we prove consistency.

\begin{theorem} [Consistency]
\longgct (\Cref{algorithm:graded_consensus_long_3_4}) satisfies consistency.
\end{theorem}
\begin{proof}
Let $p_i$ be any correct process that decides a pair $(v, 1)$ (line~\ref{line:graded_consensus_long_decide}).
Hence, $p_i$ has previously decided a pair $(\mathcal{H} = \mathsf{hash}(v), 1)$ from $\mathcal{AW}$ (line~\ref{line:graded_consensus_long_decide_aw}).
Due to the consistency property of $\mathcal{AW}$, any correct process that decides from it does decide $(\mathcal{H}, \cdot)$.
Therefore, because of the $\mathsf{hash}(\cdot)$ function's collision resistance, any correct process that decides from \longgct does decide $v$ (line~\ref{line:graded_consensus_long_decide}).
\end{proof}

The following theorem proves the integrity property.

\begin{theorem} [Integrity]
\longgct (\Cref{algorithm:graded_consensus_long_3_4}) satisfies integrity.
\end{theorem}
\begin{proof}
The integrity property of \longgct follows directly from the integrity property of $\mathcal{AW}$.
\end{proof}

Finally, we prove the termination property.

\begin{theorem} [Termination]
\longgct (\Cref{algorithm:graded_consensus_long_3_4}) satisfies termination.
\end{theorem}
\begin{proof}
Let us assume that all correct processes propose and no correct process ever abandons \longgct.
In this case, the termination property of $\mathcal{RB}$ ensures that every correct process eventually delivers a value from it (line~\ref{line:graded_consensus_long_deliver_rb}), and proposes to $\mathcal{AW}$ (line~\ref{line:graded_consensus_long_propose_aw} or line~\ref{line:graded_consensus_long_propose_bot_aw}).
Similarly, the termination property of $\mathcal{AW}$ guarantees that every correct process eventually decides from it (line~\ref{line:graded_consensus_long_decide_aw}).
We now separate two cases that can occur at any correct process $p_i$:
\begin{compactitem}
    \item Let $p_i$ decide $\bot_{\mathit{reb}}$ from $\mathcal{AW}$.
    In this case, $p_i$ decides from \longgct at line~\ref{line:graded_consensus_long_decide_1}, thus satisfying termination.

    \item Let $p_i$ decide $\mathcal{H} \neq \bot_{\mathit{reb}}$ from $\mathcal{AW}$.
    Due to the safety property of $\mathcal{AW}$, a correct process has previously proposed $\mathcal{H}$ to $\mathcal{AW}$ (line~\ref{line:graded_consensus_long_propose_aw}), which implies that a correct process has delivered value $v \neq \bot_{\mathit{reb}}$ from $\mathcal{RB}$ (line~\ref{line:graded_consensus_long_deliver_rb}) such that $\mathsf{hash}(v) = \mathcal{H}$ (due to the $\mathsf{hash}(\cdot)$ function's collision resistance).
    Therefore, the rebuilding validity property of $\mathcal{RB}$ ensures that $p_i$ eventually rebuilds $v$ (line~\ref{line:graded_consensus_wait_for_rebuild}) and decides $v$ at line~\ref{line:graded_consensus_long_decide}.
\end{compactitem}
As termination is satisfied in both cases, the proof is concluded.
\end{proof}

\smallskip
\noindent \textbf{Proof of complexity.}
First, we prove the number of bits correct processes exchange in \longgct.

\begin{theorem} [Exchanged bits]
A correct process sends $O(L + n \log (n) \kappa)$ bits in \longgct.
\end{theorem}
\begin{proof}
Let $p_i$ be any correct process.
Process $p_i$ sends $O(L + n\log(n) \kappa)$ bits in $\mathcal{RB}$ (see \Cref{section:rebuilding_broadccast}).
Moreover, process $p_i$ sends $O(n\kappa)$ bits in $\mathcal{AW}$ (see~\cite{AttiyaWelch23}).
Hence, process $p_i$ sends $O(L + n \log(n) \kappa)$ bits.
\end{proof}

Finally, the following theorem proves that \longgct takes at most $11$ asynchronous rounds before all correct processes decide.

\begin{theorem}[Asynchronous rounds]
Assuming all correct processes propose to \longgct and no correct process abandons \longgct, \longgct takes $11$ asynchronous rounds before all correct processes decide.
\end{theorem}
\begin{proof}
Each correct process that participates in \longgct and does not abandon it incurs 2 asynchronous rounds in $\mathcal{RB}$ (see \Cref{section:rebuilding_broadccast}), followed by 9 asynchronous rounds in $\mathcal{AW}$ (see~\cite{AttiyaWelch23}).
\end{proof}

\subsection{\longgcf: Pseudocode \& Proof of Correctness and Complexity} \label{subsection:its_graded_consensus}

In this subsection, we introduce \longgcf (\Cref{algorithm:its_graded_consensus}), our implementation of graded consensus that exchanges $O\big( nL + n^2 \log(n) \big)$ bits while relying on no cryptographic primitives.
\longgcf tolerates up to $t < n / 5$ Byzantine processes.
Under the hood, \longgcf utilizes the \reduceacool algorithm (see \Cref{section:acool_reduction}) and the AW graded consensus algorithm (see \Cref{section:existing_primitives}).

\begin{algorithm}
\caption{\longgcf: Pseudocode (for process $p_i$)}
\label{algorithm:its_graded_consensus}
\footnotesize
\begin{algorithmic} [1] 
\State \textbf{Uses:}
\State \hskip2em \reduceacool , \textbf{instance} $\mathcal{ACOOL}$ \BlueComment{see \Cref{section:acool_reduction}}
\State \hskip2em AW graded consensus \cite{AttiyaWelch23}, \textbf{instance} $\mathcal{AW}$ \BlueComment{see \Cref{section:existing_primitives}}

\medskip
\State \textbf{Local variables:}
\State \hskip2em $\mathsf{Value}$ $\mathit{proposal}_i \gets \bot$
\State \hskip2em $\mathsf{Value}$ $\mathit{reduction\_output}_i \gets \bot$ 

\medskip
\State \textbf{upon} $\mathsf{propose}(v \in \mathsf{Value})$: \label{line:longgcf_propose}
\State \hskip2em $\mathit{proposal}_i \gets v$
\State \hskip2em let $(\mathit{success}, \mathit{reduction\_output}_i) \gets \mathcal{ACOOL}(v)$ \label{line:longgcf_input_to_acool}
\State \hskip2em \textbf{if} $\mathit{success} = 1$:
\State \hskip4em \textbf{invoke} $\mathcal{AW}.\mathsf{propose}(\mathsf{HAPPY})$ \label{line:longgcf_propose_happy_to_aw}
\State \hskip2em \textbf{else:}
\State \hskip4em \textbf{invoke} $\mathcal{AW}.\mathsf{propose}(\mathsf{SAD})$ \label{line:longgcf_propose_sad_to_aw}

\medskip
\State \textbf{upon} $\mathcal{AW}.\mathsf{decide}(v' \in \{\mathsf{HAPPY}, \mathsf{SAD}\}, g' \in \{0, 1\})$: \label{line:longgcf_decide_from_aw}
\State \hskip2em \textbf{if} $v' = \mathsf{SAD}$:
\State \hskip4em \textbf{trigger} $\mathsf{decide}(\mathit{proposal}_i, 0)$ \label{line:longgcf_decide_own_proposal}
\State \hskip2em \textbf{else:}
\State \hskip4em \textbf{wait until} $\mathit{reduction\_output_i} \ne \bot$
\State \hskip4em \textbf{trigger} $\mathsf{decide}(\mathit{reduction\_output}_i, g')$ \label{line:longgcf_decide_output}

\end{algorithmic}
\end{algorithm}

\smallskip
\noindent \textbf{Pseudocode description.}
When a correct process $p_i$ proposes a value $v$ to \longgcf (line~\ref{line:longgcf_propose}), it forwards $v$ to the $\mathcal{ACOOL}$ instance of the \reduceacool algorithm (line~\ref{line:longgcf_input_to_acool}).
Once $p_i$ obtains a pair $(\mathit{success}, \mathit{reduction\_output}_i)$ from $\mathcal{ACOOL}$, $p_i$ checks if $\mathit{success} = 1$.
If so, $p_i$ proposes $\mathsf{HAPPY}$ to the $\mathcal{AW}$ instance of the AW graded consensus algorithm (line~\ref{line:longgcf_propose_happy_to_aw}).
Otherwise, $p_i$ proposes $\mathsf{SAD}$ to $\mathcal{AW}$ (line~\ref{line:longgcf_propose_sad_to_aw}).
When $p_i$ decides a pair $(v', g')$ from $\mathcal{AW}$ (line~\ref{line:longgcf_decide_from_aw}), $p_i$ performs the following logic:
\begin{compactitem}
    \item If $v' = \mathsf{SAD}$, then $p_i$ decides its proposal to \longgcf with grade $0$ (line~\ref{line:longgcf_decide_own_proposal}) as $p_i$ knows that it is impossible that all correct processes have previously proposed the same value to \longgcf (due to the strong validity property of $\mathcal{ACOOL}$ and $\mathcal{AW}$).

    \item Otherwise, $p_i$ decides $(\mathit{reduction\_output}_i, g')$ from \longgcf (line~\ref{line:longgcf_decide_output}).
\end{compactitem}

\smallskip
\noindent \textbf{Proof of correctness.}
We start by proving that \longgcf satisfies strong validity.

\begin{theorem} [Strong validity]
\longgcf (\Cref{algorithm:its_graded_consensus}) satisfies strong validity.
\end{theorem}
\begin{proof}
Suppose all correct processes that propose to \longgcf do so with the same value $v$.
Hence, all correct processes that input a value to $\mathcal{ACOOL}$ do input $v$ (line~\ref{line:longgcf_input_to_acool}).
The strong validity property of $\mathcal{ACOOL}$ ensures that each correct process $p_i$ that receives an output from $\mathcal{ACOOL}$ receives $(1, v)$.
Therefore, all correct processes that propose to $\mathcal{AW}$ do so with $\mathsf{HAPPY}$ (line~\ref{line:longgcf_propose_happy_to_aw}).
The strong validity property of $\mathcal{AW}$ ensures that all correct processes decide $(\mathsf{HAPPY}, 1)$ from $\mathcal{AW}$ (line~\ref{line:longgcf_decide_from_aw}).
Finally, all correct processes that decide from \longgcf do so with $(v, 1)$ (line~\ref{line:longgcf_decide_output}), which concludes the proof.
\end{proof}

Next, we prove external validity.

\begin{theorem} [External validity]
\longgcf (\Cref{algorithm:its_graded_consensus}) satisfies external validity.
\end{theorem}
\begin{proof}
If a correct process $p_i$ decides at line~\ref{line:longgcf_decide_own_proposal}, the decision is valid since the process has previously proposed a valid value to \longgcf.
Suppose $p_i$ decides some value $v$ at line~\ref{line:longgcf_decide_output}.
In this case, $p_i$ has previously decided $\mathsf{HAPPY}$ from $\mathcal{AW}$, which implies that some correct process $p_j$ has proposed $\mathsf{HAPPY}$ to $\mathcal{AW}$ at line~\ref{line:longgcf_propose_happy_to_aw} (due to the safety property of $\mathcal{AW}$).
Hence, $p_j$ has received $(1, v')$ from $\mathcal{ACOOL}$ (line~\ref{line:longgcf_input_to_acool}), for some value $v'$.
Importantly, the agreement property of $\mathcal{ACOOL}$ shows that $v' = v$.
Moreover, the safety property of $\mathcal{ACOOL}$ shows that $v$ was previously proposed to \longgcf by a correct process.
As no correct process proposes an invalid value to \longgcf, $v$ is valid.
\end{proof}

The following theorem proves consistency.

\begin{theorem} [Consistency]
\longgcf (\Cref{algorithm:its_graded_consensus}) satisfies consistency.
\end{theorem}
\begin{proof}
If any correct process $p_i$ decides $(v, 1)$ from \longgcf, it does so at line~\ref{line:longgcf_decide_output}.
This implies that $p_i$ has previously decided $(\mathsf{HAPPY}, 1)$ from $\mathcal{AW}$ (line~\ref{line:longgcf_decide_from_aw}).
Due to the safety property of $\mathcal{AW}$, a correct process has proposed $\mathsf{HAPPY}$ to $\mathcal{AW}$, which implies that any correct process that receives a pair from $\mathcal{ACOOL}$ does receive $(\cdot, v)$. (due to the agreement property of $\mathcal{ACOOL}$).
Moreover, any correct process that decides from $\mathcal{AW}$ does decide $(\mathsf{HAPPY}, \cdot)$ due to the consistency property of $\mathcal{AW}$.
Therefore, any correct process that decides from \longgcf does decide $v$ at line~\ref{line:longgcf_decide_output}, which concludes the proof.
\end{proof}

We now prove integrity.

\begin{theorem}[Integrity]
    \longgcf (\Cref{algorithm:its_graded_consensus}) satisfies integrity.
\end{theorem}
\begin{proof}
    Follows directly from integrity of $\mathcal{AW}$.
\end{proof}

Finally, we prove termination.

\begin{theorem} [Termination]
\longgcf (\Cref{algorithm:its_graded_consensus}) satisfies termination.
\end{theorem}
\begin{proof}
The termination property follows directly from termination of $\mathcal{ACOOL}$ and $\mathcal{AW}$.
\end{proof}

\smallskip
\noindent \textbf{Proof of complexity.}
We prove that any correct process sends $O\big( L + n\log(n) \big)$ bits.

\begin{theorem} [Exchanged bits]
Any correct process sends $O\big( L + n\log(n) \big)$ bits in \longgcf.
\end{theorem}
\begin{proof}
Any correct process sends $O\big( L + n \log(n) \big)$ bits in $\mathcal{ACOOL}$.
Moreover, any correct process sends $O(n)$ bits in $\mathcal{AW}$.
Therefore, any correct process sends $O\big( L + n \log(n) \big)$ bits.
\end{proof}

Lastly, we prove that \longgcf requires $14$ asynchronous rounds before all correct processes decide.

\begin{theorem}[Asynchronous rounds]
Assuming all correct processes propose to \longgcf and no correct process abandons \longgcf, \longgcf takes $14$ asynchronous rounds before all correct processes decide.
\end{theorem}
\begin{proof}
$\mathcal{ACOOL}$ incurs $5$ asynchronous rounds, whereas $\mathcal{AW}$ incurs $9$ asynchronous rounds.
Therefore, \longgcf requires $14$ asynchronous rounds until all correct processes decide.
\end{proof}
\section{Validation Broadcast: Concrete Implementations} \label{section:validation_broadcast_concrete_implementations}

In this section, we present concrete implementations of the validation broadcast primitive we utilize in \name to obtain Byzantine agreement algorithms with various bit complexity.
Concretely, \Cref{tab:validation-broadcast-summary} outlines the characteristics of three validation broadcast implementations we introduce.

\begin{table}[ht]
\centering
\footnotesize
\begin{tabular}{ |p{3.9cm}|p{2.3cm}|p{2.8cm}|p{2cm}|p{1.5cm}|p{1.8cm}|  }
 \hline
 \centering Algorithm & \centering Section & \centering Exchanged bits & \centering Async. rounds & \centering Resilience & \centering Cryptography \tabularnewline
 \hline
 \hline

\centering \textbf{ \shortvbt (\Cref{algorithm:echo_basic}) } & \centering \Cref{subsection:validation_broadcast_basic_implementation} & \centering $O(n^2L)$ & \centering $4$ & \centering $t < n / 3$ & \centering None \tabularnewline

 \hline

 \centering \textbf{ \longvbt (\Cref{algorithm:validation_broadcast_long_3_4}) }  & \centering \Cref{subsection:validation_bcast_long_3_4}                 & \centering $O( nL + n^2\log(n)\kappa )$ & \centering $6$ & \centering $t < n / 3$ & \centering Hash  \tabularnewline
 
 \hline

 \centering \textbf{ \longvbf (\Cref{algorithm:its_validation_broadcast}) }  & \centering \Cref{subsection:its_validation_broadcast}                 & \centering $O\big( nL + n^2\log(n) \big)$ & \centering $15$ & \centering $t < n / 5$ & \centering None  \tabularnewline
 
 \hline
\end{tabular}
    \caption{Relevant aspects of the three validation broadcast algorithms we propose.
    \\($L$ denotes the bit-size of a value, whereas $\kappa$ denotes the bit-size of a hash value.)}
\label{tab:validation-broadcast-summary}
\end{table}

\subsection{Review of the Specification of Validation Broadcast}\label{app:validation-broadcast-definition}
Let us recall the definition of the validation broadcast primitive.
The following interface is exposed:
\begin{compactitem}
    \item \textbf{request} $\mathsf{broadcast}(v \in \mathsf{Value})$: a process broadcasts value $v$.

    \item \textbf{request} $\mathsf{abandon}$: a process abandons (i.e., stops participating in) validation broadcast.

    \item \textbf{indication} $\mathsf{validate}(v' \in \mathsf{Value})$: a process validates value $v'$.

    \item \textbf{indication} $\mathsf{completed}$: a process is notified that validation broadcast has completed.
\end{compactitem}
Every correct process broadcasts at most once.
Not all correct processes are guaranteed to broadcast their value.
Recall that each process $p_i$ is associated with its default value $\mathsf{def}(p_i) \in \mathsf{Value}$. 

The validation broadcast primitive guarantees the following properties: 
\begin{compactitem}
    \item \emph{Strong validity:} If all correct processes that broadcast do so with the same value $v$, then no correct process validates any value $v' \neq v$.


    \item \emph{Safety:} If a correct process $p_i$ validates a value $v'$, then a correct process has previously broadcast $v'$ or $v' = \mathsf{def}(p_i)$.

    \item \emph{Integrity:} No correct process receives a $\mathsf{completed}$ indication unless it has previously broadcast.

    \item \emph{Termination:} If all correct processes broadcast their value and no correct process abandons validation broadcast, then every correct process eventually receives a $\mathsf{completed}$ indication.

    \item \emph{Totality:} If any correct process receives a $\mathsf{completed}$ indication at time $\tau$, then every correct process validates a value by time $\max(\tau, \text{GST}) + 2\delta$.
\end{compactitem}







We underline that a correct process might validate a value even if (1) it has not previously broadcast its input value, or (2) it has previously abandoned the primitive, or (3) it has previously received a $\mathsf{completed}$ indication. Moreover, a correct process may validate multiple values, and two correct processes may validate different values. 


\subsection{\shortvbt: Pseudocode \& Proof of Correctness and Complexity} \label{subsection:validation_broadcast_basic_implementation}

The pseudocode of \shortvbt is given in \Cref{algorithm:echo_basic}.
Recall that \shortvbt (1) tolerates up to $t < n / 3$ Byzantine processes, (2) uses no cryptography (i.e., is resilient against a computationally unbounded adversary), and (3) exchanges $O(n^2L)$ bits.
\shortvbt internally utilizes the reducing broadcast primitive (see \Cref{section:existing_primitives}).
 
\begin{algorithm}
\caption{\shortvbt: Pseudocode (for process $p_i$)}
\label{algorithm:echo_basic}
\footnotesize
\begin{algorithmic} [1]
\State \textbf{Uses:}
\State \hskip2em Reducing broadcast~\cite{DBLP:journals/acta/MostefaouiR17}, \textbf{instance} $\mathcal{RB}$ \BlueComment{see \Cref{section:existing_primitives}}

\medskip
\State \textbf{Rules:} 
\State \hskip2em Only one \textsc{init} message is processed per process. 

\medskip
\State \textbf{Local variables:}
\State \hskip2em $\mathsf{Map}(\mathsf{Value} \cup \{\bot_{rd}\} \to \mathsf{Boolean})$ $\mathit{echo}_i \gets \{\mathit{false}, \mathit{false}, ..., \mathit{false}\}$

\medskip
\State \textbf{Local functions:}
\State \hskip2em $\mathsf{init}(v \in \mathsf{Value} \cup \{\bot_{rd}\}) \gets $ the set of processes from which $p_i$ has received an $\langle \textsc{init}, v \rangle$ message \label{line:short_vb_init_function}

\State \hskip2em $\mathsf{total\_init} \gets \bigcup_{v} \mathsf{init}(v)$ \label{line:short_vb_total_init_function}

\State \hskip2em $\mathsf{most\_frequent} \gets v$ such that $|\mathsf{init}(v)| \geq |\mathsf{init}(v')|$, for every $v' \in \mathsf{Value} \cup \{\bot_{rd}\}$ \label{line:most_frequent_function_validation_broadcast}

\State \hskip2em $\mathsf{echo}(v \in \mathsf{Value} \cup \{\bot_{rd}\}) \gets $ the set of processes from which $p_i$ has received an 
$\langle \textsc{echo}, v \rangle$ message \label{line:echo_function}

\medskip
\State \textbf{upon} $\mathsf{broadcast}(v \in \mathsf{Value})$: \label{line:echo_broadcast}

\State \hskip2em \textbf{invoke} $\mathcal{RB}.\mathsf{broadcast}(v)$ \label{line:echo_reducing_broadcast}

\medskip
\State \textbf{upon} $\mathcal{RB}.\mathsf{deliver}(v \in \mathsf{Value} \cup \{ \bot_{rd} \})$: \label{line:echo_reducing_deliver}
\State \hskip2em \textbf{broadcast} $\langle \textsc{init}, v \rangle$ \label{line:broadcast_init_1_short_vb}

\medskip
\State \textbf{upon} exists $v \in \mathsf{Value} \cup \{\bot_{rd}\}$ such that $|\mathsf{init}(v)| \geq t + 1$ and $\mathit{echo}_i[v] = \mathit{false}$: \label{line:echo_receive_plurality}
\State \hskip2em $\mathit{echo}_i[v] \gets \mathit{true}$
\State \hskip2em \textbf{broadcast} $\langle \textsc{echo}, v \rangle$ \label{line:echo_broadcast_2}

\medskip
\State \textbf{upon} $|\mathsf{total\_init}| - |\mathsf{init}(\mathsf{most\_frequent})| \geq t + 1$ and $\mathit{echo}_i[\bot_{rd}] = \mathit{false}$: \label{line:echo_most_frequent}
\State \hskip2em $\mathit{echo}_i[\bot_{rd}] \gets \mathit{true}$
\State \hskip2em \textbf{broadcast} $\langle \textsc{echo}, \bot_{rd} \rangle$ \label{line:echo_broadcast_3}

\medskip
\State \textbf{upon} exists $v \in \mathsf{Value} \cup \{\bot_{rd}\}$ such that $|\mathsf{echo}(v)| \geq 2t + 1$:
\BlueComment{can be triggered only if $p_i$ has previously broadcast}
\label{line:echo_receive_quorum}
\State \hskip2em \textbf{trigger} $\mathsf{completed}$ \label{line:echo_trigger_completed}

\medskip
\State \textbf{upon} exists $v \in \mathsf{Value} \cup \{\bot_{rd}\}$ such that $|\mathsf{echo}(v)| \geq t + 1$: \label{line:echo_receive_plurality_2}\BlueComment{can be triggered anytime}
\State \hskip2em \textbf{if} $v = \bot_{rd}$:
\State \hskip4em $v \gets \mathsf{def}(p_i)$ \label{line:echo_validate_broadcast}
\State \hskip2em \textbf{trigger} $\mathsf{validate}(v)$ \label{line:echo_trigger_validate}

\end{algorithmic}
\end{algorithm}

\smallskip
\noindent \textbf{Pseudocode description.}
We describe \shortvbt's pseudocode from the perspective of a correct process $p_i$.
Process $p_i$ relies on the following local functions:
\begin{compactitem}
    \item $\mathsf{init}(v)$: returns the set of processes from which $p_i$ has received an \textsc{init} message with value $v$ (line~\ref{line:short_vb_init_function}).

    \item $\mathsf{total\_init}$: returns the union of $\mathsf{init}(v)$, for every value $v$ (line~\ref{line:short_vb_total_init_function}).


    \item $\mathsf{most\_frequent}$: returns the most frequent value according to the $\mathsf{init}(\cdot)$ function (line~\ref{line:most_frequent_function_validation_broadcast}).

    \item $\mathsf{echo}(v)$: returns the set of processes from which $p_i$ has received an \textsc{echo} message with value $v$ (line~\ref{line:echo_function}).
\end{compactitem}
When $p_i$ broadcasts its value $v$ (line~\ref{line:echo_broadcast}), $p_i$ disseminates $v$ using the instance $\mathcal{RB}$ of the reducing broadcast primitive (line~\ref{line:echo_reducing_broadcast}).
Once $p_i$ delivers a value from $\mathcal{RB}$ (line~\ref{line:echo_reducing_deliver}), $p_i$ broadcast an \textsc{init} message for that value (line~\ref{line:broadcast_init_1_short_vb}).
If $p_i$ receives an \textsc{init} message for the same value $v'$ from $t + 1$ processes (line~\ref{line:echo_receive_plurality}), $p_i$ broadcasts an \textsc{echo} message for $v'$ (line~\ref{line:echo_broadcast_2}) unless it has already done so.
If $|\mathsf{total\_init}| - |\mathsf{init}(\mathsf{most\_frequent})| \geq t + 1$ (line~\ref{line:echo_most_frequent}), process $p_i$ broadcasts an \textsc{echo} message for $\bot_{rd}$ (line~\ref{line:echo_broadcast_3}) as it knows that it is impossible that all correct processes have broadcast the same value.
When $p_i$ receives $2t + 1$ \textsc{echo} messages for the same value (line~\ref{line:echo_receive_quorum}), $p_i$ completes \shortvbt (line~\ref{line:echo_trigger_completed}).
Finally, when $p_i$ receives $t + 1$ \textsc{echo} messages for the same value $v$ (line~\ref{line:echo_receive_plurality_2}), $p_i$ validates a value $v^*$ according to the following logic:
\begin{compactitem}
    \item If $v = \bot_{rd}$, then $v^* = \mathsf{def}(p_i)$ (line~\ref{line:echo_validate_broadcast}).

    \item Otherwise, $v^* = v$ (line~\ref{line:echo_trigger_validate}).
\end{compactitem}

\smallskip
\noindent \textbf{Proof of correctness.}
We start by proving strong validity.

\begin{theorem} [Strong validity]
\shortvbt (\Cref{algorithm:echo_basic}) satisfies strong validity.
\end{theorem}
\begin{proof}
Suppose all correct processes that broadcast do so with the same value $v$.
Hence, due to the validity and safety properties of $\mathcal{RB}$, all correct processes that deliver a value from $\mathcal{RB}$ deliver $v$.
Therefore, no correct process sends an \textsc{echo} message for a non-$v$ value at line~\ref{line:echo_broadcast_2} as the rule at line~\ref{line:echo_receive_plurality} never activates.
Similarly, the rule at line~\ref{line:echo_most_frequent} never activates as there can be at most $t$ \textsc{init} messages for non-$v$ values received by any correct process, which implies that no correct process sends an \textsc{echo} message for a non-$v$ value at line~\ref{line:echo_broadcast_3}.
Hence, due to the check line \ref{line:echo_receive_plurality_2}, a correct process can only validate value $v$ (line~\ref{line:echo_trigger_validate}).
\end{proof}

Next, we prove safety.

\begin{theorem} [Safety]
\shortvbt (\Cref{algorithm:echo_basic}) satisfies safety.
\end{theorem}
\begin{proof}
Consider any correct process $p_i$ that validates a certain value $v'$ line~\ref{line:echo_trigger_validate}.
Necessarily, $p_i$ has received $t + 1$ $\langle \textsc{echo}, v \rangle$ messages line~\ref{line:echo_receive_plurality_2}, for some $v \in \mathsf{Value} \cup \{\bot_{rd}\}$.
We now consider two possibilities:
\begin{compactitem}
    \item Let $v = \bot_{rd}$.
    In this case, $v' = \mathsf{def}(p_i)$.

    \item Let $v \neq \bot_{rd}$.
    In this case, a correct process has previously delivered $v$ from $\mathcal{RB}$.
    Due to the safety property of $\mathcal{RB}$, a correct process has broadcast $v$.
\end{compactitem}
Safety is satisfied as it holds in both possible scenarios.
\end{proof}

The following theorem proves integrity.

\begin{theorem} [Integrity]
\shortvbt (\Cref{algorithm:echo_basic}) satisfies integrity.
\end{theorem}
\begin{proof}
The integrity property follows from the fact that the check at line~\ref{line:echo_receive_quorum} is only performed if $p_i$ has previously broadcast.
\end{proof}

Next, we prove termination.

\begin{theorem} [Termination]
\shortvbt (\Cref{algorithm:echo_basic}) satisfies termination.
\end{theorem}
\begin{proof}
Assuming that all correct processes propose and no correct process ever abandons \shortvbt, all correct processes eventually deliver a value from $\mathcal{RB}$ (due to $\mathcal{RB}$'s termination property).
At this point, we separate two possibilities:
\begin{compactitem}
    \item Let there exist a value $v \in \mathsf{Value} \cup \{\bot_{rd}\}$ such that at least $t + 1$ correct processes deliver $v$ from $\mathcal{RB}$.
    In this case, all correct processes eventually broadcast an \textsc{echo} message for $v$ (line~\ref{line:echo_broadcast_2}), which means that all correct processes eventually receive $2t + 1$ \textsc{echo} messages for $v$ and complete \shortvbt.
    
    \item Otherwise, every correct process $p_i$ eventually sends an \textsc{echo} message for $\bot_{rd}$ (line~\ref{line:echo_broadcast_3}).
    Indeed, consider the point in time at which $p_i$ receives \textsc{init} messages from all correct processes.
    Then, $|\mathsf{total\_init}| \geq 2t + 1 + f$, where $f$ is the number of faulty processes $p_i$ has heard from.
    Given that no value delivered from $\mathcal{RB}$ by at least $t + 1$ correct processes exists, $|\mathsf{init}(\mathsf{most\_frequent})| \leq t + f$.
    As $|\mathsf{total\_init}| - |\mathsf{init}(\mathsf{most\_frequent})| \geq 2t + 1 + f - t - f \geq t + 1$, the rule at line~\ref{line:echo_most_frequent} activates and $p_i$ broadcasts an \textsc{echo} message for $\bot_{rd}$ (line~\ref{line:echo_broadcast_3}).
    Therefore, all correct processes eventually receive $n -t \geq 2t + 1$ \textsc{echo} messages for $\bot_{rd}$ (line~\ref{line:echo_receive_quorum}) and complete \shortvbt (line~\ref{line:echo_trigger_completed}).
\end{compactitem}
Termination is satisfied as it holds in both possible cases.
\end{proof}

Finally, we prove totality.

\begin{theorem} [Totality]\label{thm:echo-basic-totality}
\shortvbt (\Cref{algorithm:echo_basic}) satisfies totality.
Concretely, if a correct process receives a $\mathsf{completed}$ indication at time $\tau$, then every correct process validates a value by time $\max(\tau, \text{GST}) + \delta$.
\end{theorem}

\begin{proof}
    Let $p_i$ be a correct process that receives a \textsf{completed} indication at time $\tau$. 
    Then, $p_i$ must have received $2t+1$ matching \textsc{echo} messages for some value $v \in \mathsf{Value} \cup \{\bot_{rd}\}$ by time $\tau$.
    At least $t+1$ of those messages are sent by correct processes.
    These messages are received by all correct processes by time $\max(\tau, \text{GST}) + \delta$.
    Therefore, every correct process validates $v$ by $\max(\tau, \text{GST}) + \delta$.
\end{proof}



\smallskip
\noindent \textbf{Proof of complexity.}
Next, we prove that any correct process sends $O(nL)$ bits in \shortvbt.

\begin{theorem} [Exchanged bits]
Any correct process sends $O(nL)$ bits in \shortvbt.
\end{theorem}
\begin{proof}
Each correct process broadcasts $O(1)$ \textsc{echo} messages (ensured by the reduction property of $\mathcal{RB}$), each with $O(L)$ bits.
Therefore, any correct process sends $O(nL)$ bits via \textsc{echo} messages.
Moreover, each correct process broadcasts only $O(1)$ \textsc{init} messages of size $O(L)$ bits. Thus, any correct process sends $O(nL)$ bits via \textsc{init} messages.
As $O(nL)$ bits are sent per-process while executing the reducing broadcast primitive (see~\cite{DBLP:journals/acta/MostefaouiR17}), any correct process sends $O(nL) + O(nL) + O(nL) = O(nL)$ bits in \shortvbt.
\end{proof}

Finally, we prove the number of asynchronous rounds \shortvbt requires.

\begin{theorem} [Asynchronous rounds]
Assuming all correct processes broadcast via \shortvbt and no correct process abandons \shortvbt, \shortvbt takes $4$ asynchronous rounds before all correct processes receive a $\mathsf{completed}$ indication.
\end{theorem}
\begin{proof}
First, we underline that $\mathcal{RB}$ requires $2$ asynchronous rounds until all correct processes deliver a value (see~\cite{DBLP:journals/acta/MostefaouiR17}).
Next, we analyze two scenarios:
\begin{compactitem}
    \item There exists a value $v \in \mathsf{Value} \cup \{\bot_{rd}\}$ such that at least $t + 1$ correct processes deliver $v$ from $\mathcal{RB}$.
    Hence, at the end of the second asynchronous round, these correct processes broadcast an \textsc{init} message for $v$.
    Therefore, every correct process broadcasts an \textsc{echo} message for $v$ at the end of the third asynchronous round.
    Thus, at the end of the fourth asynchronous round, every correct process receives $n - t \geq 2t + 1$ \textsc{echo} messages for $v$, and completes \shortvbt.

    \item There does not exist a value $v \in \mathsf{Value} \cup \{\bot_{rd}\}$ such that at least $t + 1$ correct processes deliver $v$ from $\mathcal{RB}$.
    In this case, every correct process sends an \textsc{echo} message for $\bot_{rd}$ at the end of the third asynchronous round.
    Therefore, all correct processes receive $n - t \geq 2t + 1$ \textsc{echo} messages for $\bot_{rd}$ at the end of the fourth asynchronous round, which concludes this case.
\end{compactitem}
The proof is concluded as it takes $4$ rounds before all correct processes complete \shortvbt.
\end{proof}

\subsection{\longvbt: Pseudocode \& Proof of Correctness and Complexity}
\label{subsection:validation_bcast_long_3_4}

This subsection presents \longvbt (\Cref{algorithm:validation_broadcast_long_3_4}), our hash-based implementation of the validation broadcast primitive.
\longvbt tolerates up to $t < n / 3$ Byzantine failures and exchanges $O(nL + n^2 \log(n) \kappa)$ bits.
\longvbt internally relies on rebuilding broadcast (see \Cref{section:rebuilding_broadccast}) and \shortvbt (see \Cref{subsection:validation_broadcast_basic_implementation}).

\begin{algorithm}
\caption{\longvbt: Pseudocode (for process $p_i$)}
\label{algorithm:validation_broadcast_long_3_4}
\footnotesize
\begin{algorithmic} [1] 
\State \textbf{Uses:}
\State \hskip2em Rebuilding broadcast, \textbf{instance} $\mathcal{RB}$ \BlueComment{see \Cref{section:rebuilding_broadccast}}
\State \hskip2em \shortvbt validation broadcast with $\mathsf{def}(p_i) = \bot$, \textbf{instance} $\mathcal{VB}$ \BlueComment{hash values are broadcast; see \Cref{subsection:validation_broadcast_basic_implementation}}


\medskip
\State \textbf{upon} $\mathsf{broadcast}(v \in \mathsf{Value})$: \label{line:validation_broadcast_long_propose}

\State \hskip2em \textbf{invoke} $\mathcal{RB}.\mathsf{broadcast}(v)$ \label{line:validation_broadcast_long_broadcast}

\medskip
\State \textbf{upon} $\mathcal{RB}.\mathsf{deliver}(v' \in \mathsf{Value} \cup \{\bot_{\mathit{reb}}\})$: \label{line:validation_broadcast_long_deliver_rb}
\State \hskip2em \textbf{if} $v' \neq \bot_{\mathit{reb}}$:
\State \hskip4em \textbf{invoke} $\mathcal{VB}.\mathsf{broadcast}\big( \mathsf{hash}(v') \big)$ \label{line:validation_broadcast_long_propose_aw}
\State \hskip2em \textbf{else:}
\State \hskip4em \textbf{invoke} $\mathcal{VB}.\mathsf{broadcast}(\bot)$ \label{line:validation_broadcast_long_broadcast_bot_vb}

\medskip
\State \textbf{upon} $\mathcal{VB}.\mathsf{completed}$: \label{line:short_vb_completed}
\State \hskip2em \textbf{trigger} $\mathsf{completed}$ \label{line:long_vb_trigger_completed}

\medskip
\State \textbf{upon} $\mathcal{VB}.\mathsf{validate}(\mathcal{H} \in \mathsf{Hash\_Value} \cup \{\bot\})$: \label{line:echos_from_short_vb}
\State \hskip2em \textbf{if} $\mathcal{H} = \bot$:
\State \hskip4em \textbf{trigger} $\mathsf{validate}\big( \mathsf{def}(p_i) \big)$ \label{line:long_vb_trigger_validate_own}
\State \hskip2em \textbf{else:}
\State \hskip4em \textbf{wait for} $\mathcal{RB}.\mathsf{rebuild}(v')$ such that $\mathsf{hash}(v') = \mathcal{H}$ \label{line:long_vb_wait_rebuild}
\State \hskip4em \textbf{trigger} $\mathsf{validate}(v')$ \label{line:long_vb_trigger_validate_theirs}

\end{algorithmic} 
\end{algorithm}

\smallskip
\noindent \textbf{Pseudocode description.}
Let us consider any correct process $p_i$.
When $p_i$ broadcasts its value (line~\ref{line:validation_broadcast_long_propose}), it disseminates that value using the $\mathcal{RB}$ instance of rebuilding broadcast (line~\ref{line:validation_broadcast_long_broadcast}).
If $p_i$ delivers a non-$\bot_{\mathit{reb}}$ value from $\mathcal{RB}$, it broadcasts the hash of the value via the $\mathcal{VB}$ instance of \shortvbt (line~\ref{line:validation_broadcast_long_propose_aw}).
If $p_i$ delivers $\bot_{\mathit{reb}}$ from $\mathcal{RB}$, it broadcasts $\bot$ via $\mathcal{VB}$ (line~\ref{line:validation_broadcast_long_broadcast_bot_vb}).
When $p_i$ completes $\mathcal{VB}$ (line~\ref{line:short_vb_completed}), it completes \longvbt (line~\ref{line:long_vb_trigger_completed}).
Finally, when $p_i$ validates $\mathcal{H} \in \mathsf{Hash\_Value} \cup \{\bot\}$ from $\mathcal{VB}$ (line~\ref{line:echos_from_short_vb}), it executes the following steps:
\begin{compactitem}
    \item If $\mathcal{H} = \bot$, $p_i$ validates $\mathsf{def}(p_i)$ from \longvbt (line~\ref{line:long_vb_trigger_validate_own}).

    \item Otherwise, $p_i$ waits until it rebuilds a value $v'$ from $\mathcal{RB}$ such that $\mathsf{hash}(v') = \mathcal{H}$ (line~\ref{line:long_vb_wait_rebuild}).
    Then, it validates $v'$ (line~\ref{line:long_vb_trigger_validate_theirs}).
\end{compactitem}

\smallskip
\noindent \textbf{Proof of correctness}
We start by proving strong validity.

\begin{theorem}[Strong validity]
\longvbt (\Cref{algorithm:validation_broadcast_long_3_4}) satisfies strong validity.
\end{theorem}

\begin{proof}
Suppose all correct processes that broadcast do so with the same value $v$.
By the strong validity and safety properties of $\mathcal{RB}$, any correct process that delivers from $\mathcal{RB}$ delivers $v$.
Thus, any correct process that broadcasts via $\mathcal{VB}$ broadcasts $\mathcal{H} = \mathsf{hash}(v) \neq \bot$.
By the strong validity of the $\mathcal{VB}$ instance, all correct processes that validate from $\mathcal{VB}$ do validate $\mathcal{H} \neq \bot$.
Finally, by the collision-resistance of the $\mathsf{hash}(\cdot)$ function, no correct process can rebuild some value $v' \ne v$ from $\mathcal{RB}$ such that $\mathsf{hash}(v') = \mathsf{hash}(v)$, thus concluding the proof.   
\end{proof}

Next, we prove the safety property.

\begin{theorem}[Safety]
\longvbt (\Cref{algorithm:validation_broadcast_long_3_4}) satisfies safety.
\end{theorem}
\begin{proof}
Let $p_i$ be a correct process that validates a value $v$. We distinguish two cases:
\begin{compactitem}
    \item Process $p_i$ validates $v$ at line~\ref{line:long_vb_trigger_validate_own}.
    In this case, $v = \mathsf{def}(p_i)$.
    
    \item Process $p_i$ validates $v$ at line~\ref{line:long_vb_trigger_validate_theirs}.
    As $\mathcal{H} \neq \bot$ and the default value for $p_i$ in $\mathcal{VB}$ is $\bot$, the safety property of $\mathcal{VB}$ guarantees that a correct process has broadcast $\mathcal{H} = \mathsf{hash}(v) \neq \bot$ via $\mathcal{VB}$.
    Therefore, due to $\mathsf{hash}(\cdot)$'s collision resistance, $v$ has been delivered from $\mathcal{RB}$ by a correct process.
    Due to the safety property of $\mathcal{RB}$, a correct process has broadcast $v$ using $\mathcal{RB}$, which implies that a correct process has broadcast $v$ using \longvbt.
\end{compactitem}
The theorem holds as its statement is true in both possible cases.
\end{proof}

The following theorem proves integrity.

\begin{theorem}[Integrity]
\longvbt (\Cref{algorithm:validation_broadcast_long_3_4}) satisfies integrity.
\end{theorem}
\begin{proof}
Let $p_i$ be a correct process that receives a \textsf{completed} indication. 
By the integrity of $\mathcal{VB}$, $p_i$ must have broadcast via $\mathcal{VB}$, 
and thus $p_i$ must have delivered from $\mathcal{RB}$. 
By the integrity of $\mathcal{RB}$, $p_i$ must have broadcast via $\mathcal{RB}$, and thus must have broadcast via \longvbt.
\end{proof}

Next, we prove the termination property.

\begin{theorem}[Termination]
\longvbt (\Cref{algorithm:validation_broadcast_long_3_4}) satisfies termination.
\end{theorem}
\begin{proof}
The termination property of \longvbt follows from the termination property of $\mathcal{RB}$ and $\mathcal{VB}$.
\end{proof}

Finally, we prove the totality property.

\begin{theorem}[Totality]\label{thm:validation_broadcast_long_3_4-totality}
\longvbt (\Cref{algorithm:validation_broadcast_long_3_4}) satisfies totality.
\end{theorem}
\begin{proof}
Suppose some correct process receives a \textsf{completed} indication at time $\tau$, then it must have received a \textsf{completed} indication from $\mathcal{VB}$ at $\tau$.
%
%
By the totality property of $\mathcal{VB}$ (\Cref{thm:echo-basic-totality}), all correct processes validate some $\mathcal{H} \in \mathsf{Hash\_Value} \cup \{\bot\}$ from $\mathcal{VB}$ by time $\max(\tau, \text{GST}) + \delta$. 
%
%
Thus, the rule at line~\ref{line:echos_from_short_vb} activates for $\mathcal{H}$ at every correct process $p_i$ by time $\max(\tau, \text{GST}) + \delta$.

If $\mathcal{H} = \bot$, then $p_i$ validates a value by time $\max(\tau, \text{GST}) + \delta$.
Otherwise, the safety property of $\mathcal{VB}$ proves that some correct process $p_j$ has broadcast $\mathcal{H}$ via $\mathcal{VB}$ by time $\tau' \leq \max(\tau, \text{GST}) + \delta$.
Thus, $p_j$ has delivered a value $v'$ from $\mathcal{RB}$ by time $\tau'$ such that $\mathsf{hash}(v') = \mathcal{H}$.
Due to the rebuilding validity of $\mathcal{RB}$, $p_i$ rebuilds and validates $v'$ by time $\max(\tau', \text{GST}) + \delta \leq \max(\tau, \text{GST}) + 2\delta$. 
\end{proof}

\smallskip
\noindent \textbf{Proof of complexity.}
We prove that any correct process sends $O(L + n \log(n) \kappa)$ bits in \longvbt.

\begin{theorem} [Exchanged bits]
Any correct process sends $O(L + n \log(n) \kappa)$ bits in \longvbt.
\end{theorem}
\begin{proof}
Correct processes only exchange bits as part of the $\mathcal{RB}$ and $\mathcal{VB}$ instances. Correct processes $\mathcal{RB}$-broadcast at most an $L$-sized value, and $\mathcal{VB}$-broadcast at most a $\kappa$-sized value (where $\kappa$ is the length of a hash). Thus, any correct process sends $O(L + n\log(n)\kappa)$ + $O(n\kappa) = O(L + n\log(n)\kappa)$ bits.
\end{proof}

Next, we prove that \longvbt requires $6$ asynchronous rounds.

\begin{theorem} [Asynchronous rounds]
Assuming all correct processes broadcast via \longvbt and no correct process abandons \longvbt, \longvbt takes $6$ asynchronous rounds before all correct processes receive a $\mathsf{completed}$ indication.
\end{theorem}
\begin{proof}
As $\mathcal{RB}$ requires $2$ asynchronous rounds (see \Cref{section:rebuilding_broadccast}) and $\mathcal{VB}$ requires $4$ asynchronous rounds (see \Cref{subsection:validation_broadcast_basic_implementation}), the theorem holds.
\end{proof}

\subsection{\longvbf: Pseudocode \& Proof of Correctness and Complexity}
\label{subsection:its_validation_broadcast}

In this subsection, we introduce \longvbf (\Cref{algorithm:its_validation_broadcast}), our implementation of validation broadcast that exchanges $O\big( nL + n^2 \log(n) \big)$ bits while relying on no cryptographic primitives.
\longvbf tolerates up to $t < n / 5$ Byzantine processes, and it follows the similar approach as \longgcf.
Specifically, \longvbf relies on (1) the \reduceacool algorithm (see \Cref{section:acool_reduction}), and (2) the AW graded consensus algorithm (see \Cref{section:existing_primitives}).

\begin{algorithm}
\caption{\longvbf: Pseudocode (for process $p_i$)}
\label{algorithm:its_validation_broadcast}
\footnotesize
\begin{algorithmic} [1] 
\State \textbf{Uses:}
\State \hskip2em \reduceacool, \textbf{instance} $\mathcal{ACOOL}$ \BlueComment{see \Cref{section:acool_reduction}}
\State \hskip2em AW graded consensus \cite{AttiyaWelch23}, \textbf{instance} $\mathcal{AW}$ \BlueComment{see \Cref{section:existing_primitives}}

\medskip
\State \textbf{upon} $\mathsf{broadcast}(v \in \mathsf{Value})$: \label{line:longvbf_broadcast}
\State \hskip2em let $(\mathit{success}, \mathit{reduction\_output}_i) \gets \mathcal{ACOOL}(v)$ \label{line:longvbf_input_to_acool}
\State \hskip2em let $[m_1, m_2, ..., m_n] \gets \mathsf{RSEnc}(\mathit{reduction\_output}_i, n, t + 1)$
\State \hskip2em \textbf{broadcast} $\langle \textsc{symbol}, m_i \rangle$  \label{line:longvbf_broadcast_symbol}
\State \hskip2em \textbf{if} $\mathit{success} = 1$:
\State \hskip4em \textbf{invoke} $\mathcal{AW}.\mathsf{propose}(\mathsf{HAPPY})$ \label{line:longvbf_propose_happy_to_aw}
\State \hskip2em \textbf{else:}
\State \hskip4em \textbf{invoke} $\mathcal{AW}.\mathsf{propose}(\mathsf{SAD})$ \label{line:longvbf_propose_sad_to_aw}

\medskip
\State \textbf{upon} $\mathcal{AW}.\mathsf{output}(v' \in \{\mathsf{HAPPY}, \mathsf{SAD}\}, g' \in \{0, 1\})$: \label{line:longvbf_decide_from_aw}

\State \hskip2em \textbf{broadcast} $\langle v' \rangle$ \label{line:longvbf_broadcast_happy} \label{line:longvbf_broadcast_sad}


\medskip
\State  \textcolor{blue}{\(\triangleright\) completion rules (triggered only if previously broadcast)}
\State  \textbf{upon} receiving $4t + 1$ \textsc{symbol} messages and $2t + 1$ $\langle \mathsf{HAPPY} \rangle$ messages: \label{line:rule_complete_happy}
\State \hskip2em \textbf{trigger} $\mathsf{completed}$ \label{line:complete_happy}

\medskip
\State \textbf{upon} receiving $2t + 1$ $\langle \mathsf{SAD} \rangle$ messages: \label{line:rule_complete_sad}
\State \hskip2em \textbf{trigger} $\mathsf{completed}$ \label{line:complete_sad}

\medskip
\State  \textcolor{blue}{\(\triangleright\) validation rules (can be triggered anytime)}

\State  \textbf{upon} receiving $3t + 1$ \textsc{symbol} messages and $t + 1$ $\langle \mathsf{HAPPY} \rangle$ messages: \label{line:rule_validate_happy}
\State \hskip2em \textbf{trigger} $\mathsf{validate}\big( \mathsf{RSDec}(t + 1, t, \text{received symbols}) \big)$  \label{line:validate_happy}

\smallskip
\State \textbf{upon} receiving $t + 1$ $\langle \mathsf{SAD} \rangle$ messages: \label{line:rule_validate_sad}
\State \hskip2em \textbf{trigger} $\mathsf{validate}\big( \mathsf{def}(p_i) \big)$ \label{line:validate_sad}

\end{algorithmic}
\end{algorithm}


\smallskip
\noindent \textbf{Pseudocode description.}
When a correct process $p_i$ broadcasts a value $v$ to \longvbf (line~\ref{line:longvbf_broadcast}), it forwards $v$ to the $\mathcal{ACOOL}$ instance of the \reduceacool algorithm (line~\ref{line:longvbf_input_to_acool}).
Once $p_i$ obtains a pair $(\mathit{success}, \mathit{reduction\_output}_i)$ from $\mathcal{ACOOL}$, $p_i$ broadcast its associated Reed-Solomon symbol for $\mathit{reduction\_output}_i$ in a $\textsc{symbol}$ message to achieve totality later. Then, $p_i$ checks if $\mathit{success} = 1$.
If so, $p_i$ proposes $\mathsf{HAPPY}$ to the $\mathcal{AW}$ instance of the AW graded consensus algorithm (line~\ref{line:longvbf_propose_happy_to_aw}).
Otherwise, $p_i$ proposes $\mathsf{SAD}$ to $\mathcal{AW}$ (line~\ref{line:longvbf_propose_sad_to_aw}).
When $p_i$ decides a pair $(v', g')$ from $\mathcal{AW}$ (line~\ref{line:longvbf_decide_from_aw}), $p_i$ broadcast $v'$, while $g'$ is ignored.
Finally, when $4t+1$ $\textsc{symbol}$ and $2t+1$ $\mathsf{HAPPY}$ messages (line~\ref{line:rule_complete_happy}) or $2t+1$ $\mathsf{SAD}$ messages (line~\ref{line:rule_complete_sad}), the $p_i$ process can trigger $\mathsf{completed}$. With $t$ fewer messages of the type mentioned above (line~\ref{line:rule_validate_happy} or line~\ref{line:rule_validate_sad}), process $p_i$ can trigger $\mathsf{validate}$ even if some correct processes have abandoned \longvbf after a certain completion, thus guaranteeing completeness.

\smallskip
\noindent \textbf{Proof of correctness.}
We start by proving the strong validity property.

\begin{theorem} [Strong validity]
\longvbf (\Cref{algorithm:its_validation_broadcast}) satisfies strong validity.
\end{theorem}
\begin{proof}
Suppose all correct processes that broadcast do so with the same value $v$.
Thus, all correct processes that input a value to $\mathcal{ACOOL}$ do input $v$ (line~\ref{line:longvbf_input_to_acool}).
The strong validity property of $\mathcal{ACOOL}$ ensures that each correct process $p_i$ that receives an output from $\mathcal{ACOOL}$ receives $(1, v)$.
Therefore, all correct processes that send a \textsc{symbol} message include a correctly-encoded RS symbol (line~\ref{line:longvbf_broadcast_symbol}).
Moreover, all correct processes that propose to $\mathcal{AW}$ do so with $\mathsf{HAPPY}$ (line~\ref{line:longvbf_propose_happy_to_aw}).
The strong validity property of $\mathcal{AW}$ ensures that all correct processes decide $(\mathsf{HAPPY}, 1)$ from $\mathcal{AW}$ (line~\ref{line:longvbf_decide_from_aw}), which implies that no correct process validates any value at line~\ref{line:validate_sad}.
Finally, if a correct process validates a value at line~\ref{line:validate_happy}, that value must be $v$ as it has received at least $2t + 1$ correctly-encoded RS symbols for $v$.
\end{proof}

The next theorem proves that \longvbf satisfies safety.

\begin{theorem} [Safety]
\longvbf (\Cref{algorithm:its_validation_broadcast}) satisfies safety.
\end{theorem}
\begin{proof}
Let $p_i$ be any correct process.
We consider the following two cases:
\begin{compactitem}
    \item Let $p_i$ validate a value $v$ at line~\ref{line:validate_sad}.
    In this case, $v = \mathsf{def}(p_i)$.

    \item Let $p_i$ validate a value $v$ at line~\ref{line:validate_happy}.
    In this case, some correct process has decided $(\mathsf{HAPPY}, \cdot)$ from $\mathcal{AW}$ (as $p_i$ has received a $\langle \mathsf{HAPPY} \rangle$ message from $t + 1$ processes).
    Therefore, the safety property of $\mathcal{AW}$ guarantees that a correct process has previously proposed $\mathsf{HAPPY}$ to $\mathcal{AW}$, which means that process has received $(1, v^*)$ from $\mathcal{ACOOL}$.
    The agreement property of $\mathcal{ACOOL}$ ensures that all correct processes that send a \textsc{symbol} message do so with a correctly-encoded RS symbol for $v^*$.
    Moreover, the safety property of $\mathcal{ACOOL}$ ensures that $v^*$ is broadcast via \longvbf by a correct process.
    As $p_i$ receives at least $2t + 1$ correctly-encoded RS symbols before validating $v$, $v = v^*$.
\end{compactitem}
The safety property is ensured as its statement holds in both possible cases.
\end{proof}

Next, we prove integrity.

\begin{theorem} [Integrity]
\longvbf (\Cref{algorithm:its_validation_broadcast}) satisfies integrity.
\end{theorem}
\begin{proof}
The statement of the theorem follows directly from the pseudocode of \longvbf.
\end{proof}

The following theorem proves termination.

\begin{theorem} [Termination]
\longvbf (\Cref{algorithm:its_validation_broadcast}) satisfies termination.
\end{theorem}
\begin{proof}
All correct processes decide from $\mathcal{AW}$ due to the termination property of $\mathcal{ACOOL}$ and $\mathcal{AW}$.
We now consider two scenarios:
\begin{compactitem}
    \item At least $2t + 1$ correct processes decide $(\mathsf{SAD}, \cdot)$ from $\mathcal{AW}$.
    In this case, every correct process eventually receives $2t + 1$ $\langle \mathsf{SAD} \rangle$ messages (line~\ref{line:rule_complete_sad}), and triggers $\mathsf{completed}$ (line~\ref{line:complete_sad}).

    \item Otherwise, every correct process eventually receives $2t + 1$ $\langle \mathsf{HAPPY} \rangle$ messages and $4t + 1$ \textsc{symbol} messages (line~\ref{line:rule_complete_happy}), and triggers $\mathsf{completed}$ (line~\ref{line:complete_happy}).
\end{compactitem}
Termination is ensured.
\end{proof}

Lastly, we prove totality.

\begin{theorem} [Totality]
\longvbf (\Cref{algorithm:its_validation_broadcast}) satisfies totality.
\end{theorem}
\begin{proof}
We consider two scenarios:
\begin{compactitem}
    \item A correct process triggers $\mathsf{completed}$ at line~\ref{line:complete_happy} at time $\tau$.
    Hence, this correct process has received $4t + 1$ RS symbols and $2t + 1$ $\langle \mathsf{HAPPY} \rangle$ messages (line~\ref{line:rule_complete_happy}) by time $\tau$.
    Therefore, every correct process receives at least $3t + 1$ symbols and $t + 1$ $\langle \mathsf{HAPPY} \rangle$ messages (line~\ref{line:rule_validate_happy}) from correct processes by time $\max(\text{GST},\tau) + \delta$, and validates a value (line~\ref{line:validate_happy}) by $\max(\tau, \text{GST}) + \delta$.

    \item A correct process triggers $\mathsf{completed}$ at line~\ref{line:complete_sad} at time $\tau$.
    Hence, this correct process has received $2t + 1$ $\langle \mathsf{SAD} \rangle$ messages (line~\ref{line:rule_complete_sad}) by time $\tau$.
    Therefore, every correct process receives $t + 1$ $\langle \mathsf{SAD} \rangle$ messages (line~\ref{line:rule_validate_sad}) by time $\max(\tau, \text{GST}) + \delta$, and validates a value (line~\ref{line:validate_sad}) by $\max(\tau, \text{GST}) + \delta$.
\end{compactitem}
As totality is ensured in both possible scenarios, the proof is concluded.
\end{proof}

\smallskip
\noindent \textbf{Proof of complexity.}
We prove that any correct process sends $O\big( L + n\log(n) \big)$ bits in \longvbf.

\begin{theorem} [Exchanged bits]
Any correct process sends $O\big( L + n\log(n) \big)$ bits in \longvbf.
\end{theorem}
\begin{proof}
Any correct process sends $O\big( L + n \log(n) \big) + O(n) = O\big( L + n \log(n) \big)$ bits via $\mathcal{ACOOL}$ and $\mathcal{AW}$.
Moreover, each correct process also sends $O\big( L + n\log(n) \big) + O(n) = O\big( L + n\log(n) \big)$ bits via $\textsc{symbol}$, $\mathsf{HAPPY}$ and $\mathsf{SAD}$ messages.
\end{proof}

Finally, we prove that \longvbf requires $15$ asynchronous rounds.

\begin{theorem} [Asynchronous rounds]
Assuming all correct processes broadcast via \longvbf and no correct process abandons \longvbf, \longvbf takes $15$ asynchronous rounds before all correct processes receive a $\mathsf{completed}$ indication.
\end{theorem}
\begin{proof}
Recall that $\mathcal{ACOOL}$ requires $5$ asynchronous rounds (see \Cref{section:acool_reduction}).
Hence, at the end of the fifth asynchronous round, each correct process (1) broadcasts a \textsc{symbol} message, and (2) proposes to $\mathcal{AW}$.
As $\mathcal{AW}$ requires $9$ asynchronous rounds (see \Cref{section:existing_primitives}), all correct processes broadcast a $\mathsf{HAPPY}$ or a $\mathsf{SAD}$ message at the end of the $14$-th asynchronous round.
Therefore, at the end of the $15$-th asynchronous round, each correct process receives (1) $4t + 1$ \textsc{symbol} and $2t + 1$ $\mathsf{HAPPY}$ messages, or (2) $2t + 1$ $\mathsf{SAD}$ messages, thus concluding the proof.
\end{proof}
\section{\name for Partially Synchronous Agreement on a Core Set (ACS)}
\label{subsec:from_ole_to_acs}

In this section, we adapt \name's transformation to Agreement on a Core Set (ACS)~\cite{AAPS23, cohen2023concurrent,das2024asynchronous,duan2023practical,shoup2024theoretical}, also known as Asynchronous Common Subset or Vector Consensus. This variant of Byzantine agreement (see Module~\ref{mod:acs}), which is the strongest achievable in non-synchronous environments~\cite{civit2023validity}, ensures that correct processes agree on a vector containing (at least) $n-t$ proposals. ACS is crucial for protocols like Atomic Broadcast~\cite{correia2006consensus,Doudou1998,neves2005solving} (also called Total Order Broadcast) and Asynchronous Secure Multi-Party Computation (AMPC), which has been studied extensively since the early 1990s~\cite{beaver1991efficient,ben1993asynchronous,Ben-Or94}. Recent implementations of AMPC include work such as \cite{abraham2024perfect} for perfect security and \cite{goyal2024towards,ji2024linear} for statistical security, improving upon earlier protocols~\cite{choudhury2017efficient,choudhury2023communication}.

We first describe how ACS reduces to reliable broadcast and index agreement on a core set (index ACS) in \Cref{section:acs_process_validation}. Next, we describe how we can achieve index ACS via oblivious leader election (OLE) in \Cref{section:acs_reduction_to_iacs_ole}, along with an implementation that assumes synchrony. Finally, we put the obtained synchronous index ACS under \name to achieve (efficient) partially synchronous index ACS and therefore, (efficient) partially synchronous ACS, in \Cref{section:acs_merging}.

\begin{module}
\caption{Agreement on a Core Set (ACS)}
\label{mod:acs}
\footnotesize
\begin{algorithmic}[1]
\Statex \textbf{Parameters:}
\Statex \hskip2em $\mathsf{Integer}$ $t$ \BlueComment{maximum number of tolerated faults}

\medskip 
\Statex \textbf{Events:}

\Statex \hskip2em \emph{request} $\mathsf{propose}(\mathsf{Value} $ $v)$: a process proposes a value $v$.

\Statex \hskip2em \emph{indication} $\mathsf{decide}(\mathsf{Vector}(\mathsf{Value} \cup \{\bot\})$ $ V' )$: a process decides a vector $V'$.

\medskip 
\Statex \textbf{Properties:}

\Statex \hskip2em \emph{Agreement:} If two correct processes $p_i$ and $p_j$ output $V$ and $V'$, respectively, then $V = V'$.

\smallskip 
\Statex \hskip2em \emph{Validity:} If a correct process $p_i$ outputs $V$, then:
\Statex \hskip4em \emph{(a)} $|\{p_j \in \Pi| V[j] \neq \bot\}| \geq n-t$, and
\Statex \hskip4em \emph{(b)} if $V[j] \neq \bot$ and $p_j$ is correct, then $V[j]$ has been proposed by $p_j$.

\smallskip
\Statex \hskip2em \emph{Termination:} Every correct process eventually decides.
\end{algorithmic}
\end{module}

\subsection{ACS to Index Agreement on a Core Set (Index ACS)}\label{section:acs_process_validation}

We now explain how ACS can be reduced to index ACS (see Module~\ref{mod:index-acs}), following the approach from \cite{das2024asynchronous}.\footnote{The corresponding primitive is referred to as Index Asynchronous Common Subset in \cite{das2024asynchronous}. However, we opt for the term Agreement on a Core Set to avoid any confusion related to the network model.}
In the context of index ACS, \emph{process validation} refers to the mechanism by which a process considers that another process has done some action correctly (e.g. reliably broadcast). Each process $p_i$ maintains a dynamic set $\mathit{Valid}_i$ of processes it has validated (in our case, from which it has reliable-delivered the input). This set can grow over time but never decrease (a process cannot be ``un''-validated).
A set of processes $V$ is said \emph{locally validated} (resp., \emph{globally validated}) if every process $p_j \in V$ has been validated by at least one correct process (resp. by all the correct processes). Let us note that if $V$ is a singleton, it boils down to the definition from \cite{das2024asynchronous}.

An essential property of the validated set $\mathit{Valid}_i$ in index ACS is $x$-\emph{Totality}. This property ensures that if a set of processes $V$ is locally validated at some time $\tau$, then it will be globally validated by some time $\tau' < \infty$. In the case of partial synchrony,  $\tau' \leq \mathsf{max}(\tau, \text{GST}) + x \delta$, where $x$ is a parameter related to the latency of process validation. In practice, this property is (typically) ensured by the $x$-Totality property of reliable broadcast (see Module~\ref{mod:reliable-broadcast}) with some constant number of rounds $x$. 


\begin{module}
\caption{Index Agreement on a Core Set (Index ACS)}
\label{mod:index-acs} 
\footnotesize
\begin{algorithmic}[1]
\Statex \textbf{Parameters:}
\Statex \hskip2em $\mathsf{Integer}$ $x$ \BlueComment{totality-related latency of process validation}

\medskip 
\Statex \textbf{Notes:}
\Statex \hskip2em Every correct process $p_i$ is \emph{equipped} with a (dynamic) set of processes it has \emph{validated}, $\mathit{Valid}_i$, that can have new processes during the execution \hphantom{allllll} of the module. The set $\mathit{Valid}_i$ satisfies $x$-\emph{Totality}: 
if a set of processes $V$ is locally validated at some time $\tau$, then it will be globally validated \hphantom{allllll} by time $\tau' < \infty$. In the case of partial synchrony, $\tau' \leq \mathsf{max}(\tau, \text{GST}) + x \delta$.

\medskip 
\Statex \textbf{Events:}

\Statex \hskip2em \emph{request} $\mathsf{propose}(\mathsf{Set}(\mathsf{Process}) \mathit{V} \subseteq \mathit{Valid}_i)$: a process $p_i$ proposes its set of validated processes.

\Statex \hskip2em \emph{indication} $\mathsf{decide}(\mathsf{Set}(\mathsf{Process})\;V')$: a process decides a subset of processes $V'$.

\medskip 
\Statex \textbf{Properties} 

\Statex \hskip2em \emph{Agreement:} If two correct processes $p_i$ and $p_j$ decide $V$ and $V'$, respectively, then $V = V'$.

\smallskip 
\Statex \hskip2em \emph{Local Validity:} If an honest process decides $V$, then $V$ has been locally validated and $|V| \geq n-t$.

\smallskip
\Statex \hskip2em \emph{Strong Validity:} If every honest process proposes the same set of processes $V$, then each honest process decides $V$. 

\smallskip 
\Statex \hskip2em \emph{Termination:} Every correct process eventually decides.

\end{algorithmic}
\end{module}

As explained in \cite{das2024asynchronous}, it is easy to solve ACS atop index ACS (see \Cref{algorithm:acs_to_index_acs}). Initially, each process broadcasts its value using the reliable broadcast primitive. When a correct process $p_i$ delivers a value from the $j$-th instance of the reliable broadcast, it stores them and insert $p_j$ to its $\mathit{Valid}_i$. When $n-t$ processes have been validated, $p_i$ then propose $\mathit{Valid}_i$ to the index ACS. Upon deciding on a set $V'$, from the index ACS, $p_i$ first waits until it has validated every process in $V'$. Then, $p_i$ define its decision as follows: for each process $p_j \in V'$, $p_i$ set its $j$-th output as the value it delivered from the $j$-th reliable broadcast, and for $p_j \notin V'$, $p_i$ set its $j$-th output as $\bot$.

The reduction of holds for the following reasons. First, the agreement property of ACS is guaranteed by the agreement properties of both the index ACS and reliable broadcast. Similarly, the validity of ACS is derived from the validity property of the index ACS and the justification property of the reliable broadcast. Finally, the termination of ACS is ensured by the termination property of the index ACS combined with the obligation and $x$-totality properties of the reliable broadcast.


Lastly, let us explain the latency and bit complexity of the reduction. From the protocol, it can be seen that the latency is $latency(\mathcal{RBC}) + latency(\mathcal{IACS}) + O(1)$, while the per-process bit complexity is $n \cdot pbit(\mathcal{RBC}) + pbit(\mathcal{IACS})$. 

\begin{algorithm}
\caption{Reducing ACS to Index ACS and Reliable Broadcast: Pseudocode (for process $p_i$)}
\label{algorithm:acs_to_index_acs}
\footnotesize
\begin{algorithmic} [1]

\State \textbf{Local variables:}
\State \hskip2em $\mathsf{Map}(\mathsf{Process} \rightarrow \mathsf{Value} \cup \{\bot\})$ $\mathit{delivered}_i \gets \{\bot, \dots, \bot\}$ \BlueComment{values delivered by reliable broadcasts} 
\State \hskip2em $\mathsf{Vector}(\mathsf{Value} \cup \{\bot\})$ $\mathit{output}_i \gets \{\bot, \dots, \bot\}$

\medskip
\State \textbf{Global variables:}
\State \hskip2em $\mathsf{Set}(\mathsf{Process})$ $\mathit{Valid}_i \gets \emptyset$ \BlueComment{set of validated processes. can be read by other modules}

\medskip
\State \textbf{Uses:}
\State \hskip2em Reliable Broadcast, \textbf{instance} $\mathcal{RBC}_j$ with source $p_j$, for each $j \in [1:n]$ \BlueComment{see Module~\ref{mod:reliable-broadcast}}
\State \hskip2em Index ACS, \textbf{instance} $\mathcal{IACS}$, associated with $\mathit{Valid}_i$ \BlueComment{see Module~\ref{mod:index-acs}}

\medskip
\State \textbf{upon} $\mathsf{propose}(v \in \mathsf{Value})$:
\State \hskip2em \textbf{invoke} $\mathcal{RBC}_i.\mathsf{broadcast}(v)$ \label{line:acs_propose_broadcast}

\medskip
\State \textbf{upon} $\mathcal{RBC}_j.\mathsf{deliver}(v' \in \mathsf{Value})$:
\State \hskip2em $\mathit{delivered}_i[p_j] \gets v'$ \label{line:acs_rbc_deliver}
\State \hskip2em $\mathit{Valid}_i \gets \mathit{Valid}_i \cup \{p_j\} $ \label{line:acs_valid_update}
\State \hskip2em \textbf{if} $|\mathit{Valid}_i| = n-t$ \textbf{then:}
\State \hskip4em \textbf{invoke} $\mathcal{IACS}.\mathsf{propose}(\mathit{Valid}_i)$\label{line:acs_iacs_propose}

\medskip
\State \textbf{upon} $\mathcal{IACS}.\mathsf{decide}(V' \subseteq \Pi)$:
\State \hskip2em \textbf{wait until} $V' \subseteq \mathit{Valid}_i $ \label{line:acs_wait_for_rbc}
\State \hskip2em \textbf{for each} $p_j \in V'$:
\State \hskip4em $\mathit{output}_i[j] \gets \mathit{delivered}_i[p_j]$ \label{line:acs_output}
\State \hskip2em \textbf{trigger} $\mathsf{decide}(\mathit{output}_i)$ \label{line:acs_decide}

\end{algorithmic}
\end{algorithm}

\begin{module}
\caption{Reliable Broadcast}
\label{mod:reliable-broadcast}
\footnotesize
\begin{algorithmic}[1]

\Statex \textbf{Parameters:}
\Statex \hskip2em $\mathsf{Process}$ $p_s$ \BlueComment{source process}
\Statex \hskip2em $\mathsf{Integer}$ $x$ \BlueComment{constant related to the totality property}

\medskip 
\Statex \textbf{Events:}

\Statex \hskip2em \emph{request} $\mathsf{broadcast}(\mathsf{Value}\ v)$: process $p_s$ broadcasts a value $v$.

\Statex \hskip2em \emph{indication} $\mathsf{deliver}(\mathsf{Value}\ v')$: a process delivers a value $v'$.

\medskip 
\Statex \textbf{Properties:}

\Statex \hskip2em \emph{Obligation:} 
If $p_s$ is correct and broadcasts a value $v$, then every correct process eventually delivers the value $v$.

\smallskip 
\Statex \hskip2em \emph{Justification:} 
If $p_s$ is correct and a correct process delivers a value $v$, then $v$ has been broadcast by $p_s$.

\smallskip 
\Statex \hskip2em \emph{Agreement:} 
If two correct processes $p_i$ and $p_j$ delivers the value $v$ and $v'$, respectively, then $v = v'$.


\smallskip
\Statex \hskip2em $x$-\emph{Totality}: If a correct process delivers a value at time $\tau$, then every correct process delivers a value by time $\max(\tau, \text{GST}) + x\delta$.

\end{algorithmic}
\end{module}

\subsection{Synchronous Index ACS to Synchronous Oblivious Leader Election (OLE)}\label{section:acs_reduction_to_iacs_ole}


We now describe how under synchrony, index ACS can be reduced to oblivious leader election (OLE) (see Module~\ref{mod:ole}). Our reduction works for (randomized) OLE protocol with a constant success probability $\rho$, in which the obtained index ACS also succeeds with the same probability.

\begin{module}
\caption{Oblivious Leader Election (OLE)}
\label{mod:ole}
\footnotesize
\begin{algorithmic}[1]
 
\Statex \textbf{Parameters:}
\Statex \hskip2em $\mathsf{Real}$ $\rho \in [0,1]$ \BlueComment{minimum probability of success}

\medskip
\Statex \textbf{Events:}

\Statex \hskip2em \emph{request} $\mathsf{call()}$: a process requests to participate in the leader election.

\Statex \hskip2em \emph{indication} $\mathsf{select}(\mathsf{Process}\ p_\ell)$: a process is notified of the elected leader $p_\ell$.

\medskip 
\Statex \textbf{Properties:}

\Statex \hskip2em \emph{$\rho$-Quality:} With probability at least $\rho$, all correct processes trigger $\mathsf{select}(p_\ell)$ for some (common) process $p_\ell$ that was correct at the time of \hphantom{allllll}the first invocation of $\mathsf{call()}$ by any correct process.

\smallskip
\Statex \hskip2em \emph{Termination:} Every correct process eventually selects a leader.

\end{algorithmic}
\end{module}

Our reduction is presented in \Cref{algorithm:synchronous_iacs}. Let us describe the reduction.
Each correct process $p_i$ maintains a dynamic set $\mathit{Valid}_i$ of validated processes which satisfy $x$-totality (see Module~\ref{mod:index-acs}). Initially, every correct process starts with a set $\mathit{Valid}_i$ containing at least $n - t$ processes. The reduction is achieved in several rounds. In the first round, each process $p_i$ broadcasts its set $\mathit{V}_i \subseteq \mathit{Valid}_i$ of size at least $n-t$. Then, $p_i$ sets a timer for $x \delta$ time units, the sufficient time to guarantee global validation of each process included in $\mathit{V}_i$ due to $x$-totality of $\mathit{Valid}_i$. In the second round, correct processes invoke the OLE module, which will elect an (optimistically common and so-far correct) leader $p_{\ell}$ to each process. Upon receiving the leader, a correct process first waits for its timer to expire to ensure global validity guaranteed by the $x$-totality.
Then, if the leader has sent a well-formed set $\mathit{V}_\ell$ (of at least $n - t$ processes) that contains only processes that have been validated by $p_i$, then $p_i$ decides $\mathit{V}_\ell$. Otherwise, $p_i$ defaults to its own set $\mathit{V}_i$.

The reduction is correct due to the following. First, termination and validity are ensured by construction. Next, with probability at least $\rho$, processes agree on a common, so-far correct leader $p_\ell$ (by the $\rho$-quality of the OLE). Due to $x$-totality guaranteed by each set $\mathit{Valid}_i$, each correct process will validate all processes included in $\mathit{V}_\ell$ before the timer expires. Consequently, with probability at least $\rho$, the correct processes agree on the leader's valid proposal $\mathit{V}_\ell$. Using the same idea, it can be proven that strong validity also holds with probability at least $\rho$.

Finally, let us note the latency and bit complexity of the reduction. Notice that a set of processes can be encoded as a binary string with $O(n)$ bits. From the algorithm, it can be seen that the latency is $latency(\mathcal{OLE}) + O(1)$ while its per-process bit complexity is $O(n^2) + pbit(\mathcal{OLE})$.


\begin{algorithm}
\caption{Synchronous Index ACS with Probability of Success \(\rho\): Pseudocode (for process $p_i$)}
\label{algorithm:synchronous_iacs}
\footnotesize
\begin{algorithmic} [1]

\State \textbf{Uses:}
\State \hskip2em Oblivious Leader Election with $\rho$-quality, \textbf{instance} $\mathcal{OLE}$ \BlueComment{see Module~\ref{mod:ole}}
\State \hskip2em Timer, \textbf{instance} $\mathit{timer}$ 

\medskip
\State \textbf{Input Parameters:}
\State \hskip2em $\mathsf{Set}(\mathsf{Process}) \mathit{V}_i$ \BlueComment{$\mathit{V}_i \subseteq \mathit{Valid}_i$}

\medskip
\State \textbf{Constants:}
\State \hskip2em $\mathsf{Integer}$ $x$ \BlueComment{related to the $x$-totality property of the process validation}

\medskip
\State \textbf{Global variables:}
\State \hskip2em $\mathsf{Set}(\mathsf{Process})$ $\mathit{Valid}_i$ \BlueComment{dynamic set of validated processes with $x$-Totality}

\medskip
\State \textbf{Round 1:}
\State \hskip2em \textbf{broadcast} $\langle \textsc{init}, \mathit{V}_i \rangle$ \label{line:iacs_broadcast_valid}
\State \hskip2em $\mathit{timer}.\mathsf{measure}(x \delta)$ 

\medskip
\State \textbf{Round 2:}

\State \hskip2em \textbf{invoke} $\mathcal{OLE}.\mathsf{call()}$ \label{line:iacs_invoke_ole}

\medskip
\State \textbf{upon} $\mathcal{OLE}.\mathsf{select}(\mathsf{Process}\;p_\ell)$: \label{line:iacs_ole_select}

\State \hskip2em \textbf{wait until} $\mathit{timer}.\mathsf{expire}()$ 

\State \hskip2em \textbf{if} $\langle \textsc{init}, V_\ell \rangle$ has been received from $p_\ell \neq \bot$ such that:
\State \hskip4em \text{(a)} $|V_\ell| \geq n-t$, \textbf{and}
\State \hskip4em \text{(b)} $V_\ell \subseteq \mathit{Valid}_i$ 
\State \hskip2em \textbf{then:}
\State \hskip4em \textbf{trigger} $\mathsf{decide}(V_\ell)$ 
\State \hskip2em \textbf{else:}

\State \hskip4em \textbf{trigger} $\mathsf{decide}(\mathit{\mathit{V}_i})$ 



\end{algorithmic}
\end{algorithm}

\subsection{Merging everything together}\label{section:acs_merging}


Now, we are ready to combine the aforementioned building blocks together to solve ACS. High-level wise, we use \Cref{algorithm:acs_to_index_acs}, where the index ACS is implemented using \Cref{algorithm:synchronous_iacs} that is compiled using \name. Specifically, we use \Cref{algorithm:synchronous_iacs} for $\mathcal{A}^S$ defined in \block. Notice that, index ACS is similar to Byzantine agreement in the sense that it is a Byzantine agreement where (1) the proposed and decision value is a set of processes, and (2) a value (set of processes) is considered valid by a process $p_i$ if it has validated it (by checking its $\mathit{Valid}_i$). Here, we only describe the modifications needed on \block to fit our index ACS usage.

\begin{itemize}
    \item In process $p_i$, each module has access to global variable $\mathit{Valid}_i$ that may grow over time.
    \item The $\mathsf{valid}(\cdot)$ predicate (defined by the external validity) employed by $p_i$ is as follow: $\mathsf{valid}(V) = |V| \ge n-t \wedge V \subseteq \mathit{Valid}_i$.
\end{itemize}

Notice that this means for the same set of processes $V$, two correct processes $p_i$ and $p_j$ may have different perceptions regarding the validity of $V$. However, thanks to the $x$-totality of $\mathit{Valid}_i$ and $\mathit{Valid}_j$ (see Module~\ref{mod:index-acs}), if $p_i$ considers $V$ to be valid at time $\tau$, then by $\max(\tau, \text{GST}) + x\delta$, $p_j$ will also consider $V$ to be valid.\footnote{A subtle point here is that each process now needs a message buffer. We cannot just drop a message with an invalid value, as a value may become valid in the future. However, this will not affect the latency or the bit complexity of the resulting protocol.} This is sufficient to allow \block to compile the index ACS and thus, allow us to solve ACS in partial synchrony.

Let us now explain the latency and the bit complexity of the resulting transformation $\name_{\mathsf{ACS}}$. First, notice that in ACS, a decision consists of $O(n)$ proposals. Thus, to distinguish their size, we denote $L_{in}$ as the size of the proposal and $L \in O(nL_{in})$ as the size of the decision.

Then, notice that in the compiled index ACS, a value is $L_{\mathsf{IACS}} \in O(n)$ bits, using the encoding we mentioned before to encode a set of processes. In our concrete implementation of the final transformation, $\name_{\mathsf{ACS}}$, we use the following implementation for these primitives, that were used in \name and \Cref{algorithm:acs_to_index_acs}:
\begin{itemize}
    \item Graded consensus (Module~\ref{mod:graded-consensus}) from \cite{AttiyaWelch23} with $O(n L_{\mathsf{IACS}}) = O(n^2)$ per-process bit complexity.
    \item Validation broadcast (Module~\ref{mod:validation-broadcast}) from \Cref{subsection:validation_broadcast_basic_implementation} with $O(n L_{\mathsf{IACS}}) = O(n^2)$ per-process bit complexity.
    \item Finisher (\Cref{subsection:finisher}) from \Cref{subsubsection:shortfin} with $O(n L_{\mathsf{IACS}}) = O(n^2)$ per-process bit complexity.
    \item Reliable broadcast (Module~\ref{mod:reliable-broadcast}) from \cite{alhaddad2022balanced,chen2024ociorcool} with $O(L_{in} + n\log(n))$ per-process bit-complexity.
\end{itemize}

Importantly, those implementations have constant latency. Hence, $\name_{\mathsf{ACS}}$ solves ACS with expected per-process bit-complexity of $O(nL_{in}+n^2) + \mathcal{B}$ and $O(\mathcal{R})$ expected latency, where $\mathcal{B}$ and $\mathcal{R}$ denote the per-process bit-complexity and the latency of the underlying (randomized) synchronous OLE protocol, respectively.

We now see the power of $\name_{\mathsf{ACS}}$ for solving ACS in partial synchrony. By using the recent OLE protocol from \cite{AsharovChandramouli24} and assuming private channels, we achieve an optimally resilient ACS against a strongly-adaptive unbounded adversary. Moreover, the protocol has expected per-process bit-complexity $O(nL_{in}+n^2\log^2(n))$ and expected constant latency. Another example is by plugging in OLE from \cite{GPV06}, on which the resulting ACS protocol is resilient against $t < n/(3+\epsilon)$ faults under a static unbounded adversary. Furthermore, the protocol has expected per-process bit-complexity of $O(nL_{in}+n^2)$ and $O(\log(n)/\epsilon^2)$ expected latency.
For a more comprehensive comparison, we summarized the results along with the existing state of the arts in \Cref{table:state_of_the_art_protocols_acs}.

\begin{table}[h!]
    \footnotesize
    \centering
    \begin{tabular}{|c|c|c|c|c|c|c|c|}
        \hline
        \textbf{Protocol} & \textbf{Relaxation} & \textbf{Security} & \textbf{Resiliency} & \textbf{Latency} & 
        $\boldsymbol{\textbf{L}}_{in}$ \textbf{for optim.} & \textbf{Remark} & \textbf{Sync} \\
                \hline \hline
            KS \cite{king2016byzantine,king2018correction} & None & Perfect & $n/(0.87\cdot 10^{9}) $ & $\tilde{O}(n^3)$ & $\Omega(n^4)$ & - & A \\
        \hline
        HPZ \cite{huang2023byzantine} & None & Perfect & $n/(3 + \epsilon)$ & $\tilde{O}(n^4/\epsilon^{8})$ & $\tilde{\Omega}(n^5/\epsilon^{8})$ & - & A \\
        \hline
        HPZ \cite{huang2023byzantine} & None & Perfect & $n/3$ & $\tilde{O}(n^{12})$ &  $\tilde{\Omega}(n^{13})$ & - & A \\
        \hline
        $\boldsymbol{\textbf{\name}}($CDGGKVZ \cite{errorFreeValidated}$)$ & None & Perfect & $n/3$ & $O(n)$ & $\Omega(n)$ & - & PS \\

\hline \hline
        Bar-Joseph \& Ben-Or \cite{Bar-JosephB98} & Crash & - & - & $\Omega(\sqrt{n/\log(n)})$ & - & LB & S \\
        \hline
        AAKS \cite{AAKS14,AAKS18} & Crash Only & Perfect & $n/2$ & $\tilde{\Omega}(n)$ & $\tilde{\Omega}(n)$ & - & A\\
        \hline
        $\boldsymbol{\textbf{\name}}($HKO \cite{hajiaghayi2024nearly}$)$ & Omission & Perfect & $n/30$ & $O(\sqrt{n}\log^3(n))$ & $\Omega(n\log^3(n))$ & - & PS \\
        
        \hline \hline
        KKKSS \cite{KKKSS08,KKKSS10} & Static Adv. & $1/n^c$ & $n/(3 + \epsilon)$ & $2^{\Theta(\log^8(n))}$ & $\tilde{\Omega}(n)$ & - & A \\
        \hline
        $\boldsymbol{\textbf{\name}_{\mathsf{acs}}}($GPV \cite{GPV06}$)$ & Static Adv. & Perfect & $n/(3 + \epsilon)$ & $O(\log(n)/\epsilon^2)$ & $\Omega(n)$ & - & PS \\
        \hline
        $\boldsymbol{\textbf{\name}_{\mathsf{acs}}}($GPV \cite{GPV06}$)$ & Static Adv. & Perfect & $n/\log^{1.58}(n)$ & $O(1)$ & $\Omega(n)$ & - & PS \\
        
        \hline \hline
        BCP \cite{BCP18} & Priv. Chan. & Perfect & $n/3$ & $O(n \log(n))$ & $\Omega(n^5)$ & - & A \\
        \hline
        BCP \cite{BCP18} & Priv. Chan. & Perfect & $n/(3+\epsilon)$ & $O(\log(n)/\epsilon)$ & $\Omega(n^5)$ & - & A \\
        \hline
        AAPS \cite{AAPS23} & Priv. Chan. & Perfect & $n/4$ & $O(1)$ & $\tilde{\Omega}(n^2)$ & - & A \\
        \hline
        AAPS \cite{AAPS23} + \cite{ChoudhuryPatra23}'s AISS & Priv. Chan. & $\mathit{neg}(\lambda)$ & $n/3$ & $O(1)$ & $\Omega(n^4 + \lambda^2 n^3)$ & - & A \\
        \hline
        $\boldsymbol{\textbf{\name}_{\mathsf{acs}}}($AC \cite{AsharovChandramouli24}$)$ & Priv. Chan. & Perfect & $n/3$ & $O(1)$ & $\Omega(n \log^2(n))$ & - & PS \\ 
        \hline
    \end{tabular}
    \caption{
Overview of state-of-the-art fully asynchronous and partially synchronous ACS protocols in different models with unbounded adversary. 
        In the Remark column, LB stands for lower bound. 
        A value different from 'Perfect' in the Security column corresponds to the probability of violating termination. 
        A (Asynchrony), S (Synchrony), and PS (Partial Synchrony) in the Sync column indicate the synchrony assumption. 
        $\epsilon$ is always a non-zero constant. 
        The value in column '$L_{in}$ for optim.' corresponds to the size  $L_{in}$ of a proposal such that the protocol achieves optimal $O(n^2 L_{in})$ total communication. 
        For protocols yielded by $\name$ and $\name_{\mathsf{acs}}$, it is also enough to achieve optimal per-process communication.
        An ACS protocol yielded by $\name$ (and not $\name_{\mathsf{acs}}$) has been obtained by classic reduction to binary agreement\cite{ben1994asynchronous,canetti1993fast,Chor1989} (instead of OLE).
    }
    \label{table:state_of_the_art_protocols_acs}
\end{table}

\clearpage
\bibliographystyle{acm}
\bibliography{references}

\end{document}